\documentclass[11pt]{article}

\usepackage{setspace,graphicx,epstopdf,amsmath,amsfonts,amssymb,amsthm,nccmath,versionPO}
\usepackage{marginnote,datetime,enumitem,subfigure,rotating,fancyvrb}
\usepackage{float}

\usepackage{titling}
\usepackage{xpatch}
\makeatletter
{\begin{titlepage}\def\@thanks{}}%
	{\end{titlepage}}
\xpatchcmd\titlepage{\setcounter{page}\@ne}{}{}{}
\xpatchcmd\endtitlepage{\setcounter{page}\@ne}{}{}{}
\makeatother

\usepackage[stable,bottom]{footmisc}
\usepackage[dvipsnames]{xcolor}
\usepackage{hyperref}
\usepackage[capitalize]{cleveref}
\hypersetup{
	colorlinks  = true,
	urlcolor    = blue,
	linkcolor   = red,
	citecolor   = blue
}

\usepackage{microtype}

\usepackage[longnamesfirst]{natbib}
\usdate

\usepackage{titlesec}
\newcommand{\periodafter}[1]{#1.}
\titleformat{\subsubsection}[runin]
{\normalfont\bfseries}{\thesubsubsection}{1em}{\periodafter}

\excludeversion{notes}		% Include notes?
\includeversion{links}      % Turn hyperlinks on?

\iflinks{}{\hypersetup{draft=true}}

\ifnotes{%
	\usepackage[margin=1in,paperwidth=10in,right=2.5in]{geometry}%
	\usepackage[textwidth=1.4in,shadow,colorinlistoftodos]{todonotes}%
}{%
	\usepackage[margin=1in]{geometry}%
	\usepackage[disable]{todonotes}%
}

\makeatletter\let\chapter\@undefined\makeatother % Undefine \chapter for todonotes

\theoremstyle{definition}
\newtheorem{theorem}{Theorem}%[section]
\newtheorem{lemma}{Lemma}%[section]
\newtheorem{assumption}{Assumption}%[section]
\newtheorem{remark}{Remark}%[section]
\newcommand{\argmin}{\mathop{\rm argmin}}

\usepackage{booktabs}
\usepackage{bbm}
\usepackage{dsfont}
\usepackage{multirow}
\usepackage{booktabs,caption}
\usepackage[flushleft]{threeparttable}
\usepackage[title]{appendix}
\usepackage{lscape}
\usepackage{longtable,rotating}
\usepackage{rotating}
\usepackage{bm,bbm}
\usepackage{mathtools}
\usepackage[font={scriptsize,onehalfspacing},labelfont=bf]{caption}
\usepackage[title]{appendix}
\usepackage{siunitx}
\usepackage{scrextend}
\usepackage{longtable}
\usepackage{slashbox}
\mathtoolsset{centercolon}
\usepackage{framed}
\usepackage{tabularx}
\newcolumntype{Y}{>{\raggedleft\arraybackslash}X}
\newcolumntype{Z}{>{\raggedright\arraybackslash}X}

\begin{document}

\setlist{noitemsep}  % Reduce space between list items (itemize, enumerate, etc.)
\onehalfspacing

\newcommand{\tit}{Realized Regularized Regressions}

\title{{\tit}\thanks{We thank Kim Christensen, Seok Young Hong, Yifan Li, Yingying Li, Olga Kolokolova, Ingmar Nolte, Manh Pham, and participants at 56th New Economic School Research Conference and 5th Frontiers of Factor Investing Conference, for their comments and suggestions. The usual disclaimer applies. }
\\[0.1cm]}

\author{
{Aleksey Kolokolov}\thanks{New Economic School, Russia. Email: \href{mailto:akolokolov@nes.ru}{akolokolov@nes.ru}. }\and
{Shifan Yu}\thanks{Corresponding author, Oxford-Man Institute of Quantitative Finance, University of Oxford, United Kingdom. Email: \href{mailto:shifan.yu@omi.ox.ac.uk}{shifan.yu@omi.ox.ac.uk}. }\\[0.4cm]}

\date{This Version: \today}
%\date{This Version: \today\\[0.5cm]
%\textsf{\textcolor{blue}{\textbf{For conference review only. Please do not circulate or distribute.}}}}

\maketitle

\begin{abstract}
\noindent %We study the asymptotic properties of spline-based penalized estimators for high-frequency regressions, in which both the response and covariates are Itô semimartingales with jumps. 
\noindent %Financial datasets are expanding rapidly, with an increasing number of variables observed at high frequencies. This raises challenges for estimation and variable selection in high-frequency regressions within a high-dimensional environment. 
We develop a continuous-time penalized regression framework for the estimation of time-varying coefficients and variable selection when both the response and covariates are Itô semimartingales with jumps. The coefficient paths are approximated by spline basis expansions and estimated via least squares from truncated high-frequency increments. In a finite-dimensional setting, we establish consistency and derive a feasible asymptotic distribution for the integrated coefficient estimator under infill asymptotics. We then extend the framework to high-dimensional settings in which the number of candidate covariates diverges, and show that a group-wise penalized estimator with a truncated $\ell_1$-penalty attains the oracle property, which delivers both consistent model selection and coefficient estimation. An empirical application to a large panel of more than two hundred high-frequency factors documents sparse factor structure across a large cross-section of stocks and industry portfolios.
%In an empirical application, we perform factor selection within the asset-pricing “factor zoo” of more than two hundred high-frequency factors, and document substantial sparsity in factor exposures across a large cross-section of stocks and industry portfolios.
\\
\vspace{0in}\\
\noindent\textbf{JEL Classifications:} C14, C22, C58, G12 \medskip\\
\noindent\textbf{Keywords:} high-frequency regressions, time-varying coefficients, penalized least squares, high-dimensional semimartingales, model selection, systematic risk
\end{abstract}

\clearpage

\section{Introduction}
\label{Sec:introduction}

In recent decades, the empirical asset pricing literature has discovered hundreds of risk factors intended to explain variation in expected returns, i.e., the so-called ``factor zoo'' \citep{cochrane2011presidential,harvey2016and}. The increased availability of high-frequency intraday data for both individual assets and aggregate factors (see, e.g., \citealp{aleti2023high}) has further enhanced opportunities for more granular analysis of time variation in factor loadings, or betas, by allowing much shorter estimation windows \citep{ait2020high,ait2025continuous}.\footnote{Some recent studies apply high-frequency factors to return prediction \citep{aleti2025intraday}, stochastic discount factor (SDF) estimation \citep{aleti2025news}, and volatility forecasting \citep{cinquetti2025volatility}.} Such advances have substantially increased the complexity of the model selection problem from both theoretical and practical perspectives. In this paper, we develop a rigorous methodology for model selection and factor identification in high-frequency regressions with many candidate covariates, while simultaneously retaining valid estimation and inference for the time-varying loadings of the relevant factors.

The standard approach to model selection in static regression models typically relies on penalized estimation, where a regularization term, such as the least absolute shrinkage and selection operator (LASSO; \citealp{tibshirani1996regression}), is added to the objective function. However, in regression models with time-varying coefficients, model selection is considerably less straightforward. A simple strategy is to apply penalized regression separately within each local estimation window, but this can lead to inconsistent outcomes in which the set of selected regressors changes across (possibly overlapping) estimation windows. This issue is especially acute for high-frequency regressions, where the window length usually shrinks to zero under infill asymptotics.

To address this issue, we propose a new approach to continuous-time nonparametric regression that is better suited to ``global'' model selection over a fixed time interval. Our approach approximates each coefficient path by polynomial splines, and therefore reduces the time-varying beta estimation to the estimation of scalar coefficients in their spline basis expansions. With time variation absorbed by the spline basis functions, it allows us to fit the model once over the entire observation interval, rather than repeatedly re-estimating local regressions. 
%While spline-based methods are well established in the statistics literature (see, e.g., \citealp{hastie1993varying,hoover1998nonparametric,huang2002varying,huang2004polynomial}), to the best of our knowledge this is the first application of this reparameterization approach to high-frequency regressions in financial econometrics. 
%Importantly, the 
The spline basis expansion also induces a natural group structure: each regressor corresponds to a group of spline coefficients. This enables group-wise model selection to identify the relevant regressors (or equivalently, the groups with nonzero expansion coefficients).
% with a nonconvex truncated $\ell_{1}$-penalty (TLP; \citealp{shen2012likelihood}) to identify the relevant regressors (equivalently, the nonzero groups of scalar spline coefficients) over the entire interval. 

When the number of regressors is fixed, we establish consistency of the time-varying beta estimator, and derive a feasible asymptotic distribution for the integrated beta estimator under infill asymptotics. In high-dimensional settings where the number of irrelevant regressors is allowed to diverge, we show that the penalized spline-based estimator with a truncated $\ell_{1}$-penalty (TLP; \citealp{shen2012likelihood}) attains the ``oracle property,'' i.e., it selects the correct set of relevant regressors with probability approaching one and estimates their coefficients as if the true model were known a priori \citep{fan2001variable}.

We investigate the finite-sample performance of the proposed approach through extensive Monte Carlo experiments based on simulated continuous-time factor models with time-varying betas and jumps. In low-dimensional cases, our spline-based estimator performs as well as the benchmark estimator of \citet{ait2020high}. As the cross-section of candidate factors grows, our estimator remains robust. Our penalized estimation procedure delivers accurate model selection with limited false discoveries, while preserving stable finite-sample estimation performance for the relevant factor exposures.

In the empirical analysis, we illustrate our approach by applying it to the high-frequency ``factor zoo'' of \citet{aleti2023high}. We find persistently sparse factor structures in every month across a cross-section of individual stocks and industry portfolios. The market factor emerges as the most stable driver of intraday comovement. A case study of risk decomposition further suggests that a factor set constructed from the most frequently selected factors in our sample achieves greater explanatory power for integrated variance than the standard six-factor benchmark.

Our paper contributes to three strands of the literature. First and foremost, it contributes to the financial econometrics literature on nonparametric regression with high-frequency data. Early studies in this area develop the realized beta estimator as the ratio of realized covariance to realized variance \citep{barndorff2004econometric,andersen2005framework,andersen2006realized}. \citet{mykland2009inference} propose the integrated beta estimator by aggregating local coefficient estimates. \citet{ait2020high} further develop this approach with local ordinary least squares (OLS) and establish feasible limit theorems in the presence of jumps. More recent studies incorporate local regularization to accommodate high-dimensional covariates \citep{chen2024realized,shin2025robust,kim2026high}. Our contribution to this literature is twofold. First, we propose a spline-based estimator for time-varying betas and establish the associated asymptotic theory. Second, we introduce a method for simultaneous coefficient estimation and model selection in high dimensions which, to the best of our knowledge, is the first to address this issue formally in the context of high-frequency regressions. The closest related work is \citet{chen2026high}, who develop high-dimensional coefficient tests for the null hypothesis that an additional block of regressors provides no incremental explanatory power. Other notable contributions in this field include, among others: (i) separate inference on factor structure, betas, and risk premia for continuous and jump components \citep{todorov2010jumps,bollerslev2016roughing,li2017jump,li2019jump,ait2025continuous}, and (ii) inference on time variation in betas \citep{patton2012does,reiss2015nonparametric,kong2018testing,kalnina2023inference}, and their cross-sectional distribution \citep{andersen2021recalcitrant,andersen2023intraday}.

Second, we contribute to the rapidly growing literature on high-dimensional problems and statistical learning in finance. As modern financial markets feature a vast set of assets, characteristics, and candidate factors, the penalized regression methods have been widely applied to (i) covariance/correlation matrix estimation and forecasting \citep{brownlees2018realized,christensen2023high,bollerslev2025forecasting}, (ii) return prediction with large predictor sets \citep{chinco2019sparse,gu2020empirical,ait2025predictable,aleti2025intraday}, (iii) factor/characteristic selection and SDF estimation for the cross-section of expected returns \citep{feng2020taming,freyberger2020dissecting,kozak2020shrinking,chen2025cross}, and (iv) implementable portfolio formation and asset allocation \citep{ao2019approaching,ding2024statistical}, among others.

Third, our work contributes to the statistics literature on varying-coefficient models. Varying-coefficient models extend classical linear regression by allowing coefficients to vary over time or with other variables \citep{cleveland1991local,hastie1993varying,fan1999statistical}, which provide a flexible and interpretable way to capture dynamic covariate effects, and have been widely adopted in longitudinal studies \citep{hoover1998nonparametric,huang2002varying,huang2004polynomial}. Work on model selection in varying-coefficient models has closely tracked advances in penalized regression. For example, \citet{wang2008variable} and \citet{wang2009shrinkage} adopt the smoothly clipped absolute deviation (SCAD; \citealp{fan2001variable}) and the adaptive LASSO \citep{zou2006adaptive}, respectively; \citet{wei2011variable} combine spline-based estimation with the group LASSO \citep{yuan2006model} in high-dimensional settings; and \citet{xue2012variable} employ the nonconvex TLP to avoid reliance on a consistent initial estimator, which can be problematic for the SCAD and LASSO approaches, as the irrepresentable conditions for the initial LASSO estimator are often implausible in high-dimensional settings \citep{zhao2006model,bach2008consistency}. Our contribution to this literature is to establish an infill asymptotic theory for spline-based estimation and model selection in a nonstandard high-frequency setting. Moreover, our asymptotic framework allows for substantially weaker assumptions on the time-varying coefficients. Whereas the classical varying-coefficient literature typically requires smoothness conditions on the coefficient functions, we allow the coefficients to evolve as sample paths of It{\^o} semimartingales with infinite-activity jumps, and hence to be substantially rougher.

%... we allow coefficient paths with different degrees of roughness, and derive limit theorems under a mild requirement that the best spline approximation error becomes negligible as the spline dimension increases and the sampling grid becomes finer.

The rest of this paper is organized as follows: In \cref{Sec:model}, we present the continuous-time regression model, and introduce the required notation and assumptions. In \cref{Sec:estimation}, we propose the estimation procedure and establish its theoretical properties. Monte Carlo experiments and empirical applications are reported in \cref{Sec:simulation,Sec:empirical}, respectively. \cref{Sec:conclusions} concludes. Proofs and supplementary results can be found in the \hyperref[AP:Proofs]{Appendix}.

%\newpage

\section{Continuous-Time Regression Model}
\label{Sec:model}

We follow the framework of \citet{ait2020high} to specify a continuous-time multiple regression model and state the assumptions. Consider the following nonparametric regression model,
\begin{equation}\label{Eq:regression}
Y_{t}=Y_{0}+\int_{0}^{t}\beta_{s-}^{\top}dX_{s}^{c}+\sum_{0\leq s\leq t}(\beta_{s-}^{J})^{\top}\Delta X_{s}+Z_{t},
\end{equation}
where $Y=(Y_{t})_{t\geq0}$ is the response process, $X=(X_{1,t},\dots,X_{p,t})^{\top}_{t\geq0}$ is a $p$-dimensional covariate process, and $Z=(Z_{t})_{t\geq0}$ is the residual process. We denote by $X^{c}$ the continuous component of $X$, and by $\Delta X_{t}=X_{t}-X_{t-}$ the jump of $X$ at time $t$. $\beta=(\beta_{1,t},\dots,\beta_{p,t})_{t\geq0}^{\top}$ and $\beta^{J}=(\beta_{1,t}^{J},\dots,\beta_{p,t}^{J})_{t\geq0}^{\top}$ collect the time-varying coefficients with respect to the continuous and discontinuous parts of $X$, respectively.

We assume that both $X$ and $Z$ are It{\^o} semimartingales defined on a filtered probability space $(\Omega,\mathcal{F},(\mathcal{F}_{t})_{t\geq0},\mathbb{P})$:
%\begin{equation}
%\label{Eq:X_Z}
%X_{t}=X_{0}+\int_{0}^{t}b_{s}ds+\int_{0}^{t}\sigma_{s}dW_{s}+\delta\star\mu_{t},\qquad
%Z_{t}=Z_{0}+\int_{0}^{t}\widetilde{b}_{s}ds+\int_{0}^{t}\widetilde{\sigma}_{s}d\widetilde{W}_{s}+\widetilde{\delta}\star\widetilde{\mu}_{t},
%\end{equation}
\begin{align}
X_{t}&=X_{0}+\int_{0}^{t}b_{s}ds+\int_{0}^{t}\sigma_{s}dW_{s}+\int_{0}^{t}\int_{\mathbb{R}^{p}}\delta(s,x)\mu(ds,dx),\label{Eq:X}\\
Z_{t}&=Z_{0}+\int_{0}^{t}\widetilde{b}_{s}ds+\int_{0}^{t}\widetilde{\sigma}_{s}d\widetilde{W}_{s}+\int_{0}^{t}\int_{\mathbb{R}}\widetilde{\delta}(s,x)\widetilde{\mu}(ds,dx),\label{Eq:Z}
\end{align}
where $X_{0}\in\mathbb{R}^{p}$ and $Z_{0}\in\mathbb{R}$ are $\mathcal{F}_{0}$-measurable, $W$ and $\widetilde{W}$ denote a $p'$-dimensional and a one-dimensional Brownian motion, respectively, the spot volatility $\sigma$ takes values on $\mathbb{R}^{p}\otimes\mathbb{R}^{p'}$, $\mu$ (resp.~$\widetilde{\mu}$) is a Poisson random measure with a compensator $\nu$ (resp.~$\widetilde{\nu}$) of the form $\nu(dt,dx)=dt\otimes\lambda(dx)$ for some $\sigma$-finite measure $\lambda$ on $\mathbb{R}^{p}$ (resp.~$\widetilde{\lambda}$ on $\mathbb{R}$). The spot covariance of $X$, denoted as $c=\sigma\sigma^{\top}$, follows another It{\^o} semimartingale defined on $(\Omega,\mathcal{F},(\mathcal{F}_{t})_{t\geq0},\mathbb{P})$:
%\begin{equation}
%\label{Eq:c}
%c_{t}=c_{0}+\int_{0}^{t}b'_{s}ds+\int_{0}^{t}\sigma'_{s}dW_{s}+(\delta'\mathbbm{1}_{\{\Vert\delta'\Vert\leq 1\}})\star(\mu-\nu)_{t}+(\delta'\mathbbm{1}_{\{\Vert\delta'\Vert>1\}})\star\mu_{t}, 
%\end{equation}
\begin{equation}
\label{Eq:c}
c_{t}=c_{0}+\int_{0}^{t}b'_{s}ds+\int_{0}^{t}\sigma'_{s}dW_{s}+\int_{0}^{t}\int_{\mathbb{R}^{p\times p}}\delta'(s,x)\mu(ds,dx),
\end{equation}
where $c_{0}\in\mathbb{R}^{p\times p}$ is $\mathcal{F}_{0}$-measurable, and $\sigma'$ takes values on $\mathbb{R}^{p\times p}\otimes\mathbb{R}^{p'}$. We use the same Brownian motion $W$ and Poisson random measure $\mu$ for both $X$ and $c$ without loss of generality, and this specification accommodates leverage effects as well as co-jumps in price and volatility.

We assume that the processes defined above satisfy the following regularity conditions:
\begin{assumption}
\label{As:Ito}
The following properties hold for the processes defined in \cref{Eq:regression,Eq:X,Eq:Z,Eq:c}:
\begin{itemize}
\item[(i)] $\beta$, $\beta^{J}$, $\sigma$, $\widetilde{\sigma}$, $\sigma'$ are adapted c{\`a}dl{\`a}g and locally bounded;
\item[(ii)] $b$, $\widetilde{b}$ and $b'$ are progressively measurable and locally bounded;
\item[(iii)] $c$ is positive definite, and both $c$ and $c^{-1}$ are locally bounded;
\item[(iv)] $\delta$, $\widetilde{\delta}$ and $\delta'$ are predictable functions;
%\item[(v)] There is a sequence $(\tau_{m})$ of stopping times increasing to $\infty$, and some sequences of deterministic nonnegative $\lambda$-integrable functions $f_{m}$ and $f'_{m}$ on $\mathbb{R}$ for each $m$, such that $\Vert\delta(\omega,t,x)\Vert^{r}\wedge1\leq f_{m}(x)$, for some $r\in[0,1)$, and $\Vert\delta'(\omega,t,x)\Vert^{2}\wedge1\leq f'_{m}(x)$ for all $(\omega,t,x)$ with $t\leq\tau_{m}(\omega)$;
\item[(v)] There exists a sequence $(\tau_{m})_{m\geq1}$ of stopping times increasing to $\infty$, and sequences $(f_{m})_{m\geq1}$ and $(f'_{m})_{m\geq1}$ of deterministic nonnegative $\lambda$-integrable functions, such that for some $r\in[0,1)$, $\Vert\delta(\omega,t,x)\Vert^{r}\wedge1\leq f_{m}(x)$ and $\Vert\delta'(\omega,t,x)\Vert^{r}\wedge1\leq f'_{m}(x)$ for all $(\omega,t,x)$ with $t\leq\tau_{m}(\omega)$;
\item[(vi)] There exists a sequence $(\widetilde{\tau}_{m})_{m\geq1}$ of stopping times increasing to $\infty$, and a sequence $(\widetilde{f}_{m})_{m\geq1}$ of deterministic nonnegative $\widetilde{\lambda}$-integrable functions, such that $\vert\widetilde{\delta}(\omega,t,x)\vert^{r}\wedge1\leq\widetilde{f}_{m}$ for all $(\omega,t,x)$ with $t\leq\widetilde{\tau}_{m}(\omega)$, where $r\in[0,1)$ is, without loss of generality, the same as in (v).
\end{itemize}
\end{assumption}

The above conditions are standard in the high-frequency literature. Conditions (i) and (ii) ensure the stochastic integrals are well defined and permit the standard localization arguments. Condition (iii) implies that $c$ is uniformly nondegenerate and locally bounded on compact intervals, which guarantees invertibility and rules out pathological cases of vanishing or explosive volatility. Conditions (v) and (vi) allow for both finite- and infinite-activity jumps, while restricting them to be of finite variation.
%The use of the same Brownian motion $W$ and Poisson random measure $\mu$ for both $X$ and $c$ is without loss of generality \citep{ait2020high}, and this specification accommodates leverage effects as well as co-jumps in price and volatility. The parameter $r$ sets a bound on the degree of jump activity. With some $r\in[0,1)$, we allow for jumps of either finite or infinite activity, but restrict them to be of finite variation, i.e., they are absolutely summable. The assumed nonsingular $c$ ensures that $\sigma$ is also an It{\^o} semimartingale; see Assumption (K'-$r$) in \citet{jacod2012discretization}. %(Assumption 4.4.4)

%The use of the same Brownian motion $W$ and Poisson random measure $\mu$ for both $X$ and $c$ is without loss of generality \citep{ait2020high}. The parameter $r$ sets a bound on the degree of jump activity. For both $X$ and $Z$, with some $r\in[0,1)$, we allow for jumps of either finite or infinite activity, but restrict them to be of finite variation, i.e., they are absolutely summable. We allow the jumps in $c$ to be of infinite variation, and therefore retain the full Grigelionis decomposition in \cref{Eq:c}; see \citet{ait2014high} for details. The assumed nonsingular $c$ ensures that $\sigma$ is also an It{\^o} semimartingale; see Assumption (K'-$r$) in \citet{jacod2012discretization}. %(Assumption 4.4.4)

To account for the time-varying coefficient, we further assume that $\beta$ follows an It{\^o} semimartingale with finite-variation jumps:
\begin{equation}
\label{Eq:beta}
\beta_{t}=\beta_{0}+\int_{0}^{t}b_{s}^{\beta}ds+\int_{0}^{t}\sigma_{s}^{\beta}dW_{s}^{\beta}+\int_{0}^{t}\int_{\mathbb{R}^{p}}\delta^{\beta}(s,x)\mu^{\beta}(ds,dx).
\end{equation}
Similar to \cref{As:Ito}, we impose the following regularity conditions:
\begin{assumption}
\label{As:Ito_beta}
$b^{\beta}$ is progressively measurable and locally bounded, $\sigma^{\beta}$ is adapted c{\`a}dl{\`a}g and locally bounded, and $\delta^{\beta}$ is a predictable function. There exists a sequence $(\tau_{m})_{m\geq1}$ of stopping times increasing to $\infty$, and a sequence $(f_{m}^{\beta})_{m\geq1}$ of deterministic nonnegative $\lambda^{\beta}$-integrable functions, such that for some $r\in[0,1)$, $\Vert\delta^{\beta}(\omega,t,x)\Vert^{r}\wedge1\leq f_{m}^{\beta}(x)$ for all $(\omega,t,x)$ with $t\leq\tau_{m}(\omega)$.
\end{assumption}

Analogous to the classical regression framework, we require an exogeneity assumption over the fixed time interval $[0,T]$. This assumption is specified separately for the continuous and jump components: the former follows \citet{barndorff2004econometric} and \citet{mykland2006anova,mykland2009inference}, whereas the latter is in line with \citet{li2017jump}.

%We impose this assumption separately on the continuous and jump components, which parallels \citet{barndorff2004econometric} and \citet{mykland2006anova,mykland2009inference} for the continuous part, and \citet{li2017jump} on jumps:
\begin{assumption}
\label{As:orthogonality}
For any $t\in[0,T]$, the following orthogonality conditions hold element-wise:
\begin{equation}
\langle X^{c}, Z^{c}\rangle_{t}=0,\qquad
\sum_{0\leq s\leq t}\Delta X_{s}\Delta Z_{s}=0,
\end{equation}
%\begin{equation}
%\langle X_{t}^{c}, Z_{t}^{c}\rangle=\bm{0},\qquad
%\langle \sum_{0\leq s\leq t}\Delta X_{s}, \sum_{0\leq s\leq t}\Delta Z_{s}\rangle=\bm{0},
%\end{equation}
where $\langle\cdot,\cdot\rangle_{t}$ denotes the quadratic covariation over $[0,t]$, and 0 is a $p$-dimensional zero vector.
\end{assumption}

Both $Y$ and $X$ are observed at discrete times $i\Delta_{n}$ for $0\leq i\leq n=\lfloor T/\Delta_{n}\rfloor$. The increments of a generic $p$-dimensional process $A$ are denoted by
\begin{equation}
\Delta_{i}^{n}A=(\Delta_{1,i}^{n}A,\dots,\Delta_{p,i}^{n}A)^{\top}=A_{i\Delta_{n}}-A_{(i-1)\Delta_{n}}.
\end{equation}
For the asymptotic theory below we work in an infill framework, i.e., $n\to\infty$ as $\Delta_{n}\to0$, while $T$ is fixed. Based on the discrete observations of both $Y$ and $X$, our goal is to estimate the time-varying coefficients $\beta$ for the continuous component of $X$. We define the integrated beta ($I\beta$) over $[0,T]$ as
\begin{equation}
I\beta_{T} = \int_{0}^{T} \beta_{s}ds.
\end{equation}

When the dimension $p$ is fixed, $I\beta_{T}$ can be estimated by aggregating local estimates of $\beta$ over $[0,T]$ \citep{mykland2009inference,ait2020high}. In the fixed-$p$ setting, the local OLS estimator of \citet{ait2020high} is semiparametrically efficient \citep{jacod2013quarticity,renault2017efficient, li2021efficient}. However, our main focus is the case where $p \to \infty$. In particular, we assume that the dimension of $X$ increases but the corresponding $\beta$ remains sparse, i.e., a large number of covariates are redundant. In such a high-dimensional setting, the local OLS estimation becomes ill-posed: the local design matrix is nearly singular or not of full column rank, so the local OLS estimator is either not uniquely defined or numerically unstable. For this reason, we seek an alternative estimation procedure whose performance is comparable to that of \citet{ait2020high} when $p$ is fixed, but which also remains valid as $p$ increases and enables automatic model selection.

One could accommodate the large $p$-case by introducing regularization into each local window. For example, \citet{chen2024realized} employ the thresholding technique of \citet{bickel2008covariance} to estimate large spot covariance matrices and to stabilize local regressions in the presence of many regressors; see also, e.g., \citet{shin2025robust} and \citet{kim2026high}. Similar ideas could be used to perform model selection in a purely local sense, but such an approach treats model selection as intrinsically timewise: the active set of regressors is allowed to fluctuate across windows, which obscures the economic interpretation of the resulting factor structure.

%\citet{ait2020high} run local OLS regressions over moving windows for the continuous components of \cref{Eq:regression}, and then aggregate the spot coefficient estimates to obtain estimates for the integrated beta. More precisely, on a shrinking time window around a time point $t$, the coefficient process $\beta_{t}$ can be treated as approximately constant, such that the continuous-time regression locally behaves like a standard linear regression of the high-frequency increments of $Y^{c}$ on those of $X^{c}$. Within the general inference framework of \citet{jacod2013quarticity}, the estimator of \citet{ait2020high} is a local maximum likelihood estimator which is semiparametrically efficient \citep{renault2017efficient,li2021efficient}. 

%With a finite dimension $p$ of $X$, the estimation method of \citet{ait2020high} enjoys strong theoretical support and demonstrates reliable finite-sample performance. In recent development with a rich set of observable high-frequency factors, however, it is common that the number of regressors $p$ is of the same order as, or even exceeds, the number of observations in each local window. In such regimes the local OLS estimation becomes ill-posed: the local design matrix is nearly singular or not of full column rank, such that the local OLS estimator is either not uniquely defined or numerically unstable. Moreover, this procedure does not perform any model selection and cannot distinguish relevant factors from spurious ones when $p$ is large.

These considerations motivate our work. The question we consider in this paper is whether one can simultaneously (i) retain the local estimation of $\beta$ (and the estimation of $I\beta$), and (ii) perform global model selection over the entire fixed time interval $[0,T]$.

\section{Estimation Method and Theoretical Results}
\label{Sec:estimation}

In this section, we present the proposed estimation method and develop its theoretical properties. \cref{Sec:estimator} outlines the proposed estimator, with technical details deferred to the subsequent subsections. \cref{Sec:spline} derives the asymptotic rate of the spline approximation error. \cref{Sec:limitTheory} establishes consistency and derives a feasible asymptotic distribution for our estimators. \cref{Sec:selection} extends the framework to high-dimensional settings with group-wise regularization and establishes model selection consistency. \cref{Sec:implementation} discusses the nonconvex optimization algorithm and the practical choice of tuning parameters.

\subsection{Estimation Method}
\label{Sec:estimator}
Our estimation method is based on approximating the coefficient process $\beta=(\beta_{1},\dots,\beta_{p})^{\top}$ on $[0,T]$ by a linear combination of $K_n$ deterministic polynomial spline basis functions with random weights, where $K_n \to \infty$ as $\Delta_n \to 0$:
%Here we outline the proposed estimator, with technical details deferred to the next section. Suppose that the coefficient process $\beta=(\beta_{1},\dots,\beta_{p})^{\top}$ over $[0,T]$ can be approximated by a linear combination of $K_n$ deterministic polynomial spline functions, where $K_n\to\infty$ as $\Delta_n\to 0$:
\begin{equation}\label{Eq:spline_approx}
%\beta_{t} \approx \sum_{k=1}^{K_{n}}(\gamma^{(k)})^{\top}B_{t}^{(k)},
\beta_{t} \approx \sum_{k=1}^{K_{n}}B_{t}^{(k)}\gamma^{(k)},
\end{equation}
where $(B_{t}^{(1)},\dots, B_{t}^{(K_{n})})$ collects the spline basis functions of $t$, and $\gamma^{(k)}=(\gamma_{1}^{(k)},\dots,\gamma_{p}^{(k)})^{\top}$ are scalar coefficients in the expansion, which are $\mathcal{F}_{T}$-measurable random variables. Pathwise, for each fixed $\omega\in\Omega$, the expansion coefficients are uniquely determined by the realized path $t\mapsto\beta_t(\omega)$ on $[0,T]$. Specifically, we employ a B-spline basis (see %, e.g., Chapter 5 in \citet{hastie2009elements} and 
Appendix \ref{AP:Bspline} for details). This leads to the following approximation for the regression model in \cref{Eq:regression}:
\begin{equation}\label{Eq:approx_reg}
Y_{t} \approx Y_{0} + \sum^{K_n}_{k=1} \int_{0}^{t} B_s^{(k)}\gamma^{(k)}dX_{s}^{c} + \sum_{0\leq s\leq t}(\beta_{s-}^{J})^{\top}\Delta X_{s} + Z_{t}.
\end{equation}
We collect the spline coefficients in the vector $\bm{\gamma}=(\gamma_{1}^{(1)},\ldots,\gamma_{1}^{(K_n)},\ldots,\gamma_{p}^{(1)},\ldots,\gamma_{p}^{(K_n)})^{\top}\in\mathbb{R}^{pK_n}$. In the absence of jumps on $[0,T]$, we can discretize \cref{Eq:approx_reg} in terms of high-frequency increments as
\begin{equation}
\label{Eq:regression_continuous}
\Delta_{i}^{n}Y \approx \sum_{k=1}^{K_{n}}(\gamma^{(k)})^{\top}B_{(i-1)\Delta_{n}}^{(k)}\Delta_{i}^{n}X + \Delta_{i}^{n}Z, %= \sum_{k=1}^{K_{n}} (\gamma^{(k)})^{\top}B_{(i-1)\Delta_{n}}^{(k)}\begin{pmatrix}
%\Delta_{1,i}^{n}X\\
%%\Delta_{2,i}^{n}X\\
%\vdots\\
%\Delta_{p,i}^{n}X\\
%\end{pmatrix} + \Delta_{i}^{n}Z,
\end{equation}
%\begin{equation}
%\label{Eq:regression_continuous}
%\Delta_{i}^{n}Y \approx\sum_{k=1}^{K_{n}}\sum_{j=1}^{p}\gamma_{j}^{(k)}B_{j,(i-1)\Delta_{n}}^{(k)}\Delta_{j,i}^{n}X + \Delta_{i}^{n}Z,
%\end{equation}
which is a static regression model with dependent variable $\Delta_{i}^{n}Y $ and regressors $B_{(i-1)\Delta_{n}}^{(k)}\Delta_{j,i}^{n}X$ for $j=1,\dots,p$ and $k=1,\dots,K_{n}$. When $p$ is fixed, the spline coefficients in $\bm{\gamma}$ can be estimated by OLS.

In practice, jumps in $X$ or $Z$ may occur on $[0,T]$. To mitigate their impact, we employ the truncation technique of \citet{mancini2009non} for both $\Delta^n_i Y$ and $\Delta_i^n X$. We then define $\widehat{\bm{\gamma}}$ as the estimator of $\bm{\gamma}$ obtained by minimizing the following truncated least squares criterion:
\begin{equation}\label{Eq:objective}
Q_{n}(\bm{\gamma}) = \sum_{i=1}^{n}\left(\Delta_{i}^{n}Y - \sum_{k=1}^{K_{n}}(\gamma^{(k)})^{\top}B_{(i-1)\Delta_{n}}^{(k)}\Delta_{i}^{n}X\right)^{2}\mathbbm{1}_{\{|\Delta_{i}^{n}Y|\leq u_{n},\,\Vert\Delta_{i}^{n}X\Vert\leq u_{n}\}},
\end{equation}
where $u_{n}>0$ is a truncation threshold satisfying $u_{n}\asymp\Delta_{n}^{\varpi}$ for some $0<\varpi<1/2$. The same truncation threshold $u_{n}>0$ is used for both $\Delta_{i}^{n}Y$ and each component of $\Delta_{i}^{n}X$ for ease of notation, while these thresholds can be different.\footnote{We use the norm truncation $\Vert\Delta_{i}^{n}X\Vert\leq u_{n}$ for notational convenience, as is standard in high-frequency regression literature \citep{reiss2015nonparametric,ait2020high}. When $p$ is fixed, it is asymptotically equivalent to component-wise truncation. In the high-dimensional analysis in \cref{Sec:selection}, we instead adopt the component-wise notation $\vert\Delta_{j,i}^{n}X\vert\leq u_{n}$. }

Given $\widehat{\bm{\gamma}}$, the corresponding estimator of $\beta_{t}=(\beta_{1,t},\dots,\beta_{p,t})^{\top}$ takes the form:
\begin{equation}
\widehat{\beta}_{t} =  \sum_{k=1}^{K_{n}}B_{t}^{(k)}\widehat{\gamma}^{(k)},\qquad\text{for any }0\leq t\leq T.
\end{equation}
The resulting estimator of $I\beta$ is defined as the Riemann sum of $\widehat{\beta}$ on $i\Delta_{n}$ \citep{mykland2009inference}:
\begin{equation}
\widehat{I\beta}_{T}=\sum_{i=1}^{n}\widehat{\beta}_{(i-1)\Delta_{n}}\Delta_{n}. 
\end{equation}
%and therefore we can construct an estimator for $I\beta$ as
%\begin{equation}
%\widehat{I\beta}_{T}=\sum_{i=1}^{n}\int_{(i-1)\Delta_{n}}^{i\Delta_{n}}\widehat{\beta}_{s}ds. 
%\end{equation}

Our proposed estimators $\widehat{\bm{\gamma}}$ and $\widehat{\beta}$ admit closed-form representations in matrix notation. We collect truncated returns of $Y$ and $X$ into the vector $\mathbf{Y}$ and matrix $\mathbf{X}$, respectively:
\begin{equation}
\mathbf{Y}=\begin{pmatrix}
\Delta_{1}^{n}Y\mathbbm{1}_{\{|\Delta_{1}^{n}Y|\leq u_{n}\}}\\
\Delta_{2}^{n}Y\mathbbm{1}_{\{|\Delta_{2}^{n}Y|\leq u_{n}\}}\\
\vdots\\
\Delta_{n}^{n}Y\mathbbm{1}_{\{|\Delta_{n}^{n}Y|\leq u_{n}\}}
\end{pmatrix}\in\mathbb{R}^{n},\qquad
\mathbf{X}=\begin{pmatrix}
\mathbf{X}_{1}^{\top}\\
\mathbf{X}_{2}^{\top}\\
\vdots\\
\mathbf{X}_{n}^{\top}
\end{pmatrix}=
\begin{pmatrix}
(\Delta_{1}^{n}X)^{\top}\mathbbm{1}_{\{\Vert\Delta_{1}^{n}X\Vert\leq u_{n}\}}\\
(\Delta_{2}^{n}X)^{\top}\mathbbm{1}_{\{\Vert\Delta_{2}^{n}X\Vert\leq u_{n}\}}\\
\vdots\\
(\Delta_{n}^{n}X)^{\top}\mathbbm{1}_{\{\Vert\Delta_{n}^{n}X\Vert\leq u_{n}\}}
\end{pmatrix}\in\mathbb{R}^{n\times p}.
\end{equation}
Define a block-diagonal matrix that collects all B-spline basis functions at time $t$ as
\begin{equation}
\mathbf{B}_{t}=\begin{pmatrix}
B_{t}^{(1)} & \dots & B_{t}^{(K_{n})} & 0 & \dots & 0 & \dots & 0 & \dots & 0\\
0 & \dots & 0 & B_{t}^{(1)} & \dots & B_{t}^{(K_{n})} & \dots & 0 & \dots & 0\\
\vdots & \ddots & \vdots & \vdots & \ddots & \vdots & \ddots & \vdots & \ddots & \vdots\\
0 & \dots & 0 & 0 & \dots & 0 & \dots & B_{t}^{(1)} & \dots & B_{t}^{(K_{n})}
\end{pmatrix}\in\mathbb{R}^{p\times pK_{n}},
\end{equation}
and the design matrix $\mathbf{R}=(\mathbf{R}_{1},\dots,\mathbf{R}_{n})^{\top}\in\mathbb{R}^{n\times pK_{n}}$ with
\begin{equation}
\mathbf{R}_{i}=\mathbf{B}_{(i-1)\Delta_{n}}^{\top}\mathbf{X}_{i}=\begin{pmatrix}
B_{(i-1)\Delta_{n}}^{(1)}\Delta_{1,i}^{n}X\mathbbm{1}_{\{\Vert\Delta_{i}^{n}X\Vert\leq u_{n}\}}\\
\vdots\\
B_{(i-1)\Delta_{n}}^{(K_{n})}\Delta_{1,i}^{n}X\mathbbm{1}_{\{\Vert\Delta_{i}^{n}X\Vert\leq u_{n}\}}\\
\vdots\\
\vdots\\
B_{(i-1)\Delta_{n}}^{(1)}\Delta_{p,i}^{n}X\mathbbm{1}_{\{\Vert\Delta_{i}^{n}X\Vert\leq u_{n}\}}\\
\vdots\\
B_{(i-1)\Delta_{n}}^{(K_{n})}\Delta_{p,i}^{n}X\mathbbm{1}_{\{\Vert\Delta_{i}^{n}X\Vert\leq u_{n}\}}\\
\end{pmatrix}\in\mathbb{R}^{pK_{n}}.
\end{equation}
Then, our estimator $\widehat{\bm{\gamma}}$ takes the following form:
\begin{equation}
\label{Eq:estimator_GammaVec}
\widehat{\bm{\gamma}}=(\mathbf{R}^{\top}\mathbf{R})^{-1}\mathbf{R}^{\top}\mathbf{Y},
\end{equation}
and the spline-based estimator of $\beta_{t}=(\beta_{1,t},\dots,\beta_{p,t})^{\top}$ can be expressed as 
\begin{equation}
\label{Eq:estimator_beta}
\widehat{\beta}_{t}=\mathbf{B}_{t}\widehat{\bm{\gamma}}=\mathbf{B}_{t}(\mathbf{R}^{\top}\mathbf{R})^{-1}\mathbf{R}^{\top}\mathbf{Y}.
\end{equation}

\subsection{Spline Approximations}
\label{Sec:spline}

We begin by defining a spline approximation to the time-varying coefficient process $\beta=(\beta_{1},\dots,\beta_{p})^{\top}$ on $[0,T]$, where each component is approximated by a polynomial spline of degree $d$. A polynomial spline of degree $d$ on a knot sequence $0=v_{0}<v_{1}<\dots<v_{N_{n}+1}=T$ is a function that coincides with a polynomial of degree $d$ on each subinterval $[v_{i-1},v_{i})$ for $1\le i\le N_{n}$ and on $[v_{N_{n}},v_{N_{n}+1}]$, and that has $d-1$ continuous derivatives when $d\ge 1$. In particular, we consider a quasi-uniform sequence of partitions satisfying
\begin{equation}
\max_{1\leq i,j\leq N_{n}+1}\frac{v_{i}-v_{i-1}}{v_{j}-v_{j-1}}\leq C<\infty.
\end{equation}
The collection of such spline functions, for a given degree and knot sequence, forms a scalar linear space $\mathcal{G}_{n}$ with a B-spline basis $\{B^{(k)}_{\cdot}\}_{k=1}^{K_{n}}$ of dimension $K_{n}=N_{n}+d+1$; see \citet{de1978practical} and \citet{schumaker2007spline} for further details. We then define the vector-valued spline space
\begin{equation}
\mathbb{G}_{n}=\mathcal{G}_{n}^{p}=\bigl\{g:[0,T]\to\mathbb{R}^{p}: g_{t}=\mathbf{B}_{t}\bm{\gamma}\ \text{for some }\bm{\gamma}\in\mathbb{R}^{pK_n}\bigr\}.
\end{equation}

%Let $\mathbb{G}_{n}$ denote the linear space of dimension $pK_{n}$ and let the collection of polynomials $\{B_{t}^{(k)}, k=1,\dots,K_n\} $ with $B_{t}^{(k)}=(B_{1,t}^{(k)},\dots, B_{p,t}^{(k)})^{\top}$ be the B-spline basis of $\mathbb{G}_{n}$.

Rather than adopting an arbitrary spline approximation of $\beta$, we work with a uniquely defined approximation under a Hilbert-space metric that is intrinsic to the continuous-time regression model. We equip the space of $\mathbb{R}^p$-valued measurable functions on $[0,T]$ with the weighted inner product
\begin{equation}
\langle f,g \rangle_{c} = \int_{0}^{T}f_{s}^{\top}c_{s}g_{s}ds,
\end{equation}
where the spot covariance process $c$ is defined in \cref{Eq:c}, and write $\Vert f\Vert_{L^{2}(c)}=\langle f,f \rangle_{c}^{1/2}$. We denote by $L^{2}(c)=\{f:\Vert f\Vert_{L^{2}(c)}<\infty\}$ the corresponding Hilbert space. For each $\omega\in\Omega$, let $\Pi_{n}\beta(\omega)$ denote the $L^{2}(c)$-orthogonal projection of the path $t\mapsto \beta_t(\omega)$ onto the spline space $\mathbb{G}_{n}$, i.e., the unique $g\in\mathbb{G}_{n}$ that minimizes
\begin{equation}
\Vert\beta(\omega)-g\Vert_{L^{2}(c)}^{2} = \int_{0}^{T}(\beta_{s}(\omega)-g_{s})^{\top}c_{s}(\omega)(\beta_{s}(\omega)-g_{s})ds .
\end{equation}
Since, for each $n$, $\mathbb{G}_n$ is a finite-dimensional subspace of $L^2(c)$, the minimizer exists and is unique. In other words, there exists a unique coefficient vector $\bm{\gamma}(\omega)\in\mathbb{R}^{pK_n}$ such that $\Pi_{n}\beta_{t}(\omega) = \mathbf{B}_{t}\bm{\gamma}(\omega)$ for all $0\leq t\leq T$. Consequently, any coefficient path admits the decomposition $\beta_{t}(\omega)=\mathbf{B}_{t}\bm{\gamma}(\omega)+e_{t}(\omega)$, where the spline approximation error $e$ is defined pathwise, and therefore $\widehat{\beta}_{t}$ in \cref{Eq:estimator_beta} estimates the projected coefficient path $\Pi_{n}\beta_{t}(\omega)$ with the standard OLS estimator $\widehat{\bm{\gamma}}$.

When $\beta$ follows a discontinuous It{\^o} semimartingale in \cref{Eq:beta}, the spline approximation error satisfies the following asymptotic order:
\begin{lemma}
\label{lemma:approximationError}
Suppose Assumptions \ref{As:Ito} and \ref{As:Ito_beta} hold. Then, as $\Delta_{n}\to0$,
\begin{equation}
\Vert e\Vert_{L^{2}}=\left(\int_{0}^{T}\Vert e_{s}\Vert^{2}ds\right)^{1/2}=O_{p}\left(\sqrt{\frac{\log K_{n}}{K_{n}}}\right), %O_{p}(K_{n}^{-1/2}\sqrt{\log K_{n}}).
\end{equation}
where $\Vert \cdot\Vert$ inside the integral is the Euclidean norm on $\mathbb{R}^{p}$.
\end{lemma}

\subsection{Estimation Consistency and Central Limit Theorem}
\label{Sec:limitTheory}

%\textcolor{purple}{($\xrightarrow{\mathbb{P}(\cdot|\beta)}$: conditional on the realized path?)}

We establish consistency of our spline-based estimator and derive the corresponding limit theorem for $\widehat{I\beta}$ when $p$ is fixed. Following the literature on varying-coefficient models (see, e.g., \citealp{hoover1998nonparametric,huang2002varying,huang2004polynomial}), we say that $\widehat{\beta}$ is a consistent estimator of $\beta$ if $\Vert\widehat{\beta}-\beta\Vert_{L^{2}}\overset{\mathbb{P}}{\longrightarrow}0$. We denote by $\asymp$ the same order of magnitude, i.e., $f\asymp g$ indicates that $K|g|\leq|f|\leq K'|g|$ for some constants $K,K'>0$.

\begin{theorem}%[Consistency]
\label{Th:consistency}
Suppose Assumptions \ref{As:Ito}--\ref{As:orthogonality} hold for the regression model in \cref{Eq:regression}, and $\Delta_{n}K_{n}\log K_{n}\to0$. Take the truncation threshold $u_{n}\asymp\Delta_{n}^{\varpi}$ for some $0<\varpi<1/2$. Then, as $\Delta_{n}\to0$, $\widehat{\beta}$ is uniquely defined with probability approaching one, and $\widehat{\beta}$ is consistent. 
\end{theorem}

\cref{Th:consistency} shows that the consistency of $\widehat{\beta}$ holds with the standard choice of truncation threshold from \citet{mancini2009non} under the mild condition $\Delta_{n}K_{n}\log K_{n}\to0$. The identifiability of $\widehat{\bm{\gamma}}$, and hence of $\widehat{\beta}$, follows from the nonsingularity of the empirical Gram matrix $\mathbf{R}^{\top}\mathbf{R}$, which is implied by that of its theoretical counterpart \citep{huang2003local}, and, in turn, is ensured by the boundedness of B-spline basis functions \citep{de1978practical}.

%\cref{Th:CLT} provides the asymptotic distribution of $\widehat{I\beta}$. We write $X^{n}\xrightarrow{\mathcal{L}-s}X$ to indicate the stable convergence in law, i.e., if for all $\mathcal{F}$-measurable processes $Y$, we have $(X^{n},Y)\overset{\mathcal{L}}{\longrightarrow}(X,Y)$. We denote by $\mathcal{MN}$ a mixed normal distribution, i.e., a normal distribution conditional on the realization of its $\mathcal{F}$-conditional variance, which is also a random variable.

We now present a central limit theorem for $\widehat{I\beta}$. We write $X^{n}\xrightarrow{\mathcal{L}-s}X$ to denote stable convergence in law, i.e., $(X^{n},Y)\overset{\mathcal{L}}{\longrightarrow}(X,Y)$ for any $\mathcal{F}$-measurable process $Y$. We denote by $\mathcal{MN}$ a mixed normal distribution, i.e., a conditionally normal distribution with random $\mathcal{F}$-conditional variance.

\begin{theorem}%[Asymptotic Normality]
\label{Th:CLT}
%Suppose Assumptions \ref{As:Ito} and \ref{As:orthogonality} hold, and \cref{As:spline} holds with some $\alpha>1/2$ and a sequence $K_{n}$ satisfying $K_{n}^{\alpha}\sqrt{\Delta_{n}}\to\infty$ and $\Delta_{n}K_{n}\log K_{n}\to0$. %or, equivalently, $K_{n}\asymp\Delta_{n}^{-\rho}$ with $(2\alpha)^{-1}<\rho<1$. 
%Let $u_{n}\asymp\Delta_{n}^{\varpi}$ for some $\varpi\in\left(\frac{1}{2(2-r)},\,\frac{1}{2}\right)$. Then, as $\Delta_{n}\to0$, it holds that $\Vert\widehat{\beta}-\beta\Vert_{L^{2}}=O_{p}(\sqrt{\Delta_{n}})$, and
Suppose Assumptions \ref{As:Ito}--\ref{As:orthogonality} hold, $\Delta_{n}K_{n}\log K_{n}\to0$, and $K_{n}\sqrt{\Delta_{n}}/\log K_{n}\to\infty$. %or, equivalently, $K_{n}\asymp\Delta_{n}^{-\rho}$ with $(2\alpha)^{-1}<\rho<1$. 
Let $u_{n}\asymp\Delta_{n}^{\varpi}$ for some $\varpi\in\left(\frac{1}{2(2-r)},\,\frac{1}{2}\right)$. Then, as $\Delta_{n}\to0$, it holds that %$\Vert\widehat{\beta}-\beta\Vert_{L^{2}}=O_{p}(\sqrt{\Delta_{n}})$, and
\begin{equation}
\label{Eq:CLT_IB}
\frac{1}{\sqrt{\Delta_{n}}}\,\bigl(\widehat{I\beta}_{T}-I\beta_{T}\bigr)\xrightarrow{\mathcal{L}-s}%\mathcal{W}_{t}^{\beta},
\mathcal{MN}\bigl(0,\,\Sigma_{T}^{\beta}\bigr),
\end{equation}
%where $\mathcal{W}^{\beta}$ is a process defined on the extension of the original probability space $(\Omega,\mathcal{F},(\mathcal{F}_{t})_{t\geq0},\mathbb{P})$ ... 
where the asymptotic covariance is given by
\begin{equation}
\Sigma_{T}^{\beta}=\left(\int_{0}^{T}\mathbf{B}_{s}ds\right)\left(\int_{0}^{T}\mathbf{B}_{s}^{\top}c_{s}\mathbf{B}_{s}ds\right)^{-1}\left(\int_{0}^{T}\widetilde{\sigma}_{s}^{2}\mathbf{B}_{s}^{\top}c_{s}\mathbf{B}_{s}ds\right)\left(\int_{0}^{T}\mathbf{B}_{s}^{\top}c_{s}\mathbf{B}_{s}ds\right)^{-1}\left(\int_{0}^{T}\mathbf{B}_{s}ds\right)^{\top}.
\end{equation}
\end{theorem}

The above stable central limit theorem is established by decomposing the normalized estimation error into a leading martingale term and a collection of remainder terms. The latter arise from spline approximation and discretization errors in the construction of $\widehat{I\beta}$, as well as from jumps in the covariate and residual processes. The leading term admits a truncated realized-covariance representation and converges stably to a mixed normal limit. This result requires stronger conditions than Theorem \ref{Th:consistency} on the truncation threshold $u_n$, as is standard in the asymptotic theory of truncated realized covariance (see, e.g., Theorem 13.2.1, \citealp{jacod2012discretization}). The condition $\Delta_{n}K_{n}\log K_{n}\to0$ follows from \cref{Th:consistency} and is analogous to standard conditions in the spline-based varying-coefficient model literature (see, e.g., Theorem 3, \citealp{huang2004polynomial}). The additional condition $K_{n}\sqrt{\Delta_{n}}/\log K_{n}\to\infty$ ensures that the bias induced by spline approximation is asymptotically negligible and thus has no impact on the central limit theorem. 

%The above stable central limit theorem (CLT) is obtained by decomposing the normalized estimation error into a dominant martingale term and a collection of remainders. The dominant component admits a truncated realized-covariance representation, so the stable mixed normality follows from the standard results of truncated functionals (see, e.g., Theorem 13.2.1, \citealp{jacod2012discretization}), and an application of the Cram{\'e}r-Wold device. The remainder terms come from the spline approximation and discretization errors by the construction of $\widehat{I\beta}$, as well as components related to both jumps in the covariate and residual processes, which are either removed by truncation or shown to be smaller order under the jump-activity conditions. 
%%The above stable central limit theorem (CLT) follows from the standard results of truncated realized covariance, see, e.g., Theorem 13.2.1 in \citet{jacod2012discretization}, and an application of the Cram{\'e}r-Wold device.  
%In particular, as the basis dimension $K_n$ increases such that $K_{n}\sqrt{\Delta_{n}}/\log K_{n}\to\infty$, the bias induced by the spline approximation error is asymptotically negligible and therefore has no impact on the CLT. 

For feasible implementation of the asymptotic distribution in \cref{Th:CLT}, we can estimate $\Sigma_{T}^{\beta}$ by a jackknife White's heteroskedasticity-consistent estimator \citep{mackinnon1985some}:

\begin{theorem}%[Feasible Inference]
\label{Th:feasiblility}
Under the conditions of \cref{Th:CLT}, $\Sigma_{T}^{\beta}$ can be consistently estimated by
\begin{equation}
\widehat{\Sigma}_{T}^{\beta}=\left(\sum_{i=1}^{n}\mathbf{B}_{(i-1)\Delta_{n}}\Delta_{n}\right)(\mathbf{R}^{\top}\mathbf{R})^{-1}\mathbf{R}^{\top}\mathbf{D}\mathbf{R}(\mathbf{R}^{\top}\mathbf{R})^{-1}\left(\sum_{i=1}^{n}\mathbf{B}_{(i-1)\Delta_{n}}\Delta_{n}\right)^{\top},%\overset{\mathbb{P}}{\longrightarrow}\Sigma_{T}^{\beta},
\end{equation}
where $\mathbf{D}=\text{diag}(\Delta_{n}^{-1}(\Delta_{i}^{n}Y-(\Delta_{i}^{n}X)^{\top}\widehat{\beta}_{(i-1)\Delta_{n}}^{(-i)})^2 \mathbbm{1}_{\{|\Delta_{i}^{n}Y|\leq u_{n},\,\Vert\Delta_{i}^{n}X\Vert\leq u_{n}\}})\in\mathbb{R}^{n\times n}$. Here, $\widehat{\beta}_{(i-1)\Delta_{n}}^{(-i)}$ denotes the leave-one-out counterpart of $\widehat{\beta}_{(i-1)\Delta_{n}}$ computed with the increments of both $X$ and $Y$ with the $i$-th interval removed.
\end{theorem}

In $\widehat{\Sigma}_{T}^{\beta}$ we use leave-one-out residuals to avoid reusing the same increments for the estimation of both coefficients and covariance matrix, which restores the conditional orthogonality between the coefficient estimate $\widehat{\beta}_{(i-1)\Delta_{n}}^{(-i)}$ and the $i$-th interval increments needed for the covariance matrix estimation consistency.\footnote{For consistency of the standard sandwich estimator based on the full-sample $\widehat{\beta}$, one needs additional uniform control of negligible leverage across all $i$, which becomes restrictive when the regression dimension $pK_n$ diverges; see, e.g., \citet{kline2020leave}, for related discussions. } For computational convenience, the leave-one-out residuals can be replaced by a block cross-fitted analogue in the spirit of \citet{chernozhukov2018double}, which requires only a finite number of refits. We leave this extension to future work.

\subsection{Model Selection with Penalized Regressions}
\label{Sec:selection}

%Up to this point you approximate the time-varying coefficients via a growing B-spline basis and estimate them from truncated high-frequency increments—yielding a standard OLS form for the continuous part—before deriving a stable CLT for the integrated path and a feasible (White-type) covariance estimator.

In this section, we assume that the $p$-dimensional coefficient process takes the following form:
\begin{equation}
\beta_{t}=(\beta_{1,t},\dots,\beta_{q,t}, 0,\dots,0)^{\top},
\end{equation}
where $q\ge1$ is the number of relevant regressors and $0<\Vert\beta_{j}\Vert_{L^{2}}<\infty$ for $j=1,\dots,q$. That is, the first $q$ regressors are relevant with nonzero coefficients, while the remaining $p-q$ regressors are redundant and have coefficients that are almost surely zero for all $t\in[0,T]$.\footnote{We assume $\beta^{J}$ has the same sparsity structure as $\beta$, i.e., $\beta_{t}^{J}=(\beta_{1,t}^{J},\dots,\beta_{q,t}^{J}, 0,\dots,0)^{\top}$. } In line with the model selection literature, we consider a high-dimensional regime in which $p=p_{n}\to\infty$, while the number of relevant regressors $q$ remains fixed. 

Our goal is to simultaneously estimate the coefficients of relevant regressors and perform the model selection. For each irrelevant regressor, the associated B-spline coefficient vector $\gamma_j=(\gamma_j^{(1)},\dots,\gamma_j^{(K_n)})^\top$ is identically zero. Consequently, consistent model selection can be achieved by imposing a group penalty on $\gamma_{j}$ for $j=1,\dots,p$. 

We consider an estimator $\widehat{\bm{\gamma}}^{*}$ defined as a global minimizer of
\begin{equation}
Q_{n}^{*}(\bm{\gamma})=\frac{1}{2}Q_{n}(\bm{\gamma})+\lambda_{n}\sum_{j=1}^{p}\rho_{n}\left(\sqrt{\gamma_{j}^{\top}\mathbf{W}_{j}\gamma_{j}}\right),\qquad\text{where }\rho_{n}(x)=\min\left(\frac{|x|}{\tau_{n}},\,1\right),
\end{equation} 
where $\rho_{n}(x)$ is the truncated $\ell_{1}$-penalty (TLP) function with a thresholding parameter $\tau_n>0$ \citep{shen2012likelihood}, and $\mathbf{W}_{j}=(\mathbf{R}_{j}^{*})^{\top}\mathbf{R}_{j}^{*}$ is the empirical Gram matrix of block $j$, with $\mathbf{R}_{j}^{*}$ the $n\times K_{n}$ submatrix of $\mathbf{R}$ corresponding to the $j$-th regressor:
\begin{equation}
\label{Eq:R_j^*}
\mathbf{R}_{j}^{*}=\begin{pmatrix}
B_{0}^{(1)}\Delta_{j,1}^{n}X\mathbbm{1}_{\{\vert\Delta_{j,1}^{n}X\vert\leq u_{n}\}} & \dots & B_{0}^{(K_{n})}\Delta_{j,1}^{n}X\mathbbm{1}_{\{\vert\Delta_{j,1}^{n}X\vert\leq u_{n}\}}\\
B_{\Delta_{n}}^{(1)}\Delta_{j,2}^{n}X\mathbbm{1}_{\{\vert\Delta_{j,2}^{n}X\vert\leq u_{n}\}} & \dots & B_{\Delta_{n}}^{(K_{n})}\Delta_{j,2}^{n}X\mathbbm{1}_{\{\vert\Delta_{j,2}^{n}X\vert\leq u_{n}\}}\\
\vdots & \ddots & \vdots\\
B_{(n-1)\Delta_{n}}^{(1)}\Delta_{j,n}^{n}X\mathbbm{1}_{\{\vert\Delta_{j,n}^{n}X\vert\leq u_{n}\}} & \dots & B_{(n-1)\Delta_{n}}^{(K_{n})}\Delta_{j,n}^{n}X\mathbbm{1}_{\{\vert\Delta_{j,n}^{n}X\vert\leq u_{n}\}}\\
\end{pmatrix}
%\begin{pmatrix}
%B_{0}^{(1)}\Delta_{j,1}^{n}X\mathbbm{1}_{\{\vert\Delta_{j,1}^{n}X\vert\leq u_{n}\}} & \dots & B_{(n-1)\Delta_{n}}^{(1)}\Delta_{j,n}^{n}X\mathbbm{1}_{\{\vert\Delta_{j,n}^{n}X\vert\leq u_{n}\}}\\
%B_{0}^{(2)}\Delta_{j,1}^{n}X\mathbbm{1}_{\{\vert\Delta_{j,1}^{n}X\vert\leq u_{n}\}} & \dots & B_{(n-1)\Delta_{n}}^{(2)}\Delta_{j,n}^{n}X\mathbbm{1}_{\{\vert\Delta_{j,n}^{n}X\vert\leq u_{n}\}}\\
%\vdots & \ddots & \vdots\\
%B_{0}^{(K_{n})}\Delta_{j,1}^{n}X\mathbbm{1}_{\{\vert\Delta_{j,1}^{n}X\vert\leq u_{n}\}} & \dots & B_{(n-1)\Delta_{n}}^{(K_{n})}\Delta_{j,n}^{n}X\mathbbm{1}_{\{\vert\Delta_{j,n}^{n}X\vert\leq u_{n}\}}\\
%\end{pmatrix}
\in\mathbb{R}^{n\times K_{n}}.
\end{equation} 
The objective function $Q_{n}^{*}(\bm{\gamma})$ adopts a general form for group selection \citep{yuan2006model,huang2012selective}, while we select the nonconvex TLP to obtain model selection consistency, which cannot be achieved with standard group LASSO without irrepresentable conditions \citep{bach2008consistency}. The TLP function is flat beyond the threshold $\tau_n$, so it penalizes small coefficients strongly while leaving sufficiently large true signals unpenalized. 

Let $\widehat{\bm{\gamma}}^{\circ}$ be the oracle estimator, i.e., the restricted minimizer of $Q_{n}(\bm{\gamma})$ subject to $\gamma_{j}\equiv0$ for $q<j\leq p$. Equivalently, $\widehat{\bm{\gamma}}^{\circ}$ is the least-squares spline estimator for the model that includes only the $q$ relevant regressors \citep{fan2001variable}.  \cref{Th:oracle_local} shows that $\widehat{\bm{\gamma}}^{\circ}$ is a local minimizer of the penalized objective $Q_{n}^{*}(\bm{\gamma})$, and \cref{Th:oracle_global} provides sufficient conditions such that the global minimizer $\widehat{\bm{\gamma}}^{*}$ coincides with $\widehat{\bm{\gamma}}^{\circ}$, which establishes the model selection consistency.

\begin{theorem}
\label{Th:oracle_local}
%Under the conditions of \cref{Th:CLT}, and $\Delta_{n}^{2\varpi-1/2}\sqrt{K_{n}\log pK_{n}}\to 0$, let $\tau_{n}\to0$ and $\lambda_{n}>0$ satisfy %$(\lambda_{n}/\tau_{n})\sqrt{K_{n}/\Delta_{n}}\to\infty$. 
%$\frac{\lambda_{n}}{\tau_{n}}\sqrt{\frac{K_{n}}{\Delta_{n}}}\to\infty$.
Under the conditions of \cref{Th:CLT}, let $\tau_{n}\to0$ and $\lambda_{n}>0$ satisfy
\begin{equation}
\label{Eq:penaltyCondition1}
\frac{\lambda_{n}}{\tau_{n}}\bigg/\left(\sqrt{\frac{\Delta_{n}\log pK_{n}}{K_{n}}}+u_{n}\log pK_{n}+u_{n}^{1-r/2}\right)\to\infty.
%\frac{\lambda_{n}}{\tau_{n}}\bigg/\left(\sqrt{\frac{\Delta_{n}\log pK_{n}}{K_{n}}}+u_{n}\log pK_{n}+u_{n}^{1-r/2}\right)\,\asymp\,\frac{\lambda_{n}}{\tau_{n}}\bigg/\left(\Delta_{n}^{\varpi}\log pK_{n}+\Delta_{n}^{\varpi(1-r/2)}\right)\to\infty.
\end{equation}
Then the oracle estimator $\widehat{\bm{\gamma}}^{\circ}$ is a local minimizer of $Q_{n}^{*}(\bm{\gamma})$ with probability approaching one. 
\end{theorem}

\cref{Th:oracle_local} is established by verifying the standard Karush-Kuhn-Tucker (KKT) local optimality conditions. The key requirement is that, uniformly over all irrelevant regressors, the effective penalty level $\lambda_{n}/\tau_{n}$ dominates the blockwise gradient of least-squares loss—a standard calibration for both convex and nonconvex penalties for model selection. The first two terms in the denominator of \cref{Eq:penaltyCondition1} follow from a Bernstein-type maximal inequality for high-dimensional martingale difference arrays (as an analogue to Lemma A.1, \citealp{van2008high}), combined with a uniform control of the predictable compensators in the third term. Altogether, the condition ensures that irrelevant blocks are shrunk to zero while the active blocks fall into the flat region of TLP and are therefore left unpenalized, such that the local KKT solution coincides with the oracle estimator $\widehat{\bm{\gamma}}^{\circ}$. 

\begin{theorem}
\label{Th:oracle_global}
Under the conditions of \cref{Th:oracle_local}, and $\Delta_{n}^{(1-\varpi(2-r))\vee2\varpi r}(K_{n}+\log p)\to\infty$, let $\tau_{n}\to0$ and $\lambda_{n}>0$ satisfy
\begin{equation}
\label{Eq:penaltyCondition2}
\lambda_{n}\to0\qquad\text{and}\qquad
\frac{\lambda_{n}}{\Delta_{n}(K_{n}+\log p)}\to\infty,
\end{equation}
Then it holds for the global minimizer $\widehat{\bm{\gamma}}^{*}$ of $Q_{n}^{*}(\bm{\gamma})$ that
\begin{equation}
\mathbb{P}(\widehat{\bm{\gamma}}^{*}=\widehat{\bm{\gamma}}^{\circ})\to1,
\end{equation}
i.e., the penalized estimator has the oracle property.
\end{theorem}

\cref{Th:oracle_global} follows from the fact that any misspecified model—whether it omits relevant groups, includes redundant ones, or both—attains a strictly larger penalized objective than the oracle model with probability approaching one. The additional conditions are calibrated to control two types of errors in opposite directions. On the one hand, letting $\lambda_{n}\to0$ ensures the penalty does not over-reward parsimony, and therefore the increase in squared loss from omitting relevant groups dominates any penalty reduction. On the other hand, if instead the model includes redundant groups, it may achieve a decrease in squared loss through overfitting, but the second condition in \cref{Eq:penaltyCondition2} guarantees that the penalty is sufficiently strong to outweigh that spurious improvement. In short, the conditions balance fidelity and parsimony: omissions are ruled out because loss inflation dominates when $\lambda_{n}$ is small, and overfitted models are also ruled out because the penalty does not vanish too quickly. Moreover, the second condition in \cref{Eq:penaltyCondition2} implies $\Delta_{n}(K_{n}+\log p)\to0$, so the dimension $p$ can diverge only at a rate slower than exponential relative to $n$. Extending the result to accommodate more rapidly growing model complexity is left for future work.

%\begin{proof}
%Let $c_j(\bm\gamma):=-\tfrac{1}{n}R_j^\top\!\big(Y-R\bm\gamma\big)$. Analogously to \citet{xue2012variable}, any local minimiser $\widehat{\bm\gamma}$ of $R_n(\bm\gamma)$
%must satisfy: for each $j=1,\ldots,p$,
%\begin{itemize}
%\item[(i)] if $\widehat{\bm\gamma}_j\neq 0$, then
%$$
%c_j(\widehat{\bm\gamma})\;+\;
%\lambda_n\,\rho_n'\!\Big(\sqrt{\widehat{\bm\gamma}_j^\top W_j\,\widehat{\bm\gamma}_j}\Big)\,
%\frac{W_j\,\widehat{\bm\gamma}_j}{\sqrt{\widehat{\bm\gamma}_j^\top W_j\,\widehat{\bm\gamma}_j}}
%\;=\;0.
%$$
%\item[(ii)]  If $\widehat{\bm\gamma}_j=0$, then
%$$
%c_j(\widehat{\bm\gamma})^\top\,W_j^{-1}\,c_j(\widehat{\bm\gamma})
%\;\le\; \big(\lambda_n\,\kappa_n\big)^2,
%\qquad \kappa_n:=\rho_n'(0+).
%$$
%\end{itemize}
%
%\end{proof}

%We consider a decomposition of the $p$-dimensional covariate process and the coefficient process, and allows $p\to\infty$. Without loss of generality, we assume there exists an integer $q\geq1$, such that $0<\mathbb{E}[|\beta_{t}|]<\infty$ for $l=1,\dots,q$, and $\mathbb{E}[\beta_{t}]=0$ for $l=q+1,\dots,p$. 
%
%
%\begin{equation}\label{Eq:objectivePenality}
%Q_{n}(\bm{\gamma})=\frac{1}{2}(\mathbf{Y}-\mathbf{R}\bm{\gamma})^{2}+\frac{\lambda_{n}}{n}\sum_{l=1}^{p}\rho(\Vert\gamma_{j}\Vert ... ),
%\end{equation}
%where $\gamma_{j}=(\gamma_{j}^{(1)},\dots,\gamma_{j}^{(K_{n})})$ and ... for each $l=1,\dots,p$.

\subsection{Practical Implementation}
\label{Sec:implementation}

In practice, the nonconvex objective  $Q_{n}^{*}(\bm{\gamma})$ is minimized by the difference-of-convex (DC) algorithm in the spirit of \citet{shen2012likelihood} and \citet{xue2012variable}. Specifically, the TLP function can be written as $\rho_{n}(x)=\rho_{1,n}(x)-\rho_{2,n}(x)$, where
%In this section, we discuss the nonconvex optimization algorithm and the tuning parameter selection used in practice. We minimize the nonconvex objective  $Q_{n}^{*}(\bm{\gamma})$ by the difference-of-convex (DC) algorithm in the spirit of \citet{shen2012likelihood} and \citet{xue2012variable}. Recall that the TLP function can be written as $\rho_{n}(x)=\rho_{1,n}(x)-\rho_{2,n}(x)$, where
\begin{equation}
\rho_{1,n}(x)=\frac{|x|}{\tau_{n}},\qquad\rho_{2,n}(x)=\max\left(\frac{|x|}{\tau_{n}}-1,\,0\right).
\end{equation}
This allows us to decompose $Q_{n}^{*}(\bm{\gamma})=Q_{1,n}^{*}(\bm{\gamma})-Q_{2,n}^{*}(\bm{\gamma})$, where
\begin{equation}
Q_{1,n}^{*}(\bm{\gamma})=\frac{1}{2}Q_{n}(\bm{\gamma})+\frac{\lambda_{n}}{\tau_{n}}\sum_{j=1}^{p}\sqrt{\gamma_{j}^{\top}\mathbf{W}_{j}\gamma_{j}},\qquad
Q_{2,n}^{*}(\bm{\gamma})=\lambda_{n}\sum_{j=1}^{p}\rho_{2,n}\left(\sqrt{\gamma_{j}^{\top}\mathbf{W}_{j}\gamma_{j}}\right).
\end{equation}
Both $Q_{1,n}^{*}(\bm{\gamma})$ and $Q_{2,n}^{*}(\bm{\gamma})$ are convex functions and, in particular, $Q_{1,n}^{*}(\bm{\gamma})$ is a standard group LASSO objective with effective penalty level $\lambda_{n}/\tau_{n}$. We initialize the algorithm at $\widehat{\bm{\gamma}}^{(0)}=0$, since the DC procedure does not require a consistent initial estimator. At iteration $m\geq1$, we keep $Q_{1,n}^{*}(\bm{\gamma})$ unchanged but approximate $Q_{2,n}^{*}(\bm{\gamma})$ by its affine minorant at $\widehat{\bm{\gamma}}^{(m-1)}$. Because $\rho_{2,n}(x)$ is piecewise linear with slope $0$ when $0\leq x<\tau_{n}$ and slope $1/\tau_{n}$ when $x>\tau_{n}$, this leads to the convex optimization for $\widehat{\bm{\gamma}}^{(m)}$:
\begin{equation}
\widehat{\bm{\gamma}}^{(m)}=\argmin_{\bm{\gamma}}\left\{\frac{1}{2}Q_{n}(\bm{\gamma})+\sum_{j=1}^{p}\lambda_{j,n}^{(m)}\sqrt{\gamma_{j}^{\top}\mathbf{W}_{j}\gamma_{j}}\right\},\quad
\text{where }\lambda_{j,n}^{(m)}=\frac{\lambda_{n}}{\tau_{n}}\mathbbm{1}_{\left\{\sqrt{(\widehat{\gamma}_{j}^{(m-1)})^{\top}\mathbf{W}_{j}\widehat{\gamma}_{j}^{(m-1)}}\leq\tau_{n}\right\}},
\end{equation}
so that blocks with $\sqrt{(\widehat{\gamma}_{j}^{(m-1)})^{\top}\mathbf{W}_{j}\widehat{\gamma}_{j}^{(m-1)}}\leq\tau_{n}$ are penalized with weight $\lambda_{n}/\tau_{n}$, whereas blocks with $\sqrt{(\widehat{\gamma}_{j}^{(m-1)})^{\top}\mathbf{W}_{j}\widehat{\gamma}_{j}^{(m-1)}}>\tau_{n}$ remain unpenalized. Therefore, the DC algorithm generates a sequence $(\widehat{\bm{\gamma}}^{(m)})_{m\ge0}$, where each update $\widehat{\bm{\gamma}}^{(m)}$ is from a standard group LASSO with group-specific tuning parameters $(\lambda_{j,n}^{(m)})_{j=1}^{p}$. In our implementation, this group LASSO is solved by the standard proximal gradient method. 

We consider a data-driven approach to selecting the tuning parameters. For the threshold $\tau_{n}$, the standard KKT conditions suggest that it should be small enough to distinguish whether $\sqrt{\gamma_{j}^{\top}\mathbf{W}_{j}\gamma_{j}}$ for each $j$ is essentially zero or not. Intuitively, $\gamma_{j}^{\top}\mathbf{W}_{j}\gamma_{j}$ measures the contribution of regressor $j$ to the integrated variance of $Y$. In our implementation, we set $\tau_{n}$ to be a small multiple, e.g., $\alpha_{\tau}=0.05$ or 0.01, of the square root of the median realized volatility (MedRV, \citealp{andersen2012jump}) of $Y$. 

We employ cubic B-splines throughout the simulation and empirical applications in this paper. Conditional on the above choice of $\tau_{n}$, we select the number of B-spline basis functions $K_{n}$ and the TLP parameter $\lambda_{n}$ (or, equivalently, the effective penalty level $\lambda_{n}/\tau_{n}$) by $K$-fold cross-validation over a reasonable grid of candidate values. The cross-validation criterion is the out-of-sample mean squared prediction error for the truncated returns of $Y$. Furthermore, to improve interpretability in finite samples, we can apply the ``one-standard-error'' rule with cross-validation, i.e., select the largest effective penalty whose cross-validation criterion is no more than one standard error above the minimum to obtain the most parsimonious model (see, e.g., Chapter 7.10, \citealp{hastie2009elements}). We implement this rule in our empirical applications in \cref{Sec:empirical}.

\section{Monte Carlo Simulations}
\label{Sec:simulation}

This section contains a Monte Carlo study to evaluate both the estimation and model selection performance of our estimators, which correspond to the theoretical results developed in \cref{Sec:limitTheory,Sec:selection}.

\subsection{Simulation Design}
\label{Sec:simulation_design}

We assess the finite‐sample performance of our estimator following a simulation design similar to that in \cite{ait2020high}. We consider a continuous-time factor model in \cref{Eq:regression} for the log-price $Y_{t}$ of a single asset, with $q=3$ relevant factors out of $p$ factors in total. The covariate process $X_t = (X_{1,t},\dots,X_{p,t})^{\top}$ evolves as
\begin{equation}
\begin{pmatrix}
dX_{1,t} \\ dX_{2,t} \\ \vdots \\dX_{p,t}
\end{pmatrix}=
\begin{pmatrix}
b_{1} \\ b_{2} \\ \vdots \\b_{p}
\end{pmatrix}dt + 
\begin{pmatrix}
\sigma_{1,t} & 0 & \dots & 0\\
0 & \sigma_{2,t} & \dots & 0\\
\vdots & \vdots & \ddots & \vdots\\
0 & 0 & \dots & \sigma_{p,t}\\
\end{pmatrix}
\begin{pmatrix}
1 & \rho_{12} & \dots & \rho_{1p}\\
\rho_{12} & 1 & \dots & \rho_{2p}\\
\vdots & \vdots & \ddots & \vdots\\
\rho_{1p} & \rho_{2p} & \dots & 1\\
\end{pmatrix}
\begin{pmatrix}
dW_{1,t} \\ dW_{2,t} \\ \vdots \\dW_{p,t}
\end{pmatrix} + 
\begin{pmatrix}
J_{1,t} \\ J_{2,t} \\ \vdots \\J_{p,t}
\end{pmatrix}dN_{t},
\end{equation}
where $b_{j}$ is the drift of factor $j$, $(W_{1,t},\dots,W_{p,t})^{\top}$ is a $p$-dimensional standard Brownian motion, the symmetric matrix with ones on the diagonal and off-diagonal entries $\rho_{jk}$ controls the dependence across factors, $J_{j,t}$ is the jump size of factor $j$ at time $t$, and $N_t$ is a Poisson process with annualized intensity $\lambda$. The factor volatilities follow the Cox-Ingersoll-Ross (CIR) dynamics:
\begin{equation}
d\sigma_{j,t}^{2}=\kappa'_{j}(\alpha'_{j}-\sigma_{j,t}^{2})dt+v'_{j}\sqrt{\sigma_{j,t}^{2}}dW'_{j,t}+J'_{j,t}dN_{t}, \qquad j=1,\dots,p.
\end{equation}
where $\kappa'_j$ is the mean-reversion speed, $\alpha'_j$ is the long-run volatility level, $v'_j$ is the volatility-of-volatility parameter, and $(W'_{j,t})_{j=1}^{p}$ are standard Brownian motions independent of $(W_t, N_t)$. For each $j = 1,\dots,p$, the factor jump sizes $J_{j,t}$ are i.i.d.~across jump times and follow a double-exponential distribution:
\begin{equation}
J_{j,t}\sim
\begin{cases}
\textcolor{white}{+}\exp(g_{j}^{+}) & \text{with probability }\omega_{j}\\
-\exp(g_{j}^{-}) & \text{with probability }1-\omega_{j}
\end{cases},
\end{equation}
where $\omega_{j}$ is the mixture probability. The jump sizes in the volatility process, $J'_{j,t}$, are exponentially distributed with mean $g'_{j}$. The Poisson process $N_t$ is common to all factors and their volatilities, so jumps occur simultaneously across components.

The idiosyncratic component $Z_t$ is simulated as a jump-diffusion,
\begin{equation}
dZ_{t}=\widetilde{b}_{t}dt+\widetilde{\sigma}_{t}d\widetilde{W}_{t}+\widetilde{J}_{t}d\widetilde{N}_{t},
\end{equation}
where $\widetilde{W}_{t}$ is a standard Brownian motion independent of $(W_t, W'_t, N_t)$, and the idiosyncratic jump sizes $\widetilde{J}_{t}$ follow another double-exponential distribution:
\begin{equation}
\widetilde{J}_{t}\sim
\begin{cases}
\textcolor{white}{+}\exp(\widetilde{g}^{+}) & \text{with probability }\widetilde{\omega}\\
-\exp(\widetilde{g}^{-}) & \text{with probability }1-\widetilde{\omega}
\end{cases},
\end{equation}
while $\widetilde{N}_{t}$ is a Poisson process with the same annualized intensity $\lambda$ and is independent of $N_{t}$.

Finally, we simulate the time-varying coefficients $(\beta_{1,t},\beta_{2,t},\beta_{3,t},0,\dots,0)^{\top}$ following an Ornstein-Uhlenbeck process:
\begin{equation}\label{Eq:OU}
d\beta_{j,t}=\kappa_{j}(\alpha_{j}-\beta_{j,t})dt + v_{j}dB_{j,t}, \qquad j = 1,2,3,
\end{equation}
where $(B_{1,t},B_{2,t},B_{3,t})^{\top}$ is a three-dimensional standard Brownian motion independent of all previously defined sources of randomness. The response process $Y_{t}$ is then obtained as
\begin{equation}
dY_{t}=\sum_{j=1}^{3}\beta_{j,t}\,dX_{j,t} + dZ_{t}.
\end{equation}
The parameters are set as follows: 
\begin{itemize}
\item[(i)] Covariates: For $j=1,2,3$, $(b_{j},\sigma_{j,0}^{2},\kappa'_{j},\alpha'_{j},v'_{j})=(0.05,0.10,5,0.06,0.35)$. For $j=4,\dots,p$, we draw $b_{j}\sim U[0.03,0.07]$, $\sigma_{j,0}^{2}\sim U[0.06,0.15]$, $\kappa'_{j}\sim U[3,5]$, $\alpha'_{j}\sim U[0.04,0.09]$, and $v'_{j}\sim U[0.3,0.4]$, where $U[a,b]$ denotes the uniform distribution on $[a,b]$. The correlation matrix is Toeplitz with $\rho_{jj'}=0.15^{|j-j'|}$ for all $j\neq j'$. For jumps in both factor prices and volatilities, $(\omega_{j}, g_{j}^{+}, g_{j}^{-}, \lambda)=(0.5 , 7\sigma_{j,0}\sqrt{\Delta_{n}}, -7\sigma_{j,0}\sqrt{\Delta_{n}}, 67)$, and $g'_{j}\sim U[0.004,0.005]$ for all $j=1,\dots,p$.
\item[(ii)] Coefficients: For $j=1,2,3$, $\kappa_{j}=2$, $(\alpha_{1},\alpha_{2},\alpha_{3})=(0.7,-0.5,0.3)$, $v_{j}=0.1$, and $\beta_{j,0}=\alpha_{j}$. 
\item[(iii)] Idiosyncratic component: $(\widetilde b,\widetilde\sigma,\widetilde\omega,\widetilde g^{+},\widetilde g^{-},\lambda)=(0,0.35,0.5,14\widetilde\sigma\sqrt{\Delta_{n}}, -14\widetilde\sigma\sqrt{\Delta_{n}}, 67)$.
\end{itemize}

We simulate 5-minute observations for both $Y$ and $X$ over a month (21 days) for each replication. %\footnote{The simulation results for 10-minute observations are reported in Appendix \ref{AP:Supp-simulation}. } 
\cref{Fig:simulation-example_Paths} illustrates a simulated path of $Y$ together with the three relevant factors, with all initial values set to zero. Most of our parameter choices are calibrated to be consistent with \citet{ait2020high} in their benchmark case with $p=q=3$. The only deviation is that we set relatively high long-run means for the coefficients of relevant factors in order to avoid a ``weak factor'' configuration, which would make true signals difficult to distinguish from noise and thus complicate the finite-sample assessment of model selection performance. Under this calibration, the three relevant factors contribute approximately 18.7\%, 7.6\% and 3.7\%, respectively, to the quadratic variation of $Y$.

\begin{figure}[!ht]
\centering
\addtolength{\leftskip} {-2cm}
\addtolength{\rightskip}{-2cm}
\includegraphics[width=\textwidth]{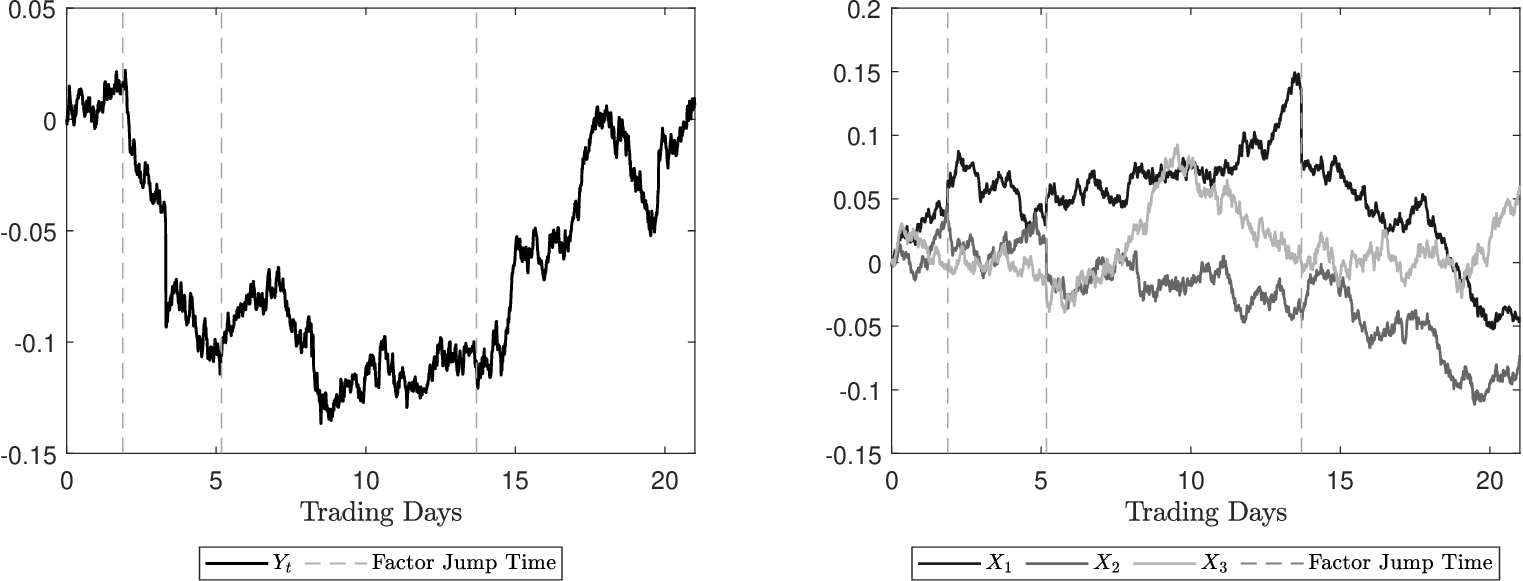}
\caption{Simulated trajectory of log-prices from the continuous-time three-factor model.}
\label{Fig:simulation-example_Paths}
\end{figure}

Following the tuning parameter selection in \cref{Sec:implementation}, for each Monte Carlo experiment, we use the first 50 simulated paths to perform 5-fold cross-validation, and then fix the parameters at the median selected values for all replications. In \cref{Sec:simulation_estimation,Sec:simulation_selection}, we report the finite-sample performance of the $\widehat{I\beta}$ estimator for the relevant factors and the model selection results, for $p$ varying from 3 to 500, which covers both the three-factor benchmark case of \citet{ait2020high} and a high-dimensional ``factor zoo'' setting.\footnote{We also consider designs with more pronounced cross-factor dependence, which is common in high-dimensional settings; see Appendix \ref{AP:Supp-simulation} for further discussion.}

%To comprehensively verify the theoretical results and further investigate the practical reliability of our estimation method, we also consider the cases with more pronounced cross-factor correlation, which is common in high-dimensional settings; see Appendix \ref{AP:Supp-simulation} for further discussion.

\subsection{Coefficient Estimation}
\label{Sec:simulation_estimation}

\cref{Table:simulation-IBeta} reports the sample bias, standard deviation and root mean square error (RMSE) of the $\widehat{I\beta}$ estimators for relevant factors.  Panel A shows that, in the three-factor benchmark case $p=q=3$, the unpenalized spline-based estimator delivers very accurate estimates for $I\beta$, with small bias and low RMSE, and performs comparably to the estimators of \citet{ait2020high}. The penalized spline-based estimator yields nearly the same results as the unpenalized one, since all three factors are truly relevant and, with our data-driven truncation threshold $\tau_{n}$, the TLP is effectively inactive in this case. Overall, in this low-dimensional case the spline-based estimator matches the performance of the local OLS benchmark, in line with our asymptotic results for fixed $p=q$.

\begin{table}%[h!]
	\centering
	\onehalfspacing
	\addtolength{\leftskip} {-2cm}
	\addtolength{\rightskip}{-2cm}
	\setlength{\tabcolsep}{6pt}
	\begin{threeparttable}
		\caption{Simulation results for model estimation}
		\scriptsize
		\begin{tabular}{lllrrrrrrrrrrrr}
			\hline\hline
			\multicolumn{15}{l}{Panel A: $p=3$}\\	
			& & & & \multicolumn{3}{c}{$\widehat{I\beta}_{1,T}$} & & \multicolumn{3}{c}{$\widehat{I\beta}_{2,T}$} & & \multicolumn{3}{c}{$\widehat{I\beta}_{3,T}$}\\
			\cline{5-7}\cline{9-11}\cline{13-15}
			Estimator & & Tuning & & Bias & Stdev & RMSE & & Bias & Stdev & RMSE & & Bias & Stdev & RMSE\\
			\hline
			Spline      & &                      & & -0.249 & 2.829 & 2.839 & &  0.161 & 2.527 & 2.531 & & -0.112 & 2.345 & 2.347 \\
			Spline TLP  & & $\alpha_{\tau}=0.05$ & & -0.250 & 2.829 & 2.839 & &  0.162 & 2.527 & 2.531 & & -0.112 & 2.345 & 2.347 \\
			Spline TLP  & & $\alpha_{\tau}=0.01$ & & -0.250 & 2.829 & 2.839 & &  0.162 & 2.527 & 2.531 & & -0.112 & 2.345 & 2.347 \\
			AKX         & & $k_{n}=78$           & & -0.257 & 2.900 & 2.910 & &  0.170 & 2.622 & 2.626 & & -0.103 & 2.447 & 2.448 \\
			AKX         & & $k_{n}=91$           & & -0.279 & 2.920 & 2.931 & &  0.165 & 2.619 & 2.623 & & -0.127 & 2.486 & 2.488 \\
			AKX         & & $k_{n}=117$          & & -0.314 & 2.929 & 2.945 & &  0.177 & 2.631 & 2.636 & & -0.091 & 2.497 & 2.498 \\
			\hline
			\multicolumn{15}{l}{Panel B: $p=10$}\\	
			& & & & \multicolumn{3}{c}{$\widehat{I\beta}_{1,T}$} & & \multicolumn{3}{c}{$\widehat{I\beta}_{2,T}$} & & \multicolumn{3}{c}{$\widehat{I\beta}_{3,T}$}\\
			\cline{5-7}\cline{9-11}\cline{13-15}
			Estimator & & Tuning & & Bias & Stdev & RMSE & & Bias & Stdev & RMSE & & Bias & Stdev & RMSE\\
			\hline
			Spline      & &                      & & -0.369 & 3.085 & 3.105 & &  0.153 & 2.825 & 2.828 & & -0.208 & 2.530 & 2.537 \\
			Spline TLP  & & $\alpha_{\tau}=0.05$ & & -0.378 & 3.050 & 3.072 & &  0.155 & 2.792 & 2.795 & & -0.194 & 2.464 & 2.471 \\
			Spline TLP  & & $\alpha_{\tau}=0.01$ & & -0.378 & 3.050 & 3.072 & &  0.155 & 2.791 & 2.794 & & -0.193 & 2.465 & 2.471 \\
			AKX         & & $k_{n}=78$           & & -0.342 & 3.318 & 3.334 & &  0.138 & 3.040 & 3.042 & & -0.238 & 2.754 & 2.763 \\
			AKX         & & $k_{n}=91$           & & -0.389 & 3.261 & 3.283 & &  0.158 & 3.080 & 3.082 & & -0.207 & 2.716 & 2.723 \\
			AKX         & & $k_{n}=117$          & & -0.368 & 3.246 & 3.265 & &  0.173 & 3.063 & 3.066 & & -0.238 & 2.705 & 2.714 \\
			\hline
			\multicolumn{15}{l}{Panel C: $p=50$}\\	
			& & & & \multicolumn{3}{c}{$\widehat{I\beta}_{1,T}$} & & \multicolumn{3}{c}{$\widehat{I\beta}_{2,T}$} & & \multicolumn{3}{c}{$\widehat{I\beta}_{3,T}$}\\
			\cline{5-7}\cline{9-11}\cline{13-15}
			Estimator & & Tuning & & Bias & Stdev & RMSE & & Bias & Stdev & RMSE & & Bias & Stdev & RMSE\\
			\hline
			Spline      & &                      & & -0.454 & 2.984 & 3.017 & &  0.128 & 3.196 & 3.197 & & -0.160 & 3.014 & 3.016 \\
			Spline TLP  & & $\alpha_{\tau}=0.05$ & & -0.467 & 2.755 & 2.793 & &  0.158 & 3.097 & 3.100 & & -0.331 & 3.200 & 3.215 \\
			Spline TLP  & & $\alpha_{\tau}=0.01$ & & -0.471 & 2.755 & 2.794 & &  0.140 & 3.085 & 3.086 & & -0.224 & 2.885 & 2.893 \\
			AKX         & & $k_{n}=78$           & & -0.293 & 4.870 & 4.877 & &  0.083 & 4.913 & 4.911 & & -0.440 & 4.817 & 4.835 \\
			AKX         & & $k_{n}=91$           & & -0.436 & 4.467 & 4.486 & &  0.164 & 4.558 & 4.559 & & -0.397 & 4.386 & 4.402 \\
			AKX         & & $k_{n}=117$          & & -0.435 & 4.023 & 4.044 & &  0.163 & 4.056 & 4.057 & & -0.154 & 3.873 & 3.874 \\
			\hline
			\multicolumn{15}{l}{Panel D: $p=100$}\\	
			& & & & \multicolumn{3}{c}{$\widehat{I\beta}_{1,T}$} & & \multicolumn{3}{c}{$\widehat{I\beta}_{2,T}$} & & \multicolumn{3}{c}{$\widehat{I\beta}_{3,T}$}\\
			\cline{5-7}\cline{9-11}\cline{13-15}
			Estimator & & Tuning & & Bias & Stdev & RMSE & & Bias & Stdev & RMSE & & Bias & Stdev & RMSE\\
			\hline
			Spline      & &                      & & -0.391 & 3.497 & 3.517 & &  0.424 & 3.128 & 3.155 & & -0.372 & 3.629 & 3.646 \\
			Spline TLP  & & $\alpha_{\tau}=0.05$ & & -0.347 & 3.026 & 3.044 & &  0.417 & 2.709 & 2.739 & & -0.432 & 3.410 & 3.436 \\
			Spline TLP  & & $\alpha_{\tau}=0.01$ & & -0.345 & 3.021 & 3.039 & &  0.414 & 2.713 & 2.743 & & -0.434 & 3.628 & 3.652 \\
			AKX         & & $k_{n}=78$           & &   --   &  --   &  --   & &   --   &  --   &  --   & &   --   &  --   &  --   \\
			AKX         & & $k_{n}=91$           & &   --   &  --   &  --   & &   --   &  --   &  --   & &   --   &  --   &  --   \\
			AKX         & & $k_{n}=117$          & & -0.304 & 7.659 & 7.662 & &  0.340 & 7.643 & 7.646 & &  0.138 & 8.965 & 8.962 \\
			\hline
			\multicolumn{15}{l}{Panel E: $p=500$}\\	
			& & & & \multicolumn{3}{c}{$\widehat{I\beta}_{1,T}$} & & \multicolumn{3}{c}{$\widehat{I\beta}_{2,T}$} & & \multicolumn{3}{c}{$\widehat{I\beta}_{3,T}$}\\
			\cline{5-7}\cline{9-11}\cline{13-15}
			Estimator & & Tuning & & Bias & Stdev & RMSE & & Bias & Stdev & RMSE & & Bias & Stdev & RMSE\\
			\hline
			Spline      & &                      & & -52.971 & 1.479 & 52.992 & & 40.906 & 1.502 & 40.933 & & -22.523 & 1.468 & 22.571 \\
			Spline TLP  & & $\alpha_{\tau}=0.05$ & &  -0.230 & 3.001 & 3.009 & &  0.388 & 3.128 & 3.150 & &  -0.296 & 3.778 & 3.788 \\
			Spline TLP  & & $\alpha_{\tau}=0.01$ & &  -0.237 & 3.003 & 3.011 & &  0.359 & 3.114 & 3.133 & &  -0.057 & 2.961 & 2.960 \\
			AKX         & & $k_{n}=78$           & &   --   &  --   &  --   & &   --   &  --   &  --   & &   --   &  --   &  --   \\
			AKX         & & $k_{n}=91$           & &   --   &  --   &  --   & &   --   &  --   &  --   & &   --   &  --   &  --   \\
			AKX         & & $k_{n}=117$          & &   --   &  --   &  --   & &   --   &  --   &  --   & &   --   &  --   &  --   \\
			\hline\hline
		\end{tabular} 
		\label{Table:simulation-IBeta}
		All these quantities are multiplied by 100. The cutoff level in the TLP function is $\tau_{n}=\alpha_{\tau}\sqrt{\text{MedRV}_{T}}$ for each simulated path of $Y$. The number of B-spline basis functions $K_{n}$ and the effective penalty level $\lambda_{n}/\tau_{n}$ are selected via 5-fold cross-validation. AKX stands for the estimator of \citet{ait2020high}, and $k_{n}$ is the size of local windows. For both the spline-based and AKX estimators, the truncation thresholds for the increments of $Y$ and each component of $X$ are given by $3\Delta_{n}^{0.47}\sqrt{\text{MedRV}_{T}}$.  All results are based on 1,000 Monte Carlo replications.	
	\end{threeparttable}
\end{table}

As the number of candidate factors increases from $p=3$ to 100, the RMSEs of the unpenalized spline-based estimator remain stable and rise only slightly across Panels A–D. Introducing the penalty has very little impact on the estimation error of $\widehat{I\beta}$, which indicates that the TLP regularizes the contribution of irrelevant factors without shrinking the true coefficients of the relevant ones, in line with the main intuition behind TLP and the oracle property in \cref{Th:oracle_global}. In particular, for $p=100$, the \citet{ait2020high} estimator uses a local window size smaller than $p$, which renders the local design matrix ill-conditioned and makes estimation infeasible. By contrast, our spline-based estimator is obtained from a single ``global'' regression of all truncated increments of $Y$ on those of $X$ over $[0,T]$, and thus it remains numerically stable and more robust than the local estimation procedure when $p$ is large.

The high-dimensional case with $p=500$ (Panel E) further highlights the role of regularization. In this case, the unpenalized spline-based estimator breaks down: the biases of $\widehat{I\beta}$ exceed those in the $p=100$ case by more than two orders of magnitude, and the RMSEs are of the same order, which indicates that OLS is no longer reliable when the number of nuisance factors is very large such that the effective number of coefficients in $\bm{\gamma}$, $pK_{n}$, exceeds the sample size $n$. In sharp contrast, the penalized spline-based estimators retain RMSEs of about 0.03 and small biases, comparable to those observed for $p=10$, 50 and 100. In conclusion, the Monte Carlo evidence in \cref{Table:simulation-IBeta} shows that our spline-based estimator matches the efficiency of the local OLS method in small dimensions but is substantially more robust when the cross-section of candidate factors becomes large, and that the specific nonconvex penalty is crucial to maintaining good finite-sample performance when $p$ is very large.

\subsection{Model Selection}
\label{Sec:simulation_selection}

In this section, we examine the finite-sample model selection performance of the penalized spline-based estimator. Consistent with \cref{Sec:selection}, “selection” refers to a factor whose corresponding block of B-spline coefficient estimates $\widehat{\gamma}_{j}=(\widehat{\gamma}_{j}^{(1)},\dots,\widehat{\gamma}_{j}^{(K_n)})^\top$ is nonzero. 

To illustrate the impact of the tuning parameters, we first consider the case $p=10$ and $q=3$. \cref{Fig:simulation-TDR_FDR_10_3} reports the true discovery rate (TDR), defined as the fraction of truly relevant factors that are correctly selected, and the false discovery rate (FDR), defined as the fraction of irrelevant factors that are incorrectly included in the active set, over a two-dimensional grid of tuning parameters, i.e., the number of B-spline basis functions $K_{n}$ and the effective penalty level $\lambda_{n}/\tau_{n}$. For small values of $\lambda_{n}/\tau_{n}$, the penalization is weak and both the TDR and FDR are high. As $\lambda_{n}/\tau_{n}$ increases, the FDR decreases sharply and the procedure becomes increasingly conservative, until, for very large penalties, the TDR starts to deteriorate because some relevant factors are also shrunk out. Hence, the ``good'' region of the grid is characterized by combinations of $(\lambda_{n}/\tau_{n},K_{n})$ for which the TDR is essentially one and the FDR is close to zero. For this example, the values of $K_{n}=4$ and $\lambda_{n}/\tau_{n}=0.07$ selected by 5-fold cross-validation fall into this region. %Moreover, the cross-validated choice of $K_{n}$ is approximately of the order $n^{1/5}$ for cubic splines, which is consistent with some rule-of-thumb choices for the B-spline basis dimension in the literature of varying-coefficient models (see, e.g., \citealp{huang2004polynomial,zhang2015varying}, albeit in some different contexts).

\begin{figure}[!ht]
\centering
\addtolength{\leftskip} {-2cm}
\addtolength{\rightskip}{-2cm}
\includegraphics[width=\textwidth]{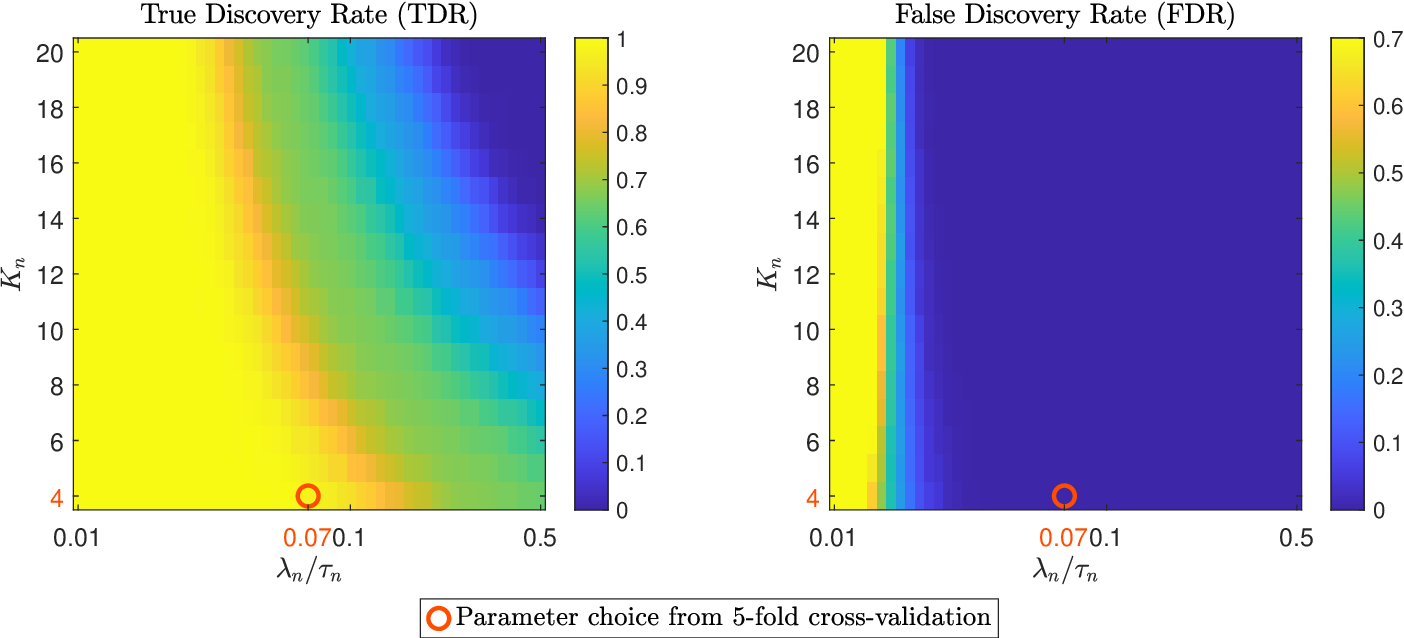}
\caption{True discovery rates (TDR, fraction of relevant factors that are correctly selected) and false discovery rates (FDR, fraction of irrelevant factors that are wrongly selected) for $p=10$ regressors with $q=3$ relevant ones, across different combinations of (i) the number of B-spline basis functions $K_{n}$, and (ii) the effective penalty level $\lambda_{n}/\tau_{n}$. The cutoff $\tau_{n}$ in the TLP function is $\tau_{n}=0.05\sqrt{\text{MedRV}_{T}}$ for each simulated path of $Y$, and each rate is computed from 1,000 Monte Carlo replications. }
\label{Fig:simulation-TDR_FDR_10_3}
\end{figure}

\cref{Fig:simulation-example_10_3} reports the empirical selection frequencies of each factor in the case with $p=10$ and $q=3$, with two choices of the truncation constant $\alpha_{\tau}=0.05$ and 0.01, and corresponding cross-validated $K_{n}$ and $\lambda_{n}/\tau_{n}$. The three relevant factors are selected in essentially all replications, while the irrelevant factors are selected only rarely when $\alpha_{\tau}=0.05$ and nearly never when $\alpha_{\tau}=0.01$. This pattern is fully in line with our model selection consistency in \cref{Th:oracle_local,Th:oracle_global}.

\begin{figure}[!ht]
\centering
\addtolength{\leftskip} {-2cm}
\addtolength{\rightskip}{-2cm}
\includegraphics[width=\textwidth]{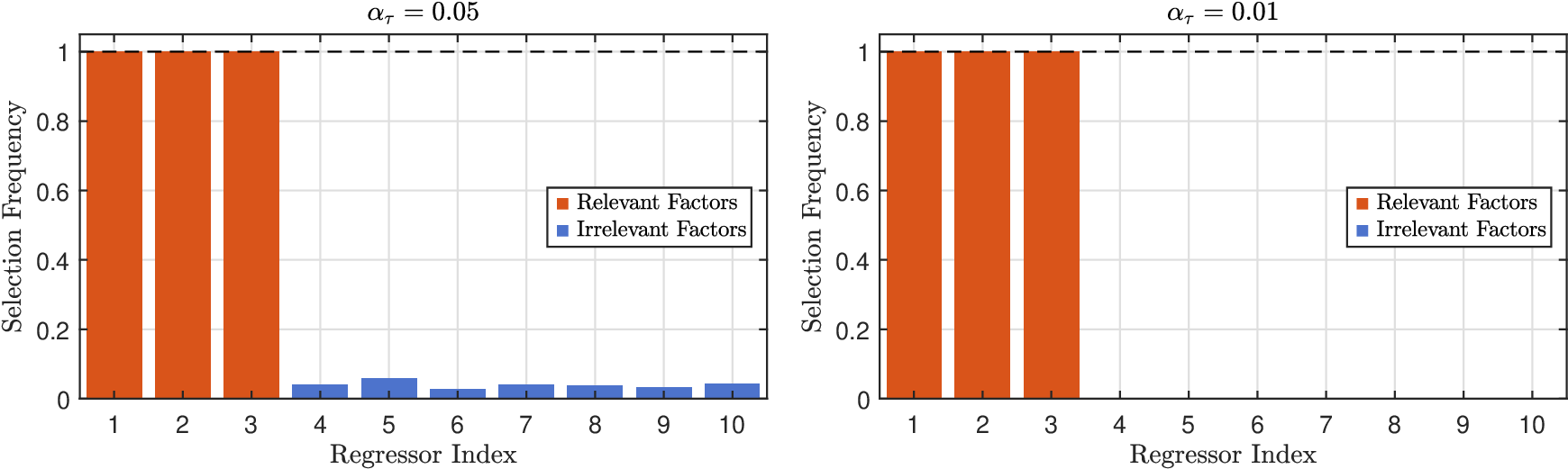}
\caption{Model selection results for $p=10$ regressors with $q=3$ relevant ones. The cutoff level in the TLP function is $\tau_{n}=\alpha_{\tau}\sqrt{\text{MedRV}_{T}}$ for each simulated path of $Y$. The number of B-spline basis functions $K_{n}$ and the effective penalty level $\lambda_{n}/\tau_{n}$ are selected via 5-fold cross-validation. All results are based on 1,000 Monte Carlo replications.}
\label{Fig:simulation-example_10_3}
\end{figure}

\cref{Table:simulation-selection} summarizes the model selection results across all cases considered in the previous section, which reports the average selection frequencies for relevant and irrelevant factors, and the frequencies of exactly correct model specification, i.e., all three relevant factors selected and all irrelevant ones excluded. For all cases and both values of $\alpha_{\tau}$, the (near-)unity selection frequencies for the relevant factors and the near-zero frequencies for the irrelevant ones, jointly with the estimation results in \cref{Table:simulation-IBeta}, indicate that the proposed penalized spline-based estimator simultaneously delivers accurate coefficient estimation and reliable recovery of sparse factor structure, which provides a direct finite-sample illustration of the oracle property established in \cref{Th:oracle_global}.

\begin{table}[h!]
	\centering
	\onehalfspacing
	\addtolength{\leftskip} {-2cm}
	\addtolength{\rightskip}{-2cm}
	\setlength{\tabcolsep}{10pt}
	\begin{threeparttable}
		\caption{Simulation results for model selection}
		\scriptsize 
		\begin{tabularx}{0.9\textwidth}{Xc *{3}{Y} c *{3}{Y}}
			\hline\hline
			& & \multicolumn{3}{c}{$\alpha_{\tau}=0.05$} & & \multicolumn{3}{c}{$\alpha_{\tau}=0.01$}\\
			\cline{3-5} \cline{7-9}
			& & Relevant & Irrelevant & Correct & & Relevant & Irrelevant & Correct\\
			\hline
			$p=3$   & & 1.000 & --    & 1.000 & & 1.000 & --    & 1.000 \\
			$p=10$  & & 1.000 & 0.039 & 0.759 & & 1.000 & 0.000 & 0.998\\
			$p=50$  & & 1.000 & 0.009 & 0.666  & & 1.000 & 0.000 & 1.000\\
			$p=100$ & & 1.000 & 0.083 & 0.006  & & 1.000 & 0.000 & 1.000\\
			$p=500$ & & 1.000 & 0.006 & 0.063 & & 1.000 & 0.000 & 0.952\\
			\hline\hline
		\end{tabularx} 
		\label{Table:simulation-selection}
		Average selection frequencies for the relevant and irrelevant factors, and frequencies of correct model specification. The cutoff level in the TLP function is $\tau_{n}=\alpha_{\tau}\sqrt{\text{MedRV}_{T}}$ for each simulated path of $Y$. The number of B-spline basis functions $K_{n}$ and the effective penalty level $\lambda_{n}/\tau_{n}$ are selected via 5-fold cross-validation. All results are based on 1,000 Monte Carlo replications.
	\end{threeparttable}
\end{table}

\section{Empirical Analysis}
\label{Sec:empirical}

In this section, we apply the methodology developed in this paper to a large number of high-frequency factors and anomalies. We start with a standard six-factor analysis for some representative stock, and compare our high-frequency estimates with their low-frequency counterparts to illustrate the benefits of exploiting granular intraday information. Then we examine whether a relatively small number of high-frequency factors from the ``factor zoo'' can explain comovement across a large cross-section of liquid U.S.~stocks and industry portfolios, where we implement simultaneous coefficient estimation and model selection by employing our penalized spline-based estimation procedure.

\subsection{Data}
\label{Sec:empirical_data}

We use the high-frequency factor zoo dataset of \citet{aleti2023high}. This minute-level dataset comprises 272 high-frequency portfolios over 1996--2020: (i) the five \citet{fama2015five} (FF5) factors plus momentum, constructed following the original Fama-French methodology; (ii) 218 characteristic-sorted factor portfolios based on the predictor libraries of \citet{chen2022open} and \citet{jensen2023there}; and (iii) 48 Fama-French industry portfolios \citep{fama1997industry}. For our empirical analysis, we consider the period from January 2016 to December 2020, and adopt a 5-minute sampling frequency for all 224 factors, individual stocks, and industry portfolios.

For the individual stocks, we start with S\&P 100 constituents and retain only firms that are consistently included in the index throughout January 2016 to December 2020. As is standard in empirical research with the Trade and Quote (TAQ) data,\footnote{We use the SAS code from \citet{holden2014liquidity} to extract all tick-by-tick transaction records matched with relevant ask/bid quotes from the daily TAQ dataset of WRDS. } we use the standard data filters as in \cite{barndorff2009realized} to eliminate data errors, remove all transactions in the original record that are later corrected, canceled or otherwise invalidated. In addition, we remove all trading days with an early market closure, and restrict our sample to transactions between 9:30:00--16:00:00 Eastern Time (ET). More importantly, we keep only stocks whose 5-minute returns are almost always non-zero (fewer than 5\% zero returns in each year). After these screens, our sample retains 57 stocks, which span 9 sectors and 21 Fama-French industry groups.\footnote{Appendix \ref{AP:Supp-empirical} reports the list of individual stocks retained in our sample, together with their sector and industry group assignments, and the fraction of zero 5-minute returns in each year. Sector classifications are based on the GICS sector code from WRDS Compustat. The Fama-French 48 industry group for each firm is assigned based on its four-digit SIC code from Compustat and the SIC-to-industry definitions from Kenneth R.~French's data library. } In addition, we include the full set of 48 Fama-French industry portfolios from \citet{aleti2023high}, for a total of 105 test assets.

\subsection{Beta Estimation for Fama-French Factors}

We estimate the betas in a standard six-factor model (FF5 plus momentum) for Apple Inc.~(AAPL) with three different approaches: (i) our unpenalized spline-based estimator, (ii) the local OLS estimator of \citet{ait2020high}, and (iii) a low-frequency OLS estimator based on daily data. Following the Monte Carlo simulations in \cref{Sec:simulation}, we select the number of B-spline basis functions $K_{n}$ via 5-fold cross-validation. For the AKX estimator, we set the local window length to $k_n = 78$.

\cref{Fig:empirical_ff6} reports monthly beta estimates for each factor. For the two high-frequency methods, the plotted series are monthly $\widehat{I\beta}$'s computed from truncated 5-minute returns of both AAPL and the six factors. For the daily-data benchmark, the OLS beta is estimated from daily returns within each month and is therefore constant over that month. We find that the spline-based estimator and the local OLS estimator deliver nearly identical beta estimates across all six factors. Moreover, both high-frequency estimators produce markedly smoother and more stable beta dynamics than the low-frequency benchmark, which exhibits substantially larger month-to-month variation and pronounced spikes. These differences are consistent with the role of jumps. Since large price moves identified as jumps are truncated and do not contribute to $\widehat{I\beta}$, the high-frequency beta estimators primarily reflect the continuous covariation between the asset and the factors. In contrast, the daily OLS estimator is directly exposed to large daily returns when jumps occur, which can induce sizable distortions in monthly beta estimates. Overall, \cref{Fig:empirical_ff6} indicates that high-frequency beta estimators are more robust and deliver more stable beta dynamics, which is consistent with the empirical evidence in \citet{ait2020high}.

\begin{figure}[!ht]
\centering
\addtolength{\leftskip} {-2cm}
\addtolength{\rightskip}{-2cm}
\includegraphics[width=\textwidth]{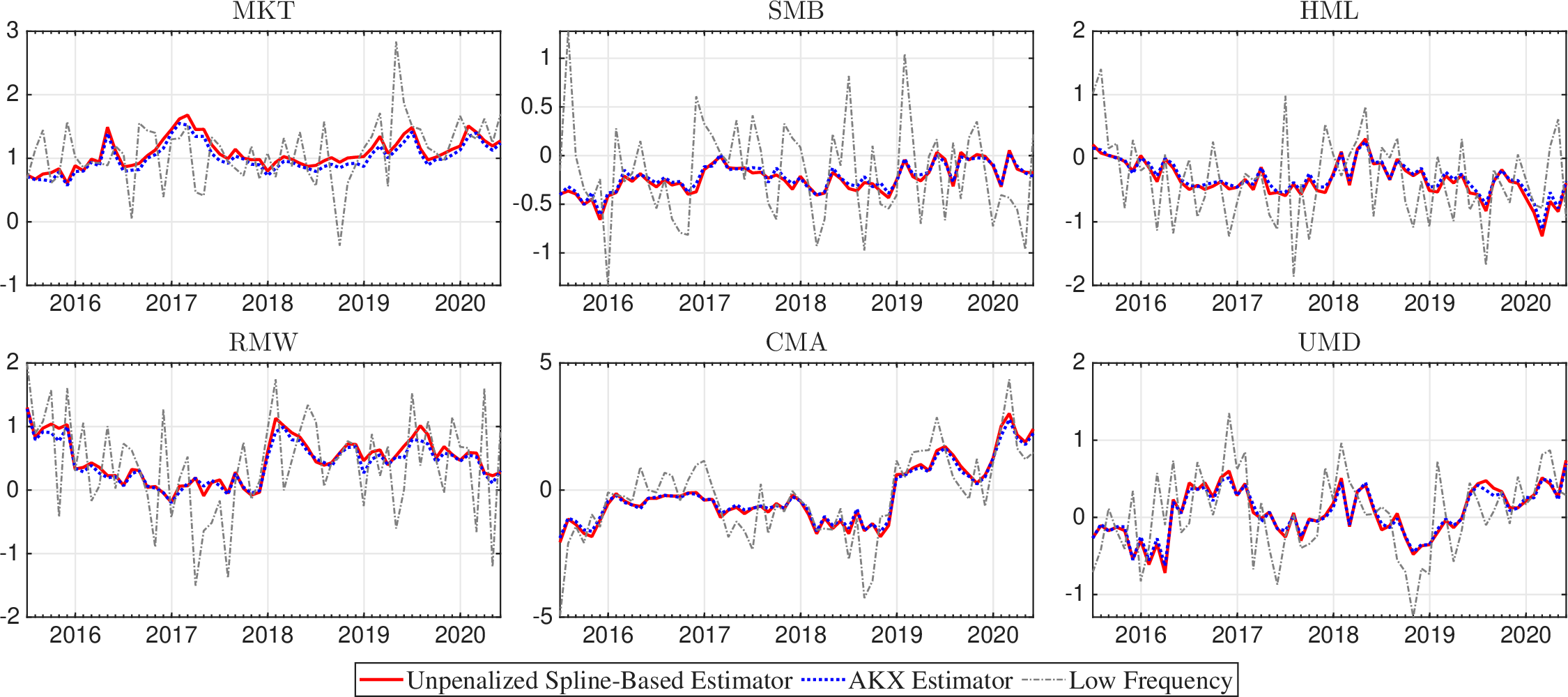}
\caption{Six-factor (FF5 + Momentum) beta estimates of AAPL from January 2016 to December 2020. We select the number of B-spline basis functions $K_{n}$ via 5-fold cross-validation. For the AKX estimator, we set the local window length to $k_n = 78$. For the two high-frequency estimators, the plotted series are monthly $\widehat{I\beta}$'s computed from truncated 5-minute returns of both AAPL and the six factors. For the low-frequency benchmark, the OLS beta is estimated from daily returns within each month and is therefore constant over that month. }
\label{Fig:empirical_ff6}
\end{figure}

\subsection{High-Frequency Factor Zoo Selection and Sparse Risk Structure}

We now apply our penalized spline-based estimator to the high-frequency factor zoo. The goal is to investigate whether the continuous component of each selected stock and industry portfolio can be explained by a relatively small set of factors with time-varying exposures. As summarized in \cref{Sec:empirical_data}, we work with the 224 high-frequency factors of \citet{aleti2023high} as candidate covariates, and study a cross-section of liquid S\&P 100 stocks together with the Fama-French 48 industry portfolios as test assets. Consistent with the Monte Carlo applications in \cref{Sec:simulation}, we select the tuning parameters, i.e., the number of B-spline basis functions $K_{n}$ and the effective penalty level $\lambda_{n}/\tau_{n}$, via 5-fold cross-validation using the one-standard-error rule. We set the TLP cutoff level $\tau_{n}=0.01\sqrt{\mathrm{MedRV}_{T}}$, where $\mathrm{MedRV}_{T}$ is the MedRV of the test asset computed over the corresponding estimation window $[0,T]$. 

\cref{Fig:empirical_selectionExample} reports the average number of selected factors for a set of representative stocks, i.e., Apple (AAPL), Microsoft (MSFT), JPMorgan Chase (JPM), Johnson \& Johnson (JNJ), Exxon Mobil (XOM), and their corresponding Fama-French industry portfolios, for each year from 2016 to 2020. The results suggest that the factor representations fitted by our penalized estimation procedure are generally parsimonious. For most of the representative stocks, the average number of selected factors falls between 15 and 25, whereas the corresponding industry portfolios tend to select fewer factors, with averages around 10. This difference is consistent with the additional diversification at the industry level, which attenuates idiosyncratic variation and concentrates the continuous comovement on a smaller set of common drivers. Furthermore, the results suggest time variation in sparsity: For AAPL and MSFT, the average number of selected factors declines over 2016--2020, while JPM, JNJ, and XOM exhibit an increase over the last two years of the sample. The most pronounced change occurs for the Petroleum and Natural Gas industry portfolio, for which our estimator selects around 30 factors on average in 2020, making it one of the industry portfolios with the largest selected sets in that year. The full results for all selected S\&P 100 stocks and industry portfolios are reported in Appendix \ref{AP:Supp-empirical}.

\begin{figure}[!ht]
\centering
\addtolength{\leftskip} {-2cm}
\addtolength{\rightskip}{-2cm}	
\includegraphics[width=\textwidth]{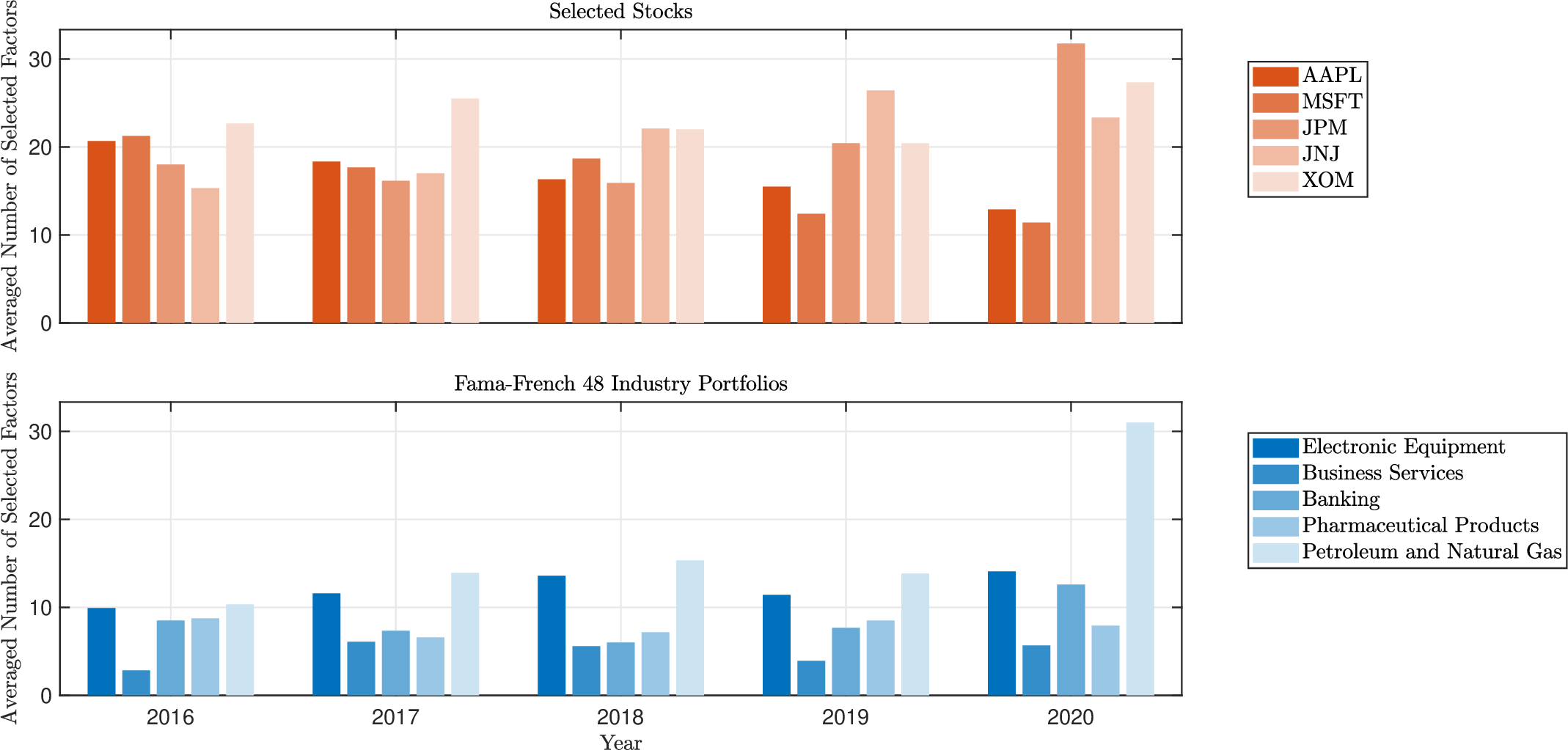}
\caption{Average number of selected factors per year for representative stocks and the corresponding Fama-French 48 industry portfolios. The candidate covariate set consists of the 224 high-frequency factors of \citet{aleti2023high}. The cutoff level in the TLP function is $\tau_{n}=0.01\sqrt{\text{MedRV}_{T}}$, where $\mathrm{MedRV}_{T}$ denotes the MedRV of the test asset computed over the corresponding estimation window $[0,T]$. The number of B-spline basis functions $K_{n}$ and the effective penalty level $\lambda_{n}/\tau_{n}$ are selected via 5-fold cross-validation using the one-standard-error rule. }
\label{Fig:empirical_selectionExample}
\end{figure}

Based on factor selections across the full cross-section of individual stocks and industry portfolios, \cref{Table:empirical_top20Factors} lists the twenty factors that are selected most frequently over the period from January 2016 to December 2020. The reported selection rate is the fraction of asset-month pairs for which the factor is selected by our penalized estimation procedure. To facilitate interpretation, we also report broad economic labels from \citet{chen2022open} and the theme classification from \citet{jensen2023there} for these frequently selected factors. To assess whether the dominant channels of intraday comovement differ between firm-level and industry-level returns, \cref{Table:empirical_top20Factors_stocks_industries} reports the same rankings separately for the selected S\&P 100 stocks and for the industry portfolios.

\begin{table}[h!]
	\centering
	\onehalfspacing
	\addtolength{\leftskip}{-2cm}
	\addtolength{\rightskip}{-2cm}
	\setlength{\tabcolsep}{4pt}
	\begin{threeparttable}
		\caption{Most frequently selected factors across all asset-month pairs}
		\scriptsize
		\begin{tabularx}{\textwidth}{l X l l r}
			\hline\hline
			Rank & Factor description & CZ classification & JKP classification & Selection rate\\
			\hline
			1  & Market & -- & -- & 0.9289 \\
			2  & R\&D-to-market & R\&D & Size & 0.1154 \\
			3  & Market correlation & -- & Seasonality & 0.1056 \\
			4  & Dollar trading volume & Volume & Size & 0.0979 \\
			5  & Pension funding status & Composite Accounting & -- & 0.0952 \\
			6  & Change in capital investment (industry-adjusted) & Investment Growth & -- & 0.0917 \\
			7  & Intrinsic value-to-market & Valuation & Profitability & 0.0884 \\
			8  & Organizational capital & R\&D & -- & 0.0884 \\
			9  & Order backlog & Sales Growth & -- & 0.0876 \\
			10 & Advertising expense & R\&D & -- & 0.0859 \\
			11 & Change in net working capital & Investment & -- & 0.0849 \\
			12 & Share turnover volatility & Liquidity & Profitability & 0.0846 \\
			13 & Cash flow volatility & -- & Low Risk & 0.0830 \\
			14 & Earnings persistence & -- & Seasonality & 0.0824 \\
			15 & Trading volume variance & Liquidity & Profitability & 0.0817 \\
			16 & Operating leverage & Other & Quality & 0.0816 \\
			17 & Mohanram G-score & Composite Accounting & -- & 0.0813 \\
			18 & Change in current operating working capital & -- & Accruals & 0.0810 \\
			19 & Years 11-15 lagged returns & Other & Seasonality & 0.0808 \\
			20 & Change in long-term investments & Investment & Seasonality & 0.0798 \\
			\hline\hline
		\end{tabularx}
		\label{Table:empirical_top20Factors}
		Twenty most frequently selected factors across the full cross-section of selected S\&P 100 stocks and Fama-French 48 industry portfolios from January 2016 to December 2020. The candidate covariate set consists of the 224 high-frequency factors of \citet{aleti2023high}. The cutoff level in the TLP function is $\tau_{n}=0.01\sqrt{\text{MedRV}_{T}}$, where $\mathrm{MedRV}_{T}$ denotes the MedRV of the test asset computed over the corresponding estimation window $[0,T]$. The number of B-spline basis functions $K_{n}$ and the effective penalty level $\lambda_{n}/\tau_{n}$ are selected via 5-fold cross-validation using the one-standard-error rule.
	\end{threeparttable}
\end{table}

\begin{table}[h!]
	\centering
	\onehalfspacing
	\addtolength{\leftskip}{-2cm}
	\addtolength{\rightskip}{-2cm}
	\setlength{\tabcolsep}{4pt}
	\begin{threeparttable}
		\caption{Most frequently selected factors for selected S\&P 100 stocks and Fama-French 48 industry portfolios}
		\scriptsize
		\begin{tabularx}{\textwidth}{l X l l r}
			\hline\hline
			\multicolumn{5}{l}{Panel A: S\&P 100 Stocks}\\
			Rank & Factor description & CZ classification & JKP classification & Selection rate\\
			\hline
			1  & Market & -- & -- & 0.9339 \\
			2  & R\&D-to-market & R\&D & Size & 0.1415 \\
			3  & Dollar trading volume & Volume & Size & 0.1395 \\
			4  & Intrinsic value-to-market & Valuation & Profitability & 0.1333 \\
			5  & Trading volume variance & Liquidity & Profitability & 0.1260 \\
			6  & Share turnover volatility & Liquidity & Profitability & 0.1257 \\
			7  & Market correlation & -- & Seasonality & 0.1208 \\
			8  & Change in capital investment (industry-adjusted) & Investment Growth & -- & 0.1181 \\
			9  & Change in net working capital & Investment & -- & 0.1152 \\
			10 & Change in long-term investments & Investment & Seasonality & 0.1140 \\
			11 & Earnings persistence & -- & Seasonality & 0.1120 \\
			12 & Pension funding status & Composite Accounting & -- & 0.1114 \\
			13 & Real dirty surplus & Composite Accounting & -- & 0.1070 \\
			14 & Mohanram G-score & Composite Accounting & -- & 0.1058 \\
			15 & Operating leverage & Other & Quality & 0.1020 \\
			16 & Years 6-10 lagged returns & Other & Seasonality & 0.1003 \\
			17 & Cash flow volatility & -- & Low Risk & 0.1000 \\
			18 & Years 11-15 lagged returns & Other & Seasonality & 0.0982 \\
			19 & Change sales minus change receivables & -- & Profit Growth & 0.0974 \\
			20 & Change in current operating working capital & -- & Accruals & 0.0968 \\
			\hline
			\multicolumn{5}{l}{Panel B: Fama-French 48 Industry Portfolios}\\
			Rank & Factor description & CZ classification & JKP classification & Selection rate\\
			\hline
			1  & Market & -- & -- & 0.9229 \\
			2  & Organizational capital & R\&D & -- & 0.0903 \\
			3  & Market correlation & -- & Seasonality & 0.0875 \\
			4  & Order backlog & Sales Growth & -- & 0.0858 \\
			5  & R\&D-to-market & R\&D & Size & 0.0844 \\
			6  & Advertising expense & R\&D & -- & 0.0809 \\
			7  & Dimson beta & -- & Low Risk & 0.0760 \\
			8  & Pension funding status & Composite Accounting & -- & 0.0760 \\
			9  & Quality minus junk & -- & Profit Growth & 0.0736 \\
			10  & Industry return of big firms & Lead Lag & -- & 0.0708 \\
			11  & Tail risk beta & Tail Risk & -- & 0.0694 \\
			12  & Change in order backlog & Accruals & -- & 0.0691 \\
			13  & Downside beta & -- & Low Risk & 0.0663 \\
			14  & Frazzini-Pedersen market beta & -- & Low Risk & 0.0632 \\
			15  & Cash flow volatility & -- & Low Risk & 0.0628 \\
			16  & Change in current operating working capital & -- & Accruals & 0.0622 \\
			17  & Change in capital investment (industry-adjusted) & Investment Growth & -- & 0.0604 \\
			18  & Years 11-15 lagged returns & Other & Seasonality & 0.0601 \\
			19  & Earnings surprise of big firms & Lead Lag & -- & 0.0594 \\
			20  & Industry concentration (assets) & Other & -- & 0.0573 \\
			\hline\hline
		\end{tabularx}
		\label{Table:empirical_top20Factors_stocks_industries}
		Twenty most frequently selected factors across the selected S\&P 100 stocks and Fama-French 48 industry portfolios, respectively, from January 2016 to December 2020. The candidate covariate set consists of the 224 high-frequency factors of \citet{aleti2023high}. The cutoff level in the TLP function is $\tau_{n}=0.01\sqrt{\text{MedRV}_{T}}$, where $\mathrm{MedRV}_{T}$ denotes the MedRV of the test asset computed over the corresponding estimation window $[0,T]$. The number of B-spline basis functions $K_{n}$ and the effective penalty level $\lambda_{n}/\tau_{n}$ are selected via 5-fold cross-validation using the one-standard-error rule.
	\end{threeparttable}
\end{table}

Both \cref{Table:empirical_top20Factors,Table:empirical_top20Factors_stocks_industries} show that the market factor is selected for the vast majority of asset-month pairs, which indicates that the market remains the most stable driver of continuous intraday comovement in our sample. Beyond the market factor, however, selection rates decline rapidly: the next-ranked factors are selected in fewer than $15\%$ of asset-months. This pattern suggests that there is no single non-market factor that plays a pervasive role across assets and time. Instead, the non-market component is accounted for by a set of factors that are heterogeneous and episodic, and is therefore well captured by a sparse representation estimated from more granular information within each month. The results also indicate that the most frequently selected factors cluster into a small number of economically coherent themes, including trading activity and liquidity, intangibles and innovation, investment, and balance-sheet-related characteristics. Moreover, the relative importance of these themes differs across firm-level and industry-level analyses. In particular, the stock-level rankings place more weight on trading activity and liquidity-related factors, together with valuation and investment-related measures, whereas the industry-level rankings tilt more toward broader risk-sensitivity and lead-lag-type predictors.

%Key points: market is always there; beyond that, continuous comovement is sparse but heterogeneous; and this heterogeneity is structured into a few economic themes that differ between firms and industries.

We further conduct a case study of total risk decomposition using two alternative factor sets: the standard six factors (FF5 plus momentum) and the twenty most frequently selected factors from \cref{Table:empirical_top20Factors}. The analysis parallels the decomposition in \citet{ait2020high}  (Fig.~6) with six factors to quantify how much of AAPL’s total intraday variation is attributable to systematic covariation with a given factor set versus the residual (idiosyncratic) component.

\cref{Fig:empirical_riskdecomp} illustrates the monthly risk decomposition for AAPL from January 2016 to December 2020. Specifically, the (annualized) quadratic variation of AAPL is estimated by realized volatility constructed from 5-minute returns within each month, and the integrated variance is estimated by the truncated realized volatility of \citet{mancini2009non} with the same truncation threshold used in our regression analysis. For each month and for each factor set, we estimate AAPL’s time-varying betas using our unpenalized spline-based regression on truncated 5-minute increments. Given the fitted beta paths, we calculate the residual increments as the difference between the truncated increments of AAPL and their corresponding predicted values. The unexplained integrated variance is then obtained as the annualized sum of squared residual increments, and the monthly $R^2$ is defined as the fraction of integrated variance explained by the regression model. 

\begin{figure}[!ht]
	\centering
	\addtolength{\leftskip} {-2cm}
	\addtolength{\rightskip}{-2cm}
	\includegraphics[width=\textwidth]{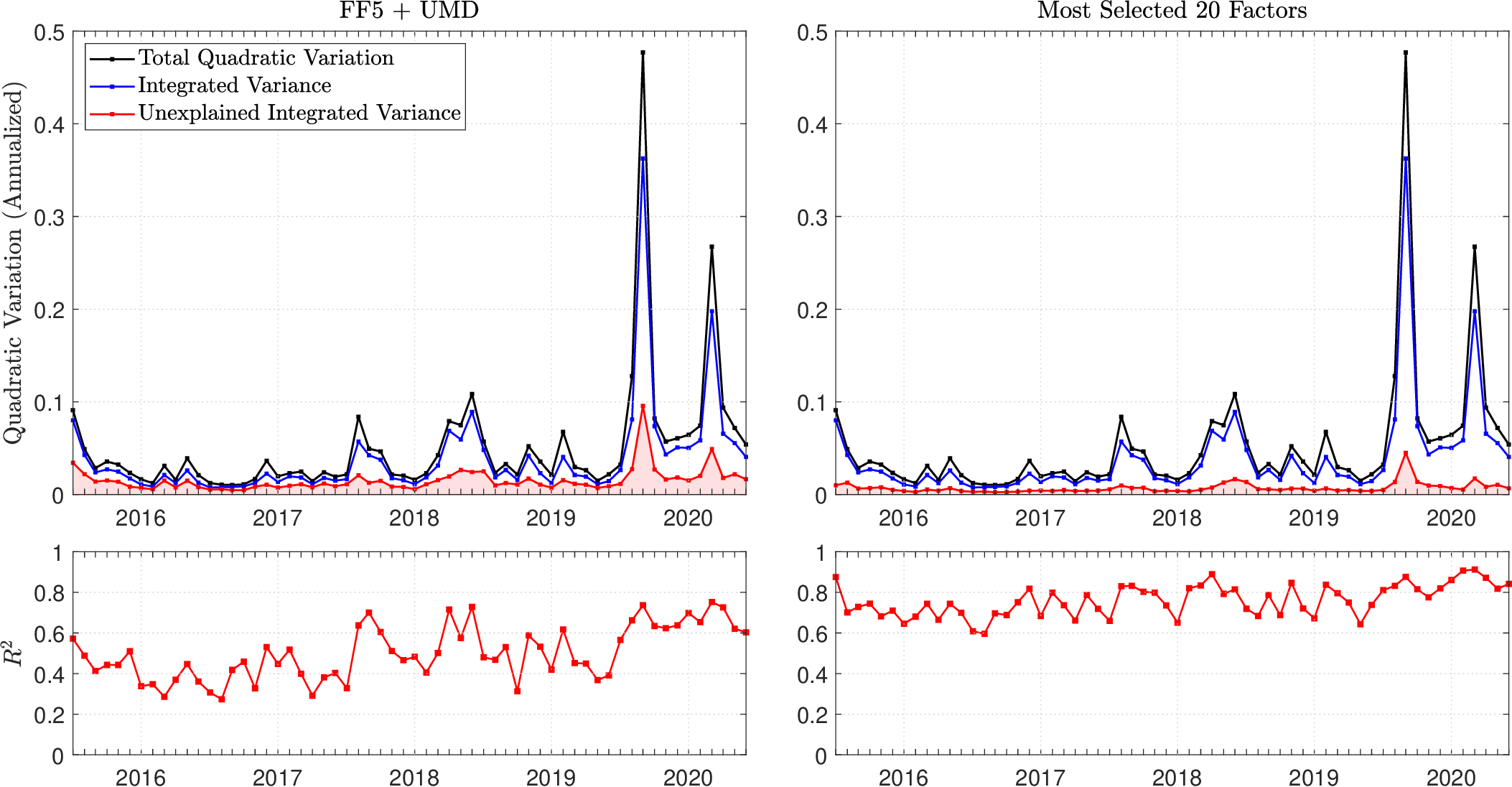}
	\caption{Monthly risk decomposition for AAPL under two alternative factor sets from January 2016 to December 2020. The left column uses the standard six factors (FF5 plus momentum), and the right column uses the twenty most frequently selected factors from \cref{Table:empirical_top20Factors}. The (annualized) total quadratic variation (black), integrated variance (blue), and unexplained integrated variance (red) are estimated by the realized volatility, truncated realized volatility, and the sum of squared residual increments, respectively. The monthly $R^2$ in the bottom panels is defined as the fraction of integrated variance explained by the model. }
	\label{Fig:empirical_riskdecomp}
\end{figure}

Relative to the six-factor benchmark, the specification based on the twenty most frequently selected factors provides a systematically tighter fit. The unexplained integrated variance is smaller, and the corresponding $R^2$ is higher in nearly every month. The results indicate that our cross-sectionally selected factors provide a more effective sparse representation for AAPL’s continuous intraday variation than the standard six-factor specification over this sample period.

\section{Conclusions}
\label{Sec:conclusions}

This paper develops a high-frequency penalized regression framework for It\^{o} semimartingales with time-varying coefficients. By approximating coefficient paths with polynomial splines, the approach provides a simple “global” alternative to local regressions over rolling windows, and makes high-dimensional model selection feasible with a group-wise $\ell_{1}$-penalty. We establish a comprehensive asymptotic theory for the spline-based estimator, and show that the penalized estimator attains the oracle property when there exists a large number of redundant factors. Monte Carlo experiments corroborate that the proposed procedure simultaneously delivers reliable coefficient estimation and model selection in finite samples. An empirical application to a large collection of high-frequency factors further indicates that continuous intraday comovement is well captured by a sparse and economically coherent set of factors. The selection results underscore the important role of the market factor and reveal a sparse theme structure: The frequently selected factors cluster into a small number of thematic classifications, such as liquidity, investment and valuation, and these themes vary systematically across stock- and industry-level analyses. %Finally, an important avenue for future research is to extend the proposed spline-based method to a long-span setting, ... ... 

%\clearpage

\bibliographystyle{apalike}
\bibliography{reference}

\clearpage

%\pagenumbering{arabic}
%Follow the requirements of the Journal of Econometrics
\numberwithin{assumption}{section}
\numberwithin{equation}{section}
\numberwithin{table}{section}
\numberwithin{figure}{section}
\numberwithin{lemma}{section}
\numberwithin{remark}{section}
\numberwithin{proposition}{section}
\numberwithin{corollary}{section}
\numberwithin{definition}{section}

\begin{appendices}

\section{Proofs}
\label{AP:Proofs}

\subsubsection*{Notation}

%We start with the notation. 
For a vector $x$, $\Vert x\Vert=\sqrt{\sum_{i}|x_{i}|^{2}}$ is the Euclidean norm ($\ell_{2}$-norm), $\Vert x\Vert_{1}=\sum_{i}|x_{i}|$ is the $\ell_{1}$-norm, and $\Vert x\Vert_{\infty}=\max_{i}|x_{i}|$ is the maximum norm ($\ell_{\infty}$-norm). For a matrix $\mathbf{A}$, $\Vert\mathbf{A}\Vert_{F}=\sqrt{\sum_{i}\sum_{j}|\mathbf{A}_{ij}|^{2}}$ is the Frobenius norm, and $\Vert\mathbf{A}\Vert=\sup_{\Vert x\Vert=1}\Vert\mathbf{A}x\Vert$ is the operator norm, which is the square root of the largest eigenvalue of the symmetric matrix $\mathbf{A}^{\top}\mathbf{A}$. For a symmetric positive definite matrix $\mathbf{A}$, $\lambda_{\text{max}}(\mathbf{A})$ and $\lambda_{\text{min}}(\mathbf{A})$ denote its largest and smallest eigenvalues, and it holds that $(\lambda_{\text{min}}(\mathbf{A}))^{-1}=\Vert \mathbf{A}^{-1}\Vert$. For a square integrable function $f$: $[0,T]\to\mathbb{R}^{p}$, $\Vert f\Vert_{L^{2}}=\sqrt{\int_{0}^{T}\Vert f_{s}\Vert^{2}ds}$ is the $L^{2}$-norm over $[0,T]$, where $\Vert\cdot\Vert$ inside the integral is the Euclidean norm on $\mathbb{R}^{p}$, and $\Vert f\Vert_{L^{2}(c)}=\langle f,f \rangle_{c}^{1/2}$, where $\langle f,g \rangle_{c} = \int_{0}^{T}f_{s}^{\top}c_{s}g_{s}ds$. %For two positive semidefinite matrices $\mathbf{A}$ and $\mathbf{B}$, the Loewner order $\mathbf{A}\preceq\mathbf{B}$ implies that $\mathbf{B}-\mathbf{A}$ is positive semidefinite. 

In all the sequel, the positive constants $\{K,K',...\}$ may change from line to line (and even within a line) but never depend on $n$ or other indices.

\subsubsection*{Localization}

With a standard localization procedure (see, e.g., Lemma 4.4.9, \citealp{jacod2012discretization}), we impose the following stronger assumption than Assumptions \ref{As:Ito} and \ref{As:Ito_beta} without loss of generality:
\begin{assumption}\label{As:Ito_Iocal}
We have Assumptions \ref{As:Ito} and \ref{As:Ito_beta} with $\tau_{1}\to\infty$. Moreover, the processes $\beta$, $\beta^{J}$, $b$, $\widetilde{b}$, $b'$, $b^{\beta}$, $\sigma$, $\widetilde{\sigma}$, $\sigma'$, $\sigma^{\beta}$, $c$, and the functions $\delta$, $\widetilde{\delta}$, $\delta'$, $\delta^{\beta}$ are bounded. 
\end{assumption}
Throughout the remainder of Appendix \ref{AP:Proofs}, references to any part of Assumptions \ref{As:Ito} and \ref{As:Ito_beta} are understood as its strengthened version under \cref{As:Ito_Iocal}.

\subsection{Properties of B-Splines}
\label{AP:Bspline}

Fix an integer degree $d\geq1$ and the partition $0=v_{0}<v_{1}<\dots,v_{N_{n}+1}=T$. Let $K_n=N_n+d+1$ and define an extended knot sequence
\begin{equation}
\bar{v}_{0}=\dots=\bar{v}_{d}=0,\qquad \bar{v}_{d+j}=v_{j} \text{ (for }j=1,\ldots,N_n), \qquad \bar{v}_{N_{n}+d+1}=\dots=\bar{v}_{K_{n}+d}=T.
\end{equation}
We define B-spline basis function $B_{t}^{(k)}\equiv B_{d,t}^{(k)}$ by the Cox-de Boor recursion: 
\begin{equation}
B_{0,t}^{(k)}=\begin{cases}
1 &\quad\text{if } t\in[\bar{v}_{k-1},\bar{v}_{k})\\
0 &\quad\text{otherwise}
\end{cases},\qquad\text{for }k=1,\dots,K_{n}+d,
\end{equation}
and $B_{r,t}^{(k)}$ for all $r=1,\dots,d$ and $k=1,\dots,K_{n}+d-r$ are calculated recursively from degree 0:
\begin{equation}
\label{Eq:Cox-deBoorRecursion}
B_{r,t}^{(k)}=\frac{t-\bar{v}_{k-1}}{\bar{v}_{k+r-1}-\bar{v}_{k-1}}B_{r-1,t}^{(k)}+\frac{\bar{v}_{k+r}-t}{\bar{v}_{k+r}-\bar{v}_{k}}B_{r-1,t}^{(k+1)},
\end{equation}
where
\begin{equation}
\begin{split}
\label{Eq:Cox-deBoorRecursion-boundries}
\frac{t-\bar{v}_{k-1}}{\bar{v}_{k+r-1}-\bar{v}_{k-1}}B_{r-1,t}^{(k)}=0,\qquad&\text{if }\bar{v}_{k+r-1}=\bar{v}_{k-1};\\
\frac{\bar{v}_{k+r}-t}{\bar{v}_{k+r}-\bar{v}_{k}}B_{r-1,t}^{(k+1)}=0,\qquad&\text{if }\bar{v}_{k+r}=\bar{v}_{k}.
\end{split}
\end{equation}
%For each $j=1,\dots,p$, $ \{B_{\cdot}^{(k)}\}_{k=1}^{K_{n}}$ is a basis for a linear space $\mathcal{G}_{n}$ of spline functions on $[0,T]$ with a particular degree and knot sequence. 
The B-spline basis functions $B_{t}^{(k)}$ have the following properties \citep{de1978practical}: 
\begin{itemize}
\item[(i)] $B_{t}^{(k)}\geq0$;
\item[(ii)] Partition of unity: $\sum_{k=1}^{K_{n}}B_{t}^{(k)}=1$;
%\item[(iii)] $\Vert B^{(k)}v^{(k)}\Vert_{L^{2}}^{2}=\int_{0}^{T}\left(\sum_{k=1}^{K_{n}}B_{s}^{(k)}v^{(k)}\right)^{2}ds\asymp K_{n}^{-1}\Vert v^{(k)}\Vert^{2}$;
\item[(iii)] Local support: %For any function $f$, define its support as the set of points where $f$ is non-zero, i.e., $\text{supp}(f)=\{x:f(x)\neq0\}$. For $B_{t}^{(k)}$ of degree $d\geq1$, $\text{supp}(B_{t}^{(k)})\subseteq$
For each $k=1,\dots,K_{n}$, $B_{t}^{(k)}$ is nonzero only on an interval covering at most $d+1$ consecutive knot intervals. For any function $f$, define its support as  $\text{supp}(f)=\{t:f(t)\neq0\}$, and denote by $|[a,b]|=b-a$ the interval length. Therefore, it holds that
\begin{equation}
|\text{supp}(B^{(k)})|\leq(d+1)H_{n}\asymp K_{n}^{-1}, \qquad\text{where }H_{n}=\max_{1\leq j\leq N_{n}+1}(v_{j}-v_{j-1}),
\end{equation}
by the quasi-uniformity of partitions. Equivalently, for any $t\in[0,T]$, at most $d+1$ basis functions in $\{B_{t}^{(k)}\}_{k=1}^{K_{n}}$ are nonzero, i.e., $|\{k:B_{t}^{(k)}>0\}|\leq d+1$.
\item[(iv)] (Lemma A.1, \citealp{huang2004polynomial}). Let $g_{t}=\sum_{k=1}^{K_{n}}B_{t}^{(k)}\gamma^{(k)}$ with $\gamma^{(k)}=(\gamma_{1}^{(k)},\dots,\gamma_{p}^{(k)})^{\top}\in\mathbb{R}^{p}$ and $\bm{\gamma}=((\gamma^{(1)})^{\top},\dots,(\gamma^{(K_{n})})^{\top})^{\top}\in\mathbb{R}^{pK_{n}}$. Then it holds that for any $\bm{\gamma}\in\mathbb{R}^{pK_{n}}$,
\begin{equation}
\Vert g\Vert_{L^{2}}^{2}\asymp\sum_{j=1}^{p}\Vert g_{j}\Vert_{L^{2}}^{2}\asymp\frac{\Vert\bm{\gamma}\Vert^{2}}{K_{n}}.
\end{equation}
\end{itemize}

%$B_{t}^{(k)}\geq0$, $\sum_{k=1}^{K_{n}}B_{t}^{(k)}=1$ (``partition of unity''), and for any $v_{j}^{(k)}\in\mathbb{R}$,
%\begin{equation}
%\label{Eq:B_spline_bounds}
%\Vert B^{(k)}v_{j}^{(k)}\Vert_{L^{2}}^{2}=\int_{0}^{T}\left(\sum_{k=1}^{K_{n}}B_{s}^{(k)}v_{j}^{(k)}\right)^{2}ds\asymp\frac{\Vert v_{j}^{(k)}\Vert^{2}}{K_{n}}.
%\end{equation}
%\begin{equation}
%\frac{K}{K_{n}}\sum_{k=1}^{K_{n}}(\gamma_{j}^{(k)})^{2}\leq\int_{0}^{T}\left(\sum_{k=1}^{K_{n}}\gamma_{j}^{(k)}B_{l,s}^{(k)}\right)^{2}ds\leq\frac{K'}{K_{n}}\sum_{k=1}^{K_{n}}(\gamma_{j}^{(k)})^{2}.
%\end{equation}

\subsection[Proof of Lemma 1]{Proof of \cref{lemma:approximationError}}
\label{AP:Proof_lemma1}

Under \cref{As:Ito} (iii), $c$ is bounded and nonsingular, i.e., there exist constants $0<\underline{\kappa}\leq \overline{\kappa}<\infty$ such that for all $x\in\mathbb{R}^p$ and any $t\in[0,T]$, it holds that
\begin{equation}
\label{Eq:c_equiv}
\underline{\kappa}\|x\|^2 \le x^\top c_{t} x \le \overline{\kappa}\|x\|^2,
\end{equation}
and thus implies the norm equivalence, i.e., for all measurable $f:[0,T]\to\mathbb{R}^p$,
\begin{equation}
\label{Eq:norm_equiv}
\sqrt{\underline{\kappa}}\Vert f\Vert_{L^{2}} \leq \Vert f\Vert_{L^{2}(c)} \leq \sqrt{\overline{\kappa}}\Vert f\Vert_{L^{2}}.
\end{equation}
Therefore, by definition of the $L^2(c)$-orthogonal projection, for any $g\in\mathbb{G}_{n}$,
\begin{equation}
\label{Eq:ArrpxError_L2norm_inequalities}
\Vert \beta-\Pi_n\beta \Vert_{L^{2}}\leq\frac{1}{\sqrt{\underline{\kappa}}}\Vert \beta-\Pi_n\beta \Vert_{L^{2}(c)}\leq \frac{1}{\sqrt{\underline{\kappa}}}\Vert \beta-g\Vert_{L^{2}(c)}\leq\sqrt{\frac{\overline{\kappa}}{\underline{\kappa}}}\Vert \beta-g\Vert_{L^{2}}.
\end{equation}
Hence it suffices to construct a single $\overline{\beta}\in\mathbb{G}_{n}$ that satisfies
\begin{equation}
\label{Eq:lemma1-0}
\Vert \beta-\overline{\beta}\Vert_{L^{2}}=O_{p}\left(\sqrt{\frac{\log K_{n}}{K_{n}}}\right).
\end{equation}
We consider a specific $\overline{\beta}$ as a local spline approximation, i.e. the quasi-interpolant (Chapter XII, \citealp{de1978practical}): For each $k=1,\dots,K_{n}$, choose an arbitrary $\tau^{(k)}\in\text{supp}(B^{(k)})$, and define
\begin{equation}
\label{Eq:quasiInterpolant}
\overline{\beta}_{t}=\sum_{k=1}^{K_{n}}\beta_{\tau^{(k)}}B_{t}^{(k)},
\end{equation}
which satisfies $\overline{\beta}\in\mathbb{G}_{n}$ by construction, and it remains to bound $\Vert \beta-\overline{\beta}\Vert_{L^{2}}$. The local support property of B-splines (Appendix \ref{AP:Bspline} (iii)) implies that, whenever $B_{t}^{(k)}>0$, $|t-\tau^{(k)}|\leq|\text{supp}(B^{(k)})|\leq (d+1)H_{n}\asymp K_{n}^{-1}$.

We denote by $\beta^{c}$ the continuous component of $\beta$, and $\beta^{j}=\beta-\beta^{c}$ collects the jumps:
\begin{equation}
\label{Eq:beta_decomp}
\beta_{t}^{c}=\beta_{0}+\int_{0}^{t}b_{s}^{\beta}ds+\int_{0}^{t}\sigma_{s}^{\beta}dW_{s}^{\beta},\qquad
\beta_{t}^{j}=\int_{0}^{t}\int_{\mathbb{R}^{p}}\delta^{\beta}(s,x)\mu^{\beta}(ds,dx).
\end{equation}
Let $\overline{\beta}^{c}$ and $\overline{\beta}^{j}$ denote the corresponding quasi-interpolants (on the same knots) of $\beta^{c}$ and $\beta^{j}$, and thus $\overline{\beta}=\overline{\beta}^{c}+\overline{\beta}^{j}$ by linearity of \cref{Eq:quasiInterpolant}. By the triangle inequality,
\begin{equation}
\label{Eq:lemma1-1}
\Vert \beta-\overline{\beta}\Vert_{L^{2}}\leq \Vert \beta^{c}-\overline{\beta}^{c}\Vert_{L^{2}} + \Vert \beta^{j}-\overline{\beta}^{j}\Vert_{L^{2}},
\end{equation}
and we bound the two terms separately. 

We start with $\Vert \beta^{c}-\overline{\beta}^{c}\Vert_{L^{2}}$. Define the modulus of continuity of $\beta^{c}$:
\begin{equation}
\omega_{\beta^{c}}(h)=\sup_{\substack{s,t\in[0,T]\\ |t-s|\le h}}\Vert\beta^c_t-\beta^c_s\Vert.
\end{equation}
By the convexity of the norm and $\sum_{k=1}^{K_{n}}B_{t}^{(k)}=1$, for any $t\in[0,T]$,
\begin{equation}
\label{Eq:convexity}
\Vert\beta_{t}^{c}-\overline{\beta}_{t}^{c}\Vert\leq\left\Vert\sum_{k=1}^{K_{n}}(\beta_{t}^{c}-\beta_{\tau^{(k)}}^{c})B_{t}^{(k)}\right\Vert\leq\sum_{k=1}^{K_{n}}\Vert\beta_{t}^{c}-\beta_{\tau^{(k)}}^{c}\Vert B_{t}^{(k)}\leq\omega_{\beta^{c}}((d+1)H_{n}),
\end{equation}
and thus
\begin{equation}
\label{Eq:L2norm_c}
\Vert \beta^{c}-\overline{\beta}^{c}\Vert_{L^{2}}\leq\sqrt{T}\sup_{t\in[0,T]}\Vert\beta_{t}^{c}-\overline{\beta}_{t}^{c}\Vert\leq\sqrt{T}\omega_{\beta^{c}}((d+1)H_{n}).
\end{equation}
To calculate the asymptotic order of $\omega_{\beta^{c}}(h)$, we decompose $\beta^{c}=A+M$ with
\begin{equation}
A_{t}=\beta_{0}+\int_{0}^{t}b_{s}^{\beta}ds,\qquad
M_{t}=\int_{0}^{t}\sigma_{s}^{\beta}dW_{s}^{\beta}.
\end{equation}
Under \cref{As:Ito_Iocal}, for $|t-s|\leq h$,
\begin{equation}
\omega_{A}(h)=\sup_{\substack{s,t\in[0,T]\\ |t-s|\le h}}\Vert A_{t}-A_{s}\Vert\leq\left\vert\int_{s}^{t}\Vert b_{u}^{\beta}\Vert du\right\vert\leq Kh.
\end{equation}
Moreover, $M$ is a continuous local martingale, which is a time-changed Brownian motion by the Dambis-Dubins-Schwarz theorem. By Lévy’s modulus of continuity for Brownian motion, it holds that when $h\to 0$,
\begin{equation}
\omega_{M}(h)=\sup_{\substack{s,t\in[0,T]\\ |t-s|\le h}}\Vert M_{t}-M_{s}\Vert = O_{p}\left(\sqrt{h\log \frac{1}{h}}\right).
\end{equation}
Therefore, we obtain
\begin{equation}
\label{Eq:modulus_beta_c}
\omega_{\beta^{c}}(h)\leq \omega_{A}(h)+\omega_{M}(h)=O_{p}\left(\sqrt{h\log \frac{1}{h}}\right).
\end{equation}
Taking $h=(d+1)H_{n}\asymp K_{n}^{-1}$, by \cref{Eq:L2norm_c}, it holds that
\begin{equation}
\label{Eq:lemma1-2}
\Vert \beta^{c}-\overline{\beta}^{c}\Vert_{L^{2}}=O_{p}\left(\sqrt{\frac{\log K_{n}}{K_{n}}}\right).
\end{equation}

Next, we verify that $\Vert \beta^{j}-\overline{\beta}^{j}\Vert_{L^{2}}=O_{p}(K_{n}^{-1/2})$. Define $V_{\beta^{j}}([a,b])$ the total variation of $\beta^{j}$ on the interval $[a,b]$, i.e.,
\begin{equation}
\label{Eq:TV_Jumps}
V_{\beta^{j}}([a,b])=\int_{a}^{b}\int_{\mathbb{R}^{p}}\Vert\delta^{\beta}(s,x)\Vert\mu^{\beta}(ds,dx).
\end{equation}
Then for any $s,t\in[a,b]$, $\Vert\beta_{t}^{j}-\beta_{s}^{j}\Vert\leq V_{\beta^{j}}([a,b])$. Therefore, whenever $B_{t}^{(k)}>0$, $\Vert\beta_{t}^{j}-\beta_{\tau^{(k)}}^{j}\Vert\leq V_{\beta^{j}}(\text{supp}(B^{(k)}))$. By the convexity of the norm, $\sum_{k=1}^{K_{n}}B_{t}^{(k)}=1$, for any $t\in[0,T]$,
\begin{equation}
\Vert\beta_{t}^{j}-\overline{\beta}_{t}^{j}\Vert\leq\sum_{k\in\{k:B_{t}^{(k)}>0\}}\Vert\beta_{t}^{j}-\beta_{\tau^{(k)}}^{j}\Vert B_{t}^{(k)}\leq\sum_{k\in\{k:B_{t}^{(k)}>0\}}B_{t}^{(k)}V_{\beta^{j}}(\text{supp}(B^{(k)})).
\end{equation}
By the Cauchy-Schwarz inequality and $|\{k:B_{t}^{(k)}>0\}|\leq d+1$,
\begin{equation}
\Vert\beta_{t}^{j}-\overline{\beta}_{t}^{j}\Vert^{2}\leq (d+1)\sum_{k\in\{k:B_{t}^{(k)}>0\}}V_{\beta^{j}}^{2}(\text{supp}(B^{(k)})).
\end{equation}
Therefore, we have
\begin{equation}
\begin{split}
\Vert \beta^{j}-\overline{\beta}^{j}\Vert_{L^{2}}^{2}&=\int_{0}^{T}\Vert\beta_{t}^{j}-\overline{\beta}_{t}^{j}\Vert^{2}dt\leq(d+1)\sum_{k=1}^{K_{n}}V_{\beta^{j}}^{2}(\text{supp}(B^{(k)}))\int_{0}^{T}\mathbbm{1}_{\{B_{t}^{(k)}>0\}}dt\\
&\leq(d+1)|\text{supp}(B^{(k)})|\sum_{k=1}^{K_{n}}V_{\beta^{j}}^{2}(\text{supp}(B^{(k)}))\leq(d+1)^{2}H_{n}\sum_{k=1}^{K_{n}}V_{\beta^{j}}^{2}(\text{supp}(B^{(k)})).
\end{split}
\end{equation}
Since supp($B^{(k)}$) covers at most $d+1$ knot intervals and each knot interval belongs to the supports of at most $d+1$ B-splines,
\begin{equation}
\begin{split}
\sum_{k=1}^{K_{n}}V_{\beta^{j}}^{2}(\text{supp}(B^{(k)}))&\leq (d+1)^{2}\sum_{j=1}^{N_{n}+1}V_{\beta^{j}}^{2}([v_{j-1},v_{j}])\\
&\leq (d+1)^{2}\left(\sum_{j=1}^{N_{n}+1}V_{\beta^{j}}([v_{j-1},v_{j}])\right)^{2}\leq (d+1)^{2}V_{\beta^{j}}^{2}([0,T]). 
\end{split}
\end{equation}
Therefore, under \cref{As:Ito_beta}, $V_{\beta^{j}}([0,T])=O_{p}(1)$, and
\begin{equation}
\label{Eq:lemma1-3}
\Vert \beta^{j}-\overline{\beta}^{j}\Vert_{L^{2}}\leq K\sqrt{H_{n}}V_{\beta^{j}}([0,T])=O_{p}(K_{n}^{-1/2}). 
\end{equation}
Combining \cref{Eq:lemma1-1,Eq:lemma1-2,Eq:lemma1-3} leads to \cref{Eq:lemma1-0}. This completes the proof.

\subsection[Proof of Theorem 1]{Proof of \cref{Th:consistency}}
\label{AP:consistency}

\subsubsection*{Identifiability}

We start with the identifiability of $\widehat{\bm{\gamma}}$ by verifying that $\mathbf{R}^{\top}\mathbf{R}$ is nonsingular. We first prove that the theoretical Gram matrix is positive definite, and then show that the theoretical and empirical Gram matrices are uniformly close over the estimation space, such that the empirical identifiability holds with $\widehat{\bm{\gamma}}$ uniquely defined; see Section 2 in \citet{huang2003local} for further details.

Define the matrix $\widetilde{\mathbf{R}}$ as the theoretical counterpart of $\mathbf{R}$:
\setlength{\arraycolsep}{2pt}
\begin{align}
\widetilde{\mathbf{R}}&=
\begin{pmatrix}
\int_{0}^{\Delta_{n}} B_{s}^{(1)}\,dX_{1,s}^{c} & \cdots & \int_{0}^{\Delta_{n}} B_{s}^{(K_{n})}\,dX_{1,s}^{c} & \cdots & \int_{0}^{\Delta_{n}} B_{s}^{(1)}\,dX_{p,s}^{c} & \cdots & \int_{0}^{\Delta_{n}} B_{s}^{(K_{n})}\,dX_{p,s}^{c} \\[1ex]
\int_{\Delta_{n}}^{2\Delta_{n}} B_{s}^{(1)}\,dX_{1,s}^{c} & \cdots & \int_{\Delta_{n}}^{2\Delta_{n}} B_{s}^{(K_{n})}\,dX_{1,s}^{c} & \cdots & \int_{\Delta_{n}}^{2\Delta_{n}} B_{s}^{(1)}\,dX_{p,s}^{c} & \cdots & \int_{\Delta_{n}}^{2\Delta_{n}} B_{s}^{(K_{n})}\,dX_{p,s}^{c} \\[1ex]
\vdots & \ddots & \vdots & \ddots & \vdots & \ddots & \vdots \\[1ex]
\int_{(n-1)\Delta_{n}}^{n\Delta_{n}} B_{s}^{(1)}\,dX_{1,s}^{c} & \cdots & \int_{(n-1)\Delta_{n}}^{n\Delta_{n}} B_{s}^{(K_{n})}\,dX_{1,s}^{c} & \cdots & \int_{(n-1)\Delta_{n}}^{n\Delta_{n}} B_{s}^{(1)}\,dX_{p,s}^{c} & \cdots & \int_{(n-1)\Delta_{n}}^{n\Delta_{n}} B_{s}^{(K_{n})}\,dX_{p,s}^{c}
\end{pmatrix}\notag\\
&=\left(\int_{(i-1)\Delta_{n}}^{i\Delta_{n}}\mathbf{B}_{s}^{\top}dX^{c}_{s}\right)_{1\leq i\leq n}^{\top}\in\mathbb{R}^{n \times pK_{n}},
\end{align}
and the corresponding Gram matrix is $\widetilde{\mathbf{R}}^{\top}\widetilde{\mathbf{R}}$. Let $g_{t}=(g_{1,t},\dots,g_{p,t})^{\top}$ with $g_{j,t}=\sum_{k=1}^{K_{n}}B_{t}^{(k)}v_{j}^{(k)}$ for any $v_{j}^{(k)}\in\mathbb{R}$ and $\bm{v}=(v_{1}^{(1)},\dots,v_{1}^{(K_{n})},\dots,v_{p}^{(1)},\dots,v_{p}^{(K_{n})})^{\top}\in\mathbb{R}^{pK_{n}}$, and $g\in\mathbb{G}_{n}$ by construction. By the properties of B-splines (Appendix \ref{AP:Bspline} (iv)) as well as bounded and nonsingular $c$, as $\Delta_{n}\to0$,
\begin{equation}
\label{Eq:Gram}
\bm{v}^{\top}\widetilde{\mathbf{R}}^{\top}\widetilde{\mathbf{R}}\bm{v}\overset{\mathbb{P}}{\longrightarrow}\int_{0}^{T}g_{s}^{\top}c_{s}g_{s}ds\asymp\sum_{j=1}^{p}\Vert g_{j}\Vert_{L^{2}}^{2}\asymp\frac{\Vert\bm{v}\Vert^{2}}{K_{n}},\qquad\text{where }\Vert\bm{v}\Vert^{2}=\sum_{j=1}^{p}\sum_{k=1}^{K_{n}}(v_{j}^{(k)})^{2},
\end{equation}
which implies that $\widetilde{\mathbf{R}}^{\top}\widetilde{\mathbf{R}}$ is positive definite. The standard law of large numbers (LLN) result for truncated functionals (Theorem 9.2.1, \citealp{jacod2012discretization}) implies that, with $u_{n}\asymp\Delta_{n}^{\varpi}$ for some $0<\varpi<1/2$,
\begin{equation}
\label{Eq:LLN}
\bm{v}^{\top}\mathbf{R}^{\top}\mathbf{R}\bm{v}\overset{\mathbb{P}}{\longrightarrow}\int_{0}^{T}g_{s}^{\top}c_{s}g_{s}ds,
\end{equation}
as $\Delta_{n}\to0$. By Lemma A.2 in \citet{huang2004polynomial}, it holds that
\begin{equation}
\label{Eq:closenessGrams}
\sup_{g\in\mathbb{G}_{n}}\left\vert\frac{\bm{v}^{\top}\mathbf{R}^{\top}\mathbf{R}\bm{v}}{\bm{v}^{\top}\widetilde{\mathbf{R}}^{\top}\widetilde{\mathbf{R}}\bm{v}}-1\right\vert=o_{p}(1),
\end{equation}
if $\Delta_{n}K_{n}\log K_{n}\to 0$. Therefore, the empirical Gram matrix $\mathbf{R}^{\top}\mathbf{R}$ is nonsingular and $\widehat{\bm{\gamma}}$ is uniquely defined with probability approaching one.

By \cref{Eq:Gram,Eq:closenessGrams}, for some $\epsilon_{n}=o_{p}(1)$ and $0\leq\epsilon_{n}<1$ with probability approaching one,
\begin{equation}
\label{Eq:quotient}
\bm{v}^{\top}\mathbf{R}^{\top}\mathbf{R}\bm{v}\geq(1-\epsilon_{n})\bm{v}^{\top}\widetilde{\mathbf{R}}^{\top}\widetilde{\mathbf{R}}\bm{v}\asymp(1-\epsilon_{n})\frac{\Vert\bm{v}\Vert^{2}}{K_{n}}.
\end{equation}
Let $\Vert\bm{v}\Vert\asymp1$, the eigenvalue bound of Rayleigh quotient implies that $\lambda_{\text{min}}(\mathbf{R}^{\top}\mathbf{R})\geq K/K_{n}$, and thus $\Vert(\mathbf{R}^{\top}\mathbf{R})^{-1}\Vert=(\lambda_{\text{min}}(\mathbf{R}^{\top}\mathbf{R}))^{-1}=O_{p}(K_{n})$. A similar but reverse inequality to \cref{Eq:quotient} indicates that $\lambda_{\text{max}}(\mathbf{R}^{\top}\mathbf{R})=O_{p}(K_{n}^{-1})$ and $\Vert\mathbf{R}\Vert=\sqrt{\lambda_{\text{max}}(\mathbf{R}^{\top}\mathbf{R})}=O_{p}(K_{n}^{-1/2})$. The probabilistic bounds for $\Vert(\mathbf{R}^{\top}\mathbf{R})^{-1}\Vert$ and $\Vert\mathbf{R}\Vert$ will be used in the following proofs.

\subsubsection*{Consistency of $\widehat{\beta}$}

The truncated increment of $Y$ in \cref{Eq:regression} over the $i$-th interval is
\begin{equation}
\label{Eq:decomp_Yi}
\begin{split}
\mathbf{Y}_{i}&=\left(\int_{(i-1)\Delta_{n}}^{i\Delta_{n}}\beta_{s-}^{\top}dX_{s}^{c}+\sum_{(i-1)\Delta_{n}\leq s\leq i\Delta_{n}}(\beta_{s-}^{J})^{\top}\Delta X_{s}+\Delta_{i}^{n}Z\right)\mathbbm{1}_{\{|\Delta_{i}^{n}Y|\leq u_{n}\}}\\
&=\underbrace{\left(\int_{(i-1)\Delta_{n}}^{i\Delta_{n}}\widetilde{\beta}_{s}^{\top}dX_{s}^{c}\right)\mathbbm{1}_{\{|\Delta_{i}^{n}Y|\leq u_{n}\}}}_{\widetilde{\mathbf{Y}}_{i}}+\underbrace{\left(\int_{(i-1)\Delta_{n}}^{i\Delta_{n}}e_{s}^{\top}dX_{s}^{c}\right)\mathbbm{1}_{\{|\Delta_{i}^{n}Y|\leq u_{n}\}}}_{\mathbf{U}_{i}}\\
&\qquad+\underbrace{\left(\sum_{(i-1)\Delta_{n}\leq s\leq i\Delta_{n}}(\beta_{s-}^{J})^{\top}\Delta X_{s}\right)\mathbbm{1}_{\{|\Delta_{i}^{n}Y|\leq u_{n}\}}}_{\mathbf{J}_{i}}+\underbrace{\Delta_{i}^{n}Z\mathbbm{1}_{\{|\Delta_{i}^{n}Y|\leq u_{n}\}}}_{\mathbf{Z}_{i}},
\end{split}
\end{equation}
where $\widetilde{\beta}=\mathbf{B}\bm{\gamma}$ is the spline approximation of $\beta$, and $e_{t}=\beta_{t-}-\widetilde{\beta}_{t}$. 

\begin{remark}[Predictability]\label{remark:Predictability}
Note that, since $H=\widetilde{\beta}$ or $e$ is not predictable, whenever we write the stochastic integral $\int H_{s}^{\top}dX_{s}^{c}$, it is understood as $\int (^{p}H_{s})^{\top}dX_{s}^{c}$, where $^{p}H$ is the predictable projection with $^{p}H_{\tau}=\mathbb{E}[H_{\tau}|\mathcal{F}_{\tau-}]$ for all predictable times $\tau$; see Theorem I.2.28 in \citet{jacod2013limit}. 

Since $\beta_{\tau-}$ is predictable under \cref{As:Ito} (i), by Proposition I.2.18 in \citet{jacod2013limit} and linearity of predictable projection, it holds that $^{p}\beta_{\tau-}=\,^{p}\widetilde{\beta}_{\tau}+\,^{p}e_{\tau}=\beta_{\tau-}$. In particular, we will apply the It{\^o} isometry and the Burkholder-Davis-Gundy (BDG) inequality to $^{p}H$. Moreover, by Jensen's inequality, $\Vert^{p}H_{s}\Vert^{2}\leq\,^{p}(\Vert H_{s}\Vert^{2})$, and thus whenever we use It{\^o} isometry,
\begin{equation}
\label{Eq:ItoIsometryRevised}
\mathbb{E}\left[\left(\int_{0}^{T}H_{s}^{\top}dX_{s}^{c}\right)^{2}\right]=\mathbb{E}\left[\int_{0}^{T}(^{p}H_{s})^{\top}c_{s}(^{p}H_{s})ds\right]\leq K\mathbb{E}\left[\int_{0}^{T}\Vert H_{s}\Vert^{2}ds\right]=K\mathbb{E}[\Vert H\Vert_{L^{2}}^{2}]. 
\end{equation}
by the boundedness of $c$ under \cref{As:Ito_Iocal}. 
\end{remark}

Then, by \cref{Eq:estimator_GammaVec},
\begin{equation}
\label{Eq:gammaBias}
\widehat{\bm{\gamma}}-\bm{\gamma}=(\mathbf{R}^{\top}\mathbf{R})^{-1}\mathbf{R}^{\top}((\widetilde{\mathbf{Y}}-\mathbf{R}\bm{\gamma})+\mathbf{U}+\mathbf{J}+\mathbf{Z}).
\end{equation}
To prove $\Vert\widehat{\beta}-\beta\Vert_{L^{2}}=o_{p}(1)$, we verify that the following four components are all $o_{p}(1)$:
\begin{equation}
\label{Eq:normDecomposition}
\begin{split}
\Vert\widehat{\beta}-\beta\Vert_{L^{2}}&\leq\Vert\widehat{\beta}-\widetilde{\beta}\Vert_{L^{2}}+\Vert e\Vert_{L^{2}}\\
&\leq\,\underbrace{\Vert\mathbf{B}(\mathbf{R}^{\top}\mathbf{R})^{-1}\mathbf{R}^{\top}(\widetilde{\mathbf{Y}}-\mathbf{R}\bm{\gamma})\Vert_{L^{2}}}_{A_{1}}\,+\,\underbrace{\Vert\mathbf{B}(\mathbf{R}^{\top}\mathbf{R})^{-1}\mathbf{R}^{\top}\mathbf{U}\Vert_{L^{2}}}_{A_{2}}\,\\
&\qquad+\,\underbrace{\Vert\mathbf{B}(\mathbf{R}^{\top}\mathbf{R})^{-1}\mathbf{R}^{\top}\mathbf{J}\Vert_{L^{2}}}_{A_{3}}\,+\,\underbrace{\Vert\mathbf{B}(\mathbf{R}^{\top}\mathbf{R})^{-1}\mathbf{R}^{\top}\mathbf{Z}\Vert_{L^{2}}}_{A_{4}}\,+\,\Vert e\Vert_{L^{2}},
\end{split}
\end{equation}
where $\Vert e\Vert_{L^{2}}=o_{p}(1)$ by \cref{lemma:approximationError}. 
%where $\Vert e\Vert_{L^{2}}=o_{p}(1)$ by Markov's inequality, i.e., for any $\varepsilon>0$,
%\begin{equation}
%\label{Eq:e_L2_op1}
%\mathbb{P}(\Vert e\Vert_{L^{2}}>\varepsilon)\leq\frac{\mathbb{E}[\Vert e\Vert_{L^{2}}]}{\varepsilon}=O(K_{n}^{-\alpha}).
%\end{equation}

We start with two lemmas. \cref{lemma:norm} is a result from the properties of B-spline basis functions and will be used to derive probabilistic bounds for the components listed above. \cref{lemma:truncationEvents} provides some results regarding the truncation technique of \citet{mancini2009non}, and will be used repeatedly in subsequent proofs. 

\begin{lemma}
\label{lemma:norm}
For any vector $\bm{r}\in\mathbb{R}^{n}$, $\Vert\mathbf{B}(\mathbf{R}^{\top}\mathbf{R})^{-1}\mathbf{R}^{\top}\bm{r}\Vert_{L^{2}}\leq K\sqrt{K_{n}}\Vert\mathbf{R}^{\top}\bm{r}\Vert\leq K'\Vert\bm{r}\Vert$. 
\end{lemma}

\begin{proof}[Proof of \cref{lemma:norm}]
It holds that, by Appendix \ref{AP:Bspline} (iv) and the Cauchy-Schwarz inequality,
\begin{equation}
\label{Eq:normInequality}
\begin{split}
\Vert\mathbf{B}(\mathbf{R}^{\top}\mathbf{R})^{-1}\mathbf{R}^{\top}\bm{r}\Vert_{L^{2}}^{2}&\asymp K_{n}^{-1}\Vert(\mathbf{R}^{\top}\mathbf{R})^{-1}\mathbf{R}^{\top}\bm{r}\Vert^{2}\leq KK_{n}^{-1}\Vert(\mathbf{R}^{\top}\mathbf{R})^{-1}\Vert^{2}\Vert\mathbf{R}^{\top}\bm{r}\Vert^{2}\\
&\leq K'K_{n}\Vert\mathbf{R}\Vert^{2}\Vert\bm{r}\Vert^{2}\leq K''\Vert\bm{r}\Vert^{2},
\end{split}
\end{equation}
where $\Vert(\mathbf{R}^{\top}\mathbf{R})^{-1}\Vert=O_{p}(K_{n})$ and $\Vert\mathbf{R}\Vert=O_{p}(K_{n}^{-1/2})$ from \cref{Eq:quotient}. 
\end{proof}

\begin{lemma}
\label{lemma:truncationEvents}
Let $u_{n}\asymp\Delta_{n}^{\varpi}$ for some $0<\varpi<1/2$. Then, as $\Delta_{n}\to0$,
\begin{itemize}
\item[(i)] $\mathbb{P}(\Vert\Delta_{i}^{n}X\Vert>u_{n})\leq K\Delta_{n}u_{n}^{-r}$, $\mathbb{P}(|\Delta_{i}^{n}Y|>u_{n})\leq K\Delta_{n}u_{n}^{-r}$;
%\item[(ii)] $\mathbb{E}[\Vert\Delta_{i}^{n}X-\Delta_{i}^{n}X^{c}\Vert^{2}\mathbbm{1}_{\{\Vert\Delta_{i}^{n}X\Vert\leq u_{n}\}}]\leq K\Delta_{n}^{2}u_{n}^{-2r}+K'\Delta_{n}u_{n}^{2-r}$.
\item[(ii)] $\Vert\Delta_{i}^{n}X-\Delta_{i}^{n}X^{c}\Vert^{2}\mathbbm{1}_{\{\Vert\Delta_{i}^{n}X\Vert\leq u_{n}\}}=O_{p}(\Delta_{n}u_{n}^{2-r})$.
\end{itemize}
\end{lemma}

\begin{proof}[Proof of \cref{lemma:truncationEvents}]

(i) Under \cref{As:Ito} (v), it holds that for $0<u<1$, 
\begin{equation}
\label{Eq:LevyTail}
%\int_{\{f_{m}>u\}}\lambda(dx)\leq Ku^{-r},
\int_{\{\Vert\delta\Vert>u\}}\lambda(dx)\leq u^{-r}\int_{\mathbb{R}^{p}}(\Vert\delta(s,x)\Vert^{r}\wedge1)\lambda(dx)\leq u^{-r}\int_{\mathbb{R}^{p}}f_{m}(x)\lambda(dx)\leq Ku^{-r},
\end{equation}
\begin{equation}
\label{Eq:smallJumps}
\int_{\{\Vert\delta\Vert\leq u\}}\Vert\delta(s,x)\Vert^2\lambda(dx)\leq u^{2-r}\int_{\mathbb{R}^{p}}(\Vert\delta(s,x)\Vert^{r}\wedge1)\lambda(dx)\leq Ku^{2-r}.
\end{equation}
%By \cref{Eq:LevyTail} and \cref{As:Ito} (v), we obtain the first-moment bound:
\begin{equation}
\label{Eq:smallJumpsl1}
\begin{split}
\int_{\{\Vert\delta\Vert\leq u\}}\Vert\delta(s,x)\Vert\lambda(dx)&=\int_{\mathbb{R}^p}\Vert\delta(s,x)\Vert\mathbbm{1}_{\{\Vert\delta(s,x)\Vert\leq u\}}\lambda(dx)=\int_{\mathbb{R}^p}\left(\int_{0}^{\Vert\delta(s,x)\Vert}\mathbbm{1}_{\{y\leq u\}}dy\right)\lambda(dx)\\
&=\int_{0}^{u}\left(\int_{\mathbb{R}^p}\mathbbm{1}_{\{\Vert\delta(s,x)\Vert>y\}}\lambda(dx)\right)dy-u\int_{\mathbb{R}^p}\mathbbm{1}_{\{\Vert\delta(s,x)\Vert>u\}}\lambda(dx)\\
&\leq\int_{0}^{u}\left(\int_{\mathbb{R}^p}\mathbbm{1}_{\{\Vert\delta(s,x)\Vert>y\}}\lambda(dx)\right)dy\\
&\leq\int_{0}^{u} y^{-r}\,dy \int_{\mathbb{R}^p}(\Vert\delta(s,x)\Vert^{r}\wedge1)\lambda(dx)\\
&\leq\left(\int_{\mathbb{R}^p} f_{m}(x)\lambda(dx)\right)\int_{0}^{u} y^{-r}dy\leq Ku^{1-r}.
\end{split}
\end{equation}
Define the count of large jumps and the sum of small jumps with the cutoff level $u_{n}$:
\begin{align}
N_{i}^{n}(u_{n})&=\int_{(i-1)\Delta_n}^{i\Delta_n}\int_{\mathbb{R}^p}\mathbbm{1}_{\{\|\delta(s,x)\|>u_{n}\}}\mu(ds,dx), \label{Eq:N}\\
S_{i}^{n}(u_{n})&=\int_{(i-1)\Delta_n}^{i\Delta_n}\int_{\{\|\delta\|\le u_{n}\}}\delta(s,x)\mu(ds,dx), \label{Eq:S}
\end{align}
and consider the following event decomposition:
\begin{equation}
E_{n}=\left\{\Vert \Delta_{i}^{n}X\Vert>u_{n},\,\Vert\Delta_{i}^{n}X^{c}\Vert \leq\frac{u_{n}}{2},\,N_{i}^{n}\Bigl(\frac{u_{n}}{2}\Bigr)=0\right\},\qquad
E_{n}^{\complement}=\{\Vert \Delta_{i}^{n}X\Vert >u_{n}\}\backslash E_{n}.
\end{equation}
By the triangle inequality, on $E_{n}$,
\begin{equation}
\left\Vert S_{i}^{n}\Bigl(\frac{u_{n}}{2}\Bigr)\right\Vert=\Vert\Delta_{i}^{n}X-\Delta_{i}^{n}X^{c}\Vert \geq \Vert \Delta_{i}^{n}X\Vert-\Vert \Delta_{i}^{n}X^{c}\Vert>\frac{u_{n}}{2},
\end{equation}
so that $\{\Vert \Delta_{i}^{n}X\Vert >u_{n}\}\subseteq\{\Vert S_{i}^{n}(u_{n}/2)\Vert>u_{n}/2\}\cup E_{n}^{\complement}$. Therefore,
\begin{equation}
\label{Eq:3Prob}
\mathbb{P}(\Vert\Delta_{i}^{n}X\Vert>u_{n})\leq\mathbb{P}\left(\Vert\Delta_{i}^{n}X^{c}\Vert>\frac{u_{n}}{2}\right)+\mathbb{P}\left(N_{i}^{n}\Bigl(\frac{u_{n}}{2}\Bigr)\geq1\right)+\mathbb{P}\left(\left\Vert S_{i}^{n}\Bigl(\frac{u_{n}}{2}\Bigr)\right\Vert>\frac{u_{n}}{2}\right),
\end{equation}
so it suffices to bound all three probabilities on the right-hand side in \cref{Eq:3Prob}. 

Choose any $q\geq 2$ that satisfies $q>\frac{1-\varpi r}{1/2-\varpi}$, and it holds that
\begin{equation}
\label{Eq:Prob1}
\mathbb{P}\left(\Vert\Delta_{i}^{n}X^{c}\Vert>\frac{u_{n}}{2}\right)\leq\frac{\mathbb{E}[\Vert\Delta_{i}^{n}X^{c}\Vert^{q}]}{(u_{n}/2)^{q}}\leq K\Delta_{n}^{q/2}u_{n}^{-q}=o(\Delta_{n}u_{n}^{-r}).
\end{equation}
by Markov's inequality and the BDG inequality.  Then, by \cref{Eq:LevyTail},
\begin{equation}
\label{Eq:Prob2}
\mathbb{P}\left(N_{i}^{n}\Bigl(\frac{u_{n}}{2}\Bigr)\geq1\right)\leq\mathbb{E}\left[N_{i}^{n}\Bigl(\frac{u_{n}}{2}\Bigr)\right]\leq\mathbb{E}\left[\int_{(i-1)\Delta_n}^{i\Delta_n}ds\int_{\{\Vert\delta\Vert>u_{n}/2\}}\lambda(dx)\right]\leq K\Delta_{n}u_{n}^{-r}.
\end{equation}
By subadditivity of the Euclidean norm,
\begin{equation}
\left\Vert S_{i}^{n}\Bigl(\frac{u_{n}}{2}\Bigr)\right\Vert\leq\int_{(i-1)\Delta_{n}}^{i\Delta_{n}}\int_{\{\Vert\delta\Vert\leq u_{n}/2\}}\Vert\delta(s,x)\Vert\mu(ds,dx),
\end{equation}
and Markov's inequality implies that
\begin{equation}
\label{Eq:Prob3}
\begin{split}
\mathbb{P}\left(\left\Vert S_{i}^{n}\Bigl(\frac{u_{n}}{2}\Bigr)\right\Vert>\frac{u_{n}}{2}\right)&\leq\mathbb{P}\left(\int_{(i-1)\Delta_{n}}^{i\Delta_{n}}\int_{\{\Vert\delta\Vert\leq u_{n}/2\}}\Vert\delta(s,x)\Vert\mu(ds,dx)>\frac{u_{n}}{2}\right)\\
&\leq\frac{2}{u_{n}}\mathbb{E}\left[\int_{(i-1)\Delta_{n}}^{i\Delta_{n}}ds\int_{\{\Vert\delta\Vert\leq u_{n}/2\}}\Vert\delta(s,x)\Vert\lambda(dx)\right]\leq K\Delta_{n}u_{n}^{-r},
\end{split}
\end{equation}
by \cref{Eq:smallJumpsl1}. Combining \cref{Eq:Prob1,Eq:Prob2,Eq:Prob3} in \cref{Eq:3Prob}, we obtain
\begin{equation}
\mathbb{P}(\Vert\Delta_{i}^{n}X\Vert>u_{n})\leq K\Delta_{n}u_{n}^{-r}. 
\end{equation}

We write $\Delta_{i}^{n}Y$ into
\begin{equation}
\Delta_{i}^{n}Y=\,\underbrace{\int_{(i-1)\Delta_{n}}^{i\Delta_{n}}\beta_{s-}^{\top}b_{s}ds+\int_{(i-1)\Delta_{n}}^{i\Delta_{n}}\beta_{s-}^{\top}\sigma_{s}dW_{s}}_{\Delta_{i}^{n}Y^{(1)}}\,\,+\,\,\underbrace{\sum_{(i-1)\Delta_{n}\leq s\leq i\Delta_{n}}(\beta_{s-}^{J})^{\top}\Delta X_{s}}_{\Delta_{i}^{n}Y^{(2)}}\,\,+\,\,\Delta_{i}^{n}Z,
\end{equation}
and by the triangle inequality,
\begin{equation}
\label{Eq:Y_triangle}
\begin{split}
\mathbb{P}(|\Delta_{i}^{n}Y|>u_{n})&\leq\mathbb{P}\left(|\Delta_{i}^{n}Y^{(1)}|>\frac{u_{n}}{4}\right)+\mathbb{P}\left(|\Delta_{i}^{n}Y^{(2)}|>\frac{u_{n}}{4}\right)\\
&\qquad+\mathbb{P}\left(|\Delta_{i}^{n}Z^{c}|>\frac{u_{n}}{4}\right)+\mathbb{P}\left(|\Delta_{i}^{n}Z-\Delta_{i}^{n}Z^{c}|>\frac{u_{n}}{4}\right).
\end{split}
\end{equation}
By bounded $\beta$ under \cref{As:Ito_Iocal} and the BDG inequality, for any $q>2$,
\begin{equation}
\mathbb{E}[|\Delta_{i}^{n}Y^{(1)}|^{q}]\leq K\mathbb{E}\left[\left(\Delta_{n}^{2}+\int_{(i-1)\Delta_{n}}^{i\Delta_{n}}\Vert\beta_{s}^{\top}\sigma_{s}\Vert^{2}ds\right)^{q/2}\right]\leq K'\Delta_{n}^{q/2}.
\end{equation}
Then by Markov's inequality,
\begin{equation}
\label{Eq:Y_triangle_Prob1}
\mathbb{P}\left(|\Delta_{i}^{n}Y^{(1)}|>\frac{u_{n}}{4}\right)\leq\frac{\mathbb{E}[|\Delta_{i}^{n}Y^{(1)}|^{q}]}{(u_{n}/4)^{q}}\leq K\Delta_{n}^{q/2}u_{n}^{-q}=o(\Delta_{n}u_{n}^{-r}),\qquad\text{if }q>\frac{1-\varpi r}{1/2-\varpi},
\end{equation}
which is always possible because $0<\varpi<1/2$ and $0\leq r<1$. 

By bounded $\beta^{J}$ under \cref{As:Ito_Iocal}, there exists some deterministic $M>0$ such that
\begin{equation}
|\Delta_{i}^{n}Y^{(2)}|\leq\sum_{(i-1)\Delta_{n}\leq s\leq i\Delta_{n}}\Vert\beta_{s}^{J}\Vert\Vert\Delta X_{s}\Vert\leq M\sum_{(i-1)\Delta_{n}\leq s\leq i\Delta_{n}}\Vert\Delta X_{s}\Vert,
\end{equation}
and thus
\begin{equation}
\label{Eq:Y_triangle_Prob2}
\begin{split}
\mathbb{P}\left(|\Delta_{i}^{n}Y^{(2)}|>\frac{u_{n}}{4}\right)&\leq\mathbb{P}\left(N_{i}^{n}\Bigl(\frac{u_{n}}{8M}\Bigr)\geq1\right)+\mathbb{P}\left(\int_{(i-1)\Delta_{n}}^{i\Delta_{n}}\int_{\{\Vert\delta\Vert\leq u_{n}/8M\}}\Vert\delta(s,x)\Vert\mu(ds,dx)>\frac{u_{n}}{8M}\right)\\
&\leq K\Delta_{n}u_{n}^{-r},
\end{split}
\end{equation}
which follows the same steps as in \cref{Eq:Prob2} and the second inequality in \cref{Eq:Prob3}. 

For the last two terms in \cref{Eq:Y_triangle}, we apply the same arguments to the one-dimensional It{\^o} semimartingale $Z$ under \cref{As:Ito} (vi). Combining them, \cref{Eq:Y_triangle_Prob1,Eq:Y_triangle_Prob2} in \cref{Eq:Y_triangle} leads to
\begin{equation}
\mathbb{P}(|\Delta_{i}^{n}Y|>u_{n})\leq K\Delta_{n}u_{n}^{-r}. 
\end{equation}
This completes the proof of \cref{lemma:truncationEvents} (i).

(ii) We consider the following event decomposition:
\begin{equation}
\label{Eq:event_decomp1}
\begin{split}
\{\Vert\Delta_{i}^{n}X\Vert\leq u_{n}\}=\{\Vert\Delta_{i}^{n}X\Vert\leq u_{n},\Vert\Delta_{i}^{n}X^{c}\Vert\leq u_{n}\}\cup\{\Vert\Delta_{i}^{n}X\Vert\leq u_{n},\Vert\Delta_{i}^{n}X^{c}\Vert>u_{n}\}.
\end{split}
\end{equation}
For each $i$, on the event $\{\Vert\Delta_{i}^{n}X\Vert\leq u_{n},\Vert\Delta_{i}^{n}X^{c}\Vert\leq u_{n}\}$,
\begin{equation}
\label{Eq:event_decomp2}
\Vert\Delta_{i}^{n}X-\Delta_{i}^{n}X^{c}\Vert\leq\Vert\Delta_{i}^{n}X\Vert+\Vert\Delta_{i}^{n}X^{c}\Vert\leq 2u_{n},
\end{equation}
and thus
\begin{equation}
\label{Eq:V_decomp}
\begin{split}
\Vert\Delta_{i}^{n}X-\Delta_{i}^{n}X^{c}\Vert^{2}\mathbbm{1}_{\{\Vert\Delta_{i}^{n}X\Vert\leq u_{n}\}}&\leq\underbrace{\Vert\Delta_{i}^{n}X-\Delta_{i}^{n}X^{c}\Vert^{2}\mathbbm{1}_{\Vert\Delta_{i}^{n}X-\Delta_{i}^{n}X^{c}\Vert\leq2u_{n}}}_{V_{i}^{(1)}}\\
&\qquad+\underbrace{\Vert\Delta_{i}^{n}X-\Delta_{i}^{n}X^{c}\Vert^{2}\mathbbm{1}_{\Vert\Delta_{i}^{n}X\Vert\leq u_{n}}\mathbbm{1}_{\Vert\Delta_{i}^{n}X^{c}\Vert>u_{n}}}_{V_{i}^{(2)}}.
\end{split}
\end{equation}
For $V_{i}^{(1)}$, since the function $x\mapsto x\wedge a$ is subadditive,
\begin{equation}
\label{Eq:V1_start}
V_{i}^{(1)}\leq(\Vert\Delta_{i}^{n}X-\Delta_{i}^{n}X^{c}\Vert\wedge2u_{n})^{2}\leq\left(\int_{(i-1)\Delta_n}^{i\Delta_n}\int_{\mathbb{R}^p}(\|\delta(s,x)\|\wedge 2u_n)\mu(ds,dx)\right)^{2}.
\end{equation}
Hence, using $(a+b)^2\leq 2a^2+2b^2$ and $\mu=\nu+(\mu-\nu)$,
\begin{equation}
\label{Eq:V1_split}
\begin{split}
\mathbb E[V_i^{(1)}]&\leq2\mathbb{E}\left[\left(\int_{(i-1)\Delta_n}^{i\Delta_n}\int_{\mathbb{R}^p} (\|\delta(s,x)\|\wedge 2u_n)(\mu-\nu)(ds,dx)\right)^2\right]\\
&\qquad+2\mathbb{E}\left[\left(\int_{(i-1)\Delta_n}^{i\Delta_n}\int_{\mathbb{R}^p}(\|\delta(s,x)\|\wedge 2u_n)\nu(ds,dx)\right)^2\right]. 
\end{split}
\end{equation}
For the first term in \cref{Eq:V1_split}, by the It\^{o} isometry for compensated Poisson integrals,
\begin{equation}
\label{Eq:V1_1}
\begin{split}
&\mathbb{E}\left[\left(\int_{(i-1)\Delta_n}^{i\Delta_n}\int_{\mathbb{R}^p} (\|\delta(s,x)\|\wedge 2u_n)(\mu-\nu)(ds,dx)\right)^2\right]\\
&\qquad=\mathbb{E}\left[\int_{(i-1)\Delta_n}^{i\Delta_n}\int_{\mathbb{R}^p} (\|\delta(s,x)\|\wedge 2u_n)^{2}\nu(ds,dx)\right]\\
&\qquad\leq\mathbb{E}\left[\int_{(i-1)\Delta_n}^{i\Delta_n}ds\left(\int_{\{\|\delta\|\le 2u_n\}}\|\delta(s,x)\|^2\lambda(dx)+4u_n^2\int_{\{\|\delta\|>2u_n\}}\lambda(dx)\right)\right]\\
&\qquad\leq K\Delta_{n}u_{n}^{2-r},
\end{split}
\end{equation}
by \cref{Eq:smallJumps,Eq:LevyTail}. For the second term in \cref{Eq:V1_split}, by Jensen's inequality,
\begin{equation}
\label{Eq:V1_2}
\begin{split}
&\mathbb{E}\left[\left(\int_{(i-1)\Delta_n}^{i\Delta_n}\int_{\mathbb{R}^p} (\|\delta(s,x)\|\wedge 2u_n)\nu(ds,dx)\right)^2\right]\\
&\qquad\leq\Delta_{n}\int_{(i-1)\Delta_n}^{i\Delta_n}\mathbb{E}\left[\left(\int_{\mathbb{R}^p}(\|\delta(s,x)\|\wedge 2u_n)\lambda(dx)\right)^2\right] ds\\
&\qquad\leq\Delta_{n}\int_{(i-1)\Delta_n}^{i\Delta_n}\mathbb{E}\left[\left(\int_{\{\|\delta\|\le 2u_n\}}\|\delta(s,x)\|\lambda(dx)+2u_n\int_{\{\|\delta\|>2u_n\}}\lambda(dx)\right)^2\right] ds\\
&\qquad\leq K\Delta_{n}\int_{(i-1)\Delta_n}^{i\Delta_n}u_n^{2-2r} ds = K\Delta_{n}^{2}u_{n}^{2-2r},
\end{split}
\end{equation}
by \cref{Eq:LevyTail,Eq:smallJumpsl1}. Combining \cref{Eq:V1_split,Eq:V1_1,Eq:V1_2}, we obtain
\begin{equation}
\label{Eq:V1_final}
\mathbb{E}[V_i^{(1)}]\leq K\Delta_n u_n^{2-r}+K'\Delta_{n}^{2}u_n^{2-2r}\le K\Delta_n u_n^{2-r}.
\end{equation}

For $V_{i}^{(2)}$, on the event $\{\Vert\Delta_{i}^{n}X\Vert\leq u_{n},\Vert\Delta_{i}^{n}X^{c}\Vert>u_{n}\}$,
\begin{equation}
\Vert\Delta_{i}^{n}X-\Delta_{i}^{n}X^{c}\Vert\leq u_{n}+\Vert\Delta_{i}^{n}X^{c}\Vert\leq2\Vert\Delta_{i}^{n}X^{c}\Vert,
\end{equation}
and thus $\mathbb{E}[V_{i}^{(2)}]\leq4\mathbb{E}\left[\Vert\Delta_{i}^{n}X^{c}\Vert^{2}\mathbbm{1}_{\{\Vert\Delta_{i}^{n}X^{c}\Vert>u_{n}\}}\right]$. It holds that, for any $q>2$,
\begin{equation}
\label{Eq:moment_Markov}
\begin{split}
\mathbb{E}\left[\Vert\Delta_{i}^{n}X^{c}\Vert^{2}\mathbbm{1}_{\{\Vert\Delta_{i}^{n}X^{c}\Vert>u_{n}\}}\right]\leq\mathbb{E}\left[\Vert\Delta_{i}^{n}X^{c}\Vert^{2}\left(\frac{\Vert\Delta_{i}^{n}X^{c}\Vert}{u_{n}}\right)^{q-2}\right]\leq\frac{\mathbb{E}[\Vert\Delta_{i}^{n}X^{c}\Vert^{q}]}{u_{n}^{q-2}}\leq\Delta_{n}^{q/2}u_{n}^{2-q}.
\end{split}
\end{equation}
by the BDG inequality. Therefore,
\begin{equation}
\label{Eq:V2_final}
\mathbb{E}[V_{i}^{(2)}]\leq\Delta_{n}^{q/2}u_{n}^{2-q}=o(\Delta_{n}u_{n}^{2-r}),\qquad\text{if }q>\frac{1-\varpi r}{1/2-\varpi},
\end{equation}
which is always possible because $0<\varpi<1/2$ and $0\leq r<1$. 

Finally, by \cref{Eq:V_decomp,Eq:V1_final,Eq:V2_final},
\begin{equation}
\mathbb{E}\left[\Vert\Delta_{i}^{n}X-\Delta_{i}^{n}X^{c}\Vert^{2}\mathbbm{1}_{\{\Vert\Delta_{i}^{n}X\Vert\leq u_{n}\}}\right]=O(\Delta_{n}u_{n}^{2-r}),
\end{equation}
and thus, by Markov's inequality,
\begin{equation}
\Vert\Delta_{i}^{n}X-\Delta_{i}^{n}X^{c}\Vert^{2}\mathbbm{1}_{\{\Vert\Delta_{i}^{n}X\Vert\leq u_{n}\}}
=O_{p}(\Delta_{n}u_{n}^{2-r}).
\end{equation}
This completes the proof of \cref{lemma:truncationEvents} (ii).

\end{proof}

Next, we verify that $A_{1}$, $A_{2}$, $A_{3}$ and $A_{4}$ in \cref{Eq:normDecomposition} are all $o_{p}(1)$:

\subsubsection*{(i) $\bm{A_{1}=o_{p}(1)}$}
\label{AP:Proof_consistency_A1}

We consider a further decomposition $\widetilde{\mathbf{Y}}=\widetilde{\mathbf{Y}}^{(1)}+\widetilde{\mathbf{Y}}^{(2)}$:
\begin{align}
\widetilde{\mathbf{Y}}_{i}^{(1)}&%=\widetilde{\beta}_{(i-1)\Delta_{n}}^{\top}\left(\int_{(i-1)\Delta_{n}}^{i\Delta_{n}}dX_{s}^{c}\right)\mathbbm{1}_{\{|\Delta_{i}^{n}Y|\leq u_{n}\}}
=\widetilde{\beta}_{(i-1)\Delta_{n}}^{\top}\Delta_{i}^{n}X^{c}\mathbbm{1}_{\{|\Delta_{i}^{n}Y|\leq u_{n}\}}, \label{Eq:Y_decomp_1}\\
\widetilde{\mathbf{Y}}_{i}^{(2)}&=\left(\int_{(i-1)\Delta_{n}}^{i\Delta_{n}}(\widetilde{\beta}_{s}-\widetilde{\beta}_{(i-1)\Delta_{n}})^{\top}dX_{s}^{c}\right)\mathbbm{1}_{\{|\Delta_{i}^{n}Y|\leq u_{n}\}}. \label{Eq:Y_decomp_2}
\end{align}
%\begin{equation}
%\widetilde{\mathbf{Y}}_{i}^{(1)}=\widetilde{\beta}_{(i-1)\Delta_{n}}^{\top}\Delta_{i}^{n}X^{c}\mathbbm{1}_{\{|\Delta_{i}^{n}Y|\leq u_{n}\}},\quad
%\widetilde{\mathbf{Y}}_{i}^{(2)}=\left(\int_{(i-1)\Delta_{n}}^{i\Delta_{n}}(\widetilde{\beta}_{s}-\widetilde{\beta}_{(i-1)\Delta_{n}})^{\top}dX_{s}^{c}\right)\mathbbm{1}_{\{|\Delta_{i}^{n}Y|\leq u_{n}\}},
%\end{equation}
The $i$-th element in $\widetilde{\mathbf{Y}}^{(1)}-\mathbf{R}\bm{\gamma}$ is given by
\begin{equation}
\label{Eq:Y_decomp_3}
\begin{split}
&\widetilde{\beta}_{(i-1)\Delta_{n}}^{\top}\Delta_{i}^{n}X^{c}\mathbbm{1}_{\{|\Delta_{i}^{n}Y|\leq u_{n}\}}-\widetilde{\beta}_{(i-1)\Delta_{n}}^{\top}\Delta_{i}^{n}X\mathbbm{1}_{\{\Vert\Delta_{i}^{n}X\Vert\leq u_{n}\}}\\
=\,\,&\underbrace{\widetilde{\beta}_{(i-1)\Delta_{n}}^{\top}\Delta_{i}^{n}X^{c}\left(\mathbbm{1}_{\{|\Delta_{i}^{n}Y|\leq u_{n}\}}-\mathbbm{1}_{\{\Vert\Delta_{i}^{n}X\Vert\leq u_{n}\}}\right)}_{M_{i}^{(1)}}\,-\,\underbrace{\widetilde{\beta}_{(i-1)\Delta_{n}}^{\top}(\Delta_{i}^{n}X-\Delta_{i}^{n}X^{c})\mathbbm{1}_{\{\Vert\Delta_{i}^{n}X\Vert\leq u_{n}\}}}_{M_{i}^{(2)}},
\end{split}
\end{equation}
and thus
\begin{equation}
\label{Eq:A1_Eq1}
\Vert\widetilde{\mathbf{Y}}^{(1)}-\mathbf{R}\bm{\gamma}\Vert^{2}=\sum_{i=1}^{n}(M_{i}^{(1)}-M_{i}^{(2)})^{2}\leq2\sum_{i=1}^{n}(M_{i}^{(1)})^{2}+2\sum_{i=1}^{n}(M_{i}^{(2)})^{2}.
\end{equation} 
%For $M_{i}^{(1)}$ and $M_{i}^{(2)}$, it holds that, by the Cauchy-Schwarz inequality and \cref{lemma:truncationEvents} (i), 
%\begin{equation}
%\label{Eq:A1_Eq2}
%\begin{split}
%\sum_{i=1}^{n}\mathbb{E}[(M_{i}^{(1)})^{2}]&\leq K\sum_{i=1}^{n}\mathbb{E}\left[\Vert\Delta_{i}^{n}X^{c}\Vert^{2}\left(\mathbbm{1}_{\{|\Delta_{i}^{n}Y|\leq u_{n}\}}-\mathbbm{1}_{\{\Vert\Delta_{i}^{n}X\Vert\leq u_{n}\}}\right)\right]\\
%&\leq K\sum_{i=1}^{n}(\mathbb{E}[\Vert\Delta_{i}^{n}X^{c}\Vert^{4}])^{1/2}(\mathbb{P}(|\Delta_{i}^{n}Y|> u_{n})+\mathbb{P}(\Vert\Delta_{i}^{n}X\Vert>u_{n}))^{1/2}\\
%&\leq K'\sqrt{\Delta_{n}}u_{n}^{-r/2}=o(1),
%\end{split}
%\end{equation}
For $M_{i}^{(1)}$ and $M_{i}^{(2)}$, since $\Vert\widetilde{\beta}\Vert_{\infty}$ is bounded by de Boor's conjecture (see, e.g., \citealp{de1973quasi,shadrin2001norm}), it holds that, by the BDG inequality and \cref{lemma:truncationEvents} (i), 
\begin{equation}
\label{Eq:A1_Eq2}
\begin{split}
\sum_{i=1}^{n}\mathbb{E}[(M_{i}^{(1)})^{2}]&\leq K\sum_{i=1}^{n}\mathbb{E}\left[\Vert\Delta_{i}^{n}X^{c}\Vert^{2}\left(\mathbbm{1}_{\{|\Delta_{i}^{n}Y|\leq u_{n}\}}-\mathbbm{1}_{\{\Vert\Delta_{i}^{n}X\Vert\leq u_{n}\}}\right)\right]\\
&\leq K\sum_{i=1}^{n}\mathbb{E}\left[\mathbb{E}[\Vert\Delta_{i}^{n}X^{c}\Vert^{2}|\mathcal{F}_{(i-1)\Delta_{n}}]\left(\mathbbm{1}_{\{|\Delta_{i}^{n}Y|\leq u_{n}\}}-\mathbbm{1}_{\{\Vert\Delta_{i}^{n}X\Vert\leq u_{n}\}}\right)\right]\\
&\leq K'\Delta_{n}\sum_{i=1}^{n}(\mathbb{P}(|\Delta_{i}^{n}Y|> u_{n})+\mathbb{P}(\Vert\Delta_{i}^{n}X\Vert>u_{n}))\\
&\leq K''\Delta_{n}u_{n}^{-r}=o(1),
\end{split}
\end{equation}
and, by \cref{lemma:truncationEvents} (ii),
\begin{equation}
\label{Eq:A1_Eq3}
\sum_{i=1}^{n}\mathbb{E}[(M_{i}^{(2)})^{2}]\leq K\sum_{i=1}^{n}\mathbb{E}\left[\Vert\Delta_{i}^{n}X-\Delta_{i}^{n}X^{c}\Vert^{2}\mathbbm{1}_{\{\Vert\Delta_{i}^{n}X\Vert\leq u_{n}\}}\right]\leq Ku_{n}^{2-r}=o(1).
\end{equation}
Therefore, $\mathbb{E}[\Vert\widetilde{\mathbf{Y}}^{(1)}-\mathbf{R}\bm{\gamma}\Vert^{2}]=o(1)$, and thus $\Vert\widetilde{\mathbf{Y}}^{(1)}-\mathbf{R}\bm{\gamma}\Vert=o_{p}(1)$ by Markov's inequality.

Next, we verify that $\Vert\widetilde{\mathbf{Y}}^{(2)}\Vert=o_{p}(1)$. 
By the It{\^o} isometry (see \cref{remark:Predictability}) and bounded $c$,  %since $\widetilde{\beta}$ is Lipschitz continuous in $t$,
\begin{equation}
\label{Eq:Discretization-1}
\mathbb{E}[\Vert\widetilde{\mathbf{Y}}^{(2)}\Vert^{2}]\leq K\sum_{i=1}^{n}\int_{(i-1)\Delta_{n}}^{i\Delta_{n}}\mathbb{E}[\Vert\widetilde{\beta}_{s}-\widetilde{\beta}_{(i-1)\Delta_{n}}\Vert^{2}]ds.
\end{equation}
Since $\widetilde{\beta}$ is absolutely continuous and has bounded derivative on $[0,T]$,
\begin{equation}
\label{Eq:Discretization-2}
\Vert\widetilde{\beta}_{s}-\widetilde{\beta}_{(i-1)\Delta_{n}}\Vert\leq (s-(i-1)\Delta_{n})\Vert\widetilde{\beta}'\Vert_{\infty},\qquad\text{where }\Vert\widetilde{\beta}'\Vert_{\infty}=\sup_{s\in[0,T]}\Vert\widetilde{\beta}'_{s}\Vert,
\end{equation}
and $\widetilde{\beta}'_{s}$ stands for the derivative of $\widetilde{\beta}$ on $s$. By \cref{Eq:Discretization-1,Eq:Discretization-2},
\begin{equation}
\label{Eq:Discretization-3}
\begin{split}
\mathbb{E}[\Vert\widetilde{\mathbf{Y}}^{(2)}\Vert^{2}]&\leq K\mathbb{E}[\Vert\widetilde{\beta}'\Vert_{\infty}^{2}]\sum_{i=1}^{n}\int_{(i-1)\Delta_{n}}^{i\Delta_{n}}(s-(i-1)\Delta_{n})^{2}ds\\
&\leq K\mathbb{E}[\Vert\widetilde{\beta}'\Vert_{\infty}^{2}]\sum_{i=1}^{n}\int_{(i-1)\Delta_{n}}^{i\Delta_{n}}\Delta_{n}^{2}ds=K\mathbb{E}[\Vert\widetilde{\beta}'\Vert_{\infty}^{2}]\sum_{i=1}^{n}\frac{\Delta_{n}^{3}}{3}\\
&\leq K'\Delta_{n}^{2}\mathbb{E}[\Vert\widetilde{\beta}'\Vert_{\infty}^{2}]=O(\Delta_{n}^{2}K_{n}^{2})=o(1).
\end{split}
\end{equation}
because $\Vert\widetilde{\beta}'\Vert_{\infty}=O_{p}(K_{n})$ and $\Delta_{n}K_{n}\log K_{n}\to0$, and thus $\Vert\widetilde{\mathbf{Y}}^{(2)}\Vert=o_{p}(1)$ by Markov's inequality.

By \cref{lemma:norm}, it holds that
\begin{equation}
\begin{split}
A_{1}&\leq\Vert\mathbf{B}(\mathbf{R}^{\top}\mathbf{R})^{-1}\mathbf{R}^{\top}(\widetilde{\mathbf{Y}}^{(1)}-\mathbf{R}\bm{\gamma})\Vert_{L^{2}}+\Vert\mathbf{B}(\mathbf{R}^{\top}\mathbf{R})^{-1}\mathbf{R}^{\top}\widetilde{\mathbf{Y}}^{(2)}\Vert_{L^{2}}\\
&\leq K\Vert\widetilde{\mathbf{Y}}^{(1)}-\mathbf{R}\bm{\gamma}\Vert + K'\Vert\widetilde{\mathbf{Y}}^{(2)}\Vert=o_{p}(1). 
\end{split}
\end{equation}

%\begin{lemma}
%\label{lemma:A2}
%Under the conditions in \cref{Th:consistency}, $A_{2}=o_{p}(1)$ as $\Delta_{n}\to0$. 
%\end{lemma}

\subsubsection*{(ii) $\bm{A_{2}=o_{p}(1)}$}

%By \cref{Eq:ItoIsometryRevised}, we obtain $\mathbb{E}[\Vert\mathbf{U}\Vert^{2}]\leq \mathbb{E}[\Vert e\Vert_{L^{2}}^{2}]$ directly. By Chebyshev's inequality, for any $\varepsilon>0$,
%\begin{equation}
%\mathbb{P}(\Vert\mathbf{U}\Vert>\varepsilon)\leq\frac{\mathbb{E}[|A_{2}|^{2}]}{\varepsilon^{2}}\leq\frac{\mathbb{E}[\Vert e\Vert_{L^{2}}^{2}]}{\varepsilon^{2}}=o(1),
%\end{equation}
%and thus $\Vert\mathbf{U}\Vert=o_{p}(1)$. Therefore, $A_{2}\leq K\Vert\mathbf{U}\Vert=o_{p}(1)$ by \cref{lemma:norm}. 
%
%
%...
%
%Fix $\varepsilon>0$. For any $\eta>0$, it holds that by Chebyshev's inequality,
%\begin{equation}
%\begin{split}
%\mathbb{P}(\Vert\mathbf{U}\Vert>\varepsilon)&\leq\mathbb{P}(\Vert e\Vert_{L^{2}}>\eta)+\mathbb{P}(\Vert\mathbf{U}\Vert>\varepsilon,\Vert e\Vert_{L^{2}}\leq\eta)\\
%&\leq\mathbb{P}(\Vert e\Vert_{L^{2}}>\eta)+\frac{\mathbb{E}[\Vert\mathbf{U}\Vert^{2}\mathbbm{1}_{\{\Vert e\Vert_{L^{2}}\leq\eta\}}]}{\varepsilon^{2}}\\
%&\leq
%\end{split}
%\end{equation}
%
%
%%\begin{proof}
Conditional on a $\sigma$-field $\mathcal{G}=\sigma(e_{s}:0\leq s\leq T)$, it holds that $\mathbb{E}[\Vert\mathbf{U}\Vert^{2}|\mathcal{G}]\leq K\Vert e\Vert_{L^{2}}^{2}$ by the It{\^o} isometry and bounded $c$. Fix $\varepsilon>0$. For any $\eta>0$, it holds that by Chebyshev's inequality,
\begin{equation}
\begin{split}
\mathbb{P}(\Vert\mathbf{U}\Vert>\varepsilon)&\leq\mathbb{P}(\Vert e\Vert_{L^{2}}>\eta)+\mathbb{P}(\Vert\mathbf{U}\Vert>\varepsilon,\Vert e\Vert_{L^{2}}\leq\eta)\\
&\leq\mathbb{P}(\Vert e\Vert_{L^{2}}>\eta)+\frac{\mathbb{E}[\Vert\mathbf{U}\Vert^{2}\mathbbm{1}_{\{\Vert e\Vert_{L^{2}}\leq\eta\}}]}{\varepsilon^{2}}\\
&\leq\mathbb{P}(\Vert e\Vert_{L^{2}}>\eta)+\frac{\mathbb{E}[\mathbb{E}[\Vert\mathbf{U}\Vert^{2}|\mathcal{G}]\mathbbm{1}_{\{\Vert e\Vert_{L^{2}}\leq\eta\}}]}{\varepsilon^{2}}\\
&\leq\mathbb{P}(\Vert e\Vert_{L^{2}}>\eta)+K\frac{\mathbb{E}[\Vert e\Vert_{L^{2}}^{2}\mathbbm{1}_{\{\Vert e\Vert_{L^{2}}\leq\eta\}}]}{\varepsilon^{2}}\\
&\leq\mathbb{P}(\Vert e\Vert_{L^{2}}>\eta)+K\frac{\eta^{2}}{\varepsilon^{2}},
\end{split}
\end{equation}
where $\mathbb{P}(\Vert e\Vert_{L^{2}}>\eta)\to0$ since $\Vert e\Vert_{L^{2}}=o_{p}(1)$ by \cref{lemma:approximationError}. Let $\eta\downarrow0$, and then $\Vert\mathbf{U}\Vert=o_{p}(1)$. Therefore, $A_{2}\leq K\Vert\mathbf{U}\Vert=o_{p}(1)$ by \cref{lemma:norm}.
%\end{proof}

%\begin{lemma}
%\label{lemma:A3}
%Under the conditions in \cref{Th:consistency}, $A_{3}=o_{p}(1)$ as $\Delta_{n}\to0$. 
%\end{lemma}

\subsubsection*{(iii) $\bm{A_{3}=o_{p}(1)}$}
\label{AP:Proof_consistency_A3}

%\begin{proof}
We define a vector $\widetilde{\mathbf{J}}\in\mathbb{R}^{n}$ with the $i$-th element
\begin{equation}
\label{Eq:J_i_tilde}
\widetilde{\mathbf{J}}_{i}=\mathbf{J}_{i}\mathbbm{1}_{\{\Vert\Delta_{i}^{n}X\Vert\leq u_{n}\}}=\left(\sum_{(i-1)\Delta_{n}\leq s\leq i\Delta_{n}}(\beta_{s-}^{J})^{\top}\Delta X_{s}\right)\mathbbm{1}_{\{|\Delta_{i}^{n}Y|\leq u_{n}\}}\mathbbm{1}_{\{\Vert\Delta_{i}^{n}X\Vert\leq u_{n}\}}.
\end{equation}
Because the design matrix $\mathbf{R}$ already includes the same truncation $\mathbbm{1}_{\{\Vert\Delta_{i}^{n}X\Vert\leq u_{n}\}}$ for covariate increments, we have $\mathbf{R}^{\top}\mathbf{J}=\mathbf{R}^{\top}\widetilde{\mathbf{J}}$. Hence, by \cref{lemma:norm}, it holds that
\begin{equation}
A_{3}\leq K\sqrt{K_{n}}\Vert\mathbf{R}^{\top}\mathbf{J}\Vert = K\sqrt{K_{n}}\Vert\mathbf{R}^{\top}\widetilde{\mathbf{J}}\Vert\leq K'\Vert\widetilde{\mathbf{J}}\Vert,
\end{equation}
%By \cref{lemma:norm}, it holds that $A_{3}\leq K\sqrt{K_{n}}\Vert\mathbf{R}^{\top}\mathbf{J}\Vert\equiv K\sqrt{K_{n}}\Vert\mathbf{R}^{\top}\widetilde{\mathbf{J}}\Vert\leq K'\Vert\widetilde{\mathbf{J}}\Vert$, since each element in $\mathbf{R}^{\top}\mathbf{J}$ satisfies
%\begin{equation}
%(\mathbf{R}^{\top}\mathbf{J})_{i}=\mathbf{B}_{(i-1)\Delta_{n}}^{\top}\mathbf{X}_{i}\mathbf{J}_{i}\equiv\mathbf{B}_{(i-1)\Delta_{n}}^{\top}\mathbf{X}_{i}\mathbf{J}_{i}\mathbbm{1}_{\{\Vert\Delta_{i}^{n}X\Vert\leq u_{n}\}}, 
%\end{equation}
so that it suffices to show that $\Vert\widetilde{\mathbf{J}}\Vert=o_{p}(1)$. By bounded $\beta^{J}$ under \cref{As:Ito_Iocal}, it holds that
\begin{align}
|\widetilde{\mathbf{J}}_{i}|^{2}&\leq\left(\sum_{(i-1)\Delta_{n}\leq s\leq i\Delta_{n}}\Vert\beta_{s-}^{J}\Vert\Vert\Delta X_{s}\Vert\right)^{2}\mathbbm{1}_{\{\Vert\Delta_{i}^{n}X\Vert\leq u_{n}\}} \leq K\left(\sum_{(i-1)\Delta_{n}\leq s\leq i\Delta_{n}}\Vert\Delta X_{s}\Vert\right)^{2}\mathbbm{1}_{\{\Vert\Delta_{i}^{n}X\Vert\leq u_{n}\}} \notag\\
&\leq K\left(\int_{(i-1)\Delta_{n}}^{i\Delta_{n}}\int_{\mathbb{R}^{p}}\Vert\delta(s,x)\Vert\mu(ds,dx)\right)^{2}\mathbbm{1}_{\{\Vert\Delta_{i}^{n}X\Vert\leq u_{n}\}} \notag\\
&\leq K\left(\int_{(i-1)\Delta_{n}}^{i\Delta_{n}}\int_{\mathbb{R}^{p}}(\Vert\delta(s,x)\Vert\wedge 2u_{n})\mu(ds,dx)+\int_{(i-1)\Delta_{n}}^{i\Delta_{n}}\int_{\{\Vert\delta\Vert>2u_{n}\}}\Vert\delta(s,x)\Vert\mu(ds,dx)\right)^{2}\mathbbm{1}_{\{\Vert\Delta_{i}^{n}X\Vert\leq u_{n}\}} \notag\\
&\leq 2K\left(\int_{(i-1)\Delta_{n}}^{i\Delta_{n}}\int_{\mathbb{R}^{p}}(\Vert\delta(s,x)\Vert\wedge 2u_{n})\mu(ds,dx)\right)^{2}\mathbbm{1}_{\{\Vert\Delta_{i}^{n}X\Vert\leq u_{n}\}} \notag\\
&\qquad+ 2K\left(\int_{(i-1)\Delta_{n}}^{i\Delta_{n}}\int_{\{\Vert\delta\Vert>2u_{n}\}}\Vert\delta(s,x)\Vert\mu(ds,dx)\right)^{2}\mathbbm{1}_{\{\Vert\Delta_{i}^{n}X\Vert\leq u_{n}\}}, \label{Eq:J_i_2u}
\end{align}
and thus
\begin{equation}
\label{Eq:EJi^2}
\begin{split}
\mathbb{E}[|\widetilde{\mathbf{J}}_{i}|^{2}]&\leq K\mathbb{E}\left[\left(\int_{(i-1)\Delta_{n}}^{i\Delta_{n}}\int_{\mathbb{R}^{p}}(\Vert\delta(s,x)\Vert\wedge 2u_{n})\mu(ds,dx)\right)^{2}\right]\\
&\qquad+K'\mathbb{E}\left[\left(\int_{(i-1)\Delta_{n}}^{i\Delta_{n}}\int_{\{\Vert\delta\Vert>2u_{n}\}}\Vert\delta(s,x)\Vert\mu(ds,dx)\right)^{2}\mathbbm{1}_{\{\Vert\Delta_{i}^{n}X\Vert\leq u_{n}\}}\right].
\end{split}
\end{equation}
Following the same steps as in \cref{Eq:V1_split,Eq:V1_1,Eq:V1_2,Eq:V1_final}, we obtain
\begin{equation}
\label{Eq:EJi^2_1}
\mathbb{E}\left[\left(\int_{(i-1)\Delta_{n}}^{i\Delta_{n}}\int_{\mathbb{R}^{p}}(\Vert\delta(s,x)\Vert\wedge 2u_{n})\mu(ds,dx)\right)^{2}\right]\leq K\Delta_{n}u_{n}^{2-r}.
\end{equation}
%To calculate the second expectation in \cref{Eq:EJi^2}, 
%\begin{equation}
%\int_{(i-1)\Delta_{n}}^{i\Delta_{n}}\int_{\{\Vert\delta\Vert>2u_{n}\}}\Vert\delta(s,x)\Vert\mu(ds,dx)\leq KN_{i}^{n}(2u_{n})=K\int_{(i-1)\Delta_n}^{i\Delta_n}\int_{\mathbb{R}^p}\mathbbm{1}_{\{\|\delta(s,x)\|>2u_{n}\}}\mu(ds,dx),
%\end{equation}
%and thus
With bounded $\delta$ under \cref{As:Ito_Iocal}, it holds that
\begin{equation}
\label{Eq:largeJumpBound}
\left(\int_{(i-1)\Delta_{n}}^{i\Delta_{n}}\int_{\{\Vert\delta\Vert>2u_{n}\}}\Vert\delta(s,x)\Vert\mu(ds,dx)\right)^{2}\mathbbm{1}_{\{\Vert\Delta_{i}^{n}X\Vert\leq u_{n}\}}\leq K(N_{i}^{n}(2u_{n}))^{2}\mathbbm{1}_{\{\Vert\Delta_{i}^{n}X\Vert\leq u_{n}\}},
\end{equation}
where $N_{i}^{n}(2u_{n})$ is defined in \cref{Eq:N}. Therefore, the left-hand side in the above inequality must be zero unless on the event $\{N_{i}^{n}(2u_{n})\geq1,\,\Vert\Delta_{i}^{n}X\Vert\leq u_{n}\}$. We consider the following event relations:
%\begin{align}
%\{N_{i}^{n}(2u_{n})\geq1,\,\Vert\Delta_{i}^{n}X\Vert\leq u_{n}\}&\subseteq\left\{N_{i}^{n}\Bigl(\frac{u_{n}}{2}\Bigr)\geq1,\,\Vert\Delta_{i}^{n}X\Vert\leq u_{n}\right\} \notag\\
%&=\left\{N_{i}^{n}\Bigl(\frac{u_{n}}{2}\Bigr)=1,\,\Vert\Delta_{i}^{n}X\Vert\leq u_{n}\right\}\cup\left\{N_{i}^{n}\Bigl(\frac{u_{n}}{2}\Bigr)\geq2,\,\Vert\Delta_{i}^{n}X\Vert\leq u_{n}\right\},
%\end{align}
\begin{equation}
\begin{split}
\{N_{i}^{n}(2u_{n})\geq1,\,\Vert\Delta_{i}^{n}X\Vert\leq u_{n}\}&\subseteq\left\{N_{i}^{n}(2u_{n})\geq1,\,\Vert\Delta_{i}^{n}X\Vert\leq u_{n},\,N_{i}^{n}\Bigl(\frac{u_{n}}{2}\Bigr)=1\right\}\\
&\qquad\cup\left\{N_{i}^{n}(2u_{n})\geq1,\,\Vert\Delta_{i}^{n}X\Vert\leq u_{n},\,N_{i}^{n}\Bigl(\frac{u_{n}}{2}\Bigr)\geq2\right\}\\
&\subseteq\left\{N_{i}^{n}(2u_{n})=1,\,\Vert\Delta_{i}^{n}X\Vert\leq u_{n}\right\}\cup\left\{N_{i}^{n}\Bigl(\frac{u_{n}}{2}\Bigr)\geq2\right\}
\end{split}
\end{equation}
On the event $\{N_{i}^{n}(2u_{n})=1,\,\Vert\Delta_{i}^{n}X\Vert\leq u_{n}\}$, by the triangle inequality,
\begin{equation}
\left\Vert\Delta_{i}^{n}X^{c}+S_{i}^{n}\Bigl(\frac{u_{n}}{2}\Bigr)\right\Vert>2u_{n}-\Vert\Delta_{i}^{n}X\Vert>u_{n},
\end{equation}
and therefore
%where
%\begin{equation}
%\begin{split}
%\left\{N_{i}^{n}(2u_{n})\geq1,\,\Vert\Delta_{i}^{n}X\Vert\leq u_{n},\,N_{i}^{n}\Bigl(\frac{u_{n}}{2}\Bigr)=1\right\}&=\left\{N_{i}^{n}(2u_{n})=1,\,\Vert\Delta_{i}^{n}X\Vert\leq u_{n}\right\}\\
%&\subseteq\left\{N_{i}^{n}(2u_{n})=1,\,\Vert\Delta_{i}^{n}X-\Delta_{i}^{n}X^{c}\Vert\geq u_{n}\right\}
%\end{split}
%\end{equation}
%\begin{equation}
%\begin{split}
%\left\{N_{i}^{n}\Bigl(\frac{u_{n}}{2}\Bigr)=1,\,\Vert\Delta_{i}^{n}X\Vert\leq u_{n}\right\}&\subseteq\left\{N_{i}^{n}\Bigl(\frac{u_{n}}{2}\Bigr)=1,\,\Vert\Delta_{i}^{n}X^{c}\Vert>\frac{u_{n}}{2}\right\}\\
%&\qquad\cup\left\{N_{i}^{n}\Bigl(\frac{u_{n}}{2}\Bigr)=1,\,\Vert\Delta_{i}^{n}X^{c}\Vert\leq\frac{u_{n}}{2},\,\left\Vert S_{i}^{n}\Bigl(\frac{u_{n}}{2}\Bigr)\right\Vert>\frac{u_{n}}{2}\right\}\\
%&\subseteq\left\{\Vert\Delta_{i}^{n}X^{c}\Vert>\frac{u_{n}}{2}\right\}\cup\left\{\left\Vert S_{i}^{n}\Bigl(\frac{u_{n}}{2}\Bigr)\right\Vert>\frac{u_{n}}{2}\right\},
%\end{split}
%\end{equation}
%with $S_{i}^{n}(u_{n}/2)$ defined in \cref{Eq:S}, and
%\begin{equation}
%\begin{split}
%\left\{N_{i}^{n}\Bigl(\frac{u_{n}}{2}\Bigr)\geq2,\,\Vert\Delta_{i}^{n}X\Vert\leq u_{n}\right\}\subseteq\left\{N_{i}^{n}\Bigl(\frac{u_{n}}{2}\Bigr)\geq2\right\},
%\end{split}
%\end{equation}
%and therefore
\begin{equation}
\label{Eq:eventDecomp_J}
\{N_{i}^{n}(2u_{n})\geq1,\,\Vert\Delta_{i}^{n}X\Vert\leq u_{n}\}\,\subseteq\,\underbrace{\left\{\Vert\Delta_{i}^{n}X^{c}\Vert>\frac{u_{n}}{2}\right\}}_{E_{n}^{(1)}}\,\cup\,\underbrace{\left\{\left\Vert S_{i}^{n}\Bigl(\frac{u_{n}}{2}\Bigr)\right\Vert>\frac{u_{n}}{2}\right\}}_{E_{n}^{(2)}}\,\cup\,\underbrace{\left\{N_{i}^{n}\Bigl(\frac{u_{n}}{2}\Bigr)\geq2\right\}}_{E_{n}^{(3)}}. 
\end{equation}
Then, by \cref{Eq:largeJumpBound},
\begin{equation}
\label{Eq:EJi^2_2_decomp}
\begin{split}
&\mathbb{E}\left[\left(\int_{(i-1)\Delta_{n}}^{i\Delta_{n}}\int_{\{\Vert\delta\Vert>2u_{n}\}}\Vert\delta(s,x)\Vert\mu(ds,dx)\right)^{2}\mathbbm{1}_{\{\Vert\Delta_{i}^{n}X\Vert\leq u_{n}\}}\right]\\
&\qquad\leq K\mathbb{E}\left[(N_{i}^{n}(2u_{n}))^{2}\mathbbm{1}_{E_{n}^{(1)}}\right]+K'\mathbb{E}\left[(N_{i}^{n}(2u_{n}))^{2}\mathbbm{1}_{E_{n}^{(2)}}\right]+K''\mathbb{E}\left[(N_{i}^{n}(2u_{n}))^{2}\mathbbm{1}_{E_{n}^{(3)}}\right].
\end{split}
\end{equation}
By the factorial moment of Poisson random variable and \cref{Eq:LevyTail},
\begin{equation}
\label{Eq:EJi^2_2_decomp3}
\begin{split}
\mathbb{E}\left[(N_{i}^{n}(2u_{n}))^{2}\mathbbm{1}_{E_{n}^{(3)}}\right]&\leq\mathbb{E}\left[\left(N_{i}^{n}\Bigl(\frac{u_{n}}{2}\Bigr)\right)^{2}\mathbbm{1}_{E_{n}^{(3)}}\right]\leq 2\mathbb{E}\left[\left(N_{i}^{n}\Bigl(\frac{u_{n}}{2}\Bigr)\right)\left(N_{i}^{n}\Bigl(\frac{u_{n}}{2}-1\Bigr)\right)\mathbbm{1}_{E_{n}^{(3)}}\right]\\
&\leq 2\mathbb{E}\left[\left(\int_{(i-1)\Delta_n}^{i\Delta_n}ds\int_{\{\Vert\delta\Vert>u_{n}/2\}}\lambda(dx)\right)^{2}\right]\leq K'\Delta_{n}^{2}u_{n}^{-2r}.
\end{split}
\end{equation}
By the Cauchy-Schwarz inequality, \cref{Eq:LevyTail,Eq:Prob2}, take some $q\geq2$ and $q>\frac{6(1-\varpi r)}{1-2\varpi}$, then
\begin{equation}
\label{Eq:EJi^2_2_decomp1}
\begin{split}
\mathbb{E}\left[(N_{i}^{n}(2u_{n}))^{2}\mathbbm{1}_{E_{n}^{(1)}}\right]&\leq\left(\mathbb{E}[(N_{i}^{n}(2u_{n}))^{4}]\right)^{1/2}\left(\mathbb{P}\left(\Vert\Delta_{i}^{n}X^{c}\Vert>\frac{u_{n}}{2}\right)\right)^{1/2}\\
%&\leq \left(\mathbb{E}\left[\left(\int_{(i-1)\Delta_n}^{i\Delta_n}ds\int_{\{\Vert\delta\Vert>2u_{n}\}}\lambda(dx)\right)^{4}\right]\right)^{1/2}\left(\mathbb{P}\left(\Vert\Delta_{i}^{n}X^{c}\Vert>\frac{u_{n}}{2}\right)\right)^{1/2}\\
&\leq K(\Delta_{n}u_{n}^{-r})^{1/2}(\Delta_{n}^{q/2}u_{n}^{-q})^{1/2}=K\Delta_{n}^{1/2+q/4}u_{n}^{-(r+q)/2}=o(\Delta_{n}^{2}u_{n}^{-2r}).
\end{split}
\end{equation}
Following a similar approach to \cref{Eq:Prob3},
\begin{align}
\mathbb{E}\left[(N_{i}^{n}(2u_{n}))^{2}\mathbbm{1}_{E_{n}^{(2)}}\right]&\leq\mathbb{E}\left[(N_{i}^{n}(2u_{n}))^{2}\left(\frac{2}{u_{n}}\int_{(i-1)\Delta_{n}}^{i\Delta_{n}}\int_{\{\Vert\delta\Vert\leq u_{n}/2\}}\Vert\delta(s,x)\Vert\mu(ds,dx)\right)\right]\notag\\
&=\frac{2}{u_{n}}\mathbb{E}\left[N_{i}^{n}(2u_{n})(N_{i}^{n}(2u_{n})-1)\left(\int_{(i-1)\Delta_{n}}^{i\Delta_{n}}\int_{\{\Vert\delta\Vert\leq u_{n}/2\}}\Vert\delta(s,x)\Vert\mu(ds,dx)\right)\right]\notag\\
&\qquad+\frac{2}{u_{n}}\mathbb{E}\left[N_{i}^{n}(2u_{n})\left(\int_{(i-1)\Delta_{n}}^{i\Delta_{n}}\int_{\{\Vert\delta\Vert\leq u_{n}/2\}}\Vert\delta(s,x)\Vert\mu(ds,dx)\right)\right] \notag\\
&\leq \frac{K}{u_{n}}(\Delta_{n}u_{n}^{-r})^{2}(\Delta_{n}u_{n}^{1-r})+\frac{K'}{u_{n}}(\Delta_{n}u_{n}^{-r})(\Delta_{n}u_{n}^{1-r})\leq K'\Delta_{n}^{2}u_{n}^{-2r}. \label{Eq:EJi^2_2_decomp2}
\end{align}
by \cref{Eq:LevyTail,Eq:smallJumpsl1}, and the fact that $\{\Vert\delta(s,x)\Vert>2u_{n}\}$ and $\{\Vert\delta(s,x)\Vert\leq u_{n}/2\}$ are disjoint.

Combining \cref{Eq:EJi^2_2_decomp1,Eq:EJi^2_2_decomp2,Eq:EJi^2_2_decomp3} in \cref{Eq:EJi^2_2_decomp} leads to
\begin{equation}
\label{Eq:EJi^2_2}
\mathbb{E}\left[\left(\int_{(i-1)\Delta_{n}}^{i\Delta_{n}}\int_{\{\Vert\delta\Vert>2u_{n}\}}\Vert\delta(s,x)\Vert\mu(ds,dx)\right)^{2}\mathbbm{1}_{\{\Vert\Delta_{i}^{n}X\Vert\leq u_{n}\}}\right]\leq K\Delta_{n}^{2}u_{n}^{-2r}.
\end{equation}
Then, combining \cref{Eq:EJi^2_1,Eq:EJi^2_2} in \cref{Eq:EJi^2}, we have
\begin{equation}
\label{Eq:J_i_moment}
\mathbb{E}[|\widetilde{\mathbf{J}}_{i}|^{2}]\leq K\Delta_{n}u_{n}^{2-r}+K'\Delta_{n}^{2}u_{n}^{-2r},
\end{equation}
and therefore
\begin{equation}
\label{Eq:J}
\mathbb{E}[\Vert\widetilde{\mathbf{J}}\Vert^{2}]\leq\sum_{i=1}^{n}\mathbb{E}[|\widetilde{\mathbf{J}}_{i}|^{2}]\leq Ku_{n}^{2-r}+K'\Delta_{n}u_{n}^{-2r}=o(1).
\end{equation}
Therefore, we have $\Vert\widetilde{\mathbf{J}}\Vert=o_{p}(1)$ by Markov's inequality and thus $A_{3}=o_{p}(1)$.

\subsubsection*{(iv) $\bm{A_{4}=o_{p}(1)}$}
\label{AP:A4}

We denote the $i'$-th element in $\mathbf{R}^{\top}\mathbf{Z}\in\mathbb{R}^{pK_{n}}$ for $i'(j,k)=(j-1)K_{n}+k$ by
\begin{equation}
\label{Eq:covarianceXZ}
C_{j}^{(k)}=\sum_{i=1}^{n}B_{(i-1)\Delta_{n}}^{(k)}\Delta_{j,i}^{n}X\Delta_{i}^{n}Z\mathbbm{1}_{\{|\Delta_{i}^{n}Y|\leq u_{n}\}}\mathbbm{1}_{\{\Vert\Delta_{i}^{n}X\Vert\leq u_{n}\}}.
\end{equation}
By Appendix \ref{AP:Bspline} (iv),
\begin{equation}
\Delta_{n}\sum_{i=1}^{n}(B_{(i-1)\Delta_{n}}^{(k)})^{2}\overset{\mathbb{P}}{\longrightarrow}\int_{0}^{T}(B_{s}^{(k)})^{2}ds\asymp K_{n}^{-1}.
\end{equation}
Under the orthogonality condition in \cref{As:orthogonality}, $X$ and $Z$ have no co-jumps, and the (fixed-weighted) realized covariance $\sum_{i=1}^{n}B_{(i-1)\Delta_{n}}^{(k)}\Delta_{j,i}^{n}X\Delta_{i}^{n}Z$ has zero limit. By the standard results of realized covariance functionals (see, e.g., Theorem 10.3.2, \citealp{jacod2012discretization}), it holds that
\begin{equation}
\mathbb{E}\left[\left(\sum_{i=1}^{n}B_{(i-1)\Delta_{n}}^{(k)}\Delta_{j,i}^{n}X\Delta_{i}^{n}Z\right)^{2}\right]=O(K_{n}^{-1}\Delta_{n}),
\end{equation}
and thus $\mathbb{E}[(C_{j}^{(k)})^{2}]=O(K_{n}^{-1}\Delta_{n})$. Therefore,
\begin{equation}
\label{Eq:R'Z}
\mathbb{E}[\Vert\mathbf{R}^{\top}\mathbf{Z}\Vert^{2}]=\sum_{l=1}^{p}\sum_{k=1}^{K_{n}}\mathbb{E}[(C_{j}^{(k)})^{2}]=O(\Delta_{n}),
\end{equation}
and thus $\Vert\mathbf{R}^{\top}\mathbf{Z}\Vert=O_{p}(\sqrt{\Delta_{n}})$. By \cref{lemma:norm}, 
\begin{equation}
A_{4}\leq K\sqrt{K_{n}}\Vert\mathbf{R}^{\top}\mathbf{Z}\Vert=O_{p}(\sqrt{K_{n}\Delta_{n}})=o_{p}(1),
\end{equation}
when $\Delta_{n}K_{n}\log K_{n}\to0$.

Therefore, all components in \cref{Eq:normDecomposition} are $o_{p}(1)$, and thus $\Vert\widehat{\beta}-\beta\Vert_{L^{2}}=o_{p}(1)$. This completes the proof.

\subsection[Proof of Theorem 2]{Proof of \cref{Th:CLT}}
\label{AP:Proof_CLT}

%We start with an estimate for the $L_2$-norm of $e$ over $[0,t]$, which is used repeatedly in the proof. 
%By Markov's inequality, for any $\varepsilon>0$, 
%\begin{equation}
%\mathbb{P}\left(\frac{\Vert e\Vert_{L^{2}}}{\sqrt{\Delta_{n}}}>\varepsilon\right)\leq\frac{\mathbb{E}[\Vert e\Vert_{L^{2}}]}{\varepsilon\sqrt{\Delta_{n}}}=O\left(\frac{K_{n}^{-\alpha}}{\varepsilon\sqrt{\Delta_{n}}}\right).
%\end{equation}
%Consequently, as $K_{n}^{\alpha}\sqrt{\Delta_{n}}\to\infty$,  $$\Vert e\Vert_{L^{2}}=o_{p}(\sqrt{\Delta_{n}}).$$  %Therefore, the triangle inequality in \cref{Eq:consistency} implies that $\Vert\widehat{\beta}-\beta\Vert_{L^{2}}\leq O_{p}(\sqrt{\Delta_{n}})$. 

We start with the following decomposition:
%
%(current)
%\begin{equation}
%\label{Eq:decomp} 
%\begin{split}
%\frac{1}{\sqrt{\Delta_{n}}}(\widehat{I\beta}_{T}-I\beta_{T})&= \frac{1}{\sqrt{\Delta_{n}}}\left(\sum_{i=1}^{n}\int_{(i-1)\Delta_{n}}^{i\Delta_{n}}\widehat{\beta}_{s}ds - \int_{0}^{T} \mathbf{B}_s\bm{\gamma}\,ds - \int_{0}^{T} e_s ds \right)\\
%&= \frac{1}{\sqrt{\Delta_{n}}}\left(\int_{0}^{T} \mathbf{B}_sds\right)(\widehat{\bm{\gamma}} - \bm{\gamma})-\frac{1}{\sqrt{\Delta_{n}}} \int_{0}^{T} e_{s} d{s}.
%\end{split}
%\end{equation}
%(original)
%\begin{align}
%\frac{1}{\sqrt{\Delta_{n}}}(\widehat{I\beta}_{T}-I\beta_{T})&= \frac{1}{\sqrt{\Delta_{n}}}\left(\sum_{i=1}^{n}\widehat{\beta}_{(i-1)\Delta_{n}}\Delta_{n} - \int_{0}^{T} \mathbf{B}_s\bm{\gamma}\,ds - \int_{0}^{T} e_s ds \right) \notag\\
%&= \frac{1}{\sqrt{\Delta_{n}}}\left(\sum^n_{i=1 }\mathbf{B}_{(i-1)\Delta_n} \widehat{\bm{\gamma}}\Delta_n -  \int_{0}^{T} \mathbf{B}_s\bm{\gamma}ds\right) - \frac{1}{\sqrt{\Delta_{n}}} \int_{0}^{T} e_s ds \notag\\
%&= \frac{1}{\sqrt{\Delta_{n}}}\left(\int_{0}^{T} \mathbf{B}_sds\right)(\widehat{\bm{\gamma}} - \bm{\gamma}) - \frac{1}{\sqrt{\Delta_{n}}}\left(\int_{0}^{T} \mathbf{B}_sds-\sum^n_{i=1 }\mathbf{B}_{(i-1)\Delta_n}\Delta_n\right)\widehat{\bm{\gamma}} \notag\\
%&\qquad\qquad-\frac{1}{\sqrt{\Delta_{n}}} \int_{0}^{T} e_s ds. \label{Eq:decomp} 
%\end{align}
%(revised)
\begin{align}
\frac{1}{\sqrt{\Delta_{n}}}(\widehat{I\beta}_{T}-I\beta_{T})&= \frac{1}{\sqrt{\Delta_{n}}}\left(\sum_{i=1}^{n}\widehat{\beta}_{(i-1)\Delta_{n}}\Delta_{n} - \int_{0}^{T} \mathbf{B}_s\bm{\gamma}\,ds - \int_{0}^{T} e_s ds \right) \notag\\
&= \frac{1}{\sqrt{\Delta_{n}}}\left(\sum^n_{i=1 }\mathbf{B}_{(i-1)\Delta_n} \widehat{\bm{\gamma}}\Delta_n -  \int_{0}^{T} \mathbf{B}_s\bm{\gamma}ds\right) - \frac{1}{\sqrt{\Delta_{n}}} \int_{0}^{T} e_s ds \notag\\
&= \frac{1}{\sqrt{\Delta_{n}}}\left(\sum_{i=1}^{n}\mathbf{B}_{(i-1)\Delta_n}\Delta_{n}\right)(\widehat{\bm{\gamma}} - \bm{\gamma}) - \frac{1}{\sqrt{\Delta_{n}}}\left( \int_{0}^{T} \mathbf{B}_sds - \sum^n_{i=1 }\mathbf{B}_{(i-1)\Delta_n}\Delta_n\right)\bm{\gamma} \notag\\
&\qquad\qquad-\frac{1}{\sqrt{\Delta_{n}}} \int_{0}^{T} e_s ds. \label{Eq:decomp}
\end{align}
To derive the asymptotic distribution of the $p$-dimensional vector in the first term by the Cram{\'e}r-Wold device, we need to show that, for any fixed linear combination with $a\in\mathbb{R}^{p}$,
\begin{equation}
\label{Eq:Cramer-Wold}
\begin{split}
H_{n}&=\frac{1}{\sqrt{\Delta_{n}}}\,a^{\top}\left(\sum_{i=1}^{n}\mathbf{B}_{(i-1)\Delta_n}\Delta_{n}\right)(\widehat{\bm{\gamma}}-\bm{\gamma})\\
&\qquad\xrightarrow{\mathcal{L}-s}\int_{0}^{T}\sqrt{\widetilde{\sigma}_{s}^{2}}\left(a^{\top}\left(\int_{0}^{T} \mathbf{B}_sds\right)\left(\int_{0}^{T}\mathbf{B}_{s}^{\top}c_{s}\mathbf{B}_{s}ds\right)^{-1}\mathbf{B}_{s}^{\top}\right)^{\top}c_{s}^{1/2}dW'_{s},
\end{split}
\end{equation}
where $W'$ is a standard $p$-dimensional Brownian motion defined on an extension of the original probability space $(\Omega,\mathcal{F},(\mathcal{F}_{t})_{t\geq0},\mathbb{P})$. We also show that the second term (the bias induced by the discretization error) and the third term (the bias induced by the spline approximation error) in \cref{Eq:decomp} are both asymptotically negligible as $\Delta_{n}\to0$. 

%In this proof, we show that the first term converges in law by verifying
%\begin{equation}
%\label{Eq:CLT_gamma}
%\frac{1}{\sqrt{\Delta_{n}}}\left(\int_{0}^{T} \mathbf{B}_sds\right)(\widehat{\bm{\gamma}}-\bm{\gamma})\xrightarrow{\mathcal{L}-s}\left(\int_{0}^{T} \mathbf{B}_sds\right)\left(\int_{0}^{T}\mathbf{B}_{s}^{\top}c_{s}\mathbf{B}_{s}ds\right)^{-1}\int_{0}^{T}\sqrt{\widetilde{\sigma}_{s}^{2}}(\mathbf{B}_{s}^{\top}c_{s}\mathbf{B}_{s})^{1/2}dW'_{s},
%\end{equation}
%and the second term (the bias induced by the spline approximation error) is asymptotic negligible. 

For the fixed linear combination in \cref{Eq:Cramer-Wold}, we consider a further decomposition by \cref{Eq:gammaBias}:
\begin{align}
H_{n} &= \underbrace{\frac{1}{\sqrt{\Delta_{n}}}\,a^{\top}\left(\sum_{i=1}^{n}\mathbf{B}_{(i-1)\Delta_n}\Delta_{n}\right)(\mathbf{R}^{\top}\mathbf{R})^{-1}\mathbf{R}^{\top}(\widetilde{\mathbf{Y}}-\mathbf{R}\bm{\gamma})}_{H_{1,n}}\,+\,\underbrace{\frac{1}{\sqrt{\Delta_{n}}}\,a^{\top}\left(\sum_{i=1}^{n}\mathbf{B}_{(i-1)\Delta_n}\Delta_{n}\right)(\mathbf{R}^{\top}\mathbf{R})^{-1}\mathbf{R}^{\top}\mathbf{U}}_{H_{2,n}} \notag\\
&+\,\underbrace{\frac{1}{\sqrt{\Delta_{n}}}\,a^{\top}\left(\sum_{i=1}^{n}\mathbf{B}_{(i-1)\Delta_n}\Delta_{n}\right)(\mathbf{R}^{\top}\mathbf{R})^{-1}\mathbf{R}^{\top}\mathbf{J}}_{H_{3,n}} \,+\, \frac{1}{\sqrt{\Delta_{n}}}\,a^{\top}\left(\sum_{i=1}^{n}\mathbf{B}_{(i-1)\Delta_n}\Delta_{n}\right)(\mathbf{R}^{\top}\mathbf{R})^{-1}\mathbf{R}^{\top}\mathbf{Z}, \label{Eq:decom}
\end{align}
where $\mathbf{Z}=(\mathbf{Z}_{1},\dots,\mathbf{Z}_{n})^{\top}$ in the fourth term includes the truncation $\mathbbm{1}_{\{|\Delta_{i}^{n}Y|\leq u_{n}\}}$ on each increment of $Y$ rather than on increments of $Z$ (see \cref{Eq:decomp_Yi}). Therefore, we rewrite the fourth term into:
\begin{align}
\frac{1}{\sqrt{\Delta_{n}}}\,a^{\top}\left(\sum_{i=1}^{n}\mathbf{B}_{(i-1)\Delta_n}\Delta_{n}\right)(\mathbf{R}^{\top}\mathbf{R})^{-1}\mathbf{R}^{\top}\mathbf{Z}&=\underbrace{\frac{1}{\sqrt{\Delta_{n}}}\,a^{\top}\left(\sum_{i=1}^{n}\mathbf{B}_{(i-1)\Delta_n}\Delta_{n}\right)(\mathbf{R}^{\top}\mathbf{R})^{-1}\mathbf{R}^{\top}\widetilde{\mathbf{Z}}}_{H_{4,n}}\notag\\
&\qquad+\underbrace{\frac{1}{\sqrt{\Delta_{n}}}\,a^{\top}\left(\sum_{i=1}^{n}\mathbf{B}_{(i-1)\Delta_n}\Delta_{n}\right)(\mathbf{R}^{\top}\mathbf{R})^{-1}\mathbf{R}^{\top}(\mathbf{Z}-\widetilde{\mathbf{Z}})}_{H_{5,n}}, \label{Eq:decom2}
\end{align}
where $\widetilde{\mathbf{Z}}=(\widetilde{\mathbf{Z}}_{1},\dots,\widetilde{\mathbf{Z}}_{n})^{\top}$ with $\widetilde{\mathbf{Z}}_{i}=\Delta_{i}^{n}Z\mathbbm{1}_{\{|\Delta_{i}^{n}Z|\leq u_{n}\}}$. Next, we verify that the leading term for the stable central limit theorem (CLT) in \cref{Eq:Cramer-Wold} is $H_{4,n}$, and the other components $H_{1,n}$, $H_{2,n}$, $H_{3,n}$, and $H_{5,n}$ are asymptotically negligible.

\subsubsection*{(i) Asymptotic distribution of $\bm{H_{4,n}}$}
\label{AP:Proof_CLT_i}

%We define the following vector of random weights:
%\begin{equation}
%\theta_{t}^{n}=\mathbf{B}_{t}(\mathbf{R}^{\top}\mathbf{R})^{-1}\left(\int_{0}^{T}\mathbf{B}_sds\right)^{\top}a\in\mathbb{R}^{p},
%\end{equation}
%By the LLN result in \cref{Eq:LLN} and the continuous mapping theorem, 
%\begin{equation}
%\theta_{t}^{n}\overset{\mathbb{P}}{\longrightarrow}\theta_{t}=\mathbf{B}_{t}\left(\int_{0}^{T}\mathbf{B}_{s}^{\top}c_{s}\mathbf{B}_{s}ds\right)^{-1}\left(\int_{0}^{T}\mathbf{B}_sds\right)^{\top}a\in\mathbb{R}^{p}.
%\end{equation}
We define the following vectors:
\begin{equation}
\label{Eq:w_n}
\bm{w}_{n}=(w_{1,n}^{(1)},\dots,w_{1,n}^{(K_{n})},\dots,w_{p,n}^{(1)},\dots,w_{p,n}^{(K_{n})})^{\top}=(\mathbf{R}^{\top}\mathbf{R})^{-1}\left(\sum_{i=1}^{n}\mathbf{B}_{(i-1)\Delta_n}\Delta_{n}\right)^{\top}a\in\mathbb{R}^{pK_{n}},
\end{equation}
\begin{equation}
\label{Eq:theta_n}
\theta_{n,t}=\mathbf{B}_{t}\bm{w}_{n}=\sum_{k=1}^{K_{n}}B_{t}^{(k)}\begin{pmatrix}
w_{1,n}^{(k)}\\
\vdots\\
w_{p,n}^{(k)}\\
\end{pmatrix}\in\mathbb{R}^{p}.
\end{equation}
By the LLN result in \cref{Eq:LLN}, a standard Riemann approximation for the deterministic integral, and the continuous mapping theorem, $\theta_{n,t}\overset{\mathbb{P}}{\longrightarrow}\theta_{t}=\mathbf{B}_{t}\bm{w}\in\mathbb{R}^{p}$ for any $0\leq t\leq T$, where
\begin{equation}
\bm{w}=(w_{1}^{(1)},\dots,w_{p}^{(K_{n})})^{\top}=\left(\int_{0}^{T}\mathbf{B}_{s}^{\top}c_{s}\mathbf{B}_{s}ds\right)^{-1}\left(\int_{0}^{T}\mathbf{B}_sds\right)^{\top}a\in\mathbb{R}^{pK_{n}}.
\end{equation}
By definition, 
\begin{equation}
H_{4,n}=\frac{1}{\sqrt{\Delta_{n}}}\sum_{i=1}^{n}\theta_{n,(i-1)\Delta_{n}}^{\top}\Delta_{i}^{n}X\Delta_{i}^{n}Z\mathbbm{1}_{\{\Vert\Delta_{i}^{n}X\Vert\leq u_{n}\}}\mathbbm{1}_{\{|\Delta_{i}^{n}Z|\leq u_{n}\}}.
\end{equation}
Under \cref{As:orthogonality}, $X$ and $Z$ have zero continuous covariation and no co-jumps. By the standard CLT for truncated realized covariances (Theorem 13.2.1, \citealp{jacod2012discretization}), and stable Slutsky's theorem, it holds that
\begin{equation}
H_{4,n}\xrightarrow{\mathcal{L}-s}\int_{0}^{T}\sqrt{\widetilde{\sigma}_{s}^{2}}\theta_{s}^{\top}c_{s}^{1/2}dW'_{s}, %=\int_{0}^{T}\sqrt{\widetilde{\sigma}_{s}^{2}}\left(a^{\top}\left(\int_{0}^{T} \mathbf{B}_sds\right)\left(\int_{0}^{T}\mathbf{B}_{s}^{\top}c_{s}\mathbf{B}_{s}ds\right)^{-1}\mathbf{B}_{s}^{\top}\right)^{\top}c_{s}^{1/2}dW'_{s},
\end{equation}
where $W'$ is a standard $p$-dimensional Brownian motion defined on an extension of the original probability space. The conditional variance of the limit equals
\begin{equation}
\begin{split}
\int_{0}^{T}\widetilde{\sigma}_{s}^{2}\theta_{s}^{\top}c_{s}\theta_{s}ds&=a^{\top}\left(\int_{0}^{T}\mathbf{B}_sds\right)\left(\int_{0}^{T}\mathbf{B}_{s}^{\top}c_{s}\mathbf{B}_{s}ds\right)^{-1}\left(\int_{0}^{T}\widetilde{\sigma}_{s}^{2}\mathbf{B}_{s}^{\top}c_{s}\mathbf{B}_{s}ds\right)\\
&\qquad\times\left(\int_{0}^{T}\mathbf{B}_{s}^{\top}c_{s}\mathbf{B}_{s}ds\right)^{-1}\left(\int_{0}^{T}\mathbf{B}_sds\right)^{\top}a\\
&=a^{\top}\Sigma_{T}^{\beta}a. 
\end{split}
\end{equation}

\subsubsection*{(ii) $\bm{H_{1,n}=o_{p}(1)}$}

Decompose $H_{1,n}$ into:
\begin{equation}
\label{Eq:H_1_decomp}
\begin{split}
H_{1,n}&=\frac{1}{\sqrt{\Delta_{n}}}\sum_{i=1}^{n}\bm{w}_{n}^{\top}\mathbf{R}_{i}\mathcal{Y}_{i}\\
&\leq\underbrace{\frac{1}{\sqrt{\Delta_{n}}}\sum_{i=1}^{n}\bm{w}_{n}^{\top}\mathbf{R}_{i}M_{i}^{(1)}}_{H_{1,n}^{(1)}}\,\,-\, \underbrace{\frac{1}{\sqrt{\Delta_{n}}}\sum_{i=1}^{n}\bm{w}_{n}^{\top}\mathbf{R}_{i}M_{i}^{(2)}}_{H_{1,n}^{(2)}} \,\,+\, \underbrace{\frac{1}{\sqrt{\Delta_{n}}}\sum_{i=1}^{n}\bm{w}_{n}^{\top}\mathbf{R}_{i}\widetilde{\mathbf{Y}}_{i}^{(2)}}_{H_{1,n}^{(3)}},
\end{split}
\end{equation}
where $\mathcal{Y}_{i}$ represents the $i$-th element of $\widetilde{\mathbf{Y}}-\mathbf{R}\bm{\gamma}$, and we employ the decomposition of $\mathcal{Y}_{i}$ in \cref{Eq:Y_decomp_1,Eq:Y_decomp_2,Eq:Y_decomp_3}:
\begin{equation}
\mathcal{Y}_{i}=M_{i}^{(1)}-M_{i}^{(2)}+\widetilde{\mathbf{Y}}_{i}^{(2)},
\end{equation}
Next, we verify that all three components on the right-hand side of \cref{Eq:H_1_decomp} are $o_{p}(1)$.

\subsubsection*{(ii.1) $\bm{H_{1,n}^{(1)}=o_{p}(1)}$}

It holds for $H_{1,n}^{(1)}$ that, by Hölder's inequality,
\begin{equation}
\label{Eq:H_1_1_Holder}
|H_{1,n}^{(1)}|=\frac{1}{\sqrt{\Delta_{n}}}\left\vert\sum_{i=1}^{n}\bm{w}_{n}^{\top}\mathbf{R}_{i}M_{i}^{(1)}\right\vert\leq\frac{1}{\sqrt{\Delta_{n}}}\sum_{i=1}^{n}|\bm{w}_{n}^{\top}\mathbf{R}_{i}||M_{i}^{(1)}|\leq\frac{1}{\sqrt{\Delta_{n}}}\sum_{i=1}^{n}\Vert\bm{w}_{n}\Vert_{\infty}\Vert\mathbf{R}_{i}\Vert_{1}|M_{i}^{(1)}|.
\end{equation}
For any $i=1,\dots,n$, by the local support of B-splines (Appendix \ref{AP:Bspline} (iii)), each $\mathbf{R}_{i}$ has only a fixed number of nonzero components, and then the equivalence between $\ell_{1}$- and $\ell_{2}$-norms implies that
\begin{equation}
|H_{1,n}^{(1)}|\leq\frac{K}{\sqrt{\Delta_{n}}}\Vert\bm{w}_{n}\Vert_{\infty}\sum_{i=1}^{n}\Vert\mathbf{R}_{i}\Vert|M_{i}^{(1)}|.
\end{equation}
By the Cauchy-Schwarz inequality,  
\begin{equation}
\label{Eq:H_1_1_CS}
\begin{split}
|H_{1,n}^{(1)}|&\leq\frac{K}{\sqrt{\Delta_{n}}}\Vert\bm{w}_{n}\Vert_{\infty}\sum_{i=1}^{n}\Vert\mathbf{B}_{(i-1)\Delta_{n}}\Vert\Vert\Delta_{i}^{n}X\Vert|M_{i}^{(1)}|\mathbbm{1}_{\{\Vert\Delta_{i}^{n}X\Vert\leq u_{n}\}}\\
&\leq\frac{K}{\sqrt{\Delta_{n}}}\Vert\bm{w}_{n}\Vert_{\infty}\left(K'\max_{1\leq k\leq K_{n}}B_{(i-1)\Delta_{n}}^{(k)}\right)\sum_{i=1}^{n}\Vert\Delta_{i}^{n}X\Vert|M_{i}^{(1)}|\mathbbm{1}_{\{\Vert\Delta_{i}^{n}X\Vert\leq u_{n}\}}\\
&\leq\frac{K}{\sqrt{\Delta_{n}}}\Vert\bm{w}_{n}\Vert_{\infty}\sum_{i=1}^{n}\Vert\Delta_{i}^{n}X\Vert|M_{i}^{(1)}|\mathbbm{1}_{\{\Vert\Delta_{i}^{n}X\Vert\leq u_{n}\}}\\
&\leq\frac{Ku_{n}}{\sqrt{\Delta_{n}}}\Vert\bm{w}_{n}\Vert_{\infty}\sum_{i=1}^{n}|M_{i}^{(1)}|.
\end{split}
\end{equation}

Next, we show that $\Vert\bm{w}_{n}\Vert_{\infty}$ is bounded on an event $E_{n}$ with $\mathbb{P}(E_{n})\to1$. Firstly, the local support of B-splines imply that
\begin{equation}
\label{Eq:norm_Aa}
\begin{split}
\left\Vert\left(\sum_{i=1}^{n}\mathbf{B}_{(i-1)\Delta_n}\Delta_{n}\right)^{\top}a\right\Vert_{\infty}&\leq K\Vert a\Vert_{\infty}\max_{1\leq k\leq K_{n}}\sum_{i=1}^{n}B_{(i-1)\Delta_n}^{(k)}\Delta_{n}\\
&\leq K\Vert a\Vert_{\infty}\sup_{t\in[0,T]}B_{t}^{(k)}\left(\Bigl\lceil\frac{|\text{supp}(B^{(k)})|}{\Delta_{n}}\Bigr\rceil+1\right)\Delta_{n}\leq\frac{K'}{K_{n}},
\end{split}
\end{equation}
where $|\text{supp}(B^{(k)})|\asymp K_{n}^{-1}$; see Appendix \ref{AP:Bspline} (iii). We define a permutation matrix $\mathbf{P}$ that reorders elements in $\mathbf{R}$ (group the $p$ components associated with each spline basis index). Specifically, the $i$-th column in the row-permutated $\mathbf{P}\mathbf{R}^{\top}\in\mathbb{R}^{pK_{n}\times n}$ is given by
\begin{equation}
\begin{pmatrix}
B_{(i-1)\Delta_{n}}^{(1)}\Delta_{1,i}^{n}X\mathbbm{1}_{\{\Vert\Delta_{i}^{n}X\Vert\leq u_{n}\}}\\
\vdots\\
B_{(i-1)\Delta_{n}}^{(1)}\Delta_{p,i}^{n}X\mathbbm{1}_{\{\Vert\Delta_{i}^{n}X\Vert\leq u_{n}\}}\\
\vdots\\
\vdots\\
B_{(i-1)\Delta_{n}}^{(K_{n})}\Delta_{1,i}^{n}X\mathbbm{1}_{\{\Vert\Delta_{i}^{n}X\Vert\leq u_{n}\}}\\
\vdots\\
B_{(i-1)\Delta_{n}}^{(K_{n})}\Delta_{p,i}^{n}X\mathbbm{1}_{\{\Vert\Delta_{i}^{n}X\Vert\leq u_{n}\}}\\
\end{pmatrix}\in\mathbb{R}^{pK_{n}}.
\end{equation}
Then we consider a rescaled and permutated Gram matrix $\mathbf{G}=K_{n}\mathbf{P}(\mathbf{R}^{\top}\mathbf{R})\mathbf{P}^{\top}\in\mathbb{R}^{pK_{n}\times pK_{n}}$. By \cref{Eq:quotient}, the eigenvalues of $\mathbf{G}$ are bounded away from 0 and infinity with probability approaching one, i.e.,
\begin{equation}
\label{Eq:event}
\mathbb{P}(E_{n})\to1,\qquad\text{where }E_{n}=\{0<\underline{\kappa}\leq\lambda_{\text{min}}(\mathbf{G})\leq\lambda_{\text{max}}(\mathbf{G})\leq\overline{\kappa}<\infty\}.
\end{equation}
Importantly, the local support property of B-spline basis implies that our rescaled matrix $\mathbf{G}$ is banded with fixed bandwidth, i.e., we have $\mathbf{G}_{ij}=0$ when $|i-j|>L$ for some $L>0$.\footnote{More strictly speaking, the matrix $\mathbf{G}$ is block-banded with block size $p$. When $p$ is fixed, this block-bandedness automatically implies ordinary (scalar) bandedness by taking a half-bandwidth $L$ large enough to cover the finitely many nonzero diagonal blocks. } A standard result for symmetric positive definite banded matrices from \citet{demko1984decay} is that the inverse has exponentially decaying off-diagonal entries. Specifically in our case, on $E_{n}$, $|(\mathbf{G}^{-1})_{ij}|\leq K\rho^{|i-j|}$ for some $\rho\in(0,1)$, and therefore
\begin{equation}
\label{Eq:norm_G_inverse}
\Vert\mathbf{G}^{-1}\Vert_{\infty}\leq\max_{i}\sum_{j=1}^{pK_{n}}|(\mathbf{G}^{-1})_{ij}|\leq K\left(1+2\sum_{m\geq1}\rho^{m}\right)=K\left(1+\frac{2\rho}{1-\rho}\right)\leq \frac{K'}{1-\rho}.
\end{equation}
Therefore, with
\begin{equation}
\mathbf{P}\bm{w}_{n}=\mathbf{P}(\mathbf{R}^{\top}\mathbf{R})^{-1}\left(\sum_{i=1}^{n}\mathbf{B}_{(i-1)\Delta_n}\Delta_{n}\right)^{\top}a=(K_{n}\mathbf{P}(\mathbf{R}^{\top}\mathbf{R})\mathbf{P}^{\top})^{-1}K_{n}\mathbf{P}\left(\sum_{i=1}^{n}\mathbf{B}_{(i-1)\Delta_n}\Delta_{n}\right)^{\top}a,
\end{equation}
it holds that on $E_{n}$, by \cref{Eq:norm_Aa,Eq:norm_G_inverse}, 
\begin{equation}
\label{Eq:w_infty}
\Vert\bm{w}_{n}\Vert_{\infty}\leq\Vert\mathbf{P}\bm{w}_{n}\Vert_{\infty}\leq\Vert\mathbf{G}^{-1}\Vert_{\infty}\left\Vert K_{n}\mathbf{P}\left(\sum_{i=1}^{n}\mathbf{B}_{(i-1)\Delta_n}\Delta_{n}\right)^{\top}a\right\Vert_{\infty}\leq K. 
\end{equation}
By \cref{Eq:H_1_1_CS,Eq:w_infty}, we have
\begin{equation}
\label{Eq:H_1_square}
\begin{split}
\mathbb{E}[|H_{1,n}^{(1)}|^{2}\mathbbm{1}_{E_{n}}]\leq\frac{Ku_{n}^{2}}{\Delta_{n}}\mathbb{E}\left[\Vert\bm{w}_{n}\Vert_{\infty}^{2}\mathbbm{1}_{E_{n}}\left(\sum_{i=1}^{n}|M_{i}^{(1)}|\right)^{2}\right]\leq\frac{K'u_{n}^{2}}{\Delta_{n}}\mathbb{E}\left[\left(\sum_{i=1}^{n}|M_{i}^{(1)}|\right)^{2}\right],
\end{split}
\end{equation}
so it suffices to show
\begin{equation}
\label{Eq:Y_sum_expectation}
\mathbb{E}\left[\left(\sum_{i=1}^{n}|M_{i}^{(1)}|\right)^{2}\right]=o\left(\frac{\Delta_{n}}{u_{n}^{2}}\right),\qquad\text{such that}\quad\mathbb{E}[|H_{1,n}^{(1)}|^{2}\mathbbm{1}_{E_{n}}]=o(1).
\end{equation}
%We employ the decomposition of $\mathcal{Y}_{i}$ in \cref{Eq:Y_decomp_1,Eq:Y_decomp_2,Eq:Y_decomp_3}:
%\begin{equation}
%\mathcal{Y}_{i}=M_{i}^{(1)}-M_{i}^{(2)}+\widetilde{\mathbf{Y}}_{i}^{(2)},
%\end{equation}
%and thus the expectation in \cref{Eq:Y_sum_expectation} can be written into a sum of three $\ell_{1}$-squared terms:
%\begin{equation}
%\label{Eq:3norms}
%\left(\sum_{i=1}^{n}|\mathcal{Y}_{i}|\right)^{2}\leq3\left(\sum_{i=1}^{n}|M_{i}^{(1)}|\right)^{2}+3\left(\sum_{i=1}^{n}|M_{i}^{(2)}|\right)^{2}+3\left(\sum_{i=1}^{n}|\widetilde{\mathbf{Y}}_{i}^{(2)}|\right)^{2}.
%\end{equation}
%
%or the $\ell_{1}$-squared term of $M_{i}^{(2)}$, we define the subset of intervals,
We define the subset of intervals,
\begin{equation}
\mathcal{I}_{n}=\{i:|\Delta_{i}^{n}Y|> u_{n}\text{ or }\Vert\Delta_{i}^{n}X\Vert>u_{n}\},
\end{equation}
and apply the Cauchy-Schwarz inequality on that subset:
\begin{equation}
\begin{split}
\left(\sum_{i=1}^{n}|M_{i}^{(1)}|\right)^{2}=\left(\sum_{i\in\mathcal{I}_{n}}|M_{i}^{(1)}|\right)^{2}\leq |\mathcal{I}_{n}|\sum_{i\in \mathcal{I}_{n}}(M_{i}^{(1)})^{2}.
\end{split}
\end{equation}
Therefore, it holds that
\begin{equation}
\label{Eq:M_1_square}
\begin{split}
\mathbb{E}\left[\left(\sum_{i=1}^{n}|M_{i}^{(1)}|\right)^{2}\right]&\leq\mathbb{E}\left[|\mathcal{I}_{n}|\sum_{i\in\mathcal{I}_{n}}(M_{i}^{(1)})^{2}\right]\leq\mathbb{E}\left[|\mathcal{I}_{n}|\left(|\mathcal{I}_{n}|\sum_{i\in\mathcal{I}_{n}}(M_{i}^{(1)})^{4}\right)^{1/2}\right]\\
&\leq\mathbb{E}\left[|\mathcal{I}_{n}|^{3/2}\left(\sum_{i=1}^{n}(M_{i}^{(1)})^{4}\right)^{1/2}\right]\leq(\mathbb{E}[|\mathcal{I}_{n}|^{3}])^{1/2}\left(\sum_{i=1}^{n}\mathbb{E}[(M_{i}^{(1)})^{4}]\right)^{1/2}\\
&\leq K(u_{n}^{-3r})^{1/2}(\Delta_{n}^{2}u_{n}^{-r})^{1/2} = K\Delta_{n}u_{n}^{-2r} = o(\Delta_{n}u_{n}^{-2}),
\end{split}
\end{equation}
where $\mathbb{E}[|\mathcal{I}_{n}|^{3}]=O(u_{n}^{-3r})$ can be proved in a similar way to \cref{Eq:Prob2}, and $\mathbb{E}[(M_{i}^{(1)})^{4}]=\Delta_{n}^{3}u_{n}^{-r}$ can be verified by the same steps as in \cref{Eq:A1_Eq2} by the BDG inequality and \cref{lemma:truncationEvents} (i). 

Combining \cref{Eq:H_1_square,Eq:M_1_square} leads to
\begin{equation}
\mathbb{E}[|H_{1,n}^{(1)}|^{2}\mathbbm{1}_{E_{n}}]=o(1).
\end{equation}
By Markov's inequality, for any $\varepsilon>0$,
\begin{equation}
\label{Eq:Chebyshev}
\mathbb{P}(|H_{1,n}^{(1)}|>\varepsilon)\leq\mathbb{P}(|H_{1,n}^{(1)}|>\varepsilon,E_{n})+\mathbb{P}(E_{n}^{\complement})\leq\frac{\mathbb{E}[|H_{1,n}^{(1)}|^{2}\mathbbm{1}_{E_{n}}]}{\varepsilon^{2}}+\mathbb{P}(E_{n}^{\complement})=o(1),
\end{equation}
and therefore $H_{1,n}^{(1)}=o_{p}(1)$.

\subsubsection*{(ii.2) $\bm{H_{1,n}^{(2)}=o_{p}(1)}$}

We follow the same steps as in (ii.1) from \cref{Eq:H_1_1_Holder} to \cref{Eq:Y_sum_expectation}, and it remains to show 
\begin{equation}
\mathbb{E}\left[\left(\sum_{i=1}^{n}|M_{i}^{(2)}|\right)^{2}\right]=o\left(\frac{\Delta_{n}}{u_{n}^{2}}\right),\qquad\text{such that}\quad\mathbb{E}[|H_{1,n}^{(2)}|^{2}\mathbbm{1}_{E_{n}}]=o(1).
\end{equation}
We consider
\begin{equation}
|M_{i}^{(2)}|\leq K\underbrace{\Vert\Delta_{i}^{n}X-\Delta_{i}^{n}X^{c}\Vert\mathbbm{1}_{\{\Vert\Delta_{i}^{n}X\Vert\leq u_{n}\}}}_{D_{i}},
\end{equation}
so that it suffices to show
\begin{equation}
\mathbb{E}\left[\left(\sum_{i=1}^{n}D_{i}\right)^{2}\right]=o\left(\frac{\Delta_{n}}{u_{n}^{2}}\right).
\end{equation}
%We consider the following event decomposition:
We consider the same event decomposition as in \cref{Eq:event_decomp1,Eq:event_decomp2,Eq:V_decomp}:
\begin{equation}
\begin{split}
\{\Vert\Delta_{i}^{n}X\Vert\leq u_{n}\}=\{\Vert\Delta_{i}^{n}X\Vert\leq u_{n},\Vert\Delta_{i}^{n}X^{c}\Vert\leq u_{n}\}\cup\{\Vert\Delta_{i}^{n}X\Vert\leq u_{n},\Vert\Delta_{i}^{n}X^{c}\Vert>u_{n}\}.
\end{split}
\end{equation}
For each $i$, on the event $\{\Vert\Delta_{i}^{n}X\Vert\leq u_{n},\Vert\Delta_{i}^{n}X^{c}\Vert\leq u_{n}\}$,
\begin{equation}
\Vert\Delta_{i}^{n}X-\Delta_{i}^{n}X^{c}\Vert\leq\Vert\Delta_{i}^{n}X\Vert+\Vert\Delta_{i}^{n}X^{c}\Vert\leq 2u_{n}. 
\end{equation}
Therefore,
\begin{equation}
D_{i}\,\leq\,\underbrace{\Vert\Delta_{i}^{n}X-\Delta_{i}^{n}X^{c}\Vert\mathbbm{1}_{\{\Vert\Delta_{i}^{n}X-\Delta_{i}^{n}X^{c}\Vert\leq 2u_{n}\}}}_{D_{i}^{(1)}}\,+\,\underbrace{\Vert\Delta_{i}^{n}X-\Delta_{i}^{n}X^{c}\Vert\mathbbm{1}_{\{\Vert\Delta_{i}^{n}X\Vert\leq u_{n}\}}\mathbbm{1}_{\{\Vert\Delta_{i}^{n}X^{c}\Vert>u_{n}\}}}_{D_{i}^{(2)}},
\end{equation}
and also
\begin{equation}
\label{Eq:Eq:D_1_D_2}
\left(\sum_{i=1}^{n}D_{i}\right)^{2}\leq2\left(\sum_{i=1}^{n}D_{i}^{(1)}\right)^{2}+2\left(\sum_{i=1}^{n}D_{i}^{(2)}\right)^{2}.
\end{equation}
For $D_{i}^{(1)}$, we follow the same steps in the proof of \cref{lemma:truncationEvents} (ii) (see Eqs.~(\ref{Eq:V1_start}) to (\ref{Eq:V1_final})):
\begin{equation}
\begin{split}
D_{i}^{(1)}&\leq\Vert\Delta_{i}^{n}X-\Delta_{i}^{n}X^{c}\Vert\wedge2u_{n}\leq\int_{(i-1)\Delta_n}^{i\Delta_n}\int_{\mathbb{R}^p}(\|\delta(s,x)\|\wedge 2u_n)\mu(ds,dx).
%&\leq\int_{(i-1)\Delta_{n}}^{i\Delta_{n}}ds\int_{\{\Vert\delta\Vert\leq 2u_{n}\}}\Vert\delta(s,x)\Vert\lambda(dx)+2u_{n}\int_{(i-1)\Delta_{n}}^{i\Delta_{n}}ds\int_{\{\Vert\delta\Vert>2u_{n}\}}\lambda(dx).
%%&\leq\int_{(i-1)\Delta_{n}}^{i\Delta_{n}}ds\int_{\{\Vert\delta\Vert\leq 2u_{n}\}}\Vert\delta(s,x)\Vert\lambda(dx)+2u_{n}\int_{(i-1)\Delta_{n}}^{i\Delta_{n}}ds\int_{\{f_{m}>2u_{n}\}}\lambda(dx).
\end{split}
\end{equation}
%By the standard Lévy tail bound in \cref{Eq:LevyTail} and \cref{As:Ito} (v),
%\begin{equation}
%\label{Eq:smallJumpsl1}
%\begin{split}
%\int_{\{\Vert\delta\Vert\leq u\}}\Vert\delta(s,x)\Vert\lambda(dx)&=\int_{\mathbb{R}^p}\Vert\delta(s,x)\Vert\mathbbm{1}_{\{\Vert\delta(s,x)\Vert\leq u\}}\lambda(dx)=\int_{\mathbb{R}^p}\left(\int_{0}^{\Vert\delta(s,x)\Vert}\mathbbm{1}_{\{y\leq u\}}dy\right)\lambda(dx)\\
%&=\int_{0}^{u}\left(\int_{\mathbb{R}^p}\mathbbm{1}_{\{\Vert\delta(s,x)\Vert>y\}}\lambda(dx)\right)dy-u\int_{\mathbb{R}^p}\mathbbm{1}_{\{\Vert\delta(s,x)\Vert>u\}}\lambda(dx)\\
%&\leq\int_{0}^{u}\left(\int_{\mathbb{R}^p}\mathbbm{1}_{\{\Vert\delta(s,x)\Vert>y\}}\lambda(dx)\right)dy\\
%&\leq\int_{0}^{u} y^{-r}\,dy \int_{\mathbb{R}^p}(\Vert\delta(s,x)\Vert^{r}\wedge1)\lambda(dx)\\
%&\leq\left(\int_{\mathbb{R}^p} f_{m}(x)\lambda(dx)\right)\int_{0}^{u} y^{-r}dy\leq Ku^{1-r}.
%\end{split}
%\end{equation}
Then we have
\begin{equation}
\label{Eq:D1_split}
\begin{split}
\mathbb{E}\left[\left(\sum_{i=1}^{n}D_{i}^{(1)}\right)^{2}\right]&\leq\mathbb{E}\left[\left(\int_{0}^{T}\int_{\mathbb{R}^p}(\|\delta(s,x)\|\wedge 2u_n)\mu(ds,dx)\right)^{2}\right]\\
&\leq2\mathbb{E}\left[\left(\int_{0}^{T}\int_{\mathbb{R}^p} (\|\delta(s,x)\|\wedge 2u_n)(\mu-\nu)(ds,dx)\right)^2\right]\\
&\qquad+2\mathbb{E}\left[\left(\int_{0}^{T}\int_{\mathbb{R}^p}(\|\delta(s,x)\|\wedge 2u_n)\nu(ds,dx)\right)^2\right].
\end{split}
\end{equation}
For the first term in \cref{Eq:D1_split}, by the It\^{o} isometry for compensated Poisson integrals,
\begin{equation}
\label{Eq:D1_1}
\begin{split}
&\mathbb{E}\left[\left(\int_{0}^{T}\int_{\mathbb{R}^p} (\|\delta(s,x)\|\wedge 2u_n)(\mu-\nu)(ds,dx)\right)^2\right]=\mathbb{E}\left[\int_{0}^{T}\int_{\mathbb{R}^p} (\|\delta(s,x)\|\wedge 2u_n)^{2}\nu(ds,dx)\right]\\
&\qquad\leq\mathbb{E}\left[\int_{0}^{T}ds\left(\int_{\{\|\delta\|\le 2u_n\}}\|\delta(s,x)\|^2\lambda(dx)+4u_n^2\int_{\{\|\delta\|>2u_n\}}\lambda(dx)\right)\right]\leq Ku_{n}^{2-r},
\end{split}
\end{equation}
by \cref{Eq:smallJumps,Eq:LevyTail}. For the second term in \cref{Eq:D1_split}, by Jensen's inequality,
\begin{equation}
\label{Eq:D1_2}
\begin{split}
&\mathbb{E}\left[\left(\int_{0}^{T}\int_{\mathbb{R}^p} (\|\delta(s,x)\|\wedge 2u_n)\nu(ds,dx)\right)^2\right]\leq T\int_{0}^{T}\mathbb{E}\left[\left(\int_{\mathbb{R}^p}(\|\delta(s,x)\|\wedge 2u_n)\lambda(dx)\right)^2\right] ds\\
&\qquad\leq T\int_{0}^{T}\mathbb{E}\left[\left(\int_{\{\|\delta\|\le 2u_n\}}\|\delta(s,x)\|\lambda(dx)+2u_n\int_{\{\|\delta\|>2u_n\}}\lambda(dx)\right)^2\right]ds\leq Ku_{n}^{2-2r},
\end{split}
\end{equation}
by \cref{Eq:LevyTail,Eq:smallJumpsl1}. Combining \cref{Eq:D1_split,Eq:D1_1,Eq:D1_2}, we obtain
\begin{equation}
\label{Eq:D_1_square}
\mathbb{E}\left[\left(\sum_{i=1}^{n}D_{i}^{(1)}\right)^{2}\right]\leq Ku_n^{2-r}+K'u_n^{2-2r}\le K u_n^{2-2r}.
\end{equation}
%Therefore, by \cref{Eq:LevyTail,Eq:smallJumpsl1},
%\begin{align}
%\mathbb{E}\left[\left(\sum_{i=1}^{n}D_{i}^{(1)}\right)^{2}\right]&\leq\mathbb{E}\left[\left(\int_{0}^{T}ds\int_{\{\Vert\delta\Vert\leq 2u_{n}\}}\Vert\delta(s,x)\Vert\lambda(dx)+2u_{n}\int_{0}^{T}ds\int_{\{f_{m}>2u_{n}\}}\lambda(dx)\right)^{2}\right] \notag\\
%&\leq \mathbb{E}[(Ku_{n}^{1-r}+K'u_{n}^{1-r})^{2}]\leq Ku_{n}^{2-2r}=o\left(\frac{\Delta_{n}}{u_{n}^{2}}\right). \label{Eq:D_1_square}
%\end{align}
For $D_{i}^{(2)}$, on the event $\{\Vert\Delta_{i}^{n}X\Vert\leq u_{n},\Vert\Delta_{i}^{n}X^{c}\Vert>u_{n}\}$,
\begin{equation}
\Vert\Delta_{i}^{n}X-\Delta_{i}^{n}X^{c}\Vert\leq u_{n}+\Vert\Delta_{i}^{n}X^{c}\Vert\leq2\Vert\Delta_{i}^{n}X^{c}\Vert,
\end{equation}
and thus it holds that
\begin{align}
\mathbb{E}\left[\left(\sum_{i=1}^{n}D_{i}^{(2)}\right)^{2}\right]&\leq\mathbb{E}\left[\left(2\sum_{i=1}^{n}\Vert\Delta_{i}^{n}X^{c}\Vert\mathbbm{1}_{\{\Vert\Delta_{i}^{n}X^{c}\Vert>u_{n}\}}\right)^{2}\right]\leq\mathbb{E}\left[4n\sum_{i=1}^{n}\Vert\Delta_{i}^{n}X^{c}\Vert^{2}\mathbbm{1}_{\{\Vert\Delta_{i}^{n}X^{c}\Vert>u_{n}\}}\right] \notag\\
&=4n\sum_{i=1}^{n}\mathbb{E}\left[\Vert\Delta_{i}^{n}X^{c}\Vert^{2}\mathbbm{1}_{\{\Vert\Delta_{i}^{n}X^{c}\Vert>u_{n}\}}\right]. \label{Eq:D_2_bound}
\end{align}
For any $q>2$,
\begin{equation}
\mathbb{E}\left[\Vert\Delta_{i}^{n}X^{c}\Vert^{2}\mathbbm{1}_{\{\Vert\Delta_{i}^{n}X^{c}\Vert>u_{n}\}}\right]\leq\mathbb{E}\left[\Vert\Delta_{i}^{n}X^{c}\Vert^{2}\left(\frac{\Vert\Delta_{i}^{n}X^{c}\Vert}{u_{n}}\right)^{q-2}\right]\leq\frac{\mathbb{E}[\Vert\Delta_{i}^{n}X^{c}\Vert^{q}]}{u_{n}^{q-2}}\leq\Delta_{n}^{q/2}u_{n}^{2-q},
\end{equation}
by the BDG inequality. Therefore, \cref{Eq:D_2_bound} turns into
\begin{equation}
 \label{Eq:D_2_square}
\mathbb{E}\left[\left(\sum_{i=1}^{n}D_{i}^{(2)}\right)^{2}\right]\leq\Delta_{n}^{q/2-2}u_{n}^{2-q}=o\left(\frac{\Delta_{n}}{u_{n}^{2}}\right),\qquad\text{if }q>\frac{3-4\varpi}{1/2-\varpi},
\end{equation}
which is always possible because $\varpi<1/2$. By \cref{Eq:Eq:D_1_D_2,Eq:D_1_square,Eq:D_2_square},
\begin{equation}
\mathbb{E}\left[\left(\sum_{i=1}^{n}D_{i}\right)^{2}\right]\leq2\mathbb{E}\left[\left(\sum_{i=1}^{n}D_{i}^{(1)}\right)^{2}\right]+2\mathbb{E}\left[\left(\sum_{i=1}^{n}D_{i}^{(2)}\right)^{2}\right]=o\left(\frac{\Delta_{n}}{u_{n}^{2}}\right),
\end{equation}
and therefore 
\begin{equation}
\label{Eq:M_2_square}
\mathbb{E}\left[\left(\sum_{i=1}^{n}|M_{i}^{(2)}|\right)^{2}\right]=o\left(\frac{\Delta_{n}}{u_{n}^{2}}\right),
\end{equation}
so that $\mathbb{E}[|H_{1,n}^{(1)}|^{2}\mathbbm{1}_{E_{n}}]=o(1)$. By the same inequality as in \cref{Eq:Chebyshev}, we conclude $H_{1,n}^{(2)}=o_{p}(1)$.

\subsubsection*{(ii.3) $\bm{H_{1,n}^{(3)}=o_{p}(1)}$}

Write $H_{1,n}^{(3)}$ into
\begin{equation}
H_{1,n}^{(3)}=\frac{1}{\sqrt{\Delta_{n}}}\sum_{i=1}^{n}\bm{w}_{n}^{\top}\mathbf{R}_{i}\left(\int_{(i-1)\Delta_{n}}^{i\Delta_{n}}(\widetilde{\beta}_{s}-\widetilde{\beta}_{(i-1)\Delta_{n}})^{\top}dX_{s}^{c}\right)\mathbbm{1}_{\{|\Delta_{i}^{n}Y|\leq u_{n}\}}.
\end{equation}
%For $H_{1,n}^{(3)}$, it holds that
%\begin{equation}
%\begin{split}
%\label{Eq:H_1_3_0}
%H_{1,n}^{(3)}&=\frac{1}{\sqrt{\Delta_{n}}}\sum_{i=1}^{n}\bm{w}_{n}^{\top}\mathbf{R}_{i}\left(\int_{(i-1)\Delta_{n}}^{i\Delta_{n}}(\widetilde{\beta}_{s}-\widetilde{\beta}_{(i-1)\Delta_{n}})^{\top}dX_{s}^{c}\right)\mathbbm{1}_{\{|\Delta_{i}^{n}Y|\leq u_{n}\}}\\
%&\leq\frac{1}{\sqrt{\Delta_{n}}}\sum_{i=1}^{n}\bm{w}_{n}^{\top}\mathbf{R}_{i}\int_{(i-1)\Delta_{n}}^{i\Delta_{n}}(\widetilde{\beta}_{s}-\widetilde{\beta}_{(i-1)\Delta_{n}})^{\top}dX_{s}^{c}.
%\end{split}
%\end{equation}
By the Cauchy-Schwarz inequality, %$|a^{\top}b|\leq\Vert a\Vert\Vert b\Vert$, 
\begin{equation}
\label{Eq:H_1_3}
|H_{1,n}^{(3)}|\leq\frac{1}{\sqrt{\Delta_{n}}}\left(\sum_{i=1}^{n}(\bm{w}_{n}^{\top}\mathbf{R}_{i})^{2}\right)^{1/2}\left(\sum_{i=1}^{n}\left(\int_{(i-1)\Delta_{n}}^{i\Delta_{n}}(\widetilde{\beta}_{s}-\widetilde{\beta}_{(i-1)\Delta_{n}})^{\top}dX_{s}^{c}\right)^{2}\right)^{1/2}.
\end{equation}
For the first multiplier on the right-hand side of \cref{Eq:H_1_3},
\begin{equation}
\label{Eq:H_1_3_1}
\begin{split}
\sum_{i=1}^{n}(\bm{w}_{n}^{\top}\mathbf{R}_{i})^{2}&=\bm{w}_{n}^{\top}(\mathbf{R}^{\top}\mathbf{R})\bm{w}_{n}=a^{\top}\left(\sum^n_{i=1 }\mathbf{B}_{(i-1)\Delta_n}\Delta_{n}\right)(\mathbf{R}^{\top}\mathbf{R})^{-1}\left(\sum^n_{i=1 }\mathbf{B}_{(i-1)\Delta_n}\Delta_{n}\right)^{\top}a\\
&\leq K\Vert a\Vert^{2}\underbrace{\left\Vert\sum^n_{i=1}\mathbf{B}_{(i-1)\Delta_n}\Delta_{n}\right\Vert^{2}}_{O(K_{n}^{-1})}\Vert(\mathbf{R}^{\top}\mathbf{R})^{-1}\Vert=O_{p}(1),
\end{split}
\end{equation}
because each basis component $\sum^n_{i=1}B_{(i-1)\Delta_n}^{(k)}\Delta_{n}\leq\sup_{t\in[0,T]}B_{t}^{(k)}\left(\Bigl\lceil\frac{|\text{supp}(B^{(k)})|}{\Delta_{n}}\Bigr\rceil+1\right)\Delta_{n}=O(K_{n}^{-1})$, and there are $K_{n}$ of them in the Euclidean norm. For the other term, by \cref{Eq:ItoIsometryRevised},
\begin{equation}
\label{Eq:ItoIsometry}
\mathbb{E}\left[\left(\left.\int_{(i-1)\Delta_{n}}^{i\Delta_{n}}(\widetilde{\beta}_{s}-\widetilde{\beta}_{(i-1)\Delta_{n}})^{\top}dX_{s}^{c}\right)^{2}\right\vert\mathcal{F}_{(i-1)\Delta_{n}}\right]\leq K\int_{(i-1)\Delta_{n}}^{i\Delta_{n}}\mathbb{E}[\Vert\widetilde{\beta}_{s}-\widetilde{\beta}_{(i-1)\Delta_{n}}\Vert^{2}]ds.
\end{equation}
By the Cauchy-Schwarz inequality, 
\begin{equation}
\Vert\widetilde{\beta}_{s}-\widetilde{\beta}_{(i-1)\Delta_{n}}\Vert^{2}\leq (s-(i-1)\Delta_{n})\int_{(i-1)\Delta_{n}}^{s}\Vert\widetilde{\beta}'_{u}\Vert^{2}du\leq\Delta_{n}\int_{(i-1)\Delta_{n}}^{i\Delta_{n}}\Vert\widetilde{\beta}'_{u}\Vert^{2}du.
\end{equation}
Note that here we adopt a different approach from \cref{Eq:Discretization-3}, and avoid using the probabilistic order bound for $\Vert\widetilde{\beta}'\Vert_{\infty}$. Then,
\begin{equation}
\label{Eq:DiscretizationBound}
\int_{(i-1)\Delta_{n}}^{i\Delta_{n}}\Vert\widetilde{\beta}_{s}-\widetilde{\beta}_{(i-1)\Delta_{n}}\Vert^{2}ds\leq\Delta_{n}^{2}\int_{(i-1)\Delta_{n}}^{i\Delta_{n}}\Vert\widetilde{\beta}'_{u}\Vert^{2}du.
\end{equation}
Combining \cref{Eq:ItoIsometry,Eq:DiscretizationBound} leads to
\begin{equation}
\mathbb{E}\left[\left(\left.\int_{(i-1)\Delta_{n}}^{i\Delta_{n}}(\widetilde{\beta}_{s}-\widetilde{\beta}_{(i-1)\Delta_{n}})^{\top}dX_{s}^{c}\right)^{2}\right\vert\mathcal{F}_{(i-1)\Delta_{n}}\right]\leq K\Delta_{n}^{2}\int_{(i-1)\Delta_{n}}^{i\Delta_{n}}\mathbb{E}[\Vert\widetilde{\beta}'_{u}\Vert^{2}]du,
\end{equation}
and thus
\begin{equation}
\mathbb{E}\left[\sum_{i=1}^{n}\left(\int_{(i-1)\Delta_{n}}^{i\Delta_{n}}(\widetilde{\beta}_{s}-\widetilde{\beta}_{(i-1)\Delta_{n}})^{\top}dX_{s}^{c}\right)^{2}\right]\leq K\Delta_{n}^{2}\int_{0}^{T}\mathbb{E}[\Vert\widetilde{\beta}'_{u}\Vert^{2}]du=K\Delta_{n}^{2}\mathbb{E}[\Vert\widetilde{\beta}'\Vert_{L^{2}}^{2}].
\end{equation}
Therefore, by Markov's inequality,
\begin{equation}
\label{Eq:H_1_3_2}
\sum_{i=1}^{n}\left(\int_{(i-1)\Delta_{n}}^{i\Delta_{n}}(\widetilde{\beta}_{s}-\widetilde{\beta}_{(i-1)\Delta_{n}})^{\top}dX_{s}^{c}\right)^{2}=O_{p}(\Delta_{n}^{2}\Vert\widetilde{\beta}'\Vert_{L^{2}}^{2}).
\end{equation}
Combining \cref{Eq:H_1_3_1,Eq:H_1_3_2} in \cref{Eq:H_1_3}, we obtain
\begin{equation}
\label{Eq:H_1_3_bound}
|H_{1,n}^{(3)}|=O_{p}(\sqrt{\Delta_{n}}\Vert\widetilde{\beta}'\Vert_{L^{2}}),
\end{equation}
and it remains to derive the probabilistic order bound for $\Vert\widetilde{\beta}'\Vert_{L^{2}}$. 

\begin{lemma}
\label{lemma:derivative}
Under \cref{As:Ito_beta}, $\Vert\widetilde{\beta}'\Vert_{L^{2}}=O_{p}(\sqrt{K_{n}\log K_{n}})$. 
\end{lemma}

\begin{proof}[Proof of \cref{lemma:derivative}]
	
Instead of working directly on $\widetilde{\beta}'$, we consider the quasi-interpolant $\overline{\beta}$ defined in \cref{Eq:quasiInterpolant} in the \hyperref[AP:Proof_lemma1]{proof} of \cref{lemma:approximationError}, and we have
\begin{equation}
\widetilde{\beta}'=\overline{\beta}' + (\widetilde{\beta}-\overline{\beta})'.
\end{equation} 
with all $\widetilde{\beta}$, $\overline{\beta}$, $\widetilde{\beta}-\overline{\beta}\in\mathbb{G}_{n}$. By the triangle inequality, 
\begin{equation}
\label{Eq:derivative}
\Vert\widetilde{\beta}'\Vert_{L^{2}}\leq\Vert\overline{\beta}'\Vert_{L^{2}}+\Vert(\widetilde{\beta}-\overline{\beta})'\Vert_{L^{2}}.
\end{equation}
We start with a general linear combination of B-spline basis functions of degree $d$: Let $g_{t}=\sum_{k=1}^{K_{n}}B_{d,t}^{(k)}r^{(k)}\in\mathbb{G}_{n}$ for any $r^{(k)}\in\mathbb{R}^{p}$. By the Cox-de Boor recursion identity in \cref{Eq:Cox-deBoorRecursion},
\begin{equation}
\label{Eq:dB}
(B_{d,t}^{(k)})'=\frac{d}{\bar{v}_{k+d-1}-\bar{v}_{k-1}}B_{d-1,t}^{(k)}-\frac{d}{\bar{v}_{k+d}-\bar{v}_{k}}B_{d-1,t}^{(k+1)}.
\end{equation}
Differentiate $g_{t}$ on $t$ with \cref{Eq:dB} plugged in, we obtain
\begin{align}
g'_{t}&=\sum_{k=1}^{K_{n}}\frac{d}{\bar{v}_{k+d-1}-\bar{v}_{k-1}}B_{d-1,t}^{(k)}r^{(k)}-\sum_{k=1}^{K_{n}}\frac{d}{\bar{v}_{k+d}-\bar{v}_{k}}B_{d-1,t}^{(k+1)}r^{(k)} \notag\\
&=\sum_{k=1}^{K_{n}}\frac{d}{\bar{v}_{k+d-1}-\bar{v}_{k-1}}B_{d-1,t}^{(k)}r^{(k)}-\sum_{k=2}^{K_{n}+1}\frac{d}{\bar{v}_{k+d-1}-\bar{v}_{k-1}}B_{d-1,t}^{(k)}r^{(k)} \notag\\
&=\frac{d}{\bar{v}_{d}-\bar{v}_{0}}r^{(1)}+\sum_{k=2}^{K_{n}}\left(\frac{d}{\bar{v}_{k+d-1}-\bar{v}_{k-1}}r^{(k)}-\frac{d}{\bar{v}_{k+d-1}-\bar{v}_{k-1}}r^{(k-1)}\right)B_{d-1,t}^{(k)}-\frac{d}{\bar{v}_{K_{n}+d}-\bar{v}_{K_{n}}}r^{(K_{n})},
\end{align}
where the first and last terms are both 0 by definition; see \cref{Eq:Cox-deBoorRecursion-boundries}. 
Therefore,
\begin{equation}
g'_{t}=\sum_{k=2}^{K_{n}}\frac{dB_{d-1,t}^{(k)}}{\bar{v}_{k+d-1}-\bar{v}_{k-1}}(r^{(k)}-r^{(k-1)})\leq K\sum_{k=1}^{K_{n}-1}(r^{(k+1)}-r^{(k)}),
\end{equation}
and thus
\begin{equation}
\label{Eq:g_order1}
\Vert g'\Vert_{L^{2}}^{2}=\int_{0}^{T}\Vert g'_{t}\Vert^{2}dt\leq KK_{n}\sum_{k=1}^{K_{n}-1}\Vert r^{(k+1)}-r^{(k)}\Vert^{2}\leq K'K_{n}\sum_{k=1}^{K_{n}}\Vert r^{(k)}\Vert^{2}=K'K_{n}\Vert \bm{r}\Vert^{2},
\end{equation}
where $\bm{r}=((r^{(1)})^{\top},\dots,(r^{(K_{n})})^{\top})^{\top}\in\mathbb{R}^{pK_{n}}$. On the other hand, by the B-spline norm equivalence in Appendix \ref{AP:Bspline} (iv),
\begin{equation}
\label{Eq:g_order2}
\Vert g\Vert_{L^{2}}^{2}\asymp\frac{\Vert\bm{r}\Vert^{2}}{K_{n}}.
\end{equation}
Combining \cref{Eq:g_order1,Eq:g_order2} leads to
\begin{equation}
\Vert g'\Vert_{L^{2}}\leq KK_{n}\Vert g\Vert_{L^{2}}. 
\end{equation}
Therefore, if $g=\widetilde{\beta}-\overline{\beta}$, 
\begin{equation}
\Vert(\widetilde{\beta}-\overline{\beta})'\Vert_{L^{2}}\leq KK_{n}\Vert \widetilde{\beta}-\overline{\beta}\Vert_{L^{2}}.
\end{equation}
By the triangle inequality, \cref{Eq:ArrpxError_L2norm_inequalities,Eq:lemma1-0},
\begin{equation}
\Vert\widetilde{\beta}-\overline{\beta}\Vert_{L^{2}}\leq\Vert\beta-\overline{\beta}\Vert_{L^{2}}+\Vert\beta-\widetilde{\beta}\Vert_{L^{2}}\leq K\Vert\beta-\overline{\beta}\Vert_{L^{2}}=O_{p}\left(\sqrt{\frac{\log K_{n}}{K_{n}}}\right).
\end{equation}
Therefore, by the above two inequalities, 
\begin{equation}
\label{Eq:derivative-1}
\Vert(\widetilde{\beta}-\overline{\beta})'\Vert_{L^{2}}=O_{p}(\sqrt{K_{n}\log K_{n}}).
\end{equation}
Then it remains to derive the probabilistic order bound for $\Vert\overline{\beta}'\Vert_{L^{2}}$.

Let $g=\overline{\beta}$, by the first inequality in \cref{Eq:g_order1},
\begin{equation}
\label{Eq:quasi_derivative_L2bound}
\Vert\overline{\beta}'\Vert_{L^{2}}^{2}\leq KK_{n}\sum_{k=1}^{K_{n}-1}\Vert\beta_{\tau^{(k+1)}}-\beta_{\tau^{(k)}}\Vert^{2}.
\end{equation}
Consider the same decomposition as in \cref{Eq:beta_decomp}.  Then we have
\begin{equation}
\label{Eq:bound_beta_decomp}
\sum_{k=1}^{K_{n}-1}\Vert\beta_{\tau^{(k+1)}}-\beta_{\tau^{(k)}}\Vert^{2}\leq 2\sum_{k=1}^{K_{n}-1}\Vert\beta_{\tau^{(k+1)}}^{c}-\beta_{\tau^{(k)}}^{c}\Vert^{2}+2\sum_{k=1}^{K_{n}-1}\Vert\beta_{\tau^{(k+1)}}^{j}-\beta_{\tau^{(k)}}^{j}\Vert^{2}.
\end{equation}
Under \cref{As:Ito_Iocal}, the BDG inequality implies that
\begin{equation}
\mathbb{E}\left[\sum_{k=1}^{K_{n}-1}\Vert\beta_{\tau^{(k+1)}}^{c}-\beta_{\tau^{(k)}}^{c}\Vert^{2}\right]\leq K\sum_{k=1}^{K_{n}-1}(\tau^{(k+1)}-\tau^{(k)})\leq KT,
\end{equation}
and thus, by Markov's inequality,
\begin{equation}
\label{Eq:bound_beta_c}
\sum_{k=1}^{K_{n}-1}\Vert\beta_{\tau^{(k+1)}}^{c}-\beta_{\tau^{(k)}}^{c}\Vert^{2}=O_{p}(1).
\end{equation}
Under \cref{As:Ito_beta} with finite-variation jumps,
\begin{equation}
\label{Eq:bound_beta_j}
\sum_{k=1}^{K_{n}-1}\Vert\beta_{\tau^{(k+1)}}^{j}-\beta_{\tau^{(k)}}^{j}\Vert^{2}\leq\left(\sum_{k=1}^{K_{n}-1}\Vert\beta_{\tau^{(k+1)}}^{j}-\beta_{\tau^{(k)}}^{j}\Vert\right)^{2}\leq V_{\beta^{j}}([0,T])^{2}=O_{p}(1),
\end{equation}
where $V_{\beta^{j}}([a,b])$ is defined in \cref{Eq:TV_Jumps}. Combining \cref{Eq:bound_beta_c,Eq:bound_beta_j} in \cref{Eq:bound_beta_decomp}, we obtain
\begin{equation}
\sum_{k=1}^{K_{n}-1}\Vert\beta_{\tau^{(k+1)}}-\beta_{\tau^{(k)}}\Vert^{2}=O_{p}(1),
\end{equation}
and thus, by \cref{Eq:quasi_derivative_L2bound},
\begin{equation}
\label{Eq:derivative-2}
\Vert\overline{\beta}'\Vert_{L^{2}}=O_{p}(\sqrt{K_{n}}).
\end{equation}
Therefore, combining \cref{Eq:derivative-1,Eq:derivative-2} in \cref{Eq:derivative}, we obtain
\begin{equation}
\Vert\widetilde{\beta}'\Vert_{L^{2}}\leq O_{p}(\sqrt{K_{n}})+O_{p}(\sqrt{K_{n}\log K_{n}})=O_{p}(\sqrt{K_{n}\log K_{n}}).
\end{equation}
This completes the proof of \cref{lemma:derivative}. 

\end{proof}

Therefore, by \cref{Eq:H_1_3_bound} and \cref{lemma:derivative},
\begin{equation}
|H_{1,n}^{(3)}|=O_{p}(\sqrt{\Delta_{n}K_{n}\log K_{n}})=o_{p}(1),
\end{equation}
since $\Delta_{n}K_{n}\log K_{n}\to 0$ as required, and thus $H_{1,n}^{(3)}=o_{p}(1)$.

By the results in (ii.1), (ii.2), and (ii.3), we conclude that $H_{1,n}=o_{p}(1)$.

\subsubsection*{(iii) $\bm{H_{2,n}=o_{p}(1)}$}

Recall that, with $\theta_{n,t}$ defined in \cref{Eq:theta_n},
\begin{equation}
H_{2,n}=\frac{1}{\sqrt{\Delta_{n}}}\sum_{i=1}^{n}\theta_{n,(i-1)\Delta_{n}}^{\top}\Delta_{i}^{n}X\left(\int_{(i-1)\Delta_{n}}^{i\Delta_{n}}e_{s}^{\top}dX_{s}^{c}\right)\mathbbm{1}_{\{\Vert\Delta_{i}^{n}X\Vert\leq u_{n}\}}\mathbbm{1}_{\{|\Delta_{i}^{n}Y|\leq u_{n}\}}.
\end{equation}
On the high-probability event $E_{n}$ in \cref{Eq:event}, it holds that $\Vert\bm{w}_{n}\Vert_{\infty}\leq K$ by \cref{Eq:w_infty}. By local support and boundedness of B-splines, on $E_{n}$,
\begin{equation}
\sup_{t\in[0,T]}\Vert\theta_{n,t}\Vert=\sup_{t\in[0,T]}\left\Vert\sum_{k=1}^{K_{n}}B_{t}^{(k)}\bm{w}_{n}\right\Vert\leq K\Vert\bm{w}_{n}\Vert_{\infty}\leq K',
\end{equation}
and $\Vert\theta'_{n}\Vert_{\infty}=O_{p}(K_{n})$ by \cref{Eq:Discretization-3}. Consider $H_{2,n}=H_{2,n}^{(1)}+H_{2,n}^{(2)}$ with
\begin{align}
H_{2,n}^{(1)}&=\frac{1}{\sqrt{\Delta_{n}}}\sum_{i=1}^{n}\theta_{n,(i-1)\Delta_{n}}^{\top}\Delta_{i}^{n}X^{c}\left(\int_{(i-1)\Delta_{n}}^{i\Delta_{n}}e_{s}^{\top}dX_{s}^{c}\right)\mathbbm{1}_{\{\Vert\Delta_{i}^{n}X\Vert\leq u_{n}\}}\mathbbm{1}_{\{|\Delta_{i}^{n}Y|\leq u_{n}\}},\\
H_{2,n}^{(2)}&=\frac{1}{\sqrt{\Delta_{n}}}\sum_{i=1}^{n}\theta_{n,(i-1)\Delta_{n}}^{\top}(\Delta_{i}^{n}X-\Delta_{i}^{n}X^{c})\left(\int_{(i-1)\Delta_{n}}^{i\Delta_{n}}e_{s}^{\top}dX_{s}^{c}\right)\mathbbm{1}_{\{\Vert\Delta_{i}^{n}X\Vert\leq u_{n}\}}\mathbbm{1}_{\{|\Delta_{i}^{n}Y|\leq u_{n}\}}.
\end{align}
Next, we show both $H_{2,n}^{(1)}$ and $H_{2,n}^{(2)}$ are $o_{p}(1)$. 

\subsubsection*{(iii.1) $\bm{H_{2,n}^{(1)}=o_{p}(1)}$}

Instead of working on $H_{2,n}^{(1)}$ directly, we consider $H_{2,n}^{(1)}=\widetilde{H}_{2,n}^{(1)}-\overline{H}_{2,n}^{(1)}$ with
\begin{align}
\widetilde{H}_{2,n}^{(1)}&=\frac{1}{\sqrt{\Delta_{n}}}\sum_{i=1}^{n}\theta_{n,(i-1)\Delta_{n}}^{\top}\Delta_{i}^{n}X^{c}\left(\int_{(i-1)\Delta_{n}}^{i\Delta_{n}}e_{s}^{\top}dX_{s}^{c}\right), \label{Eq:H_2_1_tilde}\\
\overline{H}_{2,n}^{(1)}&=\frac{1}{\sqrt{\Delta_{n}}}\sum_{i=1}^{n}\theta_{n,(i-1)\Delta_{n}}^{\top}\Delta_{i}^{n}X^{c}\left(\int_{(i-1)\Delta_{n}}^{i\Delta_{n}}e_{s}^{\top}dX_{s}^{c}\right)\left(1-\mathbbm{1}_{\{\Vert\Delta_{i}^{n}X\Vert\leq u_{n}\}}\mathbbm{1}_{\{|\Delta_{i}^{n}Y|\leq u_{n}\}}\right),
\end{align}
and show that both $\widetilde{H}_{2,n}^{(1)}$ and $\overline{H}_{2,n}^{(1)}$ are $o_{p}(1)$. 

We start with $\widetilde{H}_{2,n}^{(1)}$. Rewrite \cref{Eq:H_2_1_tilde} into $\widetilde{H}_{2,n}^{(1)}=\frac{1}{\sqrt{\Delta_{n}}}\sum_{i=1}^{n}U_{i}$ with 
\begin{equation}
U_{i}=\theta_{n,(i-1)\Delta_{n}}^{\top}\Delta_{i}^{n}X^{c}\left(\int_{(i-1)\Delta_{n}}^{i\Delta_{n}}e_{s}^{\top}dX_{s}^{c}\right). 
\end{equation}
By the conditional covariance for continuous martingale increments (with the predictability justified in \cref{remark:Predictability} via predictable projection),
\begin{equation}
\mathbb{E}[U_{i}|\mathcal{F}_{(i-1)\Delta_{n}}]=\theta_{n,(i-1)\Delta_{n}}^{\top}\int_{(i-1)\Delta_{n}}^{i\Delta_{n}}c_{s}e_{s}ds. 
\end{equation}
Consider
\begin{equation}
\begin{split}
\widetilde{H}_{2,n}^{(1)}=\underbrace{\frac{1}{\sqrt{\Delta_{n}}}\sum_{i=1}^{n}\mathbb{E}[U_{i}|\mathcal{F}_{(i-1)\Delta_{n}}]}_{A_{2,n}^{(1)}}\,+\,\underbrace{\frac{1}{\sqrt{\Delta_{n}}}\sum_{i=1}^{n}(U_{i}-\mathbb{E}[U_{i}|\mathcal{F}_{(i-1)\Delta_{n}}])}_{M_{2,n}^{(1)}}.
\end{split}
\end{equation}
For $A_{2,n}^{(1)}$, we have
\begin{equation}
A_{2,n}^{(1)}=\frac{1}{\sqrt{\Delta_{n}}}\int_{0}^{T}\theta_{n,s}^{\top}c_{s}e_{s}ds-\frac{1}{\sqrt{\Delta_{n}}}\sum_{i=1}^{n}\int_{(i-1)\Delta_{n}}^{i\Delta_{n}}(\theta_{n,s}-\theta_{n,(i-1)\Delta_{n}})^{\top}c_{s}e_{s}ds,
\end{equation}
where the first term is 0 by $L^{2}(c)$-orthogonality between $e$ and $\mathbb{G}_{n}$. For the second term, on $E_{n}$,
\begin{equation}
\sup_{(i-1)\Delta_{n}\leq s\leq i\Delta_{n}}\Vert\theta_{n,s}-\theta_{n,(i-1)\Delta_{n}}\Vert\leq\Delta_{n}\Vert\theta'_{n}\Vert_{\infty}. 
\end{equation}
Therefore, by the Cauchy-Schwarz inequality and bounded $c$ under \cref{As:Ito_Iocal},
\begin{equation}
\begin{split}
\mathbb{E}[|A_{2,n}^{(1)}|\mathbbm{1}_{E_{n}}]&\leq \frac{K}{\sqrt{\Delta_{n}}}\mathbb{E}\left[\sum_{i=1}^{n}\int_{(i-1)\Delta_{n}}^{i\Delta_{n}}\Vert\theta_{n,s}-\theta_{n,(i-1)\Delta_{n}}\Vert\Vert e_{s}\Vert ds\right]\\
&\leq K\sqrt{\Delta_{n}}\mathbb{E}[\Vert\theta'_{n}\Vert_{\infty}\Vert e\Vert_{L^{2}}]=O(\sqrt{\Delta_{n}K_{n}\log K_{n}})=o(1),
\end{split}
\end{equation}
and thus, by Markov's inequality analogous to \cref{Eq:Chebyshev} (but with first-order moment), we obtain $A_{2,n}^{(1)}=o_{p}(1)$. Moreover, for $M_{2,n}^{(1)}$, by the same steps as in \cref{Eq:A1_Eq2}, together with \cref{Eq:ItoIsometryRevised},
\begin{equation}
\begin{split}
\mathbb{E}[(M_{2,n}^{(1)})^{2}\mathbbm{1}_{E_{n}}]&\leq\frac{K}{\Delta_{n}}\sum_{i=1}^{n}\mathbb{E}[U_{i}^{2}]\leq\frac{K}{\Delta_{n}}\sum_{i=1}^{n}\mathbb{E}\left[\Vert\Delta_{i}^{n}X^{c}\Vert^{2}\left(\int_{(i-1)\Delta_{n}}^{i\Delta_{n}}e_{s}^{\top}dX_{s}^{c}\right)^{2}\right]\\
&\leq\frac{K}{\Delta_{n}}\sum_{i=1}^{n}\mathbb{E}\left[\mathbb{E}[\Vert\Delta_{i}^{n}X^{c}\Vert^{2}|\mathcal{F}_{(i-1)\Delta_{n}}]\left(\int_{(i-1)\Delta_{n}}^{i\Delta_{n}}e_{s}^{\top}dX_{s}^{c}\right)^{2}\right]\\
&\leq K\mathbb{E}[\Vert e\Vert_{L^{2}}^{2}]=o(1),
\end{split}
\end{equation}
and thus $M_{2,n}^{(1)}=o_{p}(1)$ by the same Markov's inequality as in \cref{Eq:Chebyshev}. Therefore, we conclude that $\widetilde{H}_{2,n}^{(1)}=A_{2,n}^{(1)}+M_{2,n}^{(1)}=o_{p}(1)$.

For $\overline{H}_{2,n}^{(1)}$, by the Cauchy-Schwarz inequality,
\begin{equation}
\begin{split}
\label{Eq:H_2_1_bar}
\mathbb{E}\left[(\overline{H}_{2,n}^{(1)})^{2}\mathbbm{1}_{E_{n}}\right]&\leq\frac{K}{\Delta_{n}}\left(\sum_{i=1}^{n}\mathbb{E}\left[\Vert\Delta_{i}^{n}X^{c}\Vert^{2}\left(1-\mathbbm{1}_{\{\Vert\Delta_{i}^{n}X\Vert\leq u_{n}\}}\mathbbm{1}_{\{|\Delta_{i}^{n}Y|\leq u_{n}\}}\right)\right]\right)\\
&\qquad\times\left(\sum_{i=1}^{n}\mathbb{E}\left[\left(\int_{(i-1)\Delta_{n}}^{i\Delta_{n}}e_{s}^{\top}dX_{s}^{c}\right)^{2}\left(1-\mathbbm{1}_{\{\Vert\Delta_{i}^{n}X\Vert\leq u_{n}\}}\mathbbm{1}_{\{|\Delta_{i}^{n}Y|\leq u_{n}\}}\right)\right]\right).
\end{split}
\end{equation}
By the same steps as in \cref{Eq:A1_Eq2}, we obtain
\begin{equation}
\label{Eq:H_2_1_bar-1}
\sum_{i=1}^{n}\mathbb{E}\left[\Vert\Delta_{i}^{n}X^{c}\Vert^{2}\left(1-\mathbbm{1}_{\{\Vert\Delta_{i}^{n}X\Vert\leq u_{n}\}}\mathbbm{1}_{\{|\Delta_{i}^{n}Y|\leq u_{n}\}}\right)\right]\leq K\Delta_{n}u_{n}^{-r}.
\end{equation}
By \cref{Eq:ItoIsometryRevised},
\begin{equation}
\mathbb{E}\left[\left.\left(\int_{(i-1)\Delta_{n}}^{i\Delta_{n}}e_{s}^{\top}dX_{s}^{c}\right)^{2}\right\vert\mathcal{F}_{(i-1)\Delta_{n}}\right]\leq K\int_{(i-1)\Delta_{n}}^{i\Delta_{n}}\Vert e_{s}\Vert^{2}ds,
\end{equation}
and thus, by \cref{lemma:truncationEvents} (i),
\begin{equation}
\label{Eq:H_2_1_bar-2}
\begin{split}
&\sum_{i=1}^{n}\mathbb{E}\left[\left(\int_{(i-1)\Delta_{n}}^{i\Delta_{n}}e_{s}^{\top}dX_{s}^{c}\right)^{2}\left(1-\mathbbm{1}_{\{\Vert\Delta_{i}^{n}X\Vert\leq u_{n}\}}\mathbbm{1}_{\{|\Delta_{i}^{n}Y|\leq u_{n}\}}\right)\right]\\
&\qquad\leq K\mathbb{E}\left[\sum_{i=1}^{n}\left(\int_{(i-1)\Delta_{n}}^{i\Delta_{n}}\Vert e_{s}\Vert^{2}ds\right)\left(1-\mathbbm{1}_{\{\Vert\Delta_{i}^{n}X\Vert\leq u_{n}\}}\mathbbm{1}_{\{|\Delta_{i}^{n}Y|\leq u_{n}\}}\right)\right]\\
&\qquad\leq K\mathbb{E}[\Vert e\Vert_{L^{2}}^{2}]\max_{1\leq i\leq n}(\mathbb{P}(|\Delta_{i}^{n}Y|> u_{n})+\mathbb{P}(\Vert\Delta_{i}^{n}X\Vert>u_{n}))\\
&\qquad\leq K'\Delta_{n}u_{n}^{-r}\mathbb{E}[\Vert e\Vert_{L^{2}}^{2}].
\end{split}
\end{equation}
Combining \cref{Eq:H_2_1_bar-1,Eq:H_2_1_bar-2} in \cref{Eq:H_2_1_bar} leads to
\begin{equation}
\mathbb{E}\left[(\overline{H}_{2,n}^{(1)})^{2}\mathbbm{1}_{E_{n}}\right]\leq K\Delta_{n}u_{n}^{-2r}\mathbb{E}[\Vert e\Vert_{L^{2}}^{2}]=o(1). 
\end{equation}
By the same Markov's inequality as in \cref{Eq:Chebyshev}, we conclude $\overline{H}_{2,n}^{(1)}=o_{p}(1)$, and therefore
\begin{equation}
H_{2,n}^{(1)}=\widetilde{H}_{2,n}^{(1)}-\overline{H}_{2,n}^{(1)}=o_{p}(1).
\end{equation}

\subsubsection*{(iii.2) $\bm{H_{2,n}^{(2)}=o_{p}(1)}$}

On $E_{n}$, by the Cauchy-Schwarz inequality,
\begin{equation}
\label{Eq:H_2_1}
\mathbb{E}\left[(H_{2,n}^{(2)})^{2}\mathbbm{1}_{E_{n}}\right]\leq\frac{K}{\Delta_{n}}\left(\sum_{i=1}^{n}\mathbb{E}\left[\Vert\Delta_{i}^{n}X-\Delta_{i}^{n}X^{c}\Vert^{2}\mathbbm{1}_{\{\Vert\Delta_{i}^{n}X\Vert\leq u_{n}\}}\right]\right)\left(\sum_{i=1}^{n}\mathbb{E}\left[\left(\int_{(i-1)\Delta_{n}}^{i\Delta_{n}}e_{s}^{\top}dX_{s}^{c}\right)^{2}\right]\right).
\end{equation}
By \cref{lemma:truncationEvents} (ii),
\begin{equation}
\label{Eq:H_2_1-1}
\sum_{i=1}^{n}\mathbb{E}\left[\Vert\Delta_{i}^{n}X-\Delta_{i}^{n}X^{c}\Vert^{2}\mathbbm{1}_{\{\Vert\Delta_{i}^{n}X\Vert\leq u_{n}\}}\right]\leq Ku_{n}^{2-r},
\end{equation}
and, by \cref{Eq:ItoIsometryRevised},
\begin{equation}
\label{Eq:H_2_1-2}
\sum_{i=1}^{n}\mathbb{E}\left[\left(\int_{(i-1)\Delta_{n}}^{i\Delta_{n}}e_{s}^{\top}dX_{s}^{c}\right)^{2}\right]\leq K\mathbb{E}[\Vert e\Vert_{L^{2}}^{2}].
\end{equation}
Combining \cref{Eq:H_2_1-1,Eq:H_2_1-2} in \cref{Eq:H_2_1}, and by \cref{lemma:approximationError},
\begin{equation}
\mathbb{E}\left[(H_{2,n}^{(2)})^{2}\mathbbm{1}_{E_{n}}\right]\leq \frac{K}{\Delta_{n}}u_{n}^{2-r}\mathbb{E}[\Vert e\Vert_{L^{2}}^{2}] = o(1),
\end{equation}
which holds for all $u_{n}\asymp\Delta_{n}^{\varpi}$ with $\varpi\in\left(\frac{1}{2(2-r)},\,\frac{1}{2}\right)$ because 
\begin{equation}
\frac{K}{\Delta_{n}}u_{n}^{2-r}\mathbb{E}[\Vert e\Vert_{L^{2}}^{2}] \asymp \Delta_{n}^{\varpi(2-r)-1}\frac{\log K_{n}}{K_{n}}=\Delta_{n}^{\varpi(2-r)-1/2}\frac{\log K_{n}}{K_{n}\sqrt{\Delta_{n}}}=o(1),
\end{equation}
where $\log K_{n}\big/(K_{n}\sqrt{\Delta_{n}})\to0$ as required, and $\Delta_{n}^{\varpi(2-r)-1/2}\to0$ because $\varpi(2-r)-1/2>0$. Therefore, by the same Markov's inequality as in \cref{Eq:Chebyshev}, $H_{2,n}^{(2)}=o_{p}(1)$.
%\begin{equation}
%\frac{K}{\Delta_{n}}u_{n}^{2-r}\mathbb{E}[\Vert e\Vert_{L^{2}}^{2}] \asymp \Delta_{n}^{\varpi(2-r)-1}\frac{\log K_{n}}{K_{n}}=\left(\Delta_{n}^{\varpi(2-r)-1/2}\sqrt{\log K_{n}}\right)\left(\frac{\sqrt{\log K_{n}}}{K_{n}\sqrt{\Delta_{n}}}\right),
%\end{equation}
%where $\sqrt{\log K_{n}}\big/(K_{n}\sqrt{\Delta_{n}})\to0$ as required, and $\Delta_{n}^{\varpi(2-r)-1/2}\sqrt{\log K_{n}}\to0$ because $\Delta_{n}\ll K_{n}^{-1}$ and $\varpi(2-r)-1/2>0$. Therefore, by the same Markov's inequality as in \cref{Eq:Chebyshev}, $H_{2,n}^{(2)}=o_{p}(1)$.

By the results in (iii.1) and (iii.2), we conclude that $H_{2,n}=o_{p}(1)$.

\subsubsection*{(iv) $\bm{H_{3,n}=o_{p}(1)}$}

By the same steps as in \cref{Eq:H_1_1_CS} and the safe substitution in \cref{Eq:J_i_tilde},
\begin{equation}
|H_{3,n}|\leq\frac{Ku_{n}}{\sqrt{\Delta_{n}}}\Vert\bm{w}_{n}\Vert_{\infty}\sum_{i=1}^{n}|\mathbf{J}_{i}|\mathbbm{1}_{\{\Vert\Delta_{i}^{n}X\Vert\leq u_{n}\}}
=\frac{Ku_{n}}{\sqrt{\Delta_{n}}}\Vert\bm{w}_{n}\Vert_{\infty}\sum_{i=1}^{n}|\widetilde{\mathbf{J}}_{i}|,
\end{equation}
and
\begin{equation}
\mathbb{E}[|H_{3,n}|\mathbbm{1}_{E_{n}}]\leq\frac{Ku_{n}}{\sqrt{\Delta_{n}}}\mathbb{E}\left[\sum_{i=1}^{n}|\widetilde{\mathbf{J}}_{i}|\right],
\end{equation}
so it suffices to show
\begin{equation}
\mathbb{E}\left[\sum_{i=1}^{n}|\widetilde{\mathbf{J}}_{i}|\right]=O(u_{n}^{1-r})=o\left(\frac{\sqrt{\Delta_{n}}}{u_{n}}\right),\qquad\text{such that}\quad\mathbb{E}[|H_{3,n}|\mathbbm{1}_{E_{n}}]=o(1).
\end{equation}
Same as in \cref{Eq:J_i_2u}, we have
\begin{equation}
\label{Eq:H_3_n-0}
\begin{split}
\sum_{i=1}^{n}|\widetilde{\mathbf{J}}_{i}|&\leq K\sum_{i=1}^{n}\left(\int_{(i-1)\Delta_{n}}^{i\Delta_{n}}\int_{\mathbb{R}^{p}}(\Vert\delta(s,x)\Vert\wedge 2u_{n})\mu(ds,dx)\right.\\
&\qquad+\left.\int_{(i-1)\Delta_{n}}^{i\Delta_{n}}\int_{\{\Vert\delta\Vert>2u_{n}\}}\Vert\delta(s,x)\Vert\mu(ds,dx)\right)\mathbbm{1}_{\{\Vert\Delta_{i}^{n}X\Vert\leq u_{n}\}} \\
&=K\int_{0}^{T}\int_{\mathbb{R}^{p}}(\Vert\delta(s,x)\Vert\wedge 2u_{n})\mu(ds,dx)\\
&\qquad+K\sum_{i=1}^{n}\left(\int_{(i-1)\Delta_{n}}^{i\Delta_{n}}\int_{\{\Vert\delta\Vert>2u_{n}\}}\Vert\delta(s,x)\Vert\mu(ds,dx)\right)\mathbbm{1}_{\{\Vert\Delta_{i}^{n}X\Vert\leq u_{n}\}}.
\end{split}
\end{equation}
Following the same steps as in \cref{Eq:V1_2}, 
\begin{equation}
\label{Eq:H_3_n-1}
\begin{split}
&\mathbb{E}\left[\int_{0}^{T}\int_{\mathbb{R}^{p}}(\Vert\delta(s,x)\Vert\wedge 2u_{n})\mu(ds,dx)\right]\leq K\mathbb{E}\left[\int_{\mathbb{R}^p} (\|\delta(s,x)\|\wedge 2u_n)\lambda(dx)\right]\\
&\qquad\leq K\mathbb{E}\left[\int_{\{\|\delta\|\le 2u_n\}}\|\delta(s,x)\|\lambda(dx)+2u_n\int_{\{\|\delta\|>2u_n\}}\lambda(dx)\right]\leq K'u_{n}^{1-r},
\end{split}
\end{equation}
by \cref{Eq:LevyTail,Eq:smallJumpsl1}. For the other term, by \cref{Eq:largeJumpBound,Eq:eventDecomp_J,Eq:EJi^2_2_decomp3,Eq:EJi^2_2_decomp1,Eq:EJi^2_2_decomp2},
\begin{equation}
\label{Eq:H_3_n-2}
\begin{split}
&\sum_{i=1}^{n}\mathbb{E}\left[\left(\int_{(i-1)\Delta_{n}}^{i\Delta_{n}}\int_{\{\Vert\delta\Vert>2u_{n}\}}\Vert\delta(s,x)\Vert\mu(ds,dx)\right)\mathbbm{1}_{\{\Vert\Delta_{i}^{n}X\Vert\leq u_{n}\}}\right]\\
&\qquad\leq K\sum_{i=1}^{n}\mathbb{E}\left[N_{i}^{n}(2u_{n})\mathbbm{1}_{\{\Vert\Delta_{i}^{n}X\Vert\leq u_{n}\}}\right]\leq K\sum_{i=1}^{n}\mathbb{E}\left[(N_{i}^{n}(2u_{n}))^{2}\mathbbm{1}_{\{\Vert\Delta_{i}^{n}X\Vert\leq u_{n}\}}\right]\\
&\qquad\leq K\Delta_{n}u_{n}^{-2r}=o(u_{n}^{1-r}). 
\end{split}
\end{equation}
Combining \cref{Eq:H_3_n-1,Eq:H_3_n-2} in \cref{Eq:H_3_n-0} leads to $\mathbb{E}[\sum_{i=1}^{n}|\widetilde{\mathbf{J}}_{i}|]=o(\sqrt{\Delta_{n}}u_{n}^{-1})$, and thus $\mathbb{E}[|H_{3,n}|\mathbbm{1}_{E_{n}}]=o(1)$. Therefore, by the same Markov's inequality as in \cref{Eq:Chebyshev}, $H_{3,n}=o_{p}(1)$.

\subsubsection*{(v) $\bm{H_{5,n}=o_{p}(1)}$}
\label{AP:Proof_CLT_v}

By the same steps as in \cref{Eq:H_1_1_CS},
\begin{align}
|H_{5,n}|&\leq\frac{Ku_{n}}{\sqrt{\Delta_{n}}}\Vert\bm{w}_{n}\Vert_{\infty}\sum_{i=1}^{n}|\Delta_{i}^{n}Z|\left\vert\mathbbm{1}_{\{|\Delta_{i}^{n}Y|\leq u_{n}\}}-\mathbbm{1}_{\{|\Delta_{i}^{n}Z|\leq u_{n}\}}\right\vert\mathbbm{1}_{\{\Vert\Delta_{i}^{n}X\Vert\leq u_{n}\}} \notag\\
&\leq\frac{Ku_{n}}{\sqrt{\Delta_{n}}}\Vert\bm{w}_{n}\Vert_{\infty}\sum_{i=1}^{n}|\Delta_{i}^{n}Z|\left(\mathbbm{1}_{\{|\Delta_{i}^{n}Y|>u_{n},\,|\Delta_{i}^{n}Z|\leq u_{n}\}}+\mathbbm{1}_{\{|\Delta_{i}^{n}Y|\leq u_{n},\,|\Delta_{i}^{n}Z|>u_{n}\}}\right)\mathbbm{1}_{\{\Vert\Delta_{i}^{n}X\Vert\leq u_{n}\}}\notag \\
&\leq\frac{Ku_{n}}{\sqrt{\Delta_{n}}}\Vert\bm{w}_{n}\Vert_{\infty}\sum_{i=1}^{n}\left(u_{n}\mathbbm{1}_{\{|\Delta_{i}^{n}Y|>u_{n}\}}+|\Delta_{i}^{n}Z|\mathbbm{1}_{\{|\Delta_{i}^{n}Y|\leq u_{n},\,|\Delta_{i}^{n}Z|>u_{n},\,\Vert\Delta_{i}^{n}X\Vert\leq u_{n}\}}\right). \label{Eq:E_H_5-0}
\end{align}
Note that the additional truncation $\mathbbm{1}_{\{\Vert\Delta_{i}^{n}X\Vert\leq u_{n}\}}$ is harmlessly included because the design matrix $\mathbf{R}$ already includes that for covariate increments. Then,
\begin{equation}
\label{Eq:E_H_5}
\mathbb{E}[|H_{5,n}|\mathbbm{1}_{E_{n}}]\leq\frac{Ku_{n}^{2}}{\sqrt{\Delta_{n}}}\sum_{i=1}^{n}\mathbb{P}(|\Delta_{i}^{n}Y|>u_{n})+\frac{K'u_{n}}{\sqrt{\Delta_{n}}}\sum_{i=1}^{n}\mathbb{E}\left[|\Delta_{i}^{n}Z|\mathbbm{1}_{\{|\Delta_{i}^{n}Y|\leq u_{n},\,|\Delta_{i}^{n}Z|>u_{n},\,\Vert\Delta_{i}^{n}X\Vert\leq u_{n}\}}\right].
\end{equation}
For the first term, by \cref{lemma:truncationEvents} (i),
\begin{equation}
\label{Eq:E_H_5-1}
\sum_{i=1}^{n}\mathbb{P}(|\Delta_{i}^{n}Y|>u_{n})=O(u_{n}^{-r}),\quad\text{and thus}\quad
\frac{Ku_{n}^{2}}{\sqrt{\Delta_{n}}}\sum_{i=1}^{n}\mathbb{P}(|\Delta_{i}^{n}Y|>u_{n})=O\left(\frac{u_{n}^{2-r}}{\sqrt{\Delta_{n}}}\right)=o(1),
\end{equation}
because $\varpi(2-r)-1/2>0$. 

Now we turn to the second term in \cref{Eq:E_H_5}. Intuitively, on the event $\{|\Delta_{i}^{n}Y|\leq u_{n},\,|\Delta_{i}^{n}Z|>u_{n},\,\Vert\Delta_{i}^{n}X\Vert\leq u_{n}\}$, 
%\begin{equation}
%A_{i,n}=\{|\Delta_{i}^{n}Y|\leq u_{n},\,|\Delta_{i}^{n}Z|>u_{n},\,\Vert\Delta_{i}^{n}X\Vert\leq u_{n}\},
%\end{equation}
the cancellation between $\Delta_{i}^{n}Y$, $\Delta_{i}^{n}X$ and $\Delta_{i}^{n}Z$ imposes a restriction on the size of $\Delta_{i}^{n}Z$. To formally establish this, we write $\Delta_{i}^{n}Y$ into
\begin{equation}
\label{Eq:unmatch_start}
\Delta_{i}^{n}Y=\,\underbrace{\int_{(i-1)\Delta_{n}}^{i\Delta_{n}}\beta_{s-}^{\top}dX_{s}^{c}}_{\Delta_{i}^{n}Y^{(1)}}\,\,+\,\,\underbrace{\sum_{(i-1)\Delta_{n}\leq s\leq i\Delta_{n}}(\beta_{s-}^{J})^{\top}\Delta X_{s}}_{\Delta_{i}^{n}Y^{(2)}}\,\,+\,\,\Delta_{i}^{n}Z.
\end{equation}
On $\{|\Delta_{i}^{n}Y|\leq u_{n},\,|\Delta_{i}^{n}Z|>u_{n},\,\Vert\Delta_{i}^{n}X\Vert\leq u_{n}\}$, it holds that
\begin{equation}
\label{Eq:unmatch_inequality}
\begin{split}
u_{n}<|\Delta_{i}^{n}Z|&=|\Delta_{i}^{n}Y-\Delta_{i}^{n}Y^{(1)}-\Delta_{i}^{n}Y^{(2)}|\\
&\leq |\Delta_{i}^{n}Y| + |\Delta_{i}^{n}Y^{(1)}| + |\Delta_{i}^{n}Y^{(2)}|\\
&=|\Delta_{i}^{n}Y| + |\Delta_{i}^{n}Y^{(1)}| + |\Delta_{i}^{n}Y^{(2)}\mathbbm{1}_{\{|\Delta_{i}^{n}Y|\leq u_{n},\,\Vert\Delta_{i}^{n}X\Vert\leq u_{n}\}}|\\
&=|\Delta_{i}^{n}Y| + |\Delta_{i}^{n}Y^{(1)}| + |\widetilde{\mathbf{J}}_{i}|\\
&\leq u_{n} + |\Delta_{i}^{n}Y^{(1)}| + |\widetilde{\mathbf{J}}_{i}|,
\end{split}
\end{equation}
where $\widetilde{\mathbf{J}}_{i}$ is defined in \cref{Eq:J_i_tilde}. Based on the above result, we consider the event decomposition:
\begin{equation}
\label{Eq:unmatch_eventInequality}
\begin{split}
&\{|\Delta_{i}^{n}Y|\leq u_{n},\,|\Delta_{i}^{n}Z|>u_{n},\,\Vert\Delta_{i}^{n}X\Vert\leq u_{n}\}\\
&\qquad\subseteq\{u_{n}<|\Delta_{i}^{n}Z|\leq 3u_{n}\}\,\cup\,\{|\Delta_{i}^{n}Y^{(1)}| + |\widetilde{\mathbf{J}}_{i}|>2u_{n}\}.
\end{split}
\end{equation}
Then, we have
\begin{align}
&\mathbb{E}\left[|\Delta_{i}^{n}Z|\mathbbm{1}_{\{|\Delta_{i}^{n}Y|\leq u_{n},\,|\Delta_{i}^{n}Z|>u_{n},\,\Vert\Delta_{i}^{n}X\Vert\leq u_{n}\}}\right] \notag\\
&\qquad\leq\mathbb{E}\left[|\Delta_{i}^{n}Z|\mathbbm{1}_{\{u_{n}<|\Delta_{i}^{n}Z|\leq 3u_{n}\}}\right]+\mathbb{E}\left[|\Delta_{i}^{n}Z|\mathbbm{1}_{\{|\Delta_{i}^{n}Y^{(1)}| + |\widetilde{\mathbf{J}}_{i}|>2u_{n}\}}\right] \notag\\
&\qquad\leq \underbrace{3u_{n}\mathbb{P}(|\Delta_{i}^{n}Z|>u_{n})}_{(1)}\,+\,\mathbb{E}\left[|\Delta_{i}^{n}Z|\mathbbm{1}_{\{|\Delta_{i}^{n}Y^{(1)}|> u_{n}\}}\right]+\mathbb{E}\left[|\Delta_{i}^{n}Z|\mathbbm{1}_{\{|\widetilde{\mathbf{J}}_{i}|>u_{n}\}}\right] \notag\\
&\qquad\leq (1)\,+\, \underbrace{u_{n}\mathbb{P}\left(|\Delta_{i}^{n}Y^{(1)}|>u_{n}\right)}_{(2)}\,+\,\underbrace{\mathbb{E}\left[|\Delta_{i}^{n}Y^{(1)}|\mathbbm{1}_{\{|\Delta_{i}^{n}Y^{(1)}|> u_{n}\}}\right]}_{(3)}\,+\,\underbrace{\mathbb{E}\left[|\widetilde{\mathbf{J}}_{i}|\mathbbm{1}_{\{|\Delta_{i}^{n}Y^{(1)}|>u_{n}\}}\right]}_{(4)} \notag\\
&\qquad\qquad+\,\underbrace{u_{n}\mathbb{P}\left(|\widetilde{\mathbf{J}}_{i}|>u_{n}\right)}_{(5)}\,+\,\underbrace{\mathbb{E}\left[|\Delta_{i}^{n}Y^{(1)}|\mathbbm{1}_{\{|\widetilde{\mathbf{J}}_{i}|> u_{n}\}}\right]}_{(6)}\,+\,\underbrace{\mathbb{E}\left[|\widetilde{\mathbf{J}}_{i}|\mathbbm{1}_{\{|\widetilde{\mathbf{J}}_{i}|>u_{n}\}}\right]}_{(7)}. \label{Eq:7terms}
\end{align}
Next, we verify that the above terms (1)--(7) are all $O(\Delta_{n}u_{n}^{1-r})$. 

Term (1): By the same result as in \cref{lemma:truncationEvents} (i) for the one-dimensional It{\^o} semimartingale $Z$,
\begin{equation}
3u_{n}\mathbb{P}(|\Delta_{i}^{n}Z|>u_{n})\leq K\Delta_{n}u_{n}^{1-r}.
\end{equation}

Term (2): By \cref{Eq:Y_triangle_Prob1},
\begin{equation}
u_{n}\mathbb{P}\left(|\Delta_{i}^{n}Y^{(1)}|>u_{n}\right)\leq K\Delta_{n}u_{n}^{1-r}.
\end{equation}

Term (3): By the same steps as in \cref{Eq:moment_Markov} and the BDG inequality, for any $q\geq2$,
\begin{equation}
\begin{split}
\mathbb{E}\left[|\Delta_{i}^{n}Y^{(1)}|\mathbbm{1}_{\{|\Delta_{i}^{n}Y^{(1)}|> u_{n}\}}\right]\leq\mathbb{E}\left[|\Delta_{i}^{n}Y^{(1)}|\left(\frac{|\Delta_{i}^{n}Y^{(1)}|}{u_{n}}\right)^{q-1}\right]\leq \frac{\mathbb{E}[|\Delta_{i}^{n}Y^{(1)}|^{q}]}{u_{n}^{q-1}}\leq K\Delta_{n}^{q/2}u_{n}^{1-q},
\end{split}
\end{equation}
and it is $O(\Delta_{n}u_{n}^{1-r})$ if $q\geq 2(1-\varpi r)/(1-2\varpi)$, which is always possible for $1/2(2-r)<\varpi<1/2$ and $0\leq r<1$. 

Term (4): By the Cauchy-Schwarz inequality, \cref{Eq:J_i_moment,Eq:Y_triangle_Prob1}, for any $q\geq 2$,
\begin{equation}
\begin{split}
\mathbb{E}\left[|\widetilde{\mathbf{J}}_{i}|\mathbbm{1}_{\{|\Delta_{i}^{n}Y^{(1)}|> u_{n}\}}\right]&\leq\left(\mathbb{E}[|\widetilde{\mathbf{J}}_{i}|^{2}]\right)^{1/2}\left(\mathbb{P}\left(|\Delta_{i}^{n}Y^{(1)}|>u_{n}\right)\right)^{1/2}\\
&\leq K(\Delta_{n}u_{n}^{2-r}\vee\Delta_{n}^{2}u_{n}^{-2r})^{1/2}(\Delta_{n}^{q/2}u_{n}^{-q})^{1/2}\\
&\leq K\Delta_{n}^{1/2+q/4}u_{n}^{1-r/2-q/2}\vee K\Delta_{n}^{1+q/4}u_{n}^{-r-q/2}\\
&=O(\Delta_{n}u_{n}^{1-r}),
\end{split}
\end{equation}
if we choose some
\begin{equation}
q\geq \frac{2(1-\varpi r)}{1-2\varpi}\bigvee\frac{4\varpi}{1-2\varpi},
\end{equation}
which is always possible for $1/2(2-r)<\varpi<1/2$ and $0\leq r<1$. 

Term (5): By Markov's inequality,
\begin{equation}
\label{Eq:E_J_i_tilde}
u_{n}\mathbb{P}\left(|\widetilde{\mathbf{J}}_{i}|>u_{n}\right)\leq u_{n}\frac{\mathbb{E}[|\widetilde{\mathbf{J}}_{i}|]}{u_{n}}\leq \mathbb{E}[|\widetilde{\mathbf{J}}_{i}|]\leq K\Delta_{n}u_{n}^{1-r},
\end{equation}
which can be calculated following the same steps as in \cref{Eq:H_3_n-1,Eq:H_3_n-2}.

Term (6): By Hölder's inequality, for some $q\geq 2$,
\begin{equation}
\begin{split}
\mathbb{E}\left[|\Delta_{i}^{n}Y^{(1)}|\mathbbm{1}_{\{|\widetilde{\mathbf{J}}_{i}|> u_{n}\}}\right]&\leq\left(\mathbb{E}[|\Delta_{i}^{n}Y^{(1)}|^{q}]\right)^{1/q}\left(\mathbb{P}\left(|\widetilde{\mathbf{J}}_{i}|>u_{n}\right)\right)^{1-1/q}\\
&\leq K\sqrt{\Delta_{n}}(\Delta_{n}u_{n}^{-r})^{1-1/q}=K\Delta_{n}^{3/2-1/q}u_{n}^{r/q-r},
\end{split}
\end{equation}
by the BDG inequality, and it is $O(\Delta_{n}u_{n}^{1-r})$ if $q\geq 2(1-\varpi r)/(1-2\varpi)$.

Term (7): Same as in \cref{Eq:J_i_2u}, we have
\begin{equation}
|\widetilde{\mathbf{J}}_{i}|\leq K(A_{i}^{(1)}+A_{i}^{(2)})\mathbbm{1}_{\{\Vert\Delta_{i}^{n}X\Vert\leq u_{n}\}},
\end{equation}
where
\begin{align}
A_{i}^{(1)}&=\int_{(i-1)\Delta_{n}}^{i\Delta_{n}}\int_{\mathbb{R}^{p}}(\Vert\delta(s,x)\Vert\wedge 2u_{n})\mu(ds,dx),\\
A_{i}^{(2)}&=\int_{(i-1)\Delta_{n}}^{i\Delta_{n}}\int_{\{\Vert\delta\Vert>2u_{n}\}}\Vert\delta(s,x)\Vert\mu(ds,dx),
\end{align}
and thus, 
\begin{equation}
\mathbb{E}\left[|\widetilde{\mathbf{J}}_{i}|\mathbbm{1}_{\{|\widetilde{\mathbf{J}}_{i}|>u_{n}\}}\right]\leq K\mathbb{E}\left[A_{i}^{(1)}\mathbbm{1}_{\{A_{i}^{(1)}>\frac{u_{n}}{2K}\}}\right]+K\mathbb{E}\left[A_{i}^{(2)}\mathbbm{1}_{\{A_{i}^{(2)}>\frac{u_{n}}{2K}\}}\mathbbm{1}_{\{\Vert\Delta_{i}^{n}X\Vert\leq u_{n}\}}\right].
\end{equation}
For the first expectation on the right-hand side, by \cref{Eq:EJi^2_1},
\begin{equation}
\mathbb{E}\left[A_{i}^{(1)}\mathbbm{1}_{\{A_{i}^{(1)}>\frac{u_{n}}{2K}\}}\right]\leq \frac{2K}{u_{n}}\mathbb{E}[(A_{i}^{(1)})^{2}]\leq K'\Delta_{n}u_{n}^{1-r}.
\end{equation}
For the second expectation, we follow the same steps as in \cref{Eq:H_3_n-2},
\begin{equation}
\begin{split}
\mathbb{E}\left[A_{i}^{(2)}\mathbbm{1}_{\{A_{i}^{(2)}>\frac{u_{n}}{2K}\}}\mathbbm{1}_{\{\Vert\Delta_{i}^{n}X\Vert\leq u_{n}\}}\right]&\leq K'\mathbb{E}\left[N_{i}^{n}(2u_{n})\mathbbm{1}_{\{N_{i}^{n}(2u_{n})\geq 1\}}\mathbbm{1}_{\{\Vert\Delta_{i}^{n}X\Vert\leq u_{n}\}}\right]\\
&\leq K'\mathbb{E}\left[(N_{i}^{n}(2u_{n}))^{2}\mathbbm{1}_{\{N_{i}^{n}(2u_{n})\geq 1\}}\mathbbm{1}_{\{\Vert\Delta_{i}^{n}X\Vert\leq u_{n}\}}\right]\\
&\leq K''\Delta_{n}^{2}u_{n}^{-2r}=o(\Delta_{n}u_{n}^{1-r}). 
\end{split}
\end{equation}
Therefore, it holds that for the term (7):
\begin{equation}
\mathbb{E}\left[|\widetilde{\mathbf{J}}_{i}|\mathbbm{1}_{\{|\widetilde{\mathbf{J}}_{i}|>u_{n}\}}\right]\leq K\Delta_{n}u_{n}^{1-r}.
\end{equation}

Therefore, combining all terms (1)--(7) in \cref{Eq:7terms} leads to
\begin{equation}
\label{Eq:E_H_5-2}
\mathbb{E}\left[|\Delta_{i}^{n}Z|\mathbbm{1}_{\{|\Delta_{i}^{n}Y|\leq u_{n},\,|\Delta_{i}^{n}Z|>u_{n},\,\Vert\Delta_{i}^{n}X\Vert\leq u_{n}\}}\right]\leq K\Delta_{n}u_{n}^{1-r}.
\end{equation}
By \cref{Eq:E_H_5,Eq:E_H_5-1,Eq:E_H_5-2}, $\mathbb{E}[|H_{5,n}|\mathbbm{1}_{E_{n}}]=o(1)$. By the same Markov's inequality as in \cref{Eq:Chebyshev}, we conclude that $H_{5,n}=o_{p}(1)$.\\

By results in (\hyperref[AP:Proof_CLT_i]{i}) to (\hyperref[AP:Proof_CLT_v]{v}), we complete the proof of \cref{Eq:Cramer-Wold}. Next, we show in (\hyperref[AP:Proof_CLT_vi]{vi}) and (\hyperref[AP:Proof_CLT_vii]{vii}) that the second term (the bias induced by the discretization error) and the third term (the bias induced by the spline approximation error) in \cref{Eq:decomp} are both asymptotically negligible and thus have no impact on the CLT. 

\subsubsection*{(vi) Discretization error}
\label{AP:Proof_CLT_vi}

Consider the discretization error term in \cref{Eq:decomp}:
\begin{equation}
D_{n}=\frac{1}{\sqrt{\Delta_{n}}}\left( \int_{0}^{T} \mathbf{B}_sds - \sum^n_{i=1 }\mathbf{B}_{(i-1)\Delta_n}\Delta_n\right)\bm{\gamma}=\frac{1}{\sqrt{\Delta_{n}}}\sum_{i=1}^{n}\int_{(i-1)\Delta_{n}}^{i\Delta_{n}}(\widetilde{\beta}_{s}-\widetilde{\beta}_{(i-1)\Delta_{n}})ds.
\end{equation}
By the Cram{\'e}r-Wold device, we need to show that, for any fixed linear combination with $a\in\mathbb{R}^{p}$, it holds that $|a^{\top}D_{n}|=o_{p}(1)$. By the Cauchy-Schwarz inequality, since $\widetilde{\beta}$ is absolutely continuous and has bounded derivative on $[0,T]$ (same as in \cref{Eq:Discretization-2}),
\begin{equation}
\begin{split}
|a^{\top}D_{n}|&\leq\Vert a\Vert\Vert D_{n}\Vert\leq \frac{K}{\sqrt{\Delta_{n}}}\left\Vert\sum_{i=1}^{n}\int_{(i-1)\Delta_{n}}^{i\Delta_{n}}(\widetilde{\beta}_{s}-\widetilde{\beta}_{(i-1)\Delta_{n}})ds\right\Vert\\
&\leq\frac{K}{\sqrt{\Delta_{n}}}\sum_{i=1}^{n}\int_{(i-1)\Delta_{n}}^{i\Delta_{n}}\left(\int_{(i-1)\Delta_{n}}^{s}\Vert\widetilde{\beta}_{u}'\Vert du\right)ds\leq \frac{K}{\sqrt{\Delta_{n}}}\sum_{i=1}^{n}\int_{(i-1)\Delta_{n}}^{i\Delta_{n}}\Delta_{n}\Vert\widetilde{\beta}_{u}'\Vert du \\
&\leq K\sqrt{\Delta_{n}}\int_{0}^{T}\Vert\widetilde{\beta}_{u}'\Vert du\leq K\sqrt{\Delta_{n}}\sqrt{T}\Vert\widetilde{\beta}'\Vert_{L^{2}}\leq K'\sqrt{\Delta_{n}}\Vert\widetilde{\beta}'\Vert_{L^{2}}.
\end{split}
\end{equation}
By \cref{lemma:derivative},
\begin{equation}
|a^{\top}D_{n}|=O_{p}(\sqrt{\Delta_{n}K_{n}\log K_{n}})=o_{p}(1).
\end{equation}

\subsubsection*{(vii) Approximation error}
\label{AP:Proof_CLT_vii}

Define the matrix-valued process $\psi=c^{-1}\mathbf{I}_{p}$, where $\mathbf{I}_{p}$ is an identity matrix of dimension $p$. Under \cref{As:Ito} (iii), consider the same $L^{2}(c)$-orthogonal projection $\Pi_{n}\psi\in\mathbb{G}_{n}$, and then
\begin{equation}
\int_{0}^{T}e_{s}ds=\int_{0}^{T}e_{s}c_{s}\psi_{s}ds=\langle e,\psi\rangle_{c}=\langle e,\psi-\Pi_{n}\psi\rangle_{c}.
\end{equation}
By the Cauchy-Schwarz inequality and the norm equivalence in \cref{Eq:norm_equiv},
\begin{equation}
\begin{split}
\left\Vert\frac{1}{\sqrt{\Delta_{n}}}\int_{0}^{T}e_{s}ds\right\Vert&\leq\frac{1}{\sqrt{\Delta_{n}}}\Vert\langle e,\psi-\Pi_{n}\psi\rangle_{c}\Vert\leq\frac{1}{\sqrt{\Delta_{n}}}\Vert e\Vert_{L^{2}(c)}\Vert\psi-\Pi_{n}\psi\Vert_{L^{2}(c)}\\
&\leq\frac{K}{\sqrt{\Delta_{n}}}\Vert e\Vert_{L^{2}}\Vert\psi-\Pi_{n}\psi\Vert_{L^{2}}.
\end{split}
\end{equation}
By \cref{lemma:approximationError}, $\Vert e\Vert_{L^{2}}=O_{p}(\sqrt{\log K_{n}/K_{n}})$. Moreover, $\psi$ is also an It{\^o} semimartingale with finite-variation jumps (because $c$ is bounded, nonsingular, and an It{\^o} semimartingale under \cref{As:Ito} (iii), so is $c^{-1}$ by It{\^o}'s formula for semimartingales with summable jumps; see, e.g., Theorem 3.2.2, \citealp{jacod2012discretization}), and then the result in \cref{lemma:approximationError} also applies to $\psi-\Pi_{n}\psi$. Therefore,
\begin{equation}
\Vert\psi-\Pi_{n}\psi\Vert_{L^{2}}=O_{p}\left(\sqrt{\frac{\log K_{n}}{K_{n}}}\right),
\end{equation}
and then
\begin{equation}
\left\Vert\frac{1}{\sqrt{\Delta_{n}}}\int_{0}^{T}e_{s}ds\right\Vert=O_{p}\left(\frac{\log K_{n}}{K_{n}\sqrt{\Delta_{n}}}\right)=o_{p}(1). 
\end{equation}
This completes the proof.

\subsection[Proof of Theorem 3]{Proof of \cref{Th:feasiblility}}

By the standard Riemann approximation and the LLN for the empirical Gram matrix in \cref{Eq:LLN},
\begin{equation}
\label{Eq:sandwich1}
\sum_{i=1}^{n}\mathbf{B}_{(i-1)\Delta_{n}}\Delta_{n}\overset{\mathbb{P}}{\longrightarrow}\int_{0}^{T}\mathbf{B}_{s}ds,\qquad
(\mathbf{R}^{\top}\mathbf{R})^{-1}\overset{\mathbb{P}}{\longrightarrow}\left(\int_{0}^{T}\mathbf{B}_{s}^{\top}c_{s}\mathbf{B}_{s}ds\right)^{-1}.
\end{equation}
Hence  it suffices to show
\begin{equation}
\mathbf{R}^{\top}\mathbf{D}\mathbf{R}\overset{\mathbb{P}}{\longrightarrow}\int_{0}^{T}\widetilde{\sigma}_{s}^{2}\mathbf{B}_{s}^{\top}c_{s}\mathbf{B}_{s}ds
\end{equation}
We define the infeasible and feasible residuals:
\begin{equation}
\varepsilon_{i}=\Delta_{i}^{n}Y-(\Delta_{i}^{n}X)^{\top}\beta_{(i-1)\Delta_{n}-},\qquad
\widehat{\varepsilon}_{i}=\Delta_{i}^{n}Y-(\Delta_{i}^{n}X)^{\top}\widehat{\beta}_{(i-1)\Delta_{n}}^{(-i)},
\end{equation}
and $\widetilde{\mathbf{D}}=\text{diag}(\Delta_{n}^{-1}\varepsilon_{i}^{2}\mathbbm{1}_{\{|\Delta_{i}^{n}Y|\leq u_{n},\,\Vert\Delta_{i}^{n}X\Vert\leq u_{n}\}})$ as the infeasible counterpart of $\mathbf{D}$. We first show that the difference between $\mathbf{R}^{\top}\mathbf{D}\mathbf{R}$ and $\mathbf{R}^{\top}\widetilde{\mathbf{D}}\mathbf{R}$ are asymptotically negligible, i.e.,
\begin{equation}
\Vert\mathbf{R}^{\top}(\mathbf{D}-\widetilde{\mathbf{D}})\mathbf{R}\Vert=o_{p}(1).
\end{equation}
Firstly, we obtain
\begin{align}
&\widehat{\varepsilon}_{i}^{2}-\varepsilon_{i}^{2}=(\Delta_{i}^{n}Y-(\Delta_{i}^{n}X)^{\top}\widehat{\beta}_{(i-1)\Delta_{n}}^{(-i)})^{2}-(\Delta_{i}^{n}Y-(\Delta_{i}^{n}X)^{\top}\beta_{(i-1)\Delta_{n}-})^{2}\notag\\
&=(2\Delta_{i}^{n}Y-(\Delta_{i}^{n}X)^{\top}\widehat{\beta}_{(i-1)\Delta_{n}}^{(-i)}-(\Delta_{i}^{n}Y-(\Delta_{i}^{n}X)^{\top}\beta_{(i-1)\Delta_{n}-})((\Delta_{i}^{n}X)^{\top}(\beta_{(i-1)\Delta_{n}-}-\widehat{\beta}_{(i-1)\Delta_{n}}^{(-i)}))\notag\\
&=(2\Delta_{i}^{n}Y-(\Delta_{i}^{n}X)^{\top}\beta_{(i-1)\Delta_{n}-}+(\Delta_{i}^{n}X)^{\top}(\beta_{(i-1)\Delta_{n}-}-\widehat{\beta}_{(i-1)\Delta_{n}}^{(-i)}))((\Delta_{i}^{n}X)^{\top}(\beta_{(i-1)\Delta_{n}-}-\widehat{\beta}_{(i-1)\Delta_{n}}^{(-i)}))\notag\\
&=2\varepsilon_{i}(\Delta_{i}^{n}X)^{\top}(\beta_{(i-1)\Delta_{n}-}-\widehat{\beta}_{(i-1)\Delta_{n}}^{(-i)})+((\Delta_{i}^{n}X)^{\top}(\beta_{(i-1)\Delta_{n}-}-\widehat{\beta}_{(i-1)\Delta_{n}}^{(-i)}))^{2}, \label{Eq:eps-eps}
\end{align}
and, by the BDG inequality and \cref{Eq:smallJumps},
\begin{equation}
\label{Eq:epsilon^4}
\begin{split}
&\mathbb{E}\left[(\varepsilon_{i})^{4}\mathbbm{1}_{\{|\Delta_{i}^{n}Y|\leq u_{n},\,\Vert\Delta_{i}^{n}X\Vert\leq u_{n}\}}\right]\\
&\leq K\mathbb{E}\left[\left(\int_{(i-1)\Delta_{n}}^{i\Delta_{n}}(\beta_{s}-\beta_{(i-1)\Delta_{n}-})^{\top}dX_{s}^{c}+\Delta_{i}^{n}Z^{c}\right)^{4}\right]\\
&\quad +K'\mathbb{E}\left[\left(\sum_{(i-1)\Delta_{n}\leq s\leq i\Delta_{n}}\left(\Delta Y_{s}+(\beta_{s}^{J}-\beta_{(i-1)\Delta_{n}-})^{\top}\Delta X_{s}\right)\right)^{4}\mathbbm{1}_{\{|\Delta_{i}^{n}Y|\leq u_{n},\,\Vert\Delta_{i}^{n}X\Vert\leq u_{n}\}}\right]\\
&\leq K\Delta_{n}^{2}+K'\Delta_{n}u_{n}^{4-r}. 
\end{split}
\end{equation}
By the triangle inequality,
\begin{equation}
\Vert\mathbf{R}^{\top}(\mathbf{D}-\widetilde{\mathbf{D}})\mathbf{R}\Vert\leq S_{1}+S_{2},
\end{equation}
where, by \cref{Eq:eps-eps} and the boundedness of B-spline basis functions,
\begin{align}
S_{1}&=\frac{K}{\Delta_{n}}\sum_{i=1}^{n}|\varepsilon_{i}|\Vert\Delta_{i}^{n}X\Vert^{3}\Vert\beta_{(i-1)\Delta_{n}-}-\widehat{\beta}_{(i-1)\Delta_{n}}^{(-i)}\Vert\mathbbm{1}_{\{|\Delta_{i}^{n}Y|\leq u_{n},\,\Vert\Delta_{i}^{n}X\Vert\leq u_{n}\}},\\
S_{2}&=\frac{K}{\Delta_{n}}\sum_{i=1}^{n}\Vert\Delta_{i}^{n}X\Vert^{4}\Vert\beta_{(i-1)\Delta_{n}-}-\widehat{\beta}_{(i-1)\Delta_{n}}^{(-i)}\Vert^{2}\mathbbm{1}_{\{\Vert\Delta_{i}^{n}X\Vert\leq u_{n}\}},
\end{align}
and it suffices to show that both $S_{1}$ and $S_{2}$ are $o_{p}(1)$. 

Fix $i$ and define the enlarged $\sigma$-field
\begin{equation}
\widetilde{\mathcal{F}}_{(i-1)\Delta_{n}}=\mathcal{F}_{(i-1)\Delta_{n}}\vee\sigma\{(\Delta_{\ell}^{n}X,\Delta_{\ell}^{n}Y):\ell\neq i\}.
\end{equation}
By construction, $\widehat{\beta}_{(i-1)\Delta_{n}}^{(-i)}$ is $\widetilde{\mathcal{F}}_{(i-1)\Delta_{n}}$-measurable. By \cref{lemma:truncationEvents} (ii),
\begin{equation}
\label{Eq:E_Delta_i_X^4_truncated}
\begin{split}
\mathbb{E}\left[\left.\Vert\Delta_{i}^{n}X\Vert^{4}\mathbbm{1}_{\{\Vert\Delta_{i}^{n}X\Vert\leq u_{n}\}}\right\vert\widetilde{\mathcal{F}}_{(i-1)\Delta_{n}}\right]&\leq K\mathbb{E}\left[\Vert\Delta_{i}^{n}X^{c}\Vert^{4}+\Vert\Delta_{i}^{n}X-\Delta_{i}^{n}X^{c}\Vert^{4}\mathbbm{1}_{\{\Vert\Delta_{i}^{n}X\Vert\leq u_{n}\}}\right]\\
&\leq K\Delta_{n}^{2}+Ku_{n}^{2}\mathbb{E}\left[\Vert\Delta_{i}^{n}X-\Delta_{i}^{n}X^{c}\Vert^{2}\mathbbm{1}_{\{\Vert\Delta_{i}^{n}X\Vert\leq u_{n}\}}\right]\\
&\leq K\Delta_{n}^{2}+K'\Delta_{n}u_{n}^{4-r}\leq K''\Delta_{n}^{2},
\end{split}
\end{equation}
and thus
\begin{equation}
\label{Eq:ES2}
\begin{split}
\mathbb{E}[S_{2}]&=\frac{K}{\Delta_{n}}\sum_{i=1}^{n}\mathbb{E}\left[\mathbb{E}\left[\left.\Vert\Delta_{i}^{n}X\Vert^{4}\mathbbm{1}_{\{\Vert\Delta_{i}^{n}X\Vert\leq u_{n}\}}\right\vert\widetilde{\mathcal{F}}_{(i-1)\Delta_{n}}\right]\Vert\beta_{(i-1)\Delta_{n}-}-\widehat{\beta}_{(i-1)\Delta_{n}}^{(-i)}\Vert^{2}\right]\\
&\leq K\Delta_{n}\sum_{i=1}^{n}\mathbb{E}\left[\Vert\beta_{(i-1)\Delta_{n}-}-\widehat{\beta}_{(i-1)\Delta_{n}}^{(-i)}\Vert^{2}\right]\asymp \mathbb{E}\left[\Vert\widehat{\beta}^{-}-\beta\Vert_{L^{2}}^{2}\right]=o(1),
\end{split}
\end{equation}
where $\Delta_{n}\sum_{i=1}^{n}\Vert\beta_{(i-1)\Delta_{n}}-\widehat{\beta}_{(i-1)\Delta_{n}}^{(-i)}\Vert^{2}$ is a Riemann sum for $\Vert\widehat{\beta}^{-}-\beta\Vert_{L^{2}}^{2}$ of the piecewise constant process $\widehat{\beta}^{-}$ that equals $\widehat{\beta}_{(i-1)\Delta_{n}}^{(-i)}$ on $((i-1)\Delta_{n},i\Delta_{n}]$, and the proof of \cref{Th:consistency} applies to each leave-one-out OLS fit. Therefore, $S_{2}=o_{p}(1)$ by Markov's inequality. 

By the Cauchy-Schwarz inequality,
\begin{equation}
\mathbb{E}[S_{1}]\leq K(\mathbb{E}[S_{2}])^{1/2}\left(\Delta_{n}^{-1}\sum_{i=1}^{n}\mathbb{E}\left[\varepsilon_{i}^{2}\Vert\Delta_{i}^{n}X\Vert^{2}\mathbbm{1}_{\{|\Delta_{i}^{n}Y|\leq u_{n},\,\Vert\Delta_{i}^{n}X\Vert\leq u_{n}\}}\right]\right)^{1/2},
\end{equation}
%\begin{equation}
%\mathbb{E}[S_{1}\mathbbm{1}_{E_{n}}]\leq K(\mathbb{E}[S_{2}\mathbbm{1}_{E_{n}}])^{1/2}\left(\Delta_{n}^{-1}\sum_{i=1}^{n}\mathbb{E}\left[\varepsilon_{i}^{2}\Vert\Delta_{i}^{n}X\Vert^{2}\mathbbm{1}_{\{|\Delta_{i}^{n}Y|\leq u_{n},\,\Vert\Delta_{i}^{n}X\Vert\leq u_{n}\}}\right]\right)^{1/2},
%\end{equation}
where, by \cref{Eq:epsilon^4,Eq:E_Delta_i_X^4_truncated},
\begin{equation}
\label{Eq:S1}
\begin{split}
&\mathbb{E}\left[\varepsilon_{i}^{2}\Vert\Delta_{i}^{n}X\Vert^{2}\mathbbm{1}_{\{|\Delta_{i}^{n}Y|\leq u_{n},\,\Vert\Delta_{i}^{n}X\Vert\leq u_{n}\}}\right]\\
&\qquad\leq\left(\mathbb{E}\left[\varepsilon_{i}^{4}\mathbbm{1}_{\{|\Delta_{i}^{n}Y|\leq u_{n},\,\Vert\Delta_{i}^{n}X\Vert\leq u_{n}\}}\right]\right)^{1/2}\left(\mathbb{E}\left[\Vert\Delta_{i}^{n}X\Vert^{4}\mathbbm{1}_{\{\Vert\Delta_{i}^{n}X\Vert\leq u_{n}\}}\right]\right)^{1/2}\\
&\qquad\leq K\Delta_{n}^{2}+K'\Delta_{n}^{3/2}u_{n}^{2-r/2}\leq K''\Delta_{n}^{2},
\end{split}
\end{equation}
and thus $S_{1}=o_{p}(1)$ by Markov's inequality. Therefore, it holds that $\Vert\mathbf{R}^{\top}(\mathbf{D}-\widetilde{\mathbf{D}})\mathbf{R}\Vert=o_{p}(1)$, and thus
\begin{equation}
\label{Eq:sandwichEquivalence}
\mathbf{R}^{\top}\mathbf{D}\mathbf{R}=\mathbf{R}^{\top}\widetilde{\mathbf{D}}\mathbf{R}+o_{p}(1).
\end{equation}
Based on the same arguments on negligible remainders as in the proofs of \cref{Th:consistency,Th:CLT}, we have $\mathbf{R}^{\top}\widetilde{\mathbf{D}}\mathbf{R}=\mathbf{R}^{\top}\text{diag}(\Delta_{n}^{-1}(\Delta_{i}^{n}Z)^{2}\mathbbm{1}_{\{|\Delta_{i}^{n}Z|\leq u_{n}\}})\mathbf{R}+o_{p}(1)$, and thus, by the standard LLN result of truncated covariances, 
\begin{equation}
\label{Eq:sandwich2}
\mathbf{R}^{\top}\widetilde{\mathbf{D}}\mathbf{R}\overset{\mathbb{P}}{\longrightarrow}\int_{0}^{T}\widetilde{\sigma}_{s}^{2}\mathbf{B}_{s}^{\top}c_{s}\mathbf{B}_{s}ds.
\end{equation}
Combining \cref{Eq:sandwich1,Eq:sandwichEquivalence,Eq:sandwich2} leads to the consistency of $\widehat{\Sigma}_{T}^{\beta}$. This completes the proof.

\subsection[Proof of Theorem 4]{Proof of \cref{Th:oracle_local}}
\label{AP:Proof_oracle_local}

Define the block-wise score $c_{j}(\bm{\gamma})=\frac{\partial Q_{n}(\bm{\gamma})}{\partial\gamma_{j}}=-(\mathbf{R}_{j}^{*})^{\top}(\mathbf{Y}-\mathbf{R}\bm{\gamma})$. The standard KKT conditions (e.g., \citealp{xue2012variable}, Eqs.~(19)–(20)) imply that any local minimizer of $Q_{n}^{*}(\bm{\gamma})$ must satisfy:
\begin{align}
c_{j}(\widehat{\bm{\gamma}})=0,\qquad\sqrt{\widehat{\gamma}_{j}^{\top}\mathbf{W}_{j}\widehat{\gamma}_{j}}>\tau_{n},\qquad&\text{for }j=1,\dots,q. \label{Eq:KKT1}\\
\sqrt{(c_{j}(\widehat{\bm{\gamma}}))^{\top}\mathbf{W}_{j}c_{j}(\widehat{\bm{\gamma}})}\leq\frac{\lambda_{n}}{\tau_{n}},\qquad\sqrt{\widehat{\gamma}_{j}^{\top}\mathbf{W}_{j}\widehat{\gamma}_{j}}=0,\qquad&\text{for }j=q+1,\dots,p. \label{Eq:KKT2}
\end{align}
Then it suffices to show that $\widehat{\bm{\gamma}}^{\circ}$ satisfies both \cref{Eq:KKT1,Eq:KKT2} with probability approaching one. 

We denote by $\widehat{\bm{\gamma}}_{q}^{\circ}$ (resp.~$\bm{\gamma}_{q}^{\circ}$) the subvector of $\widehat{\bm{\gamma}}^{\circ}=((\widehat{\bm{\gamma}}_{q}^{\circ})^{\top},(\widehat{\bm{\gamma}}_{-q}^{\circ})^{\top})^{\top}$ (resp.~$\bm{\gamma}^{\circ}$) that includes only the first $qK_{n}$ elements, and by $\mathbf{R}_{q}=(\mathbf{R}_{1}^{*},\dots,\mathbf{R}_{q}^{*})^{\top}\in\mathbb{R}^{n\times qK_{n}}$ the submatrix of $\mathbf{R}$ corresponding to the first $q$ regressors. Define the oracle least-squares spline estimator,
\begin{equation}
\widehat{\bm{\gamma}}_{q}^{\circ}=\argmin_{\bm{\gamma}:\,\gamma_{j}\equiv0,\,\forall q<j\leq p}Q_{n}(\bm{\gamma}),\qquad
\widehat{\bm{\gamma}}_{-q}^{\circ}\equiv 0,\qquad
\widehat{\bm{\gamma}}_{q}^{\circ}=(\mathbf{R}_{q}^{\top}\mathbf{R}_{q})^{-1}\mathbf{R}_{q}^{\top}\mathbf{Y}.
\end{equation}
Define two standard projection matrices: the orthogonal projection matrix $\mathbf{H}_{q}=\mathbf{R}_{q}(\mathbf{R}_{q}^{\top}\mathbf{R}_{q})^{-1}\mathbf{R}_{q}^{\top}\in\mathbb{R}^{n\times n}$ onto span($\mathbf{R}_{q}$), and the residualizer $\mathbf{M}_{q}=\mathbf{I}_{n}-\mathbf{H}_{q}$, where $\mathbf{I}_{n}$ is an identity matrix of dimension $n$. Both $\mathbf{H}_{q}$ and $\mathbf{M}_{q}$ are symmetric and idempotent. Then the oracle residual is given by
\begin{equation}
\label{Eq:oracleResidual}
\mathbf{Y}-\mathbf{R}\widehat{\bm{\gamma}}^{\circ}=\mathbf{Y}-\mathbf{R}_{q}\widehat{\bm{\gamma}}_{q}^{\circ}=\mathbf{M}_{q}\mathbf{Y}.
\end{equation}

\subsubsection*{Active blocks}

By definition,  $\widehat{\bm{\gamma}}^{\circ}$ minimizes $Q_{n}(\bm{\gamma})$ over the restricted model $\{\bm{\gamma}:\bm{\gamma}_{-q}\equiv0\}$. Hence $c_{j}(\widehat{\bm{\gamma}}^{\circ})=0$ holds for all $1\leq j\leq q$. %It remains to verify $\sqrt{(\widehat{\gamma}_{j}^{\circ})^{\top}\mathbf{W}_{j}\widehat{\gamma}_{j}^{\circ}}>\tau_{n}$ for all $1\leq j\leq q$ with probablity approaching 1.  
Same as \cref{Eq:Gram}, it holds that
%\begin{equation}
%\label{Eq:KKT_inequality_active}
%\sqrt{(\widehat{\gamma}_{j}^{\circ})^{\top}\mathbf{W}_{j}\widehat{\gamma}_{j}^{\circ}}\asymp\Vert\mathbf{B}_{j}\widehat{\gamma}_{j}^{\circ}\Vert_{L^{2}}\overset{\mathbb{P}}{\longrightarrow}\Vert\beta_{j}\Vert_{L^{2}}>0,
%\end{equation}
\begin{equation}
\label{Eq:KKT_inequality_active}
(\widehat{\gamma}_{j}^{\circ})^{\top}\mathbf{W}_{j}\widehat{\gamma}_{j}^{\circ}\overset{\mathbb{P}}{\longrightarrow}\int_{0}^{T}(\widehat{\beta}_{j,s}^{\circ})^{\top}c_{s}\widehat{\beta}_{j,s}^{\circ}ds\asymp\Vert\widehat{\beta}_{j}^{\circ}\Vert_{L^{2}}^{2},\qquad\text{where}\quad\widehat{\beta}_{j}^{\circ}=\mathbf{B}_{j}\widehat{\gamma}_{j}^{\circ}.
\end{equation}
%For any active block $1\leq j\leq q$, the oracle estimator $\widehat{\beta}_{j}^{\circ}$ is the unpenalized fit on the true model with a fixed number of regressors. 
By \cref{Th:consistency} and the triangle inequality, for any active block $1\leq j\leq q$, 
\begin{equation}
\label{Eq:KKT_inequality_active2}
\left\vert\Vert\widehat{\beta}_{j}^{\circ}\Vert_{L^{2}}-\Vert\beta_{j}\Vert_{L^{2}}\right\vert\leq\Vert\widehat{\beta}_{j}^{\circ}-\beta_{j}\Vert_{L^{2}}\overset{\mathbb{P}}{\longrightarrow}0.
\end{equation}
so $\sqrt{(\widehat{\gamma}_{j}^{\circ})^{\top}\mathbf{W}_{j}\widehat{\gamma}_{j}^{\circ}}\overset{\mathbb{P}}{\longrightarrow}\Vert\beta_{j}\Vert_{L^{2}}>0$, and thus it satisfies $\sqrt{(\widehat{\gamma}_{j}^{\circ})^{\top}\mathbf{W}_{j}\widehat{\gamma}_{j}^{\circ}}>\tau_{n}$ for any $\tau_{n}\to0$.  Hence the active-block KKT conditions hold. 

\subsubsection*{Inactive blocks}
 
For  $q+1\leq j\leq p$, $\widehat{\gamma}_{j}^{\circ}\equiv 0$ by construction. Under the condition
\begin{equation}
\frac{\lambda_{n}}{\tau_{n}}\bigg/\left(\sqrt{\frac{\Delta_{n}\log pK_{n}}{K_{n}}}+u_{n}\log pK_{n}+u_{n}^{1-r/2}\right)\to\infty,
%\frac{\tau_{n}}{\lambda_{n}}\left(\sqrt{\frac{\Delta_{n}\log pK_{n}}{K_{n}}}+\Delta_{n}^{\varpi}\log pK_{n}\right)\to0 \quad\Leftrightarrow\quad
%\sqrt{\frac{\Delta_{n}\log pK_{n}}{K_{n}}}+\Delta_{n}^{\varpi}\log pK_{n}\ll\frac{\lambda_{n}}{\tau_{n}},
\end{equation}
it suffices to show
\begin{equation}
\max_{q+1\leq j\leq p}\sqrt{(c_{j}(\widehat{\bm{\gamma}}^{\circ}))^{\top}\mathbf{W}_{j}c_{j}(\widehat{\bm{\gamma}}^{\circ})}=O_{p}\left(\sqrt{\frac{\Delta_{n}\log pK_{n}}{K_{n}}}+u_{n}\log pK_{n}+u_{n}^{1-r/2}\right),
\end{equation}
which implies  \cref{Eq:KKT2}. By \cref{Eq:oracleResidual}, we have the residualized score for any inactive blocks,
\begin{equation}
c_{j}(\widehat{\bm{\gamma}}^{\circ})=-(\mathbf{R}_{j}^{*})^{\top}\mathbf{M}_{q}\mathbf{Y}=-(\widetilde{\mathbf{R}}_{j}^{*})^{\top}\mathbf{Y},\qquad
\text{where}\quad\widetilde{\mathbf{R}}_{j}^{*}=\mathbf{M}_{q}\mathbf{R}_{j}^{*},
\end{equation}
for any $q+1\leq j\leq p$. By the Cauchy-Schwarz inequality,
\begin{equation}
\label{Eq:oracle_score}
\sqrt{(c_{j}(\widehat{\bm{\gamma}}^{\circ}))^{\top}\mathbf{W}_{j}c_{j}(\widehat{\bm{\gamma}}^{\circ})}=\Vert\mathbf{W}_{j}^{1/2}(\widetilde{\mathbf{R}}_{j}^{*})^{\top}\mathbf{Y}\Vert\leq\Vert\mathbf{W}_{j}^{1/2}\Vert\Vert(\widetilde{\mathbf{R}}_{j}^{*})^{\top}\mathbf{Y}\Vert,
%\sqrt{(c_{j}(\widehat{\bm{\gamma}}^{\circ}))^{\top}\mathbf{W}_{j}c_{j}(\widehat{\bm{\gamma}}^{\circ})}=\Vert\mathbf{W}_{j}^{1/2}(\mathbf{R}_{j}^{*})^{\top}(\mathbf{Y}-\mathbf{R}\widehat{\bm{\gamma}}^{\circ})\Vert\leq\Vert\mathbf{W}_{j}^{1/2}\Vert\Vert(\mathbf{R}_{j}^{*})^{\top}(\mathbf{Y}-\mathbf{R}\widehat{\bm{\gamma}}^{\circ})\Vert,
\end{equation}
where $\Vert\mathbf{W}_{j}^{1/2}\Vert=\sqrt{\lambda_{\text{max}}(\mathbf{W}_{j})}=O_{p}(K_{n}^{-1/2})$ as shown following \cref{Eq:quotient}. It suffices to show
\begin{equation}
\max_{q+1\leq j\leq p}\Vert(\widetilde{\mathbf{R}}_{j}^{*})^{\top}\mathbf{Y}\Vert=O_{p}\left(\sqrt{\Delta_{n}\log pK_{n}}+\sqrt{K_{n}}u_{n}\log pK_{n}+\sqrt{K_{n}}u_{n}^{1-r/2}\right).
%\max_{q+1\leq j\leq p}\Vert(\mathbf{R}_{j}^{*})^{\top}(\mathbf{Y}-\mathbf{R}\widehat{\bm{\gamma}}^{\circ})\Vert=O_{p}\left(\sqrt{\Delta_{n}\log pK_{n}}+\sqrt{K_{n}}u_{n}\log pK_{n}+\sqrt{K_{n}}u_{n}^{1-r/2}\right).
\end{equation}
We consider the decomposition of $\mathbf{Y}$ as in \cref{Eq:decomp_Yi} (with the first two components combined):
\begin{equation}
\label{Eq:Y_decomp}
\begin{split}
\mathbf{Y}_{i}&=\underbrace{\left(\int_{(i-1)\Delta_{n}}^{i\Delta_{n}}\beta_{s-}^{\top}dX_{s}^{c}\right)\mathbbm{1}_{\{|\Delta_{i}^{n}Y|\leq u_{n}\}}}_{\mathbf{C}_{i}}+\,\underbrace{\left(\sum_{(i-1)\Delta_{n}\leq s\leq i\Delta_{n}}(\beta_{s-}^{J})^{\top}\Delta X_{s}\right)\mathbbm{1}_{\{|\Delta_{i}^{n}Y|\leq u_{n}\}}}_{\mathbf{J}_{i}}\\&\qquad\quad+\,\underbrace{\Delta_{i}^{n}Z\mathbbm{1}_{\{|\Delta_{i}^{n}Y|\leq u_{n}\}}}_{\mathbf{Z}_{i}}.
\end{split}
\end{equation}
Therefore, it suffices to verify, for each $\Xi\in\{\mathbf{C},\,\mathbf{J},\,\mathbf{Z}\}$,
\begin{equation}
\label{Eq:Xi}
\max_{q+1\leq j\leq p}\Vert(\widetilde{\mathbf{R}}_{j}^{*})^{\top}\Xi\Vert=O_{p}\left(\sqrt{\Delta_{n}\log pK_{n}}+\sqrt{K_{n}}u_{n}\log pK_{n}+\sqrt{K_{n}}u_{n}^{1-r/2}\right).
\end{equation}
We firstly show $\Vert\widetilde{\mathbf{R}}_{j}^{*}\Vert_{\infty}$ is bounded for any $q+1\leq j\leq p$.

\subsubsection*{(i) Bounded $\bm{\Vert\widetilde{\mathbf{R}}_{j}^{*}\Vert_{\infty}}$}

%We start with a deterministic inequality for the $\ell_{\infty}$-norm of $\widetilde{\mathbf{R}}_{j}^{*}$:
Let $h_{ii}=(\mathbf{H}_{q})_{ii}$ and $h_{\text{max}}=\max_{1\leq i\leq n}h_{ii}$. For any $\bm{r}\in\mathbb{R}^{n}$, $|(\mathbf{H}_{q}\bm{r})_{i}|\leq\sqrt{h_{ii}}\Vert\bm{r}\Vert$, and thus $\Vert\mathbf{H}_{q}\bm{r}\Vert_{\infty}\leq\sqrt{h_{\text{max}}}\Vert\bm{r}\Vert$. By definition,
\begin{equation}
\label{Eq:stability}
\Vert\mathbf{M}_{q}\bm{r}\Vert_{\infty}\leq\Vert(\mathbf{I}_{n}-\mathbf{H}_{q})\bm{r}\Vert_{\infty}\leq\Vert\bm{r}\Vert_{\infty}+\sqrt{h_{\text{max}}}\Vert\bm{r}\Vert.
\end{equation}
Moreover, it holds that
\begin{equation}
\label{Eq:stability2}
\begin{split}
h_{\text{max}}&\leq\left(\max_{1\leq i\leq n}\Vert\mathbf{R}_{i}\Vert^{2}\right)\Vert(\mathbf{R}_{q}^{\top}\mathbf{R}_{q})^{-1}\Vert\\
&\leq\left(\max_{1\leq i\leq n}\sum_{j=1}^{q}\sum_{k=1}^{K_{n}}(B_{(i-1)\Delta_{n}}^{(k)}\Delta_{j,i}^{n}X)^{2}\mathbbm{1}_{\{|\Delta_{j,i}^{n}X|\leq u_{n}\}}\right)\Vert(\mathbf{R}_{q}^{\top}\mathbf{R}_{q})^{-1}\Vert\\
&\leq\left(\max_{1\leq i\leq n}qu_{n}^{2}\sum_{k=1}^{K_{n}}(B_{(i-1)\Delta_{n}}^{(k)})^{2}\right)\Vert(\mathbf{R}_{q}^{\top}\mathbf{R}_{q})^{-1}\Vert=O_{p}(K_{n}u_{n}^{2}),
\end{split}
\end{equation}
by the boundedness and local support of B-splines, and $\Vert(\mathbf{R}_{q}^{\top}\mathbf{R}_{q})^{-1}\Vert=O_{p}(K_{n})$. Therefore, by \cref{Eq:stability,Eq:stability2},
\begin{equation}
\label{Eq:stability_bound}
\Vert\widetilde{\mathbf{R}}_{j}^{*}\Vert_{\infty}\leq\Vert\mathbf{R}_{j}^{*}\Vert_{\infty}+\sqrt{h_{\text{max}}}\Vert\mathbf{R}_{j}^{*}\Vert\leq Ku_{n}+K'\sqrt{K_{n}}u_{n}\frac{1}{\sqrt{K_{n}}}\leq K''u_{n}.
\end{equation}

\subsubsection*{(ii) Maximal inequality for martingale difference array}

We start with a high-dimensional martingale maximal inequality from standard Freedman’s inequality, as an analogue to Lemma A.1 in \citet{van2008high} from Bernstein’s inequality:
\begin{lemma}
\label{lemma:BernsteinMax}
Let $(\mathcal{F}_{i})_{i=0}^{n}$ be a filtration and $(M_{i}^{(\ell)})$ be a triangular array indexed by $1\leq\ell\leq m$. Assume there exist $a_{n}$, $b_{n}>0$ such that for all $1\leq\ell\leq m$,
\begin{itemize}
\item[(i)] $\mathbb{E}[M_{i}^{(\ell)}|\mathcal{F}_{i-1}]=0$ for all $i$;
\item[(ii)] $|M_{i}^{(\ell)}|\leq b_{n}$ almost surely for all $i$;
\item[(iii)] $\sum_{i=1}^{n}\mathbb{E}[(M_{i}^{(\ell)})^{2}|\mathcal{F}_{i-1}]\leq a_{n}$.
\end{itemize}
Then, 
%for any $\varepsilon>0$,
%\begin{equation}
%\mathbb{P}\left(\max_{1\leq\ell\leq m}\left\vert\sum_{i=1}^{n}M_{i}^{(\ell)}\right\vert\geq\varepsilon\right)\leq2m\exp\left(-\frac{\varepsilon^{2}}{2(a_{n}+b_{n}\varepsilon/3)}\right),
%\end{equation}
%and thus
\begin{equation}
\max_{1\leq\ell\leq m}\left\vert\sum_{i=1}^{n}M_{i}^{(\ell)}\right\vert=O_{p}\left(\sqrt{a_{n}\log m}+b_{n}\log m\right).
\end{equation}
\end{lemma}

\begin{proof}[Proof of \cref{lemma:BernsteinMax}]
By Freedman’s inequality (see, e.g., Theorem 1.1, \citealp{tropp2011freedman}), for any $\varepsilon>0$,
\begin{equation}
\mathbb{P}\left(\left\vert\sum_{i=1}^{n}M_{i}^{(\ell)}\right\vert\geq\varepsilon\right)\leq2\exp\left(-\frac{\varepsilon^{2}}{2(a_{n}+b_{n}\varepsilon/3)}\right).
\end{equation}
By the union bound over all $1\leq\ell\leq m$,
\begin{equation}
\label{Eq:union}
\mathbb{P}\left(\max_{1\leq\ell\leq m}\left\vert\sum_{i=1}^{n}M_{i}^{(\ell)}\right\vert\geq\varepsilon\right)\leq 2m\exp\left(-\frac{\varepsilon^{2}}{2(a_{n}+b_{n}\varepsilon/3)}\right).
\end{equation}
Choose $\varepsilon=C(\sqrt{a_{n}\log m}+b_{n}\log m)$ with a sufficiently large constant $C>0$. 

(a) If $\sqrt{a_{n}\log m}\geq b_{n}\log m$, we have $\varepsilon\asymp C\sqrt{a_{n}\log m}$, $\sqrt{a_{n}}\geq b_{n}\sqrt{\log m}$, and $a_{n}+b_{n}\varepsilon/3\asymp a_{n}$, so the right-hand side of \cref{Eq:union} satisfies
\begin{equation}
\label{Eq:maximalResult}
2m\exp\left(-\frac{\varepsilon^{2}}{2(a_{n}+b_{n}\varepsilon/3)}\right)\leq 2m\exp(-KC^{2}\log m)=2m^{1-KC^{2}}\to 0,
\end{equation}
for $C$ large enough.

(b) If $\sqrt{a_{n}\log m}\leq b_{n}\log m$, we have $\varepsilon\asymp Cb_{n}\log m$, $a_{n}\leq b_{n}^{2}\log m$, and $a_{n}+b_{n}\varepsilon/3\asymp b_{n}^{2}\log m$, so the right-hand side of \cref{Eq:union} has the same bound as in \cref{Eq:maximalResult}. 

This completes the proof of \cref{lemma:BernsteinMax}.

\end{proof}

Next, we use \cref{lemma:BernsteinMax} to prove \cref{Eq:Xi} for each $\Xi$. We start with some notation: For the $j$-th block design submatrix $\mathbf{R}_{j}^{*}$ in \cref{Eq:R_j^*}, we write each element as
\begin{equation}
R_{j,k,i}=B_{(i-1)\Delta_{n}}^{(k)}\Delta_{j,i}^{n}X\mathbbm{1}_{\{\vert\Delta_{j,i}^{n}X\vert\leq u_{n}\}},
\end{equation}
and the $((j-1)K_{n}+k)$-column $R_{j,k}=(R_{j,k,i})_{i=1}^{n}\in\mathbb{R}^{n}$. For $\widetilde{\mathbf{R}}_{j}^{*}=\mathbf{M}_{q}\mathbf{R}_{j}^{*}$, we write the corresponding ``residualized'' column as
\begin{equation}
\widetilde{R}_{j,k}=\mathbf{M}_{q}R_{j,k}=(\widetilde{R}_{j,k,i})_{i=1}^{n}\in\mathbb{R}^{n},
\end{equation}
where each element $|\widetilde{R}_{j,k,i}|\leq Ku_{n}$ by \cref{Eq:stability_bound}.

\subsubsection*{(ii.1) $\bm{\Xi=\mathbf{Z}}$}
\label{AP:Xi_1}

Define 
\begin{equation}
\widetilde{M}_{i}^{(j,k)}=\widetilde{R}_{j,k,i}\Delta_{i}^{n}Z\mathbbm{1}_{\{|\Delta_{i}^{n}Y|\leq u_{n}\}},
\end{equation}
and
\begin{equation}
\label{Eq:Xi_Z}
M_{i}^{(j,k)}=\widetilde{M}_{i}^{(j,k)}-\mathbb{E}\left[\left.\widetilde{M}_{i}^{(j,k)}\right\vert\mathcal{F}_{(i-1)\Delta_{n}}\right],\quad\text{so that}\quad
\mathbb{E}\left[\left.M_{i}^{(j,k)}\right\vert\mathcal{F}_{(i-1)\Delta_{n}}\right]=0.
\end{equation}
With finite-variation jumps in $Z$ and bounded $\widetilde{\delta}$ under \cref{As:Ito_Iocal}, it holds that $|\widetilde{M}_{i}^{(j,k)}|\leq Ku_{n}$, and thus $|M_{i}^{(j,k)}|\leq 2Ku_{n}$. For the conditional variance, by It{\^o}'s isometry and Markov's inequality, $\mathbb{E}[(\Delta_{i}^{n}Z)^{2}|\breve{\mathcal{F}}_{(i-1)\Delta_{n}}]=O_{p}(\Delta_{n})$ with an enlarged $\sigma$-field $\breve{\mathcal{F}}_{(i-1)\Delta_{n}}=\mathcal{F}_{(i-1)\Delta_{n}}\cup\sigma(\mathbf{M}_{q},\Delta_{i}^{n}X)$. Then, since $\widetilde{R}_{j,k,i}$ is $\breve{\mathcal{F}}_{(i-1)\Delta_{n}}$-measurable, by the law of iterated conditional expectations,
\begin{equation}
\label{Eq:conditionalVar}
\begin{split}
\sum_{i=1}^{n}\mathbb{E}\left[\left.(M_{i}^{(j,k)})^{2}\right\vert\mathcal{F}_{(i-1)\Delta_{n}}\right]&\leq\sum_{i=1}^{n}\mathbb{E}\left[\left.(\widetilde{M}_{i}^{(j,k)})^{2}\right\vert\mathcal{F}_{(i-1)\Delta_{n}}\right]\\
&\leq\sum_{i=1}^{n}\mathbb{E}\left[\left.\mathbb{E}\left[\left.(\Delta_{i}^{n}Z)^{2}\right\vert\breve{\mathcal{F}}_{(i-1)\Delta_{n}}\right]\widetilde{R}_{j,k,i}^{2}\right\vert\mathcal{F}_{(i-1)\Delta_{n}}\right]\\
&\leq K\Delta_{n}\sum_{i=1}^{n}\mathbb{E}\left[\left.\widetilde{R}_{j,k,i}^{2}\right\vert\mathcal{F}_{(i-1)\Delta_{n}}\right],
\end{split}
\end{equation}
where, since $\mathbf{M}_{q}$ is an orthogonal projection,
\begin{equation}
\label{Eq:conditionalVar2}
\mathbb{E}\left[\sum_{i=1}^{n}\mathbb{E}\left[\left.\widetilde{R}_{j,k,i}^{2}\right\vert\mathcal{F}_{(i-1)\Delta_{n}}\right]\right]=\sum_{i=1}^{n}\mathbb{E}[\widetilde{R}_{j,k,i}^{2}]=\mathbb{E}[\Vert\widetilde{R}_{j,k}\Vert^{2}]\leq\mathbb{E}[\Vert R_{j,k}\Vert^{2}]=O(K_{n}^{-1}),
\end{equation}
and thus $\sum_{i=1}^{n}\mathbb{E}[(M_{i}^{(j,k)})^{2}|\mathcal{F}_{(i-1)\Delta_{n}}]=O_{p}(K_{n}^{-1}\Delta_{n})$ by Markov's inequality. 
%As shown in Appendix \ref{AP:consistency} (\hyperref[AP:A4]{iv}), $\sum_{i=1}^{n}\mathbb{E}[(M_{i}^{(j,k)})^{2}|\mathcal{F}_{(i-1)\Delta_{n}}]\asymp\mathbb{E}[(C_{j}^{(k)})^{2}]=O(K_{n}^{-1}\Delta_{n})$ by the standard results of realized covariances.
Therefore, by \cref{lemma:BernsteinMax},
\begin{equation}
\label{Eq:Xi_Z_1}
\max_{q+1\leq j\leq p}\max_{1\leq k\leq K_{n}}\left\vert\sum_{i=1}^{n}M_{i}^{(j,k)}\right\vert=O_{p}\left(\sqrt{\frac{\Delta_{n}\log pK_{n}}{K_{n}}}+u_{n}\log pK_{n}\right).
\end{equation}
Next, we bound the predictable compensator uniformly. Write
\begin{equation}
\sum_{i=1}^{n}\mathbb{E}\left[\left.\widetilde{M}_{i}^{(j,k)}\right\vert\mathcal{F}_{(i-1)\Delta_{n}}\right]=\sum_{i=1}^{n}\mathbb{E}\left[\left.\widetilde{R}_{j,k,i}\Delta_{i}^{n}Z\mathbbm{1}_{\{|\Delta_{i}^{n}Y|\leq u_{n}\}}\right\vert\mathcal{F}_{(i-1)\Delta_{n}}\right]
\end{equation}
By the exogeneity condition under \cref{As:orthogonality}, the untruncated product has zero conditional expectation: $\mathbb{E}[\widetilde{R}_{j,k,i}\Delta_{i}^{n}Z|\mathcal{F}_{(i-1)\Delta_{n}}]=0$. Hence, we flip the indicator function
\begin{equation}
\begin{split}
\mathbb{E}\left[\left.\widetilde{M}_{i}^{(j,k)}\right\vert\mathcal{F}_{(i-1)\Delta_{n}}\right]=-\mathbb{E}\left[\left.\widetilde{R}_{j,k,i}\Delta_{i}^{n}Z\mathbbm{1}_{\{|\Delta_{i}^{n}Y|>u_{n}\}}\right\vert\mathcal{F}_{(i-1)\Delta_{n}}\right],
\end{split}
\end{equation}
and then, by the Cauchy-Schwarz inequality, \cref{lemma:truncationEvents} (i), and Markov's inequality,
\begin{align}
\sum_{i=1}^{n}\left\vert\mathbb{E}\left[\left.\widetilde{M}_{i}^{(j,k)}\right\vert\mathcal{F}_{(i-1)\Delta_{n}}\right]\right\vert&\leq Ku_{n}\sum_{i=1}^{n}\mathbb{E}\left[\left.|\Delta_{i}^{n}Z|\mathbbm{1}_{\{|\Delta_{i}^{n}Y|>u_{n}\}}\right\vert\mathcal{F}_{(i-1)\Delta_{n}}\right] \notag\\
&\leq Ku_{n}\left(\sum_{i=1}^{n}\mathbb{E}[(\Delta_{i}^{n}Z)^{2}|\mathcal{F}_{(i-1)\Delta_{n}}]\right)^{1/2}\left(\sum_{i=1}^{n}\mathbb{P}(|\Delta_{i}^{n}Y|>u_{n}|\mathcal{F}_{(i-1)\Delta_{n}})\right)^{1/2}\notag\\
&\leq K'u_{n}\left(\sum_{i=1}^{n}O_{p}(\Delta_{n})\right)^{1/2}\left(\sum_{i=1}^{n}O_{p}(\Delta_{n}u_{n}^{-r})\right)^{1/2}=O_{p}(u_{n}^{1-r/2}), \label{Eq:Xi_Z_2}
\end{align}
for any $q+1\leq j\leq p$ and $1\leq k\leq K_{n}$. Combining \cref{Eq:Xi_Z_1,Eq:Xi_Z_2} in \cref{Eq:Xi_Z}, we obtain
\begin{equation}
\label{Eq:Xi_Z_final}
\begin{split}
\max_{q+1\leq j\leq p}\Vert(\widetilde{\mathbf{R}}_{j}^{*})^{\top}\mathbf{Z}\Vert&\leq\sqrt{K_{n}}\max_{q+1\leq j\leq p}\max_{1\leq k\leq K_{n}}\left\vert\sum_{i=1}^{n}\widetilde{M}_{i}^{(j,k)}\right\vert\\
&\leq \sqrt{K_{n}}\max_{q+1\leq j\leq p}\max_{1\leq k\leq K_{n}}\left(\left\vert\sum_{i=1}^{n}M_{i}^{(j,k)}\right\vert+\left\vert\sum_{i=1}^{n}\mathbb{E}\left[\left.\widetilde{M}_{i}^{(j,k)}\right\vert\mathcal{F}_{(i-1)\Delta_{n}}\right]\right\vert\right)\\
&=O_{p}\left(\sqrt{\Delta_{n}\log pK_{n}}+\sqrt{K_{n}}u_{n}\log pK_{n}+\sqrt{K_{n}}u_{n}^{1-r/2}\right).
\end{split}
\end{equation}

\subsubsection*{(ii.2) $\bm{\Xi=\mathbf{C}}$}
\label{AP:Xi_2}

Define 
\begin{equation}
\widetilde{M}_{i}^{(j,k)}=\widetilde{R}_{j,k,i}\left(\int_{(i-1)\Delta_{n}}^{i\Delta_{n}}\beta_{s-}^{\top}dX_{s}^{c}\right)\mathbbm{1}_{\{|\Delta_{i}^{n}Y|\leq u_{n}\}},
\end{equation}
and $M_{i}^{(j,k)}=\widetilde{M}_{i}^{(j,k)}-\mathbb{E}[\widetilde{M}_{i}^{(j,k)}\vert\mathcal{F}_{(i-1)\Delta_{n}}]$.
%\begin{equation}
%\label{Eq:Xi_C}
%M_{i}^{(j,k)}=\widetilde{M}_{i}^{(j,k)}-\mathbb{E}\left[\left.\widetilde{M}_{i}^{(j,k)}\right\vert\mathcal{F}_{(i-1)\Delta_{n}}\right],\quad\text{so that}\quad
%\mathbb{E}\left[\left.M_{i}^{(j,k)}\right\vert\mathcal{F}_{(i-1)\Delta_{n}}\right]=0.
%\end{equation}
By bounded $\beta$ under \cref{As:Ito_Iocal}, $|M_{i}^{(j,k)}|\leq Ku_{n}$. The remaining steps follow exactly as in (\hyperref[AP:Xi_1]{ii.1}), and we conclude
\begin{equation}
\max_{q+1\leq j\leq p}\Vert(\widetilde{\mathbf{R}}_{j}^{*})^{\top}\mathbf{C}\Vert=O_{p}\left(\sqrt{\Delta_{n}\log pK_{n}}+\sqrt{K_{n}}u_{n}\log pK_{n}+\sqrt{K_{n}}u_{n}^{1-r/2}\right).
\end{equation}

\subsubsection*{(ii.3) $\bm{\Xi=\mathbf{J}}$}
\label{AP:Xi_3}

Using the safe replacement $\widetilde{\mathbf{J}}$ for $\mathbf{J}$ in \cref{Eq:J_i_tilde},
\begin{equation}
\begin{split}
\widetilde{M}_{i}^{(j,k)}&=\widetilde{R}_{j,k,i}\left(\sum_{(i-1)\Delta_{n}\leq s\leq i\Delta_{n}}(\beta_{s-}^{J})^{\top}\Delta X_{s}\right)\mathbbm{1}_{\{|\Delta_{i}^{n}Y|\leq u_{n}\}}\\
&=\widetilde{R}_{j,k,i}\left(\sum_{(i-1)\Delta_{n}\leq s\leq i\Delta_{n}}(\beta_{s-}^{J})^{\top}\Delta X_{s}\right)\mathbbm{1}_{\{|\Delta_{j,i}^{n}X|\leq u_{n},\,|\Delta_{i}^{n}Y|\leq u_{n}\}},
\end{split}
\end{equation}
since $\widetilde{R}_{j,k,i}$ already includes the truncation $\mathbbm{1}_{\{|\Delta_{j,i}^{n}X|\leq u_{n}\}}$, and $M_{i}^{(j,k)}=\widetilde{M}_{i}^{(j,k)}-\mathbb{E}[\widetilde{M}_{i}^{(j,k)}\vert\mathcal{F}_{(i-1)\Delta_{n}}]$. With bounded $\beta^{J}$ and $\delta$ under \cref{As:Ito_Iocal} and finite-variation jumps in $X$, $|M_{i}^{(j,k)}|\leq Ku_{n}$. For the conditional variance, we follow the same steps as in (\hyperref[AP:Xi_1]{ii.1}). By the result in \cref{Eq:J_i_moment} and Markov's inequality,
\begin{equation}
\mathbb{E}\left[\left.\left(\sum_{(i-1)\Delta_{n}\leq s\leq i\Delta_{n}}(\beta_{s-}^{J})^{\top}\Delta X_{s}\right)^{2}\mathbbm{1}_{\{|\Delta_{j,i}^{n}X|\leq u_{n},\,|\Delta_{i}^{n}Y|\leq u_{n}\}}\right\vert\breve{\mathcal{F}}_{(i-1)\Delta_{n}}\right]=O_{p}(\Delta_{n}u_{n}^{2-r}+\Delta_{n}^{2}u_{n}^{-2r}),
\end{equation}
and thus $\sum_{i=1}^{n}\mathbb{E}[(M_{i}^{(j,k)})^{2}|\mathcal{F}_{(i-1)\Delta_{n}}]=O_{p}(K_{n}^{-1}(\Delta_{n}u_{n}^{2-r}+\Delta_{n}^{2}u_{n}^{-2r}))$. Therefore, by \cref{lemma:BernsteinMax},
\begin{equation}
\begin{split}
\max_{q+1\leq j\leq p}\max_{1\leq k\leq K_{n}}\left\vert\sum_{i=1}^{n}M_{i}^{(j,k)}\right\vert&=O_{p}\left(\sqrt{\frac{(\Delta_{n}u_{n}^{2-r}+\Delta_{n}^{2}u_{n}^{-2r})\log pK_{n}}{K_{n}}}+u_{n}\log pK_{n}\right)\\
&=O_{p}\left(\sqrt{\frac{\Delta_{n}\log pK_{n}}{K_{n}}}+u_{n}\log pK_{n}\right). \label{Eq:Xi_J_1}
\end{split}
\end{equation}
We bound the predictable compensator uniformly with the flipped indicator function as in \cref{Eq:Xi_Z_2},
\begin{align}
&\sum_{i=1}^{n}\left\vert\mathbb{E}\left[\left.\widetilde{M}_{i}^{(j,k)}\right\vert\mathcal{F}_{(i-1)\Delta_{n}}\right]\right\vert \notag\\
&\qquad\leq Ku_{n}\sum_{i=1}^{n}\mathbb{E}\left[\left.\left\vert\sum_{(i-1)\Delta_{n}\leq s\leq i\Delta_{n}}(\beta_{s-}^{J})^{\top}\Delta X_{s}\right\vert\mathbbm{1}_{\{|\Delta_{j,i}^{n}X|\leq u_{n},\,|\Delta_{i}^{n}Y|> u_{n}\}}\right\vert\mathcal{F}_{(i-1)\Delta_{n}}\right] \notag\\
&\qquad\leq K'u_{n}\left(\sum_{i=1}^{n}\mathbb{E}\left[\left.\left(\sum_{(i-1)\Delta_{n}\leq s\leq i\Delta_{n}}\Delta X_{s}\right)^{2}\right\vert\mathcal{F}_{(i-1)\Delta_{n}}\right]\right)^{1/2}\left(\sum_{i=1}^{n}\mathbb{P}(|\Delta_{i}^{n}Y|>u_{n}|\mathcal{F}_{(i-1)\Delta_{n}})\right)^{1/2}\notag\\
&\qquad=O_{p}(u_{n}^{1-r/2}). \label{Eq:Xi_J_2}
\end{align}
Combining \cref{Eq:Xi_J_1,Eq:Xi_J_2}, as in \cref{Eq:Xi_Z_final}, leads to
\begin{equation}
\max_{q+1\leq j\leq p}\Vert(\widetilde{\mathbf{R}}_{j}^{*})^{\top}\mathbf{J}\Vert=O_{p}\left(\sqrt{\Delta_{n}\log pK_{n}}+\sqrt{K_{n}}u_{n}\log pK_{n}+\sqrt{K_{n}}u_{n}^{1-r/2}\right).
\end{equation}
Finally, \cref{Eq:Xi} holds by all results in (\hyperref[AP:Xi_1]{ii.1}), (\hyperref[AP:Xi_2]{ii.2}) and (\hyperref[AP:Xi_3]{ii.3}), and this completes the proof.

\subsection[Proof of Theorem 5]{Proof of \cref{Th:oracle_global}}

We define the oracle index set  $A^{\circ}=\{1,\dots,q\}$. Let $\Gamma_{1}=\{A:A\subseteq\{1,\dots,p\}, A^{\circ}\subset A, A\neq A^{\circ}\}$ collect all index sets of overfitting models, $\Gamma_{2}=\{A:A\subseteq\{1,\dots,p\}, A\subset A^{\circ}, A\neq A^{\circ}\}$ collect those of underfitting models, and $\Gamma_{3}=\{A:A\subseteq\{1,\dots,p\}, A^{\circ}\not\subset A, A\not\subset A^{\circ}, A\neq A^{\circ}\}$ collect those of other misspecified models. For any $\bm{\gamma}$, $A(\bm{\gamma})=\{j:\sqrt{\gamma_{j}^{\top}\mathbf{W}_{j}\gamma_{j}}\geq\tau_{n}\}$ must fall into one of $\Gamma_{\ell}$ for $\ell\in\{1,2,3\}$ if $A(\bm{\gamma})\neq A^{\circ}$. %We denote by $A^{\complement}(\bm{\gamma})=\{j:\sqrt{\gamma_{j}^{\top}\mathbf{W}_{j}\gamma_{j}}<\tau_{n}\}=\{1,\dots,p\}\backslash A(\bm{\gamma})$, and by $|A|$ the cardinality of $A$. 
$|A|$ denotes the cardinality of $A$. We aim to show that, with probability approaching one, 
\begin{equation}
\min_{\bm{\gamma}: A(\bm{\gamma})\neq A^{\circ}}Q_{n}^{*}(\bm{\gamma})>Q_{n}^{*}(\widehat{\bm{\gamma}}^{\circ}),
\end{equation}
which ensures that the global minimizer $\widehat{\bm{\gamma}}^{*}$ of $Q_{n}^{*}(\bm{\gamma})$ coincides with the oracle estimator $\widehat{\bm{\gamma}}^{\circ}$. It suffices to show that, for $\ell\in\{1,2,3\}$,
\begin{equation}
\label{Eq:Gamma123}
\sum_{A\in\Gamma_{\ell}}\mathbb{P}\left(\min_{\bm{\gamma}: A(\bm{\gamma})=A}Q_{n}^{*}(\bm{\gamma})-Q_{n}^{*}(\widehat{\bm{\gamma}}^{\circ})\leq 0\right)\to0,
\end{equation}
such that, combining with \cref{Th:oracle_local}, it holds that the global minimizer $\widehat{\bm{\gamma}}^{*}$ exists and satisfies $\widehat{\bm{\gamma}}^{*}=\widehat{\bm{\gamma}}^{\circ}$ with probability approaching one.  Firstly, we have
\begin{equation}
%\min_{\bm{\gamma}: A(\bm{\gamma})=A}Q_{n}^{*}(\bm{\gamma})=\frac{1}{2}\mathbf{Y}^{\top}(\mathbf{I}_{n}-\mathbf{P}_{A})\mathbf{Y}+\lambda_{n}\sum_{j\in A}\rho_{n}\left(\sqrt{\gamma_{j}^{\top}\mathbf{W}_{j}\gamma_{j}}\right),
\begin{split}
\min_{\bm{\gamma}: A(\bm{\gamma})=A}Q_{n}^{*}(\bm{\gamma})&=\min_{\bm{\gamma}: A(\bm{\gamma})=A}\left(\frac{1}{2}Q_{n}(\bm{\gamma})+\lambda_{n}\sum_{j\in A(\bm{\gamma})}\rho_{n}\left(\sqrt{\gamma_{j}^{\top}\mathbf{W}_{j}\gamma_{j}}\right)+\lambda_{n}\sum_{j\in \{1,\dots,p\}\backslash A(\bm{\gamma})}\rho_{n}\left(\sqrt{\gamma_{j}^{\top}\mathbf{W}_{j}\gamma_{j}}\right)\right)\\
&=\frac{1}{2}\min_{\bm{\gamma}: A(\bm{\gamma})=A}Q_{n}(\bm{\gamma})+\lambda_{n}|A|,
\end{split}
\end{equation}
and thus comparing any model $A$ to the oracle model $A^{\circ}$ reduces to comparing unpenalized losses and constant penalties:
\begin{equation}
\label{Eq:DeltaQstar}
\min_{\bm{\gamma}: A(\bm{\gamma})=A}Q_{n}^{*}(\bm{\gamma})-Q_{n}^{*}(\widehat{\bm{\gamma}}^{\circ})=\frac{1}{2}\left(\min_{\bm{\gamma}: A(\bm{\gamma})=A}Q_{n}(\bm{\gamma})-Q_{n}(\widehat{\bm{\gamma}}^{\circ})\right)+\lambda_{n}(|A|-q). 
\end{equation}
For any index set $A$, let $\mathbf{R}_{A}=(\mathbf{R}_{j}^{*})_{j\in A}\in\mathbb{R}^{n\times K_{n}|A|}$ be the submatrix of $\mathbf{R}$ formed by blocks $j\in A$. We define the orthogonal projection matrix $\mathbf{H}_{A}=\mathbf{R}_{A}(\mathbf{R}_{A}^{\top}\mathbf{R}_{A})^{-1}\mathbf{R}_{A}^{\top}\in\mathbb{R}^{n\times n}$ onto span($\mathbf{R}_{A}$), and the residualizer $\mathbf{M}_{A}=\mathbf{I}_{n}-\mathbf{H}_{A}$, where $\mathbf{I}_{n}$ is an identity matrix of dimension $n$.  Both $\mathbf{H}_{A}$ and $\mathbf{M}_{A}$ are symmetric and idempotent.

\subsubsection*{Overfitting models}

For any overfitting models with $A\in\Gamma_{1}$, the added blocks in the design matrix $\mathbf{R}_{A}$ corresponding to the index set $A\backslash A^{\circ}$ leads to a decrease in the OLS residual sums of squares. We denote the residual vector by $\widehat{\bm{\varepsilon}}=\mathbf{M}_{A^{\circ}}\mathbf{Y}=\mathbf{Y}-\mathbf{R}\widehat{\bm{\gamma}}^{\circ}=\mathbf{Y}-\mathbf{R}_{A^{\circ}}\bm{\gamma}_{A^{\circ}}\in\mathbb{R}^{n}$. For any irrelevant group $j>q$, we define the spurious improvement in the unpenalized loss from adding block $j$ as
\begin{equation}
\label{Eq:Delta_Q_j}
\Delta Q_{n,j}=Q_{n}(\widehat{\bm{\gamma}}^{\circ})-\min_{\bm{\gamma}: A(\bm{\gamma})\subseteq A^{\circ}\cup\{j\}}Q_{n}(\bm{\gamma})=\mathbf{Y}^{\top}(\mathbf{M}_{A^{\circ}}-\mathbf{M}_{A^{\circ}\cup\{j\}})\mathbf{Y}.
\end{equation}
The following \cref{lemma:overfitOneBlock} provides a uniform bound for $\Delta Q_{n,j}$:

\begin{lemma}
\label{lemma:overfitOneBlock}
Under the conditions of \cref{Th:oracle_global}, there exists a constant $K>0$ independent of $n$ and $p$ such that, as $\Delta_{n}\to0$,
\begin{equation}
\label{Eq:E_n}
\mathbb{P}(E_{n})=\mathbb{P}\left(\max_{q+1\leq j\leq p}\Delta Q_{n,j}\leq K\Delta_{n}(K_{n}+\log p)\right)\to 1.
\end{equation}
\end{lemma}

\begin{proof}[Proof of \cref{lemma:overfitOneBlock}]
%It holds by definition that 
%\begin{equation}
%\Delta Q_{n,j}\leq\mathbf{Y}^{\top}\mathbf{M}_{A^{\circ}}\mathbf{Y}-\mathbf{Y}^{\top}\mathbf{M}_{A^{\circ}\cup\{j\}}\mathbf{Y}.
%\end{equation}
Define the residualized block as $\widetilde{\mathbf{R}}_{j}^{*}=\mathbf{M}_{A^{\circ}}\mathbf{R}_{j}^{*}\in\mathbb{R}^{n\times K_{n}}$, the projection matrix $\mathbf{H}_{j}^{*}=\widetilde{\mathbf{R}}_{j}^{*}((\widetilde{\mathbf{R}}_{j}^{*})^{\top}\widetilde{\mathbf{R}}_{j}^{*})^{-1}(\widetilde{\mathbf{R}}_{j}^{*})^{\top}\in\mathbb{R}^{n\times n}$ onto span$(\widetilde{\mathbf{R}}_{j}^{*})$, and the corresponding residualizer $\mathbf{M}_{j}^{*}=\mathbf{I}_{n}-\mathbf{H}_{j}^{*}$. By the Frisch-Waugh-Lovell theorem (Theorem 2.1, \citealp{davidson2004econometric}), the OLS residual sums of squares from the overfitted model with the irrelevant block $j$ satisfies
\begin{equation}
\label{Eq:FWL}
\mathbf{Y}^{\top}\mathbf{M}_{A^{\circ}\cup\{j\}}\mathbf{Y}=\widehat{\bm{\varepsilon}}^{\top}\mathbf{M}_{j}^{*}\widehat{\bm{\varepsilon}},
\end{equation}
so that
\begin{equation}
\label{Eq:FWL2}
\Delta Q_{n,j}= \widehat{\bm{\varepsilon}}^{\top}\widehat{\bm{\varepsilon}}-\widehat{\bm{\varepsilon}}^{\top}\mathbf{M}_{j}^{*}\widehat{\bm{\varepsilon}}=\widehat{\bm{\varepsilon}}^{\top}\mathbf{H}_{j}^{*}\widehat{\bm{\varepsilon}}=\Vert(\mathbf{P}_{j}^{*})^{\top}\widehat{\bm{\varepsilon}}\Vert^{2},%\qquad
%\text{where}\quad\mathbf{P}_{j}^{*}=\widetilde{\mathbf{R}}_{j}^{*}((\widetilde{\mathbf{R}}_{j}^{*})^{\top}\widetilde{\mathbf{R}}_{j}^{*})^{-1/2},
\end{equation}
where $\mathbf{P}_{j}^{*}=\widetilde{\mathbf{R}}_{j}^{*}((\widetilde{\mathbf{R}}_{j}^{*})^{\top}\widetilde{\mathbf{R}}_{j}^{*})^{-1/2}$, and it suffices to show
\begin{equation}
\label{Eq:FWL3}
\max_{q+1\leq j\leq p}\Vert(\mathbf{P}_{j}^{*})^{\top}\widehat{\bm{\varepsilon}}\Vert^{2}\leq K\Delta_{n}(K_{n}+\log p),
\end{equation}
with probability approaching one.

Consider a further decomposition based on \cref{Eq:decomp_Yi}:
\begin{equation}
\begin{split}
\mathbf{Y}_{i}&=\underbrace{\left(\int_{(i-1)\Delta_{n}}^{i\Delta_{n}}\widetilde{\beta}_{s}^{\top}dX_{s}^{c}\right)\mathbbm{1}_{\{|\Delta_{i}^{n}Y|\leq u_{n}\}}}_{\widetilde{\mathbf{Y}}_{i}}\,+\,\underbrace{\int_{(i-1)\Delta_{n}}^{i\Delta_{n}}e_{s}^{\top}dX_{s}^{c}}_{\mathbf{U}_{i}^{(1)}}\\
&\qquad-\,\underbrace{\left(\int_{(i-1)\Delta_{n}}^{i\Delta_{n}}e_{s}^{\top}dX_{s}^{c}\right)\mathbbm{1}_{\{|\Delta_{i}^{n}Y|> u_{n}\}}}_{\mathbf{U}_{i}^{(2)}}\,+\,\underbrace{\left(\sum_{(i-1)\Delta_{n}\leq s\leq i\Delta_{n}}(\beta_{s-}^{J})^{\top}\Delta X_{s}\right)\mathbbm{1}_{\{|\Delta_{i}^{n}Y|\leq u_{n}\}}}_{\mathbf{J}_{i}}\\
&\qquad\quad+\,\underbrace{\Delta_{i}^{n}Z^{c}}_{\mathbf{Z}_{i}^{(1)}}\,-\,\underbrace{\Delta_{i}^{n}Z^{c}\mathbbm{1}_{\{|\Delta_{i}^{n}Y|> u_{n}\}}}_{\mathbf{Z}_{i}^{(2)}}\,+\,\underbrace{(\Delta_{i}^{n}Z-\Delta_{i}^{n}Z^{c})\mathbbm{1}_{\{|\Delta_{i}^{n}Y|\leq u_{n}\}}}_{\mathbf{Z}_{i}^{(3)}}.
\end{split}
\end{equation}
We denote the leading continuous-martingale term by $\Theta=\mathbf{U}^{(1)}+\mathbf{Z}^{(1)}$, and all remainders by $\Xi\in\{\widetilde{\mathbf{Y}}-\mathbf{R}_{A^{\circ}}\bm{\gamma}_{A^{\circ}},\,\mathbf{U}^{(2)},\,\mathbf{J},\,\mathbf{Z}^{(2)},\,\mathbf{Z}^{(3)}\}$. Then, it suffices to show
\begin{equation}
\label{Eq:spuriousLossImprove1}
\max_{q+1\leq j\leq p}\Vert(\mathbf{P}_{j}^{*})^{\top}\mathbf{M}_{A^{\circ}}\Theta\Vert^{2}\leq K\Delta_{n}(K_{n}+\log p),
\end{equation}
with probability approaching one, and, for each $\Xi$,
\begin{equation}
\label{Eq:spuriousLossImprove2}
\max_{q+1\leq j\leq p}\Vert(\mathbf{P}_{j}^{*})^{\top}\mathbf{M}_{A^{\circ}}\Xi\Vert^{2}=o_{p}(\Delta_{n}(K_{n}+\log p)).
\end{equation}

We firstly verify \cref{Eq:spuriousLossImprove1}. Since $\Theta$ contains only continuous semimartingale increments, $\Vert(\mathbf{P}_{j}^{*})^{\top}\mathbf{M}_{A^{\circ}}\Theta\Vert^{2}$ can be written into a quadratic form of (conditional) Gaussian random variables, and \cref{Eq:spuriousLossImprove1} can be proved with the Hanson-Wright inequality (\citealp{hanson1971bound}; Theorem 6.2.1, \citealp{vershynin2018high}). Specifically, we assume $b,\widetilde{b}\equiv0$ without loss of generality as the drift only contributes a negligible order of $\Delta_{n}$, and thus, conditional on a $\sigma$-field $\mathcal{G}=\sigma(c_{s},\widetilde{\sigma}_{s},e_{s}:0\leq s\leq T)$, $\Theta_{i}$ are independent Gaussian variables with mean 0 and variances
\begin{equation}
\text{Var}(\Theta_{i}|\mathcal{G})=\int_{(i-1)\Delta_{n}}^{i\Delta_{n}}e_{s}^{\top}c_{s}e_{s}ds+\int_{(i-1)\Delta_{n}}^{i\Delta_{n}}\widetilde{\sigma}_{s}^{2}ds=\bar{v}_{i}\Delta_{n},
\end{equation}
for all $i=1,\dots,n$, with $0\leq \bar{v}_{i}\leq \overline{K}$ for some $\overline{K}>0$ under \cref{As:Ito_Iocal}. Therefore, we denote the conditional covariance matrix of $\Theta$ by
\begin{equation}
\bm{\Sigma}=\text{Cov}(\Theta|\mathcal{G})=\text{diag}(\bar{v}_{1}\Delta_{n},\dots,\bar{v}_{n}\Delta_{n}),
%=\Delta_{n} \begin{pmatrix}
%\bar{v}_{1} & 0 & \dots & 0\\
%0 & \bar{v}_{2} & \dots & 0\\
%\vdots & \vdots & \ddots & \vdots\\
%0 & 0 & \dots & \bar{v}_{n}\\
%\end{pmatrix},
\end{equation}
and thus $\Theta=\mathbf{\Sigma}^{1/2}\Gamma$, where $\Gamma$ is a $n\times 1$ vector of standard Gaussian random variables and is independent of the conditioning. Then, conditional on the same $\sigma$-field, $\Vert(\mathbf{P}_{j}^{*})^{\top}\mathbf{M}_{A^{\circ}}\Theta\Vert^{2}$ can be written into a quadratic form of Gaussian random variables: With symmetric and idempotent $\mathbf{M}_{A^{\circ}}$,
\begin{equation}
\label{Eq:QuadraticGaussian}
\Vert(\mathbf{P}_{j}^{*})^{\top}\mathbf{M}_{A^{\circ}}\Theta\Vert^{2}=\Gamma^{\top}\mathbf{\Sigma}^{1/2}\mathbf{M}_{A^{\circ}}\mathbf{H}_{j}^{*}\mathbf{M}_{A^{\circ}}\mathbf{\Sigma}^{1/2}\Gamma=\Gamma^{\top}\bm{\Lambda}_{j}\Gamma,\quad\text{with}\quad\bm{\Lambda}_{j}=\mathbf{\Sigma}^{1/2}\mathbf{H}_{j}^{*}\mathbf{\Sigma}^{1/2}.
\end{equation}
By definition of the block-wise orthogonal projection matrix, $\Vert\mathbf{H}_{j}^{*}\Vert=1$ and rank$(\mathbf{H}_{j}^{*})=K_{n}$. By the submultiplicativity of operator norms,
\begin{equation}
\Vert\bm{\Lambda}_{j}\Vert\leq\Vert\bm{\Sigma}^{1/2}\Vert\Vert\mathbf{H}_{j}^{*}\Vert\Vert\bm{\Sigma}^{1/2}\Vert=\lambda_{\text{max}}(\bm{\Sigma})=\max_{1\leq i\leq n}\bar{v}_{i}\Delta_{n}\leq \overline{K}\Delta_{n},
\end{equation}
and for symmetric $\bm{\Lambda}_{j}$, with $\text{rank}(\bm{\Lambda}_{j})=\text{rank}(\mathbf{H}_{j}^{*})=K_{n}$,
\begin{equation}
\Vert\bm{\Lambda}_{j}\Vert_{F}^{2}\leq\text{rank}(\bm{\Lambda}_{j})\Vert\bm{\Lambda}_{j}\Vert^{2}\leq \overline{K}^{2}K_{n}\Delta_{n}^{2}.
\end{equation}
By the Hanson-Wright inequality, for any $t>0$ and some constant $k>0$,
\begin{equation}
\label{Eq:HW}
\mathbb{P}(\vert\Gamma^{\top}\bm{\Lambda}_{j}\Gamma-\mathbb{E}[\Gamma^{\top}\bm{\Lambda}_{j}\Gamma]\vert>t\,\vert\,\mathcal{G})\leq2\exp\left(-k\min\left\{\frac{t^{2}}{\Vert\bm{\Lambda}_{j}\Vert_{F}^{2}},\,\frac{t}{\Vert\bm{\Lambda}_{j}\Vert}\right\}\right),
\end{equation}
where
\begin{equation}
\label{Eq:HW1}
\mathbb{E}[\Gamma^{\top}\bm{\Lambda}_{j}\Gamma]=\text{tr}(\bm{\Lambda}_{j})=\text{tr}(\mathbf{\Sigma}^{1/2}\mathbf{H}_{j}^{*}\mathbf{\Sigma}^{1/2})=\text{tr}(\mathbf{H}_{j}^{*}\mathbf{\Sigma})\leq\text{tr}(\mathbf{H}_{j}^{*})\Vert\bm{\Sigma}\Vert\leq K_{n}\Vert\bm{\Sigma}\Vert\leq \overline{K}K_{n}\Delta_{n}.
\end{equation}
Let $t=\overline{K}\Delta_{n}(\sqrt{K_{n}c\log p}+c\log p)$ for some constant $c>0$ satisfying $kc>1$. Then,
\begin{align}
\frac{t^{2}}{\Vert\bm{\Lambda}_{j}\Vert_{F}^{2}}&\geq\frac{\overline{K}^{2}\Delta_{n}^{2}(\sqrt{K_{n}c\log p}+c\log p)^{2}}{\overline{K}^{2}K_{n}\Delta_{n}^{2}}\geq\frac{(\sqrt{K_{n}c\log p}+c\log p)^{2}}{K_{n}}\geq c\log p, \label{Eq:HW2}\\
\frac{t}{\Vert\bm{\Lambda}_{j}\Vert}&\geq\frac{\overline{K}\Delta_{n}(\sqrt{K_{n}c\log p}+c\log p)}{\overline{K}\Delta_{n}}\geq c\log p. \label{Eq:HW3}
\end{align}
Combining \cref{Eq:HW1,Eq:HW2,Eq:HW3} in \cref{Eq:HW} leads to
\begin{equation}
\mathbb{P}\left(\left.\Gamma^{\top}\bm{\Lambda}_{j}\Gamma>\overline{K}\Delta_{n}(K_{n}+\sqrt{K_{n}c\log p}+c\log p)\right\vert\mathcal{G}\right)\leq 2e^{-kc\log p}=2p^{-kc}.
\end{equation}
Applying a union bound over all irrelevant blocks, we obtain
\begin{equation}
\mathbb{P}\left(\left.\max_{q+1\leq j\leq p}\Gamma^{\top}\bm{\Lambda}_{j}\Gamma>\overline{K}\Delta_{n}(K_{n}+\sqrt{K_{n}c\log p}+c\log p)\right\vert\mathcal{G}\right)\leq 2p^{1-kc}\to0,
\end{equation}
when $p\to\infty$ and $1-kc<0$. Taking expectations on both sides implies the same union bound holds unconditionally. By \cref{Eq:QuadraticGaussian}, it holds that with probability approaching one,
\begin{equation}
\begin{split}
\max_{q+1\leq j\leq p}\Vert(\mathbf{P}_{j}^{*})^{\top}\mathbf{M}_{A^{\circ}}\Theta\Vert^{2}\leq\max_{q+1\leq j\leq p}\Gamma^{\top}\bm{\Lambda}_{j}\Gamma\leq\overline{K}\Delta_{n}(K_{n}+\sqrt{K_{n}c\log p}+c\log p).
\end{split}
\end{equation}
Furthermore, with $\sqrt{K_{n}\log p}\leq(K_{n}+\log p)/2$,
\begin{equation}
\max_{q+1\leq j\leq p}\Vert(\mathbf{P}_{j}^{*})^{\top}\mathbf{M}_{A^{\circ}}\Theta\Vert^{2}\leq K\Delta_{n}(K_{n}+\log p),
\end{equation}
with probability approaching one, i.e., \cref{Eq:spuriousLossImprove1} holds.

Next, we show \cref{Eq:spuriousLossImprove2} holds for each $\Xi\in\{\widetilde{\mathbf{Y}},\,\mathbf{U}^{(2)},\,\mathbf{J},\,\mathbf{Z}^{(2)},\,\mathbf{Z}^{(3)}\}$.  By definition,
\begin{equation}
\label{Eq:PMXi}
\Vert(\mathbf{P}_{j}^{*})^{\top}\mathbf{M}_{A^{\circ}}\Xi\Vert^{2}=\Xi^{\top}\mathbf{M}_{A^{\circ}}\mathbf{H}_{j}^{*}\mathbf{M}_{A^{\circ}}\Xi\leq\Vert\mathbf{M}_{A^{\circ}}\Xi\Vert^{2}\leq\Vert\Xi\Vert^{2}.
\end{equation}
Next, we show $\|\Xi\|^2=o_p(\Delta_n(K_n+\log p))$ for each $\Xi$.

%For $\Xi=\mathbf{Z}^{(2)},\mathbf{Z}^{(3)}$, we will show directly $\|\Xi\|^2=o_p(\Delta_n(K_n+\log p))$. For $\Xi=\mathbf{C},\,\mathbf{J}$, the Euclidean norm is not small, and we must exploit oracle-span cancellation: the leading parts of $\mathbf{C}$ and $\mathbf{J}$ lie in $\mathrm{span}(\mathbf{R}_{A^\circ})$ and are annihilated by $\mathbf{M}_{A^\circ}$, leaving only asymptotically negligible remainders.

\subsubsection*{(i) $\bm{\Xi=\mathbf{Z}^{(2)}}$}
\label{AP:Proof_oracle1}

By definition,
\begin{equation}
\Vert\mathbf{Z}^{(2)}\Vert^{2}=\sum_{i=1}^{n}(\Delta_{i}^{n}Z^{c})^{2}\mathbbm{1}_{\{|\Delta_{i}^{n}Y|> u_{n}\}}. 
\end{equation}
By Hölder's and BDG inequalities, as well as \cref{lemma:truncationEvents} (i), for any $q\geq 2$,
\begin{equation}
\label{Eq:Holder}
\begin{split}
\mathbb{E}\left[(\Delta_{i}^{n}Z^{c})^{2}\mathbbm{1}_{\{|\Delta_{i}^{n}Y|>u_{n}\}}\right]&\leq\left(\mathbb{E}[(\Delta_{i}^{n}Z^{c})^{2q}]\right)^{1/q}\left(\mathbb{P}(|\Delta_{i}^{n}Y|>u_{n})\right)^{1-1/q}\\
&\leq K\Delta_{n}(\Delta_{n}u_{n}^{-r})^{1-1/q}=K\Delta_{n}^{2-1/q}u_{n}^{r/q-r},
\end{split}
\end{equation}
and it is $O(\Delta_{n}u_{n}^{2-r})$ if we choose $q\geq(1-\varpi r)/(1-2\varpi)$, which is always possible for $1/2(2-r)<\varpi<1/2$ and $0\leq r<1$. By Markov's inequality,
\begin{equation}
\label{Eq:Z^2}
\Vert\mathbf{Z}^{(2)}\Vert^{2}=O_{p}(u_{n}^{2-r})=o_{p}(\Delta_{n}(K_{n}+\log p)),
\end{equation}
when $\Delta_{n}^{1-\varpi(2-r)}(K_{n}+\log p)\to\infty$ as required.

%We bound its expectation by H{\"o}lder's inequality. For any $q>2$,
%\begin{equation}
%\mathbb{E}\left[(\Delta_{i}^{n}Z^{c})^{2}\mathbbm{1}_{\{|\Delta_{i}^{n}Y|> u_{n}\}}\right]\leq\left(\mathbb{E}[(\Delta_{i}^{n}Z^{c})^{q}]\right)^{2/q}\left(\mathbb{P}(|\Delta_{i}^{n}Y|> u_{n})\right)^{1-2/q}\leq K\Delta_{n}^{2-2/q}u_{n}^{-r(1-2/q)},
%\end{equation}
%by the BDG inequality and \cref{lemma:truncationEvents} (i). By Markov's inequality,
%\begin{equation}
%\Vert\mathbf{Z}^{(2)}\Vert^{2}=O_{p}\left(\Delta_{n}^{1-2/q}u_{n}^{-r(1-2/q)}\right).
%\end{equation}
%Therefore, it holds that $\Vert\mathbf{Z}^{(2)}\Vert^{2}=o_{p}(\Delta_{n}(K_{n}+\log p))$ if
%\begin{equation}
%K_{n}+\log p\gg \Delta_{n}^{-2/q}u_{n}^{-r(1-2/q)}=\Delta_{n}^{-\varpi r-\frac{2}{q}(1-\varpi r)}.
%\end{equation}
%With $\varpi r<1/2$, choosing $q$ sufficiently large makes the above condition arbitrarily close to $(K_n+\log p)\Delta_n^{\varpi r}\to\infty$, which is implied by $K_n\sqrt{\Delta_n}/\log K_n\to\infty$.

\subsubsection*{(ii) $\bm{\Xi=\mathbf{Z}^{(3)}}$}
\label{AP:Proof_oracle2}

By definition,
\begin{equation}
\label{Eq:Z^3}
\Vert\mathbf{Z}^{(3)}\Vert^{2}=\sum_{i=1}^{n}(\Delta_{i}^{n}Z-\Delta_{i}^{n}Z^{c})^{2}\left(\mathbbm{1}_{\{|\Delta_{i}^{n}Z|\leq u_{n}\}}+\left\vert\mathbbm{1}_{\{|\Delta_{i}^{n}Y|\leq u_{n}\}}-\mathbbm{1}_{\{|\Delta_{i}^{n}Z|\leq u_{n}\}}\right\vert\right),
%\begin{split}
%\Vert\mathbf{Z}^{(3)}\Vert^{2}&=\sum_{i=1}^{n}(\Delta_{i}^{n}Z-\Delta_{i}^{n}Z^{c})^{2}\mathbbm{1}_{\{|\Delta_{i}^{n}Y|\leq u_{n}\}}\\
%&=\sum_{i=1}^{n}(\Delta_{i}^{n}Z-\Delta_{i}^{n}Z^{c})^{2}\mathbbm{1}_{\{|\Delta_{i}^{n}Z|\leq u_{n}\}}\\
%&\qquad+\sum_{i=1}^{n}(\Delta_{i}^{n}Z-\Delta_{i}^{n}Z^{c})^{2}\left\vert\mathbbm{1}_{\{|\Delta_{i}^{n}Y|\leq u_{n}\}}-\mathbbm{1}_{\{|\Delta_{i}^{n}Z|\leq u_{n}\}}\right\vert,
%\end{split}
\end{equation}
where, by \cref{lemma:truncationEvents} (ii) and Markov's inequality,
\begin{equation}
\label{Eq:Z^3_1}
\sum_{i=1}^{n}(\Delta_{i}^{n}Z-\Delta_{i}^{n}Z^{c})^{2}\mathbbm{1}_{\{|\Delta_{i}^{n}Z|\leq u_{n}\}}=O_{p}(u_{n}^{2-r}),
\end{equation}
and, following the same steps as in Appendix \ref{AP:Proof_CLT} (\hyperref[AP:Proof_CLT_v]{v}) for truncation mismatch, 
\begin{equation}
\label{Eq:Z^3_2}
\sum_{i=1}^{n}(\Delta_{i}^{n}Z-\Delta_{i}^{n}Z^{c})^{2}\left\vert\mathbbm{1}_{\{|\Delta_{i}^{n}Y|\leq u_{n}\}}-\mathbbm{1}_{\{|\Delta_{i}^{n}Z|\leq u_{n}\}}\right\vert=O_{p}(u_{n}^{2-r}\vee\Delta_{n}u_{n}^{-2r}).
\end{equation}
Specifically, we denote $X^{\circ}=(X_{1,t},\dots,X_{q,t})^{\top}_{t\geq0}$, and follow \cref{Eq:E_H_5-0}:
\begin{equation}
\label{Eq:2sum}
\begin{split}
&\sum_{i=1}^{n}(\Delta_{i}^{n}Z-\Delta_{i}^{n}Z^{c})^{2}\left\vert\mathbbm{1}_{\{|\Delta_{i}^{n}Y|\leq u_{n}\}}-\mathbbm{1}_{\{|\Delta_{i}^{n}Z|\leq u_{n}\}}\right\vert\\
&\quad\leq\sum_{i=1}^{n}(u_{n}-\Delta_{i}^{n}Z^{c})^{2}\mathbbm{1}_{\{|\Delta_{i}^{n}Y|>u_{n}\}}+\sum_{i=1}^{n}(\Delta_{i}^{n}Z-\Delta_{i}^{n}Z^{c})^{2}\mathbbm{1}_{\{|\Delta_{i}^{n}Y|\leq u_{n},\,|\Delta_{i}^{n}Z|>u_{n},\,\Vert\Delta_{i}^{n}X^{\circ}\Vert\leq u_{n}\}}.
\end{split}
\end{equation}
Note that $X^{\circ}$ has only a finite dimension of nonzero components (relevant regressors), so the fixed dimensional result in Appendix \ref{AP:Proof_CLT} (\hyperref[AP:Proof_CLT_v]{v}) still holds. As in the previous proofs, the OLS regression in \cref{Eq:objective} already incorporates the truncation indicator $\mathbbm{1}_{\{\Vert\Delta_{i}^{n}X^{\circ}\Vert\leq u_{n}\}}$, so it does no harm to include it in any $\Xi$. For the first sum in \cref{Eq:2sum},
\begin{equation}
\sum_{i=1}^{n}(u_{n}-\Delta_{i}^{n}Z^{c})^{2}\mathbbm{1}_{\{|\Delta_{i}^{n}Y|>u_{n}\}}\leq 2\sum_{i=1}^{n}u_{n}^{2}\mathbbm{1}_{\{|\Delta_{i}^{n}Y|>u_{n}\}}+2\underbrace{\sum_{i=1}^{n}(\Delta_{i}^{n}Z^{c})^{2}\mathbbm{1}_{\{|\Delta_{i}^{n}Y|>u_{n}\}}}_{\Vert\mathbf{Z}^{(2)}\Vert^{2}}.
\end{equation}
By \cref{lemma:truncationEvents} (i) and Markov's inequality, 
\begin{equation}
\sum_{i=1}^{n}u_{n}^{2}\mathbb{P}(|\Delta_{i}^{n}Y|>u_{n})=O(u_{n}^{2-r})\quad\Rightarrow\quad\sum_{i=1}^{n}u_{n}^{2}\mathbbm{1}_{\{|\Delta_{i}^{n}Y|>u_{n}\}}=O_{p}(u_{n}^{2-r}).
\end{equation}
%By Hölder's and BDG inequalities, as well as \cref{lemma:truncationEvents} (i), for any $q\geq 2$,
%\begin{equation}
%\label{Eq:Holder}
%\begin{split}
%\mathbb{E}\left[(\Delta_{i}^{n}Z^{c})^{2}\mathbbm{1}_{\{|\Delta_{i}^{n}Y|>u_{n}\}}\right]&\leq\left(\mathbb{E}[(\Delta_{i}^{n}Z^{c})^{2q}]\right)^{1/q}\left(\mathbb{P}(|\Delta_{i}^{n}Y|>u_{n})\right)^{1-1/q}\\
%&\leq K\Delta_{n}(\Delta_{n}u_{n}^{-r})^{1-1/q}=K\Delta_{n}^{2-1/q}u_{n}^{r/q-r},
%\end{split}
%\end{equation}
%and it is $O(\Delta_{n}u_{n}^{2-r})$ if we choose $q\geq(1-\varpi r)/(1-2\varpi)$, which is always possible for $1/2(2-r)<\varpi<1/2$ and $0\leq r<1$. By Markov's inequality,
%\begin{equation}
%\sum_{i=1}^{n}(\Delta_{i}^{n}Z^{c})^{2}\mathbbm{1}_{\{|\Delta_{i}^{n}Y|>u_{n}\}}=O_{p}(u_{n}^{2-r}),
%\end{equation}
and, with the result for $\Vert\mathbf{Z}^{(2)}\Vert^{2}$ in \cref{Eq:Z^2}, the first sum in \cref{Eq:2sum} is $O_{p}(u_{n}^{2-r})$. 

For the second sum, we follow the same steps as from \cref{Eq:unmatch_start,Eq:unmatch_inequality,Eq:unmatch_eventInequality,Eq:7terms}. On $\{|\Delta_{i}^{n}Y|\leq u_{n},\,|\Delta_{i}^{n}Z|>u_{n},\,\Vert\Delta_{i}^{n}X^{\circ}\Vert\leq u_{n}\}$, it holds that, by \cref{Eq:unmatch_inequality},
\begin{equation}
u_{n}^{2}\leq |\Delta_{i}^{n}Z|^{2}\leq (u_{n} + |\Delta_{i}^{n}Y^{(1)}| + |\widetilde{\mathbf{J}}_{i}|)^{2}\leq 3(u_{n}^{2}+|\Delta_{i}^{n}Y^{(1)}|^{2}+|\widetilde{\mathbf{J}}_{i}|^{2}),
\end{equation}
and thus
\begin{equation}
|\Delta_{i}^{n}Z-\Delta_{i}^{n}Z^{c}|^{2}\leq K(u_{n}^{2}+|\Delta_{i}^{n}Y^{(1)}|^{2}+|\widetilde{\mathbf{J}}_{i}|^{2}+|\Delta_{i}^{n}Z^{c}|^{2}). 
\end{equation}
By the same event decomposition as in \cref{Eq:unmatch_eventInequality},
\begin{align}
&\mathbb{E}\left[|\Delta_{i}^{n}Z-\Delta_{i}^{n}Z^{c}|^{2}\mathbbm{1}_{\{|\Delta_{i}^{n}Y|\leq u_{n},\,|\Delta_{i}^{n}Z|>u_{n},\,\Vert\Delta_{i}^{n}X^{\circ}\Vert\leq u_{n}\}}\right] \notag\\
&\qquad\leq\mathbb{E}\left[|\Delta_{i}^{n}Z-\Delta_{i}^{n}Z^{c}|^{2}\mathbbm{1}_{\{u_{n}<|\Delta_{i}^{n}Z|\leq 3u_{n}\}}\right]+\mathbb{E}\left[|\Delta_{i}^{n}Z-\Delta_{i}^{n}Z^{c}|^{2}\mathbbm{1}_{\{|\Delta_{i}^{n}Y^{(1)}| + |\widetilde{\mathbf{J}}_{i}|>2u_{n}\}}\right] \notag\\
&\qquad\leq 2\underbrace{(3u_{n})^{2}\mathbb{P}(|\Delta_{i}^{n}Z|>u_{n})}_{(1)}\,+\,2\underbrace{\mathbb{E}\left[|\Delta_{i}^{n}Z^{c}|\mathbbm{1}_{\{|\Delta_{i}^{n}Z|>u_{n}\}}\right]}_{(2)}\notag\\
&\qquad\qquad+\mathbb{E}\left[|\Delta_{i}^{n}Z-\Delta_{i}^{n}Z^{c}|^{2}\mathbbm{1}_{\{|\Delta_{i}^{n}Y^{(1)}|> u_{n}\}}\right]+\mathbb{E}\left[|\Delta_{i}^{n}Z-\Delta_{i}^{n}Z^{c}|^{2}\mathbbm{1}_{\{|\widetilde{\mathbf{J}}_{i}|>u_{n}\}}\right]\notag\\
&\qquad \leq(1)+(2)\,+\, K\underbrace{u_{n}^{2}\mathbb{P}\left(|\Delta_{i}^{n}Y^{(1)}|>u_{n}\right)}_{(3)}\,+\,K\underbrace{\mathbb{E}\left[|\Delta_{i}^{n}Y^{(1)}|^{2}\mathbbm{1}_{\{|\Delta_{i}^{n}Y^{(1)}|> u_{n}\}}\right]}_{(4)} \notag\\
&\qquad\qquad+K\underbrace{\mathbb{E}\left[|\widetilde{\mathbf{J}}_{i}|^{2}\mathbbm{1}_{\{|\Delta_{i}^{n}Y^{(1)}|>u_{n}\}}\right]}_{(5)}\,+\,K\underbrace{\mathbb{E}\left[|\Delta_{i}^{n}Z^{c}|^{2}\mathbbm{1}_{\{|\Delta_{i}^{n}Y^{(1)}|>u_{n}\}}\right]}_{(6)} + K\underbrace{u_{n}^{2}\mathbb{P}\left(|\widetilde{\mathbf{J}}_{i}|>u_{n}\right)}_{(7)}\notag\\
&\qquad\qquad+K\underbrace{\mathbb{E}\left[|\Delta_{i}^{n}Y^{(1)}|^{2}\mathbbm{1}_{\{|\widetilde{\mathbf{J}}_{i}|> u_{n}\}}\right]}_{(8)}\,+\,K\underbrace{\mathbb{E}\left[|\widetilde{\mathbf{J}}_{i}|^{2}\mathbbm{1}_{\{|\widetilde{\mathbf{J}}_{i}|>u_{n}\}}\right]}_{(9)}\,+\,K\underbrace{\mathbb{E}\left[|\Delta_{i}^{n}Z^{c}|^{2}\mathbbm{1}_{\{|\widetilde{\mathbf{J}}_{i}|>u_{n}\}}\right]}_{(10)}. \label{Eq:10terms}
\end{align}
Next, we verify that the above terms (1)--(10) are all $O(\Delta_{n}u_{n}^{2-r}\vee\Delta_{n}^{2}u_{n}^{-2r})$. 

Term (1): By the same result as in \cref{lemma:truncationEvents} (i) for the one-dimensional It{\^o} semimartingale $Z$,
\begin{equation}
u_{n}^{2}\mathbb{P}(|\Delta_{i}^{n}Z|>u_{n})\leq K\Delta_{n}u_{n}^{2-r}.
\end{equation}

Term (2): Following the same steps as in \cref{Eq:Holder},
\begin{equation}
\mathbb{E}\left[(\Delta_{i}^{n}Z^{c})^{2}\mathbbm{1}_{\{|\Delta_{i}^{n}Z|>u_{n}\}}\right]=O(\Delta_{n}u_{n}^{2-r}).
\end{equation}

%By Hölder's and BDG inequalities, as well as \cref{lemma:truncationEvents} (i), for any $q\geq 2$,
%\begin{equation}
%\begin{split}
%\mathbb{E}\left[(\Delta_{i}^{n}Z^{c})^{2}\mathbbm{1}_{\{|\Delta_{i}^{n}Z|>u_{n}\}}\right]&\leq\left(\mathbb{E}[(\Delta_{i}^{n}Z^{c})^{2q}]\right)^{1/q}\left(\mathbb{P}(|\Delta_{i}^{n}Z|>u_{n})\right)^{1-1/q}\\
%&\leq K\Delta_{n}(\Delta_{n}u_{n}^{-r})^{1-1/q}=K\Delta_{n}^{2-1/q}u_{n}^{r/q-r},
%\end{split}
%\end{equation}
%and it is $O(\Delta_{n}u_{n}^{2-r})$ if we choose $q\geq(1-\varpi r)/(1-2\varpi)$, which is always possible for $1/2(2-r)<\varpi<1/2$ and $0\leq r<1$. 

Term (3): By \cref{Eq:Y_triangle_Prob1},
\begin{equation}
u_{n}^{2}\mathbb{P}\left(|\Delta_{i}^{n}Y^{(1)}|>u_{n}\right)\leq K\Delta_{n}u_{n}^{2-r}.
\end{equation}

%Term (3): By \cref{Eq:Y_triangle_Prob1}, the BDG and Cauchy-Schwarz inequalities,
%\begin{equation}
%\begin{split}
%\mathbb{E}\left[|\Delta_{i}^{n}Y^{(1)}|^{2}\mathbbm{1}_{\{|\Delta_{i}^{n}Y^{(1)}|> u_{n}\}}\right]&\leq\left(\mathbb{E}[|\Delta_{i}^{n}Y^{(1)}|^{4}]\right)^{1/2}\left(\mathbb{P}(|\Delta_{i}^{n}Y^{(1)}|> u_{n})\right)^{1/2}\\
%&\leq K\Delta_{n}^{3/2}u_{n}^{1/2-r/2}=o(\Delta_{n}u_{n}^{2-r}).
%\end{split}
%\end{equation}

Term (4): By the same steps as in \cref{Eq:moment_Markov} and the BDG inequality, for any $q\geq2$,
\begin{equation}
\begin{split}
\mathbb{E}\left[|\Delta_{i}^{n}Y^{(1)}|^{2}\mathbbm{1}_{\{|\Delta_{i}^{n}Y^{(1)}|> u_{n}\}}\right]\leq\mathbb{E}\left[|\Delta_{i}^{n}Y^{(1)}|^{2}\left(\frac{|\Delta_{i}^{n}Y^{(1)}|}{u_{n}}\right)^{q-2}\right]\leq \frac{\mathbb{E}[|\Delta_{i}^{n}Y^{(1)}|^{q}]}{u_{n}^{q-2}}\leq K\Delta_{n}^{q/2}u_{n}^{2-q},
\end{split}
\end{equation}
and it is $O(\Delta_{n}u_{n}^{2-r})$ if $q\geq 2(1-\varpi r)/(1-2\varpi)$, which is always possible for $1/2(2-r)<\varpi<1/2$ and $0\leq r<1$. 

Term (5): By \cref{Eq:J_i_moment},
\begin{equation}
\mathbb{E}\left[|\widetilde{\mathbf{J}}_{i}|^{2}\mathbbm{1}_{\{|\Delta_{i}^{n}Y^{(1)}|> u_{n}\}}\right]\leq\mathbb{E}[|\widetilde{\mathbf{J}}_{i}|^{2}]\leq K\Delta_{n}u_{n}^{2-r}+K'\Delta_{n}^{2}u_{n}^{-2r}.
\end{equation}

Term (6): Following the same steps as in \cref{Eq:Holder}, together with \cref{Eq:Y_triangle_Prob1}, 
\begin{equation}
\mathbb{E}\left[|\Delta_{i}^{n}Z^{c}|^{2}\mathbbm{1}_{\{|\Delta_{i}^{n}Y^{(1)}|>u_{n}\}}\right]=O(\Delta_{n}u_{n}^{2-r}).
\end{equation}

Term (7): By \cref{Eq:E_J_i_tilde},
\begin{equation}
u_{n}^{2}\mathbb{P}\left(|\widetilde{\mathbf{J}}_{i}|>u_{n}\right)\leq K\Delta_{n}u_{n}^{2-r}.
\end{equation}

Term (8): By Hölder's and BDG inequalities, together with \cref{Eq:E_J_i_tilde}, for some $q\geq 2$,
\begin{equation}
\begin{split}
\mathbb{E}\left[|\Delta_{i}^{n}Y^{(1)}|^{2}\mathbbm{1}_{\{|\widetilde{\mathbf{J}}_{i}|> u_{n}\}}\right]&\leq\left(\mathbb{E}[|\Delta_{i}^{n}Y^{(1)}|^{2q}]\right)^{1/q}\left(\mathbb{P}\left(|\widetilde{\mathbf{J}}_{i}|>u_{n}\right)\right)^{1-1/q}\\
&\leq K\Delta_{n}(\Delta_{n}u_{n}^{-r})^{1-1/q}=K\Delta_{n}^{2-1/q}u_{n}^{r/q-r},
\end{split}
\end{equation}
and it is $O(\Delta_{n}u_{n}^{2-r})$ if we choose $q\geq(1-\varpi r)/(1-2\varpi)$, which is always possible for $1/2(2-r)<\varpi<1/2$ and $0\leq r<1$. 

Term (9): Same as term (5), we have
\begin{equation}
\mathbb{E}\left[|\widetilde{\mathbf{J}}_{i}|^{2}\mathbbm{1}_{\{|\widetilde{\mathbf{J}}_{i}|> u_{n}\}}\right]\leq\mathbb{E}[|\widetilde{\mathbf{J}}_{i}|^{2}]\leq K\Delta_{n}u_{n}^{2-r}+K'\Delta_{n}^{2}u_{n}^{-2r}.
\end{equation}

Term (10): Following the same steps as in \cref{Eq:Holder}, together with \cref{Eq:E_J_i_tilde},
\begin{equation}
\mathbb{E}\left[|\Delta_{i}^{n}Z^{c}|^{2}\mathbbm{1}_{\{|\widetilde{\mathbf{J}}_{i}|>u_{n}\}}\right]=O(\Delta_{n}u_{n}^{2-r}).
\end{equation}

Combining all terms (1)--(10) in \cref{Eq:10terms} leads to
\begin{equation}
\mathbb{E}\left[|\Delta_{i}^{n}Z-\Delta_{i}^{n}Z^{c}|^{2}\mathbbm{1}_{\{|\Delta_{i}^{n}Y|\leq u_{n},\,|\Delta_{i}^{n}Z|>u_{n},\,\Vert\Delta_{i}^{n}X^{\circ}\Vert\leq u_{n}\}}\right]=O(\Delta_{n}u_{n}^{2-r}\vee\Delta_{n}^{2}u_{n}^{-2r}),
\end{equation}
and thus \cref{Eq:Z^3_2} holds. Combining \cref{Eq:Z^3_1,Eq:Z^3_2} in Eqs.~(\ref{Eq:Z^3}) and (\ref{Eq:PMXi}), we obtain
\begin{equation}
\Vert(\mathbf{P}_{j}^{*})^{\top}\mathbf{M}_{A^{\circ}}\mathbf{Z}^{(3)}\Vert^{2}=O_{p}(u_{n}^{2-r}\vee\Delta_{n}u_{n}^{-2r})=o_{p}(\Delta_{n}(K_{n}+\log p)),
\end{equation}
when $\Delta_{n}^{(1-\varpi(2-r))\vee2\varpi r}(K_{n}+\log p)\to\infty$ as required.

\subsubsection*{(iii) $\bm{\Xi=\widetilde{\mathbf{Y}}-\mathbf{R}_{A^{\circ}}\bm{\gamma}_{A^{\circ}}}$}

Consider the same decomposition as in Appendix \ref{AP:consistency} (\hyperref[AP:Proof_consistency_A1]{i}), i.e., \cref{Eq:Y_decomp_1,Eq:Y_decomp_2}, 
\begin{equation}
\label{Eq:Y}
\widetilde{\mathbf{Y}}-\mathbf{R}_{A^{\circ}}\bm{\gamma}_{A^{\circ}}=(\widetilde{\mathbf{Y}}^{(1)}-\mathbf{R}_{A^{\circ}}\bm{\gamma}_{A^{\circ}})+\widetilde{\mathbf{Y}}^{(2)}.
\end{equation}
Since the oracle index set $A^{\circ}$ has fixed cardinality $q$, the arguments used in the proofs of \cref{Th:consistency,Th:CLT} apply verbatim to the restricted oracle regression. By \cref{Eq:A1_Eq1,Eq:A1_Eq2,Eq:A1_Eq3} and Markov's inequality, 
\begin{equation}
\label{Eq:Y1}
\Vert\widetilde{\mathbf{Y}}^{(1)}-\mathbf{R}_{A^{\circ}}\bm{\gamma}_{A^{\circ}}\Vert^{2}=O_{p}(\Delta_{n}u_{n}^{-r}\vee u_{n}^{2-r})=O_{p}(u_{n}^{2-r}).
\end{equation}
and, by \cref{Eq:Discretization-3},
\begin{equation}
\label{Eq:Y2}
\Vert\widetilde{\mathbf{Y}}^{(2)}\Vert^{2}=O_{p}(\Delta_{n}^{2}K_{n}^{2}).
\end{equation}
Combining \cref{Eq:Y1,Eq:Y2} in \cref{Eq:Y} leads to
\begin{equation}
\Vert\widetilde{\mathbf{Y}}-\mathbf{R}_{A^{\circ}}\bm{\gamma}_{A^{\circ}}\Vert^{2}=O_{p}(u_{n}^{2-r}\vee \Delta_{n}^{2}K_{n}^{2})=o_{p}(\Delta_{n}(K_{n}+\log p)),
\end{equation}
when $\Delta_{n}^{1-\varpi(2-r)}(K_{n}+\log p)\to\infty$, and since $\Delta_{n}K_{n}\to0$,
\begin{equation}
\Delta_{n}^{2}K_{n}^{2}=o(\Delta_{n}K_{n})=o(\Delta_{n}(K_{n}+\log p)).
\end{equation}

\subsubsection*{(iv) $\bm{\Xi=\mathbf{U}^{(2)}}$}

By Hölder's and BDG inequalities, as well as \cref{lemma:truncationEvents} (i), for any $q\geq 2$,
\begin{equation}
\label{Eq:C2}
\begin{split}
\mathbb{E}[\mathbf{U}_{i}^{2}]&=\mathbb{E}\left[\left(\int_{(i-1)\Delta_{n}}^{i\Delta_{n}}e_{s}^{\top}dX_{s}^{c}\right)^{2}\mathbbm{1}_{\{|\Delta_{i}^{n}Y|>u_{n}\}}\right]\\
&\leq\left(\mathbb{E}\left[\left(\int_{(i-1)\Delta_{n}}^{i\Delta_{n}}e_{s}^{\top}dX_{s}^{c}\right)^{2q}\right]\right)^{1/q}\left(\mathbb{P}(|\Delta_{i}^{n}Y|>u_{n})\right)^{1-1/q}\\
&\leq K\Delta_{n}(\Delta_{n}u_{n}^{-r})^{1-1/q}=K\Delta_{n}^{2-1/q}u_{n}^{r/q-r},
\end{split}
\end{equation}
and it is $O(\Delta_{n}u_{n}^{2-r})$ if we choose $q\geq(1-\varpi r)/(1-2\varpi)$, which is always possible for $1/2(2-r)<\varpi<1/2$ and $0\leq r<1$. Therefore,
\begin{equation}
\label{Eq:C3}
\Vert\mathbf{U}\Vert^{2}=O_{p}(u_{n}^{2-r})=o_{p}(\Delta_{n}(K_{n}+\log p)),
\end{equation}
when $\Delta_{n}^{1-\varpi(2-r)}(K_{n}+\log p)\to\infty$.

\subsubsection*{(v) $\bm{\Xi=\mathbf{J}}$}
\label{AP:Proof_oracle5}

Same as in (\hyperref[AP:Proof_oracle2]{ii}), we employ the safe replacement $\widetilde{\mathbf{J}}$ with the additional truncation indicator $\mathbbm{1}_{\{\Vert\Delta_{i}^{n}X^{\circ}\Vert\leq u_{n}\}}$. By \cref{Eq:J} and Markov's inequality, 
\begin{equation}
\Vert\widetilde{\mathbf{J}}\Vert^{2}=O_{p}(u_{n}^{2-r}\vee\Delta_{n}u_{n}^{-2r})=o_{p}(\Delta_{n}(K_{n}+\log p)),
\end{equation}
when $\Delta_{n}^{(1-\varpi(2-r))\vee2\varpi r}(K_{n}+\log p)\to\infty$.

Combining results in (\hyperref[AP:Proof_oracle1]{i})--(\hyperref[AP:Proof_oracle5]{v}) in \cref{Eq:PMXi} completes the proof of \cref{lemma:overfitOneBlock}.

%Notes for the unfinished proof: When $Z$ has no jumps, $\widehat{\bm{\varepsilon}}^{\top}\mathbf{H}_{j}^{*}\widehat{\bm{\varepsilon}}$ can be written into a quadratic form of Gaussian random variables, and \cref{lemma:overfitOneBlock} can be proved with the Hanson-Wright inequality (\citealp{hanson1971bound}; Theorem 6.2.1, \citealp{vershynin2018high}). For the general case with jumps in $Z$, it suffices to show the additional terms from the decomposition are negligible. \textcolor{purple}{(to be finished)}

\end{proof}

For any index set $A$, we define the total loss reduction from adding all blocks in $A\backslash A^{\circ}$ as
\begin{equation}
\Delta Q_{n}(A)=Q_{n}(\widehat{\bm{\gamma}}^{\circ})-\min_{\bm{\gamma}: A(\bm{\gamma})=A}Q_{n}(\bm{\gamma}). 
\end{equation}

\begin{lemma}
\label{lemma:extraSSR}
Under the conditions of \cref{Th:oracle_global}, it holds that for any $A\in\Gamma_{1}$,
\begin{equation}
\Delta Q_{n}(A)\leq(|A|-q)\max_{q+1\leq j\leq p}\Delta Q_{n,j}
\end{equation}
\end{lemma}

\begin{proof}[Proof of \cref{lemma:extraSSR}]
We denote by $Q_{n}(A)=\min_{\bm{\gamma}: A(\bm{\gamma})=A}Q_{n}(\bm{\gamma})$, where $Q_{n}(A^{\circ})=Q_{n}(\widehat{\bm{\gamma}}^{\circ})$. Fix $A\in\Gamma_{1}$, and enumerate the redundant groups as $A\backslash A^{\circ}=\{j_{1},\dots,j_{|A|-q}\}$. Define the nested sequence of index sets 
\begin{equation}
\overline{A}_{0}=A^{\circ},\qquad
\overline{A}_{k}=A^{\circ}\cup\{j_{1},\dots,j_{k}\},\qquad\text{where }k=1,\dots,|A|-q.
\end{equation}
The total improvement from $A^{\circ}$ to $A$ can be written as a sum of incremental improvements:
\begin{equation}
\Delta Q_{n}(A)=\mathbf{Y}^{\top}(\mathbf{M}_{A^{\circ}}-\mathbf{M}_{A})\mathbf{Y}=\sum_{k=1}^{|A|-q}\mathbf{Y}^{\top}(\mathbf{M}_{\overline{A}_{k-1}}-\mathbf{M}_{\overline{A}_{k}})\mathbf{Y}.
\end{equation}
We define the residualized block $\breve{\mathbf{R}}_{j_{k}}^{*}=\mathbf{M}_{\overline{A}_{k-1}}\mathbf{R}_{j_{k}}^{*}\in\mathbb{R}^{n\times K_{n}}$, the orthogonal projection matrix $\breve{\mathbf{H}}_{j_{k}}^{*}=\breve{\mathbf{R}}_{j_{k}}^{*}((\breve{\mathbf{R}}_{j_{k}}^{*})^{\top}\breve{\mathbf{R}}_{j_{k}}^{*})^{-1}(\breve{\mathbf{R}}_{j_{k}}^{*})^{\top}\in\mathbb{R}^{n\times n}$. We follow the same steps as in \cref{Eq:FWL,Eq:FWL2}, by the Frisch-Waugh-Lovell theorem,
\begin{equation}
\label{Eq:FWL_A}
\mathbf{Y}^{\top}(\mathbf{M}_{\overline{A}_{k-1}}-\mathbf{M}_{\overline{A}_{k}})\mathbf{Y}=(\mathbf{M}_{\overline{A}_{k-1}}\mathbf{Y})^{\top}\breve{\mathbf{H}}_{j_{k}}^{*}(\mathbf{M}_{\overline{A}_{k-1}}\mathbf{Y})=\mathbf{Y}^{\top}\breve{\mathbf{H}}_{j_{k}}^{*}\mathbf{Y}.
\end{equation}
Since $A^{\circ}\subseteq\overline{A}_{k-1}$ for all $k=1,\dots,|A|-q$, $\text{span}(\breve{\mathbf{R}}_{j_{k}}^{*})=\text{span}(\mathbf{M}_{\overline{A}_{k-1}}\mathbf{R}_{j_{k}}^{*})\subseteq\text{span}(\mathbf{M}_{A^{\circ}}\mathbf{R}_{j_{k}}^{*})$, and thus the projector $\breve{\mathbf{H}}_{j_{k}}^{*}$ is onto a smaller subspace than $\mathbf{H}_{j_{k}}^{*}$, so that
\begin{equation}
\mathbf{Y}^{\top}\breve{\mathbf{H}}_{j_{k}}^{*}\mathbf{Y}\leq\mathbf{Y}^{\top}\mathbf{H}_{j_{k}}^{*}\mathbf{Y},\qquad\text{where}\quad
\mathbf{H}_{j_{k}}^{*}=(\mathbf{M}_{A^{\circ}}\mathbf{R}_{j_{k}}^{*})((\mathbf{M}_{A^{\circ}}\mathbf{R}_{j_{k}}^{*})^{\top}\mathbf{M}_{A^{\circ}}\mathbf{R}_{j_{k}}^{*})^{-1}\mathbf{M}_{A^{\circ}}\mathbf{R}_{j_{k}}^{*}.
\end{equation}
Therefore, by the definition of $\Delta Q_{n,j}$ in \cref{Eq:Delta_Q_j} and the same steps as in \cref{Eq:FWL_A},
\begin{equation}
\Delta Q_{n}(A)\leq\sum_{k=1}^{|A|-q}\mathbf{Y}^{\top}\mathbf{H}_{j_{k}}^{*}\mathbf{Y}\leq(|A|-q)\max_{q+1\leq j\leq p}\Delta Q_{n,j}.
\end{equation}
This completes the proof of \cref{lemma:extraSSR}. 

%...
%
%...
%
%we define the residualized block as $\widetilde{\mathbf{R}}_{j}^{*}=\mathbf{M}_{A^{\circ}}\mathbf{R}_{j}^{*}\in\mathbb{R}^{n\times K_{n}}$, the projection matrix $\mathbf{H}_{j}^{*}=\widetilde{\mathbf{R}}_{j}^{*}((\widetilde{\mathbf{R}}_{j}^{*})^{\top}\widetilde{\mathbf{R}}_{j}^{*})^{-1}(\widetilde{\mathbf{R}}_{j}^{*})^{\top}\in\mathbb{R}^{n\times n}$ onto span$(\widetilde{\mathbf{R}}_{j}^{*})$, and the corresponding residualizer $\mathbf{M}_{j}^{*}=\mathbf{I}_{n}-\mathbf{H}_{j}^{*}$. 
%
%
%\begin{equation}
%\Delta Q_{n}(A)=Q_{n}(\overline{A}_{0})-Q_{n}(\overline{A}_{m})=\sum_{k=1}^{m}(Q_{n}(\overline{A}_{k-1})-Q_{n}(\overline{A}_{k})). 
%\end{equation}
%We claim that for each $k=1,\dots,m$ and each $j\notin \overline{A}_{k-1}$,
%\begin{equation}
%\label{Eq:marginalDecrease}
%Q_{n}(S_{k-1})-Q_{n}(S_{k-1}\cup\{j\})\leq Q_{n}(S_{0})-Q_{n}(S_{0}\cup\{j\})=\Delta Q_{n,j}.
%\end{equation}
%Adding a given block $j$ can only improve the fit less when we start from a larger model. %To verify \cref{Eq:marginalDecrease} formally,  ...  \textcolor{purple}{(with projection matrices, to be finished)}
%
%Therefore, we have
%\begin{equation}
%\Delta Q_{n}(A)=\sum_{k=1}^{m}(Q_{n}(S_{k-1})-Q_{n}(S_{k}))\leq(|A|-q)\max_{q+1\leq j\leq p}\Delta Q_{n,j}.
%\end{equation}
%This completes the proof of \cref{lemma:extraSSR}. 
\end{proof}

By \cref{lemma:overfitOneBlock,lemma:extraSSR}, it holds for any $A\in\Gamma_{1}$ on the event $E_{n}$ in \cref{Eq:E_n} that
\begin{equation}
\label{Eq:overfit_SSR_improvement}
\Delta Q_{n}(A)\leq K(|A|-q)\Delta_{n}(K_{n}+\log p).
\end{equation}
By \cref{Eq:DeltaQstar,Eq:overfit_SSR_improvement}, on $E_{n}$,
\begin{equation}
\min_{\bm{\gamma}: A(\bm{\gamma})=A}Q_{n}^{*}(\bm{\gamma})-Q_{n}^{*}(\widehat{\bm{\gamma}}^{\circ})%\geq\frac{1}{2}K(|A|-q)\Delta_{n}(K_{n}+\log p)+\lambda_{n}(|A|-q)
\geq (|A|-q)\left(-\frac{K}{2}\Delta_{n}(K_{n}+\log p)+\lambda_{n}\right)>0,
\end{equation}
since $|A|-q\geq1$ and $\lambda_{n}\gg \Delta_{n}(K_{n}+\log p)$ as required. Consequently, by \cref{lemma:overfitOneBlock},
\begin{equation}
\label{Eq:Gamma1}
\mathbb{P}\left(\exists A\in\Gamma_{1}:\min_{\bm{\gamma}: A(\bm{\gamma})=A}Q_{n}^{*}(\bm{\gamma})-Q_{n}^{*}(\widehat{\bm{\gamma}}^{\circ})\leq 0\right)\leq\mathbb{P}(E_{n}^{\complement})\to0.
\end{equation}

\subsubsection*{Underfitting models}

For any underfitting models with $A\in\Gamma_{2}$, the difference between unpenalized losses in \cref{Eq:DeltaQstar} satisfies
\begin{equation}
Q_{n}(\widehat{\bm{\gamma}}^{\circ})-\min_{\bm{\gamma}: A(\bm{\gamma})=A}Q_{n}(\bm{\gamma})=\mathbf{Y}^{\top}(\mathbf{M}_{A^{\circ}}-\mathbf{M}_{A})\mathbf{Y},
\end{equation}
where the right-hand side corresponds to the decrease of OLS residual sums of squares from adding missed relevant regressors in $A''=A^{\circ}\backslash A\neq\varnothing$ to $A$. 

Consider the decomposition of the design matrix $\mathbf{R}=(\mathbf{R}_{A}^{\top},\mathbf{R}_{A''}^{\top},\mathbf{R}_{\{1,\dots,p\}\backslash A^{\circ}}^{\top})^{\top}$ and the oracle estimator $\widehat{\bm{\gamma}}^{\circ}=((\widehat{\bm{\gamma}}_{A}^{\circ})^{\top},(\widehat{\bm{\gamma}}_{A''}^{\circ})^{\top},0^{\top})^{\top}$, and thus
\begin{equation}
\mathbf{M}_{A}\mathbf{Y}=\mathbf{M}_{A}(\mathbf{R}_{A}\widehat{\bm{\gamma}}_{A}^{\circ}+\mathbf{R}_{A''}\widehat{\bm{\gamma}}_{A''}^{\circ}+\widehat{\bm{\varepsilon}})=\mathbf{M}_{A}\mathbf{R}_{A''}\widehat{\bm{\gamma}}_{A''}^{\circ}+\widehat{\bm{\varepsilon}},
\end{equation}
because $\mathbf{M}_{A}\mathbf{R}_{A}=0$ and $\mathbf{M}_{A}\widehat{\bm{\varepsilon}}=\widehat{\bm{\varepsilon}}$ by definition. With similar applications of the Frisch-Waugh-Lovell theorem as in \cref{Eq:FWL}, it holds that
\begin{align}
\mathbf{Y}^{\top}(\mathbf{M}_{A}-\mathbf{M}_{A^{\circ}})\mathbf{Y}&=\mathbf{Y}^{\top}(\mathbf{M}_{A}-\mathbf{M}_{A\cup A''})\mathbf{Y} \notag\\
%&=(\mathbf{M}_{A}\mathbf{Y})^{\top}\mathbf{M}_{A}\mathbf{Y}\\
%&\qquad-(\mathbf{M}_{A}\mathbf{Y})^{\top}(\mathbf{I}_{n}-\mathbf{M}_{A}\mathbf{R}_{A''}((\mathbf{M}_{A}\mathbf{R}_{A''})^{\top}\mathbf{M}_{A}\mathbf{R}_{A''})^{-1}(\mathbf{M}_{A}\mathbf{R}_{A''})^{\top})\mathbf{M}_{A}\mathbf{Y}\\
&=(\mathbf{M}_{A}\mathbf{Y})^{\top}\mathbf{M}_{A}\mathbf{R}_{A''}((\mathbf{M}_{A}\mathbf{R}_{A''})^{\top}\mathbf{M}_{A}\mathbf{R}_{A''})^{-1}(\mathbf{M}_{A}\mathbf{R}_{A''})^{\top}\mathbf{M}_{A}\mathbf{Y} \notag\\
&=(\widehat{\bm{\gamma}}_{A''}^{\circ})^{\top}\mathbf{R}_{A''}^{\top}\mathbf{M}_{A}\mathbf{R}_{A''}\widehat{\bm{\gamma}}_{A''}^{\circ}+2(\widehat{\bm{\gamma}}_{A''}^{\circ})^{\top}\mathbf{R}_{A''}^{\top}\widehat{\bm{\varepsilon}}+\widehat{\bm{\varepsilon}}^{\top}\mathbf{R}_{A''}(\mathbf{R}_{A''}^{\top}\mathbf{M}_{A}\mathbf{R}_{A''})^{-1}\mathbf{R}_{A''}^{\top}\widehat{\bm{\varepsilon}} \notag\\
&\geq (\widehat{\bm{\gamma}}_{A''}^{\circ})^{\top}\mathbf{R}_{A''}^{\top}\mathbf{M}_{A}\mathbf{R}_{A''}\widehat{\bm{\gamma}}_{A''}^{\circ}+2(\widehat{\bm{\gamma}}_{A''}^{\circ})^{\top}\mathbf{R}_{A''}^{\top}\widehat{\bm{\varepsilon}},
\end{align}
where the nonnegative quadratic term of $\widehat{\bm{\varepsilon}}$ is dropped for a lower bound. Since $\widehat{\bm{\varepsilon}}$ is orthogonal to the span of $\mathbf{R}_{A^{\circ}}$, we have $\mathbf{R}_{A''}^{\top}\widehat{\bm{\varepsilon}}=\bm{0}$ and thus the cross term vanishes. Therefore, it holds that
\begin{equation}
\label{Eq:underfit_SSR}
\mathbf{Y}^{\top}(\mathbf{M}_{A}-\mathbf{M}_{A^{\circ}})\mathbf{Y}\geq(\widehat{\bm{\gamma}}_{A''}^{\circ})^{\top}\mathbf{R}_{A''}^{\top}\mathbf{M}_{A}\mathbf{R}_{A''}\widehat{\bm{\gamma}}_{A''}^{\circ}\geq K\sum_{j\in A''}(\widehat{\gamma}_{j}^{\circ})^{\top}\mathbf{W}_{j}\widehat{\gamma}_{j}^{\circ}\geq K(q-|A|)\min_{j\in A''}(\widehat{\gamma}_{j}^{\circ})^{\top}\mathbf{W}_{j}\widehat{\gamma}_{j}^{\circ}. 
\end{equation}
By \cref{Eq:KKT_inequality_active,Eq:KKT_inequality_active2},  $\sqrt{(\widehat{\gamma}_{j}^{\circ})^{\top}\mathbf{W}_{j}\widehat{\gamma}_{j}^{\circ}}\overset{\mathbb{P}}{\longrightarrow}\Vert\beta_{j}\Vert_{L^{2}}>0$ for all $j\in A^{\circ}$.
%\begin{equation}
%(\widehat{\gamma}_{j}^{\circ})^{\top}\mathbf{W}_{j}\widehat{\gamma}_{j}^{\circ}\overset{\mathbb{P}}{\longrightarrow}\int_{0}^{T}(\widehat{\beta}_{j,s}^{\circ})^{\top}c_{s}\widehat{\beta}_{j,s}^{\circ}ds\asymp\Vert\widehat{\beta}_{j}^{\circ}\Vert_{L^{2}}^{2}>0.
%\end{equation}
Therefore, for some $b_{\text{min}}=K\min_{j\in A^{\circ}}\Vert\beta_{j}\Vert_{L^{2}}^{2}>0$, it holds that
\begin{equation}
\label{Eq:underfit_SSR_unionBound}
\mathbb{P}\left(\min_{j\in A^{\circ}}(\widehat{\gamma}_{j}^{\circ})^{\top}\mathbf{W}_{j}\widehat{\gamma}_{j}^{\circ}\geq b_{\text{min}}\right)\to 1.
\end{equation}
By \cref{Eq:DeltaQstar,Eq:underfit_SSR}, 
\begin{equation}
\min_{\bm{\gamma}: A(\bm{\gamma})=A}Q_{n}^{*}(\bm{\gamma})-Q_{n}^{*}(\widehat{\bm{\gamma}}^{\circ})\geq \frac{K}{2}|A''|\min_{j\in A''}(\widehat{\gamma}_{j}^{\circ})^{\top}\mathbf{W}_{j}\widehat{\gamma}_{j}^{\circ}-\lambda_{n}|A''|.  
\end{equation}
With the union bound in \cref{Eq:underfit_SSR_unionBound} and $\lambda_{n}\to0$, it holds that
\begin{equation}
\label{Eq:Gamma2}
\mathbb{P}\left(\exists A\in\Gamma_{2}:\min_{\bm{\gamma}: A(\bm{\gamma})=A}Q_{n}^{*}(\bm{\gamma})-Q_{n}^{*}(\widehat{\bm{\gamma}}^{\circ})\leq 0\right)\leq1-\mathbb{P}\left(\min_{j\in A^{\circ}}(\widehat{\gamma}_{j}^{\circ})^{\top}\mathbf{W}_{j}\widehat{\gamma}_{j}^{\circ}\geq b_{\text{min}}\right)\to0.
\end{equation}

\subsubsection*{Misspecified models}

For any misspecified models with $A\in\Gamma_{3}$, we denote by $A'=A\cap A^{\circ}$ the index set of included relevant regressors, by $A''=A^{\circ}\backslash A\neq\varnothing$ the index set of missed relevant regressors, and by $A'''=A\backslash A^{\circ}\neq\varnothing$ the index set of spurious regressors, and thus $A^{\circ}=A'\cup A''$ and $A=A'\cup A'''$. It holds that
\begin{equation}
\begin{split}
Q_{n}(\widehat{\bm{\gamma}}^{\circ})-\min_{\bm{\gamma}: A(\bm{\gamma})=A}Q_{n}(\bm{\gamma})&=\mathbf{Y}^{\top}(\mathbf{M}_{A^{\circ}}-\mathbf{M}_{A})\mathbf{Y}\\
&=\mathbf{Y}^{\top}(\mathbf{M}_{A'\cup A''}-\mathbf{M}_{A'})\mathbf{Y}-\mathbf{Y}^{\top}(\mathbf{M}_{A'\cup A'''}-\mathbf{M}_{A'})\mathbf{Y}\\
&= \mathbf{Y}^{\top}(\mathbf{M}_{A'\cup A''}-\mathbf{M}_{A'})\mathbf{Y}+\mathbf{Y}^{\top}(\mathbf{M}_{A'}-\mathbf{M}_{A'\cup A'''})\mathbf{Y},
\end{split}
\end{equation}
where the first term is the decrease of OLS residual sums of squares from adding missed relevant regressors, and the second term is the increase from excluding spurious regressors. For the second term, we have
\begin{align}
\mathbf{Y}^{\top}(\mathbf{M}_{A'}-\mathbf{M}_{A'\cup A'''})\mathbf{Y}&=\mathbf{Y}^{\top}(\mathbf{M}_{A'}-\mathbf{M}_{A'\cup A''}+\mathbf{M}_{A'\cup A''}-\mathbf{M}_{A'\cup A''\cup A'''}+\mathbf{M}_{A'\cup A''\cup A'''}-\mathbf{M}_{A'\cup A'''})\mathbf{Y} \notag\\
&\leq \mathbf{Y}^{\top}(\mathbf{M}_{A'}-\mathbf{M}_{A'\cup A''})\mathbf{Y}+\mathbf{Y}^{\top}(\mathbf{M}_{A'\cup A''}-\mathbf{M}_{A'\cup A''\cup A'''})\mathbf{Y},
\end{align}
since $\mathbf{Y}^{\top}(\mathbf{M}_{A'\cup A''\cup A'''}-\mathbf{M}_{A'\cup A'''})\mathbf{Y}\leq 0$, and therefore
\begin{equation}
Q_{n}(\widehat{\bm{\gamma}}^{\circ})-\min_{\bm{\gamma}: A(\bm{\gamma})=A}Q_{n}(\bm{\gamma})\leq\mathbf{Y}^{\top}(\mathbf{M}_{A'\cup A''}-\mathbf{M}_{A'\cup A''\cup A'''})\mathbf{Y}=\mathbf{Y}^{\top}(\mathbf{M}_{A^{\circ}}-\mathbf{M}_{A^{\circ}\cup A'''})\mathbf{Y}. 
\end{equation}
By \cref{lemma:overfitOneBlock,lemma:extraSSR}, it holds that, on the event $E_{n}$,
\begin{equation}
Q_{n}(\widehat{\bm{\gamma}}^{\circ})-\min_{\bm{\gamma}: A(\bm{\gamma})=A}Q_{n}(\bm{\gamma})%\leq |A'''|\max_{q+1\leq j\leq p}\widehat{\bm{\varepsilon}}^{\top}\mathbf{H}_{j}^{*}\widehat{\bm{\varepsilon}}\\
\leq K|A'''|\Delta_{n}(K_{n}+\log p).
\end{equation}
Bring it back to \cref{Eq:DeltaQstar}, on the event $E_{n}$,
\begin{equation}
\begin{split}
\min_{\bm{\gamma}: A(\bm{\gamma})=A}Q_{n}^{*}(\bm{\gamma})-Q_{n}^{*}(\widehat{\bm{\gamma}}^{\circ})%\geq\frac{1}{2}K(|A|-q)\Delta_{n}(K_{n}+\log p)+\lambda_{n}(|A|-q)
&\geq -\frac{K}{2}|A'''|\Delta_{n}(K_{n}+\log p)+\lambda_{n}(|A'| + |A'''| -q)\\
&\geq |A'''|\left(-\frac{K}{2}\Delta_{n}(K_{n}+\log p)+\lambda_{n}\right)-\lambda_{n}|A''|> 0,
\end{split}
\end{equation}
since $|A'''|\geq1$, $\lambda_{n}\gg \Delta_{n}(K_{n}+\log p)$, and $\lambda_{n}\to0$. Consequently, by \cref{lemma:overfitOneBlock},
\begin{equation}
\label{Eq:Gamma3}
\mathbb{P}\left(\exists A\in\Gamma_{3}:\min_{\bm{\gamma}: A(\bm{\gamma})=A}Q_{n}^{*}(\bm{\gamma})-Q_{n}^{*}(\widehat{\bm{\gamma}}^{\circ})\leq0\right)\leq\mathbb{P}(E_{n}^{\complement})\to0.
\end{equation}
In conclusion, \cref{Eq:Gamma1,Eq:Gamma2,Eq:Gamma3} are sufficient for \cref{Eq:Gamma123}. This completes the proof.

\section{Supplementary Materials}
\label{AP:Supp}

\subsection{Supplementary Simulation Results}
\label{AP:Supp-simulation}

In addition to the simulation results in \cref{Sec:simulation} (\cref{Table:simulation-IBeta,Table:simulation-selection}), this appendix reports supplementary Monte Carlo evidence on the sensitivity of the proposed estimation procedure. %We consider (i) a less aggressive truncation threshold, (ii) discontinuous coefficient paths, and (iii) imperfect multicollinearity among factors. 
Unless stated otherwise, the data-generating process and tuning choices follow the baseline design in \cref{Sec:simulation_design}, and results are reported in \cref{Table:simulation-IBeta-threshold4}--\ref{Table:simulation-selection-multicollinearity}.

\subsubsection*{Alternative truncation threshold}

\cref{Table:simulation-IBeta-threshold4,Table:simulation-selection-threshold4} report the model estimation and selection results, respectively, with a less aggressive truncation threshold $4\Delta_{n}^{0.47}\sqrt{\text{MedRV}_{T}}$ for the increments of $Y$ and each component of $X$, with the same simulation design used in \cref{Sec:simulation}.

\subsubsection*{Imperfect multicollinearity}

%We introduce imperfect multicollinearity by letting the correlation matrix of the redundant factors to be random: the off-diagonal entries of their correlation block are drawn uniformly from $[-0.80,0.80]$. In contrast, the first three relevant factors follow the same low-correlation Toeplitz structure used in \cref{Sec:simulation_design}, while the entire correlation matrix is adjusted to be positive definite as needed. This design keeps the relevant factors relatively independent—and only weakly correlated with the redundant factors—such that the relevant-factor dependence remains comparable to the baseline design in the main text, while it allows strong collinearity among the redundant factors. 

%We introduce imperfect multicollinearity by drawing the pairwise correlations among redundant factors, and their correlations with the relevant factors, independently from $U[-0.8,0.8]$. The correlations within the set of relevant factors follow the same low-correlation Toeplitz structure used in \cref{Sec:simulation_design}, and the entire correlation matrix is adjusted to be positive definite. Therefore, the relevant-factor dependence remains comparable to the baseline design in the main text, while it allows strong dependence among the redundant factors and nontrivial correlations between relevant and redundant factors. The model estimation and selection results are reported in Tables B.7 and B.8, respectively. 

To introduce imperfect multicollinearity while preserving the baseline dependence among the relevant factors, we construct the $p\times p$ correlation matrix in block form. We fix the relevant block to have the same Toeplitz structure as in \cref{Sec:simulation_design}, i.e., 
\begin{equation}
\bm{\rho}_{11}=\left(0.15^{|i-j|}\right)_{1\le i,j\le q}.
\end{equation}
We then generate the redundant-factor block $\bm{\rho}_{22}$ from a simple latent-factor construction, which is positive definite by construction and can exhibit relatively strong within-block dependence. Specifically, we draw a loading matrix $\textbf{L}\in\mathbb{R}^{(p-q)\times k}$, where $k$ is the latent-factor dimension, and set
\begin{equation}
\bm{\rho}_{22}=\textbf{L}\textbf{L}^{\top}+\mathrm{diag}(1-\Vert \textbf{L}_{i\cdot}\Vert^{2}),
\end{equation}
where $\Vert \textbf{L}_{i\cdot}\Vert$ is the Euclidean norm of the $i$-th row of $\textbf{L}$. Next, we draw a candidate cross-correlation block $\bm{\rho}_{12}^{*}\in\mathbb{R}^{q\times(p-q)}$ with entries i.i.d.~$U[-1,1]$. Since an arbitrary choice of $\bm{\rho}_{12}$ may violate positive definiteness of the full matrix, we scale the cross-block by a single factor and set $\bm{\rho}_{12}=\kappa \bm{\rho}_{12}^{*}$, where $\kappa\in(0,1]$ is chosen as large as possible subject to the Schur-complement condition, i.e.,
\begin{equation}
\bm{\rho}_{22}-\kappa^{2}(\bm{\rho}_{12}^{*})^{\top}\bm{\rho}_{11}^{-1}\bm{\rho}_{12}^{*}
\end{equation}
is positive definite. The final correlation matrix
\begin{equation}
\bm{\rho}=\begin{pmatrix}\bm{\rho}_{11} & \bm{\rho}_{12}\\ \bm{\rho}_{12}^{\top} & \bm{\rho}_{22}\end{pmatrix}
\end{equation}
is symmetric with unit diagonal and positive definite. Hence, the relevant-factor dependence remains comparable to the baseline design in the main text, while it allows stronger dependence among the redundant factors and nontrivial correlations between relevant and redundant factors. \cref{Fig:simulation_cross_corr} plots the histogram of all off-diagonal entries of $\bm{\rho}$ for the case $p=500$. The model estimation and selection results are reported in \cref{Table:simulation-IBeta-multicollinearity,Table:simulation-selection-multicollinearity}, respectively.

\begin{figure}[!ht]
\centering
\addtolength{\leftskip} {-2cm}
\addtolength{\rightskip}{-2cm}
\includegraphics[width=0.85\textwidth]{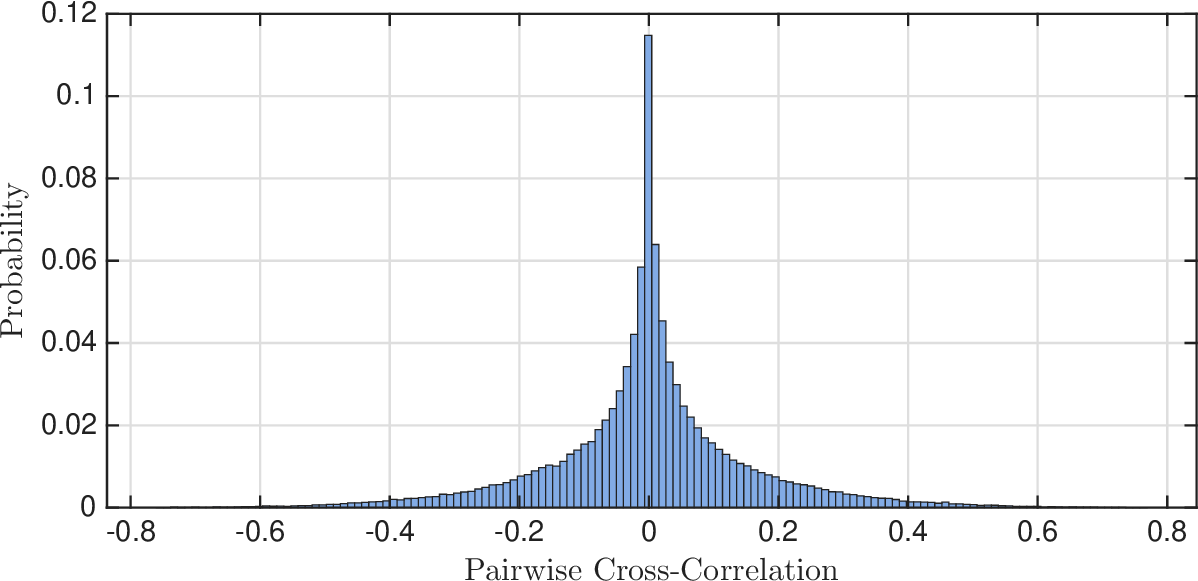}
\caption{Histogram of all off-diagonal entries of the correlation matrix $\bm{\rho}$ for the case $p=500$. }
\label{Fig:simulation_cross_corr}
\end{figure}

\newpage

\begin{table}%[h!]
	\centering
	\onehalfspacing
	\addtolength{\leftskip} {-2cm}
	\addtolength{\rightskip}{-2cm}
	\setlength{\tabcolsep}{6pt}
	\begin{threeparttable}
		\caption{Simulation results for model estimation with alternative truncation threshold}
		\scriptsize
		\begin{tabular}{lllrrrrrrrrrrrr}
			\hline\hline
			\multicolumn{15}{l}{Panel A: $p=3$}\\	
			& & & & \multicolumn{3}{c}{$\widehat{I\beta}_{1,T}$} & & \multicolumn{3}{c}{$\widehat{I\beta}_{2,T}$} & & \multicolumn{3}{c}{$\widehat{I\beta}_{3,T}$}\\
			\cline{5-7}\cline{9-11}\cline{13-15}
			Estimator & & Tuning & & Bias & Stdev & RMSE & & Bias & Stdev & RMSE & & Bias & Stdev & RMSE\\
			\hline
			Spline      & &                      & &  0.023 & 2.793 & 2.792 & & -0.024 & 2.537 & 2.536 & & -0.022 & 2.331 & 2.330 \\
			Spline TLP  & & $\alpha_{\tau}=0.05$ & &  0.023 & 2.793 & 2.792 & & -0.024 & 2.537 & 2.536 & & -0.022 & 2.331 & 2.330 \\
			Spline TLP  & & $\alpha_{\tau}=0.01$ & &  0.023 & 2.793 & 2.792 & & -0.025 & 2.537 & 2.536 & & -0.021 & 2.331 & 2.330 \\
			AKX         & & $k_{n}=78$           & &  0.000 & 2.871 & 2.869 & &  0.000 & 2.632 & 2.631 & & -0.010 & 2.451 & 2.450 \\
			AKX         & & $k_{n}=91$           & & -0.018 & 2.887 & 2.885 & & -0.021 & 2.621 & 2.619 & & -0.027 & 2.475 & 2.474 \\
			AKX         & & $k_{n}=117$          & & -0.048 & 2.894 & 2.893 & & -0.003 & 2.640 & 2.638 & &  0.003 & 2.495 & 2.494 \\
			\hline
			\multicolumn{15}{l}{Panel B: $p=10$}\\	
			& & & & \multicolumn{3}{c}{$\widehat{I\beta}_{1,T}$} & & \multicolumn{3}{c}{$\widehat{I\beta}_{2,T}$} & & \multicolumn{3}{c}{$\widehat{I\beta}_{3,T}$}\\
			\cline{5-7}\cline{9-11}\cline{13-15}
			Estimator & & Tuning & & Bias & Stdev & RMSE & & Bias & Stdev & RMSE & & Bias & Stdev & RMSE\\
			\hline
			Spline      & &                      & & -0.144 & 3.032 & 3.034 & & -0.026 & 2.806 & 2.805 & & -0.095 & 2.507 & 2.508 \\
			Spline TLP  & & $\alpha_{\tau}=0.05$ & & -0.152 & 3.004 & 3.006 & & -0.024 & 2.776 & 2.774 & & -0.085 & 2.450 & 2.451 \\
			Spline TLP  & & $\alpha_{\tau}=0.01$ & & -0.152 & 3.004 & 3.006 & & -0.024 & 2.775 & 2.774 & & -0.085 & 2.451 & 2.451 \\
			AKX         & & $k_{n}=78$           & & -0.122 & 3.325 & 3.326 & & -0.003 & 3.022 & 3.020 & & -0.165 & 2.771 & 2.775 \\
			AKX         & & $k_{n}=91$           & & -0.187 & 3.221 & 3.225 & &  0.018 & 3.102 & 3.101 & & -0.094 & 2.663 & 2.663 \\
			AKX         & & $k_{n}=117$          & & -0.123 & 3.233 & 3.234 & & -0.004 & 3.088 & 3.086 & & -0.123 & 2.688 & 2.690 \\
			\hline
			\multicolumn{15}{l}{Panel C: $p=50$}\\	
			& & & & \multicolumn{3}{c}{$\widehat{I\beta}_{1,T}$} & & \multicolumn{3}{c}{$\widehat{I\beta}_{2,T}$} & & \multicolumn{3}{c}{$\widehat{I\beta}_{3,T}$}\\
			\cline{5-7}\cline{9-11}\cline{13-15}
			Estimator & & Tuning & & Bias & Stdev & RMSE & & Bias & Stdev & RMSE & & Bias & Stdev & RMSE\\
			\hline
			Spline      & &                      & & -0.184 & 2.898 & 2.902 & & -0.019 & 3.158 & 3.156 & & -0.070 & 2.941 & 2.940 \\
			Spline TLP  & & $\alpha_{\tau}=0.05$ & & -0.199 & 2.667 & 2.673 & &  0.014 & 3.066 & 3.065 & & -0.182 & 2.993 & 2.997 \\
			Spline TLP  & & $\alpha_{\tau}=0.01$ & & -0.201 & 2.669 & 2.675 & & -0.004 & 3.053 & 3.051 & & -0.092 & 2.647 & 2.647 \\
			AKX         & & $k_{n}=78$           & & -0.151 & 5.135 & 5.135 & &  0.009 & 5.095 & 5.092 & & -0.314 & 4.887 & 4.895 \\
			AKX         & & $k_{n}=91$           & & -0.046 & 4.671 & 4.669 & & -0.030 & 4.691 & 4.689 & & -0.263 & 4.235 & 4.241 \\
			AKX         & & $k_{n}=117$          & & -0.127 & 3.894 & 3.895 & & -0.042 & 4.140 & 4.138 & & -0.124 & 3.866 & 3.866 \\
			\hline
			\multicolumn{15}{l}{Panel D: $p=100$}\\	
			& & & & \multicolumn{3}{c}{$\widehat{I\beta}_{1,T}$} & & \multicolumn{3}{c}{$\widehat{I\beta}_{2,T}$} & & \multicolumn{3}{c}{$\widehat{I\beta}_{3,T}$}\\
			\cline{5-7}\cline{9-11}\cline{13-15}
			Estimator & & Tuning & & Bias & Stdev & RMSE & & Bias & Stdev & RMSE & & Bias & Stdev & RMSE\\
			\hline
			Spline      & &                      & & -0.108 & 3.351 & 3.351 & &  0.261 & 2.964 & 2.974 & & -0.241 & 3.487 & 3.494 \\
			Spline TLP  & & $\alpha_{\tau}=0.05$ & & -0.082 & 2.933 & 2.932 & &  0.280 & 2.601 & 2.615 & & -0.342 & 3.380 & 3.396 \\
			Spline TLP  & & $\alpha_{\tau}=0.01$ & & -0.085 & 2.928 & 2.928 & &  0.270 & 2.597 & 2.610 & & -0.282 & 3.348 & 3.358 \\
			AKX         & & $k_{n}=78$           & &   --   &  --   &  --   & &   --   &  --   &  --   & &   --   &  --   &  --   \\
			AKX         & & $k_{n}=91$           & &   --   &  --   &  --   & &   --   &  --   &  --   & &   --   &  --   &  --   \\
			AKX         & & $k_{n}=117$          & & -0.112 & 8.294 & 8.291 & &  0.530 & 7.133 & 7.149 & & -0.346 & 9.036 & 9.038 \\
			\hline
			\multicolumn{15}{l}{Panel E: $p=500$}\\	
			& & & & \multicolumn{3}{c}{$\widehat{I\beta}_{1,T}$} & & \multicolumn{3}{c}{$\widehat{I\beta}_{2,T}$} & & \multicolumn{3}{c}{$\widehat{I\beta}_{3,T}$}\\
			\cline{5-7}\cline{9-11}\cline{13-15}
			Estimator & & Tuning & & Bias & Stdev & RMSE & & Bias & Stdev & RMSE & & Bias & Stdev & RMSE\\
			\hline
			Spline      & &                      & & -45.012 & 1.572 & 45.039 & & 34.732 & 1.564 & 34.768 & & -20.000 & 1.538 & 20.059 \\
			Spline TLP  & & $\alpha_{\tau}=0.05$ & &  -0.065 & 2.545 &  2.544 & &  0.105 & 2.626 &  2.626 & &  -0.058 & 2.555 &  2.555 \\
			Spline TLP  & & $\alpha_{\tau}=0.01$ & &  -0.065 & 2.543 &  2.542 & &  0.103 & 2.625 &  2.626 & &  -0.041 & 2.470 &  2.469 \\
			AKX         & & $k_{n}=78$           & &   --   &  --   &  --   & &   --   &  --   &  --   & &   --   &  --   &  --   \\
			AKX         & & $k_{n}=91$           & &   --   &  --   &  --   & &   --   &  --   &  --   & &   --   &  --   &  --   \\
			AKX         & & $k_{n}=117$          & &   --   &  --   &  --   & &   --   &  --   &  --   & &   --   &  --   &  --   \\
			\hline\hline
		\end{tabular} 
		\label{Table:simulation-IBeta-threshold4}
		All these quantities are multiplied by 100. The cutoff level in the TLP function is $\tau_{n}=\alpha_{\tau}\sqrt{\text{MedRV}_{T}}$ for each simulated path of $Y$. The number of B-spline basis functions $K_{n}$ and the effective penalty level $\lambda_{n}/\tau_{n}$ are selected via 5-fold cross-validation. AKX stands for the estimator of \citet{ait2020high}. For both the spline-based and AKX estimators, the truncation thresholds for the increments of $Y$ and each component of $X$ are given by $4\Delta_{n}^{0.47}\sqrt{\text{MedRV}_{T}}$. All results are based on 1,000 Monte Carlo replications.	
	\end{threeparttable}
\end{table}

\begin{table}[h!]
	\centering
	\onehalfspacing
	\addtolength{\leftskip} {-2cm}
	\addtolength{\rightskip}{-2cm}
	\setlength{\tabcolsep}{10pt}
	\begin{threeparttable}
			\caption{Simulation results for model selection with alternative truncation threshold}
			\scriptsize 
			\begin{tabularx}{0.9\textwidth}{Xc *{3}{Y} c *{3}{Y}}
					\hline\hline
					& & \multicolumn{3}{c}{$\alpha_{\tau}=0.05$} & & \multicolumn{3}{c}{$\alpha_{\tau}=0.01$}\\
					\cline{3-5} \cline{7-9}
					& & Relevant & Irrelevant & Correct & & Relevant & Irrelevant & Correct\\
					\hline
					$p=3$   & & 1.000 & --    & 1.000 & & 1.000 & --    & 1.000 \\
					$p=10$  & & 1.000 & 0.041 & 0.745 & & 1.000 & 0.000 & 1.000\\
					$p=50$  & & 1.000 & 0.011 & 0.604  & & 1.000 & 0.000 & 1.000\\
					$p=100$ & & 1.000 & 0.038 & 0.019 & & 1.000 & 0.000 & 1.000\\
					$p=500$ & & 1.000 & 0.010 & 0.014 & & 1.000 & 0.000 & 0.788\\
					\hline\hline
				\end{tabularx} 
			\label{Table:simulation-selection-threshold4}
			Average selection frequencies for the relevant and irrelevant factors, and frequencies of correct model specification. The cutoff level in the TLP function is $\tau_{n}=\alpha_{\tau}\sqrt{\text{MedRV}_{T}}$ for each simulated path of $Y$. The number of B-spline basis functions $K_{n}$ and the effective penalty level $\lambda_{n}/\tau_{n}$ are selected via 5-fold cross-validation. The truncation thresholds for the increments of $Y$ and each component of $X$ are given by $4\Delta_{n}^{0.47}\sqrt{\text{MedRV}_{T}}$.  All results are based on 1,000 Monte Carlo replications.
		\end{threeparttable}
\end{table}

\begin{table}%[h!]
	\centering
	\onehalfspacing
	\addtolength{\leftskip} {-2cm}
	\addtolength{\rightskip}{-2cm}
	\setlength{\tabcolsep}{5.5pt}
	\begin{threeparttable}
		\caption{Simulation results for model estimation with imperfect multicollinearity}
		\scriptsize
		\begin{tabular}{lllrrrrrrrrrrrr}
			\hline\hline
			\multicolumn{15}{l}{Panel A: $p=3$}\\	
			& & & & \multicolumn{3}{c}{$\widehat{I\beta}_{1,T}$} & & \multicolumn{3}{c}{$\widehat{I\beta}_{2,T}$} & & \multicolumn{3}{c}{$\widehat{I\beta}_{3,T}$}\\
			\cline{5-7}\cline{9-11}\cline{13-15}
			Estimator & & Tuning & & Bias & Stdev & RMSE & & Bias & Stdev & RMSE & & Bias & Stdev & RMSE\\
			\hline
			Spline      & &                      & & -0.249 &  2.829 &  2.839 & &  0.161 &  2.527 &  2.531 & & -0.112 &  2.345 &  2.347 \\
			Spline TLP  & & $\alpha_{\tau}=0.05$ & & -0.250 &  2.829 &  2.839 & &  0.162 &  2.527 &  2.531 & & -0.112 &  2.345 &  2.347 \\
			Spline TLP  & & $\alpha_{\tau}=0.01$ & & -0.250 &  2.829 &  2.839 & &  0.162 &  2.527 &  2.531 & & -0.112 &  2.345 &  2.347 \\
			AKX         & & $k_{n}=78$           & & -0.257 &  2.900 &  2.910 & &  0.170 &  2.622 &  2.626 & & -0.103 &  2.447 &  2.448 \\
			AKX         & & $k_{n}=91$           & & -0.279 &  2.920 &  2.931 & &  0.165 &  2.619 &  2.623 & & -0.127 &  2.486 &  2.488 \\
			AKX         & & $k_{n}=117$          & & -0.314 &  2.929 &  2.945 & &  0.177 &  2.631 &  2.636 & & -0.091 &  2.497 &  2.498 \\
			\hline
			\multicolumn{15}{l}{Panel B: $p=10$}\\	
			& & & & \multicolumn{3}{c}{$\widehat{I\beta}_{1,T}$} & & \multicolumn{3}{c}{$\widehat{I\beta}_{2,T}$} & & \multicolumn{3}{c}{$\widehat{I\beta}_{3,T}$}\\
			\cline{5-7}\cline{9-11}\cline{13-15}
			Estimator & & Tuning & & Bias & Stdev & RMSE & & Bias & Stdev & RMSE & & Bias & Stdev & RMSE\\
			\hline
			Spline      & &                      & & -0.337 &  5.401 &  5.409 & &  0.288 &  9.212 &  9.212 & & -0.165 & 52.775 & 52.748 \\
			Spline TLP  & & $\alpha_{\tau}=0.05$ & & -0.302 &  3.033 &  3.046 & &  0.257 &  2.889 &  2.899 & & -0.081 &  2.552 &  2.552 \\
			Spline TLP  & & $\alpha_{\tau}=0.01$ & & -0.302 &  3.032 &  3.046 & &  0.258 &  2.889 &  2.899 & & -0.089 &  2.529 &  2.530 \\
			AKX         & & $k_{n}=78$           & &  0.038 & 11.424 & 11.418 & & -0.439 & 22.124 & 22.117 & &  4.002 &129.714 &129.711 \\
			AKX         & & $k_{n}=91$           & & -0.006 & 10.902 & 10.897 & & -0.393 & 20.901 & 20.895 & &  3.525 &122.475 &122.465 \\
			AKX         & & $k_{n}=117$          & & -0.412 & 10.325 & 10.328 & &  0.405 & 19.823 & 19.817 & & -0.967 &115.436 &115.382 \\
			\hline
			\multicolumn{15}{l}{Panel C: $p=50$}\\	
			& & & & \multicolumn{3}{c}{$\widehat{I\beta}_{1,T}$} & & \multicolumn{3}{c}{$\widehat{I\beta}_{2,T}$} & & \multicolumn{3}{c}{$\widehat{I\beta}_{3,T}$}\\
			\cline{5-7}\cline{9-11}\cline{13-15}
			Estimator & & Tuning & & Bias & Stdev & RMSE & & Bias & Stdev & RMSE & & Bias & Stdev & RMSE\\
			\hline
			Spline      & &                      & & -1.297 &  7.651 &  7.756 & &  1.573 &  8.529 &  8.668 & & -2.245 & 13.414 & 13.594 \\
			Spline TLP  & & $\alpha_{\tau}=0.05$ & & -0.226 &  2.845 &  2.852 & &  0.225 &  3.217 &  3.224 & & -0.196 &  2.670 &  2.676 \\
			Spline TLP  & & $\alpha_{\tau}=0.01$ & & -0.224 &  2.845 &  2.852 & &  0.224 &  3.231 &  3.237 & & -0.193 &  2.757 &  2.762 \\
			AKX         & & $k_{n}=78$           & &  0.124 & 44.185 & 44.163 & &  0.216 & 48.798 & 48.774 & &  0.928 & 92.308 & 92.267 \\
			AKX         & & $k_{n}=91$           & &  0.714 & 37.624 & 37.612 & & -0.801 & 41.825 & 41.812 & &  2.677 & 79.908 & 79.913 \\
			AKX         & & $k_{n}=117$          & & -0.465 & 27.295 & 27.285 & &  0.548 & 30.044 & 30.034 & & -0.405 & 57.546 & 57.518 \\
			\hline
			\multicolumn{15}{l}{Panel D: $p=100$}\\	
			& & & & \multicolumn{3}{c}{$\widehat{I\beta}_{1,T}$} & & \multicolumn{3}{c}{$\widehat{I\beta}_{2,T}$} & & \multicolumn{3}{c}{$\widehat{I\beta}_{3,T}$}\\
			\cline{5-7}\cline{9-11}\cline{13-15}
			Estimator & & Tuning & & Bias & Stdev & RMSE & & Bias & Stdev & RMSE & & Bias & Stdev & RMSE\\
			\hline
			Spline      & &                      & & -6.568 &  8.022 & 10.364 & &  5.921 &  6.179 &  8.556 & & -3.756 &  6.944 &  7.892 \\
			Spline TLP  & & $\alpha_{\tau}=0.05$ & & -0.531 &  3.448 &  3.487 & &  0.599 &  2.771 &  2.834 & & -0.277 &  3.048 &  3.059 \\
			Spline TLP  & & $\alpha_{\tau}=0.01$ & & -0.532 &  3.449 &  3.488 & &  0.587 &  2.767 &  2.828 & & -0.233 &  2.921 &  2.929 \\
			AKX         & & $k_{n}=78$           & &   --   &   --   &   --   & &   --   &   --   &   --   & &   --   &   --   &   --   \\
			AKX         & & $k_{n}=91$           & &   --   &   --   &   --   & &   --   &   --   &   --   & &   --   &   --   &   --   \\
			AKX         & & $k_{n}=117$          & &  3.499 &100.218 &100.229 & & -3.375 & 92.011 & 92.027 & &  2.155 & 58.581 & 58.591 \\
			\hline
			\multicolumn{15}{l}{Panel E: $p=500$}\\	
			& & & & \multicolumn{3}{c}{$\widehat{I\beta}_{1,T}$} & & \multicolumn{3}{c}{$\widehat{I\beta}_{2,T}$} & & \multicolumn{3}{c}{$\widehat{I\beta}_{3,T}$}\\
			\cline{5-7}\cline{9-11}\cline{13-15}
			Estimator & & Tuning & & Bias & Stdev & RMSE & & Bias & Stdev & RMSE & & Bias & Stdev & RMSE\\
			\hline
			Spline      & &                      & & -45.035 &  1.801 & 45.071 & & 35.018 &  1.902 & 35.069 & & -20.468 &  2.029 & 20.568 \\
			Spline TLP  & & $\alpha_{\tau}=0.05$ & &  -0.676 &  3.072 &  3.144 & &  0.836 &  3.398 &  3.498 & &  -1.109 &  5.477 &  5.585 \\
			Spline TLP  & & $\alpha_{\tau}=0.01$ & &  -0.668 &  3.064 &  3.135 & &  0.756 &  3.379 &  3.461 & &  -0.466 &  4.387 &  4.410 \\
			AKX         & & $k_{n}=78$           & &   --   &   --   &   --   & &   --   &   --   &   --   & &   --   &   --   &   --   \\
			AKX         & & $k_{n}=91$           & &   --   &   --   &   --   & &   --   &   --   &   --   & &   --   &   --   &   --   \\
			AKX         & & $k_{n}=117$          & &   --   &   --   &   --   & &   --   &   --   &   --   & &   --   &   --   &   --   \\
			\hline\hline
		\end{tabular} 
		\label{Table:simulation-IBeta-multicollinearity}
		All these quantities are multiplied by 100. The modified simulation design allows dependence among the redundant factors and also correlations between relevant and redundant factors. The cutoff level in the TLP function is $\tau_{n}=\alpha_{\tau}\sqrt{\text{MedRV}_{T}}$ for each simulated path of $Y$. The number of B-spline basis functions $K_{n}$ and the effective penalty level $\lambda_{n}/\tau_{n}$ are selected via 5-fold cross-validation. AKX stands for the estimator of \citet{ait2020high}. For both the spline-based and AKX estimators, the truncation thresholds for the increments of $Y$ and each component of $X$ are given by $3\Delta_{n}^{0.47}\sqrt{\text{MedRV}_{T}}$.  All results are based on 1,000 Monte Carlo replications.	
	\end{threeparttable}
\end{table}

\begin{table}[h!]
	\centering
	\onehalfspacing
	\addtolength{\leftskip} {-2cm}
	\addtolength{\rightskip}{-2cm}
	\setlength{\tabcolsep}{10pt}
	\begin{threeparttable}
		\caption{Simulation results for model selection with imperfect multicollinearity}
		\scriptsize 
		\begin{tabularx}{0.9\textwidth}{Xc *{3}{Y} c *{3}{Y}}
			\hline\hline
			& & \multicolumn{3}{c}{$\alpha_{\tau}=0.05$} & & \multicolumn{3}{c}{$\alpha_{\tau}=0.01$}\\
			\cline{3-5} \cline{7-9}
			& & Relevant & Irrelevant & Correct & & Relevant & Irrelevant & Correct\\
			\hline
			$p=3$   & & 1.000 & --    & 1.000 & & 1.000 & --    & 1.000 \\
			$p=10$  & & 1.000 & 0.010 & 0.931 & & 1.000 & 0.000 & 0.997\\
			$p=50$  & & 1.000 & 0.003 & 0.878  & & 1.000 & 0.000 & 1.000\\
			$p=100$ & & 1.000 & 0.012 & 0.348 & & 1.000 & 0.000 & 1.000\\
			$p=500$ & & 1.000 & 0.062 & 0.000 & & 0.999 & 0.000 & 0.984\\
			\hline\hline
		\end{tabularx} 
		\label{Table:simulation-selection-multicollinearity}
		Average selection frequencies for the relevant and irrelevant factors, and frequencies of correct model specification. The modified simulation design allows dependence among the redundant factors and also correlations between relevant and redundant factors. The cutoff level in the TLP function is $\tau_{n}=\alpha_{\tau}\sqrt{\text{MedRV}_{T}}$ for each simulated path of $Y$. The number of B-spline basis functions $K_{n}$ and the effective penalty level $\lambda_{n}/\tau_{n}$ are selected via 5-fold cross-validation. The truncation thresholds for the increments of $Y$ and each component of $X$ are given by $3\Delta_{n}^{0.47}\sqrt{\text{MedRV}_{T}}$.  All results are based on 1,000 Monte Carlo replications.
	\end{threeparttable}
\end{table}

\clearpage

\subsection{Supplementary Empirical Results}
\label{AP:Supp-empirical}

\subsubsection*{Sample information}

\cref{Table:empirical_dataDescription} reports the list of individual stocks in our sample, including their sector and industry group assignments. Sector classifications are based on the GICS sector code from WRDS Compustat. The Fama-French 48 industry group for each firm is assigned based on its four-digit SIC code from Compustat and the SIC-to-industry definitions from Kenneth R. French's data library. \cref{Table:empirical_zeroReturns} reports the fraction of zero 5-minute returns in each year for all 57 stocks.

\begin{table}[h!]
	\centering
	\onehalfspacing
	\addtolength{\leftskip} {-2cm}
	\addtolength{\rightskip}{-2cm}
	\setlength{\tabcolsep}{4.5pt}
	\begin{threeparttable}
		\caption{Descriptive information of selected S\&P 100 stocks}
		\scriptsize
		\begin{tabularx}{\textwidth}{lllllY}
			\hline\hline
			Symbol & Security & GICS & Sector & SIC & Fama-French 48 Industry\\
			\hline
			AAPL & Apple Inc & 45 & Information Technology & 3663 & Electronic Equipment (36)\\
			ABBV & AbbVie Inc & 35 & Health Care & 2836 & Pharmaceutical Products (13)\\
			ACN & Accenture Plc & 45 & Information Technology & 8742 & Business Services (34)\\
			AMGN & Amgen Inc & 35 & Health Care & 2836 & Pharmaceutical Products (13)\\
			AMZN & Amazon.com Inc & 25 & Consumer Discretionary & 5961 & Retail (42)\\
			AXP & American Express Co & 40 & Financials & 6141 & Banking (44)\\
			BA & Boeing Co & 20 & Industrials & 3721 & Aircraft (24)\\
			BIIB & Biogen Inc & 35 & Health Care & 2836 & Pharmaceutical Products (13)\\
			BLK & BlackRock Inc & 40 & Financials & 6282 & Trading (47)\\
			BMY & Bristol-Myers Squibb Co & 35 & Health Care & 2834 & Pharmaceutical Products (13)\\
			C & Citigroup Inc & 40 & Financials & 6199 & Banking (44)\\
			CAT & Caterpillar Inc & 20 & Industrials & 3531 & Machinery (21)\\
			COF & Capital One Financial Corp & 40 & Financials & 6141 & Banking (44)\\
			COP & ConocoPhillips & 10 & Energy & 1311 & Petroleum and Natural Gas (30)\\
			COST & Costco Wholesale Corp & 30 & Consumer Staples & 5399 & Retail (42)\\
			CVS & CVS Health Corp & 35 & Health Care & 8000 & Healthcare (11)\\
			CVX & Chevron Corp & 10 & Energy & 2911 & Petroleum and Natural Gas (30)\\
			DIS & Walt Disney Co & 50 & Communication Services & 4888 & Communication (32)\\
			EMR & Emerson Electric Co & 20 & Industrials & 3823 & Measuring and Control Equipment (37)\\
			FDX & FedEx Corp & 20 & Industrials & 4513 & Transportation (40)\\
			GD & General Dynamics Corp & 20 & Industrials & 3721 & Aircraft (24)\\
			GILD & Gilead Sciences Inc & 35 & Health Care & 2836 & Pharmaceutical Products (13)\\
			GOOG & Alphabet Inc & 50 & Communication Services & 7370 & Business Services (34)\\
			GS & Goldman Sachs Group Inc & 40 & Financials & 6211 & Trading (47)\\
			HD & Home Depot Inc & 25 & Consumer Discretionary & 5211 & Retail (42)\\
			HON & Honeywell International Inc & 20 & Industrials & 9997 & --\\
			IBM & IBM & 45 & Information Technology & 7370 & Business Services (34)\\
			JNJ & Johnson \& Johnson & 35 & Health Care & 2834 & Pharmaceutical Products (13)\\
			JPM & JPMorgan Chase \& Co & 40 & Financials & 6020 & Banking (44)\\
			LLY & Eli Lilly \& Co & 35 & Health Care & 2834 & Pharmaceutical Products (13)\\
			LMT & Lockheed Martin Corp & 20 & Industrials & 3760 & Defense (26)\\
			LOW & Lowe's Cos Inc & 25 & Consumer Discretionary & 5211 & Retail (42)\\
			MA & Mastercard Inc & 40 & Financials & 7389 & Business Services (34)\\
			MCD & McDonald's Corp & 25 & Consumer Discretionary & 5812 & Restaurants, Hotels, Motels (43)\\
			MDT & Medtronic Plc & 35 & Health Care & 3845 & Medical Equipment (12)\\
			MMM & 3M Co & 20 & Industrials & 9997 & --\\
			MO & Altria Group Inc & 30 & Consumer Staples & 2111 & Tobacco Products (5)\\
			MRK & Merck \& Co Inc & 35 & Health Care & 2834 & Pharmaceutical Products (13)\\
			MSFT & Microsoft Corp & 45 & Information Technology & 7372 & Business Services (34)\\
			NKE & Nike Inc & 25 & Consumer Discretionary & 3021 & Apparel (10)\\
			PEP & PepsiCo Inc & 30 & Consumer Staples & 2080 & Beer \& Liquor (4)\\
			\hline\hline
		\end{tabularx}
		\caption*{\small Continued on next page.}
		\label{Table:empirical_dataDescription}
	\end{threeparttable}
\end{table}

\begin{table}[h!]
	\centering
	\onehalfspacing
	\addtolength{\leftskip} {-2cm}
	\addtolength{\rightskip}{-2cm}
	\setlength{\tabcolsep}{4.5pt}
	\begin{threeparttable}
		\scriptsize
		\begin{tabularx}{\textwidth}{lllllY}
			\hline\hline
			Symbol & Security & GICS & Sector & SIC & Fama-French 48 Industry\\
			\hline
			PG & Procter \& Gamble Co & 30 & Consumer Staples & 2840 & Consumer Goods (9)\\
			PM & Philip Morris International Inc & 30 & Consumer Staples & 2111 & Tobacco Products (5)\\
			QCOM & Qualcomm Inc & 45 & Information Technology & 3674 & Electronic Equipment (36)\\
			SBUX & Starbucks Corp & 25 & Consumer Discretionary & 5812 & Restaurants, Hotels, Motels (43)\\
			SLB & Schlumberger Ltd & 10 & Energy & 1389 & Petroleum and Natural Gas (30)\\
			SPG & Simon Property Group Inc & 60 & Real Estate & 6798 & Trading (47)\\
			TGT & Target Corp & 30 & Consumer Staples & 5331 & Retail (42)\\
			UNH & UnitedHealth Group Inc & 35 & Health Care & 6324 & Insurance (45)\\
			UNP & Union Pacific Corp & 20 & Industrials & 4011 & Transportation (40)\\
			UPS & UPS & 20 & Industrials & 4210 & Transportation (40)\\
			V & Visa Inc & 40 & Financials & 6099 & Banking (44)\\
			VZ & Verizon Communications Inc & 50 & Communication Services & 4812 & Communication (32)\\
			WBA & Walgreens Boots Alliance Inc & 30 & Consumer Staples & 5912 & Retail (42)\\
			WFC & Wells Fargo \& Co & 40 & Financials & 6020 & Banking (44)\\
			WMT & Walmart Inc & 30 & Consumer Staples & 5331 & Retail (42)\\
			XOM & Exxon Mobil Corp & 10 & Energy & 2911 & Petroleum and Natural Gas (30)\\
			\hline\hline
		\end{tabularx}
Individual stocks used in the empirical analysis. Sector classifications are based on the GICS sector code from WRDS Compustat. The Fama-French 48 industry group for each firm is assigned based on its four-digit SIC code from Compustat and the SIC-to-industry definitions from Kenneth R. French's data library. 
	\end{threeparttable}
\end{table}

\begin{table}[h!]
	\centering
	\onehalfspacing
	\addtolength{\leftskip}{-2cm}
	\addtolength{\rightskip}{-2cm}
	\setlength{\tabcolsep}{5pt}
	\begin{threeparttable}
		\caption{Proportions of zero 5-minute returns (\%) for selected S\&P 100 stocks}
		\label{Table:empirical_zeroReturns}
		\scriptsize
		\begin{tabularx}{\textwidth}{l*{5}{Y}|l*{5}{Y}}
			\hline\hline
			Symbol & 2016 & 2017 & 2018 & 2019 & 2020 & Symbol & 2016 & 2017 & 2018 & 2019 & 2020\\
			\hline
			AAPL & 1.75 & 1.59 & 0.88 & 1.09 & 0.39 & LLY & 3.82 & 4.71 & 3.81 & 2.98 & 1.45\\
			ABBV & 3.35 & 3.45 & 1.99 & 2.86 & 1.60 & LMT & 1.24 & 1.39 & 0.99 & 0.92 & 0.47\\
			ACN & 3.05 & 3.19 & 2.00 & 2.04 & 1.22 & LOW & 3.69 & 3.30 & 2.49 & 2.70 & 1.13\\
			AMGN & 1.59 & 1.72 & 1.17 & 1.50 & 0.88 & MA & 3.49 & 3.58 & 1.37 & 1.44 & 0.64\\
			AMZN & 0.37 & 0.21 & 0.16 & 0.17 & 0.13 & MCD & 3.09 & 2.98 & 1.87 & 1.73 & 1.02\\
			AXP & 4.65 & 4.71 & 3.16 & 3.40 & 1.75 & MDT & 3.44 & 4.01 & 3.12 & 2.79 & 1.81\\
			BA & 2.05 & 1.67 & 0.58 & 0.61 & 0.46 & MMM & 2.14 & 2.21 & 1.40 & 1.69 & 0.99\\
			BIIB & 0.69 & 0.86 & 1.01 & 1.27 & 0.88 & MO & 3.90 & 4.13 & 3.11 & 3.06 & 3.03\\
			BLK & 0.63 & 0.87 & 0.80 & 0.85 & 0.58 & MRK & 3.93 & 4.37 & 3.23 & 3.23 & 2.01\\
			BMY & 3.40 & 4.13 & 3.71 & 4.17 & 2.55 & MSFT & 3.03 & 3.39 & 1.63 & 1.52 & 0.69\\
			C & 3.97 & 3.51 & 2.83 & 3.32 & 1.91 & NKE & 3.52 & 4.39 & 3.30 & 2.99 & 1.78\\
			CAT & 3.35 & 2.93 & 1.42 & 2.27 & 1.24 & PEP & 3.44 & 3.55 & 2.67 & 2.80 & 1.57\\
			COF & 4.24 & 3.58 & 2.97 & 3.44 & 2.01 & PG & 2.95 & 3.72 & 2.53 & 2.70 & 1.93\\
			COP & 3.59 & 4.16 & 3.40 & 3.40 & 2.43 & PM & 3.19 & 3.50 & 2.73 & 3.13 & 2.44\\
			COST & 2.06 & 1.88 & 1.29 & 1.12 & 0.56 & QCOM & 3.81 & 3.52 & 3.08 & 2.85 & 1.37\\
			CVS & 3.07 & 3.18 & 2.92 & 2.74 & 1.80 & SBUX & 3.39 & 4.09 & 3.48 & 3.04 & 1.66\\
			CVX & 2.37 & 2.80 & 1.88 & 2.46 & 1.34 & SLB & 3.06 & 3.42 & 2.79 & 4.00 & 3.43\\
			DIS & 2.51 & 2.62 & 2.09 & 1.71 & 0.83 & SPG & 1.65 & 2.03 & 1.92 & 1.88 & 1.19\\
			EMR & 4.61 & 4.63 & 3.34 & 4.06 & 2.54 & TGT & 3.67 & 3.57 & 2.82 & 2.71 & 1.32\\
			FDX & 1.95 & 1.75 & 0.99 & 1.41 & 0.78 & UNH & 2.27 & 1.91 & 1.13 & 1.17 & 0.58\\
			GD & 2.01 & 1.77 & 1.48 & 1.87 & 1.31 & UNP & 3.02 & 2.97 & 1.88 & 2.00 & 1.21\\
			GILD & 2.52 & 3.37 & 2.77 & 3.51 & 2.06 & UPS & 3.71 & 3.40 & 2.07 & 2.66 & 1.13\\
			GOOG & 0.42 & 0.51 & 0.38 & 0.45 & 0.26 & V & 3.22 & 3.29 & 1.85 & 1.86 & 0.82\\
			GS & 1.90 & 1.40 & 1.11 & 1.67 & 0.83 & VZ & 3.96 & 3.73 & 3.26 & 3.82 & 2.58\\
			HD & 2.18 & 2.47 & 1.43 & 1.34 & 0.74 & WBA & 3.84 & 4.32 & 3.32 & 3.65 & 2.72\\
			HON & 2.85 & 3.44 & 2.01 & 2.22 & 1.19 & WFC & 3.96 & 3.82 & 3.41 & 4.22 & 2.61\\
			IBM & 2.30 & 2.35 & 1.70 & 2.27 & 1.44 & WMT & 4.26 & 3.89 & 2.56 & 2.88 & 1.23\\
			JNJ & 3.25 & 2.82 & 1.88 & 2.38 & 1.11 & XOM & 2.96 & 3.40 & 2.59 & 3.03 & 1.77\\
			JPM & 3.23 & 2.92 & 1.83 & 2.26 & 1.19 &  & -- & -- & -- & -- & --\\
			\hline\hline
		\end{tabularx}
	\end{threeparttable}
\end{table}

\newpage

\subsubsection*{Factor selection} 

In addition to the factor selection results for representative stocks in \cref{Fig:empirical_selectionExample} and the pooled rankings in \cref{Table:empirical_top20Factors,Table:empirical_top20Factors_stocks_industries}, we report the full sample results on the average number of selected factors for all selected S\&P 100 stocks and Fama-French industry portfolios in \cref{Table:empirical_factors_stocks,Table:empirical_factors_industries}, respectively.\\[1.2cm]

\begin{table}[h!]
	\centering
	\onehalfspacing
	\addtolength{\leftskip}{-2cm}
	\addtolength{\rightskip}{-2cm}
	\setlength{\tabcolsep}{5pt}
	\begin{threeparttable}
		\caption{Average number of selected factors per year for selected S\&P 100 stocks}
		\label{Table:empirical_factors_stocks}
		\scriptsize
		\begin{tabularx}{\textwidth}{l*{5}{Y}|l*{5}{Y}}
			\hline\hline
			Symbol & 2016 & 2017 & 2018 & 2019 & 2020 & Symbol & 2016 & 2017 & 2018 & 2019 & 2020\\
			\hline
			AAPL & 20.67 & 18.33 & 16.33 & 15.50 & 12.92 & LLY  & 6.50  & 6.08  & 6.50  & 12.00 & 9.00 \\
			ABBV & 22.08 & 17.25 & 27.17 & 6.83  & 14.83 & LMT  & 7.67  & 4.75  & 6.83  & 8.08  & 7.42 \\
			ACN  & 3.33  & 1.08  & 2.50  & 2.00  & 1.67  & LOW  & 6.08  & 3.25  & 4.50  & 4.25  & 6.83 \\
			AMGN & 12.33 & 9.75  & 5.75  & 8.83  & 6.83  & MA   & 7.42  & 13.58 & 15.75 & 17.83 & 18.25\\
			AMZN & 24.75 & 22.33 & 15.17 & 16.83 & 13.58 & MCD  & 9.42  & 6.83  & 2.50  & 10.42 & 6.08 \\
			AXP  & 2.67  & 3.08  & 6.58  & 9.50  & 2.83  & MDT  & 1.75  & 1.08  & 2.17  & 2.92  & 1.67 \\
			BA   & 12.67 & 12.92 & 16.42 & 18.67 & 22.83 & MMM  & 2.58  & 2.67  & 7.33  & 3.58  & 3.50 \\
			BIIB & 3.92  & 4.67  & 4.25  & 1.92  & 2.92  & MO   & 15.92 & 12.83 & 21.17 & 20.92 & 9.08 \\
			BLK  & 2.33  & 2.00  & 3.33  & 3.75  & 2.50  & MRK  & 15.00 & 9.58  & 13.67 & 18.83 & 15.75\\
			BMY  & 18.50 & 7.00  & 6.83  & 8.00  & 18.67 & MSFT & 21.25 & 17.67 & 18.67 & 12.42 & 11.42\\
			C    & 19.25 & 20.50 & 17.25 & 10.33 & 8.67  & NKE  & 19.00 & 5.00  & 9.08  & 4.25  & 7.75 \\
			CAT  & 3.58  & 3.92  & 7.58  & 7.83  & 4.67  & PEP  & 10.58 & 11.25 & 15.67 & 14.75 & 8.67 \\
			COF  & 4.67  & 4.92  & 5.25  & 5.25  & 4.08  & PG   & 19.33 & 19.17 & 14.00 & 16.42 & 16.33\\
			COP  & 7.67  & 10.42 & 12.75 & 10.75 & 8.67  & PM   & 18.50 & 23.08 & 14.00 & 22.67 & 8.25 \\
			COST & 2.67  & 0.92  & 3.17  & 6.08  & 7.92  & QCOM & 14.83 & 7.17  & 7.25  & 13.25 & 5.25 \\
			CVS  & 10.92 & 8.58  & 4.25  & 8.58  & 3.92  & SBUX & 7.17  & 2.42  & 1.08  & 3.58  & 4.17 \\
			CVX  & 25.58 & 23.67 & 23.33 & 17.42 & 28.67 & SLB  & 4.58  & 5.17  & 9.67  & 4.50  & 4.50 \\
			DIS  & 21.75 & 12.67 & 16.50 & 33.00 & 35.17 & SPG  & 0.92  & 0.17  & 1.58  & 0.58  & 3.17 \\
			EMR  & 2.08  & 4.17  & 4.75  & 3.75  & 3.08  & TGT  & 1.83  & 3.42  & 5.83  & 3.33  & 1.83 \\
			FDX  & 1.17  & 0.42  & 2.08  & 2.92  & 3.00  & UNH  & 25.50 & 26.83 & 13.17 & 20.67 & 20.67\\
			GD   & 1.83  & 3.42  & 3.00  & 2.42  & 2.75  & UNP  & 6.25  & 3.75  & 7.67  & 14.25 & 5.10 \\
			GILD & 25.00 & 7.75  & 10.83 & 5.42  & 6.75  & UPS  & 1.83  & 0.92  & 2.50  & 4.67  & 1.42 \\
			GOOG & 7.00  & 12.83 & 15.08 & 15.50 & 12.25 & V    & 20.08 & 21.58 & 18.75 & 19.83 & 17.50\\
			GS   & 5.92  & 7.92  & 11.00 & 7.92  & 3.92  & VZ   & 29.75 & 20.67 & 21.83 & 20.50 & 15.17\\
			HD   & 22.83 & 23.92 & 32.08 & 24.58 & 20.42 & WBA  & 13.75 & 6.58  & 3.58  & 2.33  & 2.67 \\
			HON  & 3.75  & 3.50  & 7.92  & 7.17  & 2.67  & WFC  & 25.58 & 26.33 & 26.75 & 23.42 & 14.67\\
			IBM  & 18.33 & 7.83  & 5.67  & 3.25  & 4.67  & WMT  & 33.92 & 27.92 & 22.58 & 28.08 & 17.00\\
			JNJ  & 15.33 & 17.00 & 22.08 & 26.42 & 23.33 & XOM  & 22.67 & 25.50 & 22.00 & 20.42 & 27.33\\
			JPM  & 18.00 & 16.17 & 15.92 & 20.42 & 31.75 &      & --    & --    & --    & --    & --   \\
			\hline\hline
		\end{tabularx}
		Average number of selected factors per year for all selected S\&P 100 stocks. The candidate covariate set consists of the 224 high-frequency factors of \citet{aleti2023high}. The cutoff level in the TLP function is $\tau_{n}=0.01\sqrt{\text{MedRV}_{T}}$, where $\mathrm{MedRV}_{T}$ denotes the MedRV of the test asset computed over the corresponding estimation window $[0,T]$. The number of B-spline basis functions $K_{n}$ and the effective penalty level $\lambda_{n}/\tau_{n}$ are selected via 5-fold cross-validation using the one-standard-error rule. 
	\end{threeparttable}
\end{table}

\begin{table}[h!]
	\centering
	\onehalfspacing
	\addtolength{\leftskip}{-2cm}
	\addtolength{\rightskip}{-2cm}
	\setlength{\tabcolsep}{5pt}
	\begin{threeparttable}
		\caption{Average number of selected factors per year for Fama-French 48 industry portfolios}
		\label{Table:empirical_factors_industries}
		\scriptsize
		\begin{tabularx}{\textwidth}{c l*{5}{Y}}
			\hline\hline
			Ind. & Industry & 2016 & 2017 & 2018 & 2019 & 2020\\
			\hline
			1  & Agriculture & 1.00 & 1.25 & 2.38 & 1.00 & 2.00\\
			2  & Food Products & 8.08 & 9.92 & 6.67 & 3.25 & 6.64\\
			3  & Candy \& Soda & 10.42 & 14.92 & 15.42 & 15.50 & 11.50\\
			4  & Beer \& Liquor & 3.00 & 9.67 & 3.00 & 1.00 & 1.30\\
			5  & Tobacco Products & 20.58 & 26.00 & 24.42 & 25.75 & 11.36\\
			6  & Recreation & 2.82 & 1.88 & 2.27 & 1.27 & 4.36\\
			7  & Entertainment & 9.08 & 10.83 & 20.00 & 33.83 & 31.83\\
			8  & Printing and Publishing & 1.00 & 1.67 & 1.50 & 2.44 & 3.82\\
			9  & Consumer Goods & 12.42 & 12.33 & 9.33 & 12.00 & 16.18\\
			10 & Apparel & 14.18 & 10.36 & 9.17 & 13.58 & 10.00\\
			11 & Healthcare & 13.83 & 5.67 & 5.50 & 7.75 & 6.08\\
			12 & Medical Equipment & 6.92 & 4.42 & 11.58 & 12.17 & 11.17\\
			13 & Pharmaceutical Products & 8.75 & 6.58 & 7.17 & 8.50 & 7.92\\
			14 & Chemicals & 13.83 & 8.58 & 14.83 & 5.58 & 5.50\\
			15 & Rubber and Plastic Products & 3.83 & 1.89 & 2.92 & 2.33 & 3.27\\
			16 & Textiles & 2.73 & 3.13 & 2.00 & 2.00 & 4.33\\
			17 & Construction Materials & 2.17 & 2.73 & 11.67 & 6.25 & 5.58\\
			18 & Construction & 5.67 & 2.91 & 7.50 & 5.75 & 6.25\\
			19 & Steel Works Etc. & 6.92 & 7.50 & 6.50 & 3.92 & 4.58\\
			20 & Fabricated Products & 3.25 & 4.08 & 3.82 & 2.33 & 3.36\\
			21 & Machinery & 11.17 & 9.08 & 12.50 & 5.75 & 9.17\\
			22 & Electrical Equipment & 2.67 & 4.36 & 3.50 & 2.75 & 4.33\\
			23 & Automobiles and Trucks & 10.42 & 9.83 & 9.50 & 8.58 & 8.92\\
			24 & Aircraft & 4.67 & 5.92 & 11.83 & 10.00 & 9.08\\
			25 & Shipbuilding, Railroad Equipment & 2.00 & 4.00 & 1.00 & 1.57 & 2.92\\
			26 & Defense & 7.25 & 4.75 & 10.33 & 9.42 & 6.33\\
			27 & Precious Metals & 4.80 & 2.00 & 2.33 & 3.00 & 1.00\\
			28 & Non-Metallic and Industrial Metal Mining & 7.75 & 10.50 & 6.75 & 2.55 & 2.08\\
			29 & Coal & 2.33 & 3.67 & 2.44 & 2.88 & 3.14\\
			30 & Petroleum and Natural Gas & 10.33 & 13.92 & 15.33 & 13.83 & 31.00\\
			31 & Utilities & 18.50 & 15.50 & 13.00 & 11.58 & 15.82\\
			32 & Communication & 15.75 & 14.75 & 11.75 & 14.83 & 10.92\\
			33 & Personal Services & 3.00 & 3.25 & 3.33 & 1.58 & 2.27\\
			34 & Business Services & 2.83 & 6.08 & 5.58 & 3.92 & 5.67\\
			35 & Computers & 12.75 & 17.00 & 12.00 & 10.25 & 9.83\\
			36 & Electronic Equipment & 9.92 & 11.58 & 13.58 & 11.42 & 14.08\\
			37 & Measuring and Control Equipment & 3.08 & 3.08 & 6.42 & 8.67 & 3.00\\
			38 & Business Supplies & 2.67 & 2.70 & 2.33 & 3.42 & 5.25\\
			39 & Shipping Containers & 1.83 & 1.88 & 2.58 & 1.09 & 2.45\\
			40 & Transportation & 18.25 & 12.33 & 15.67 & 8.33 & 13.67\\
			41 & Wholesale & 1.25 & 1.00 & 1.58 & 2.33 & 3.64\\
			42 & Retail & 6.42 & 11.67 & 8.33 & 7.42 & 7.83\\
			43 & Restaurants, Hotels, Motels & 8.75 & 2.25 & 5.67 & 5.83 & 21.45\\
			44 & Banking & 8.50 & 7.33 & 6.00 & 7.67 & 12.58\\
			45 & Insurance & 6.92 & 10.00 & 5.92 & 12.50 & 12.25\\
			46 & Real Estate & 3.33 & 2.40 & 1.67 & 1.17 & 2.92\\
			47 & Trading & 5.17 & 5.50 & 5.25 & 4.42 & 8.92\\
			48 & Other & 4.10 & 1.29 & 1.42 & 2.75 & 2.08\\
			\hline\hline
		\end{tabularx}
		Average number of selected factors per year for Fama-French industry portfolios. The candidate covariate set consists of the 224 high-frequency factors of \citet{aleti2023high}. The cutoff level in the TLP function is $\tau_{n}=0.01\sqrt{\text{MedRV}_{T}}$, where $\mathrm{MedRV}_{T}$ denotes the MedRV of the test asset computed over the corresponding estimation window $[0,T]$. The number of B-spline basis functions $K_{n}$ and the effective penalty level $\lambda_{n}/\tau_{n}$ are selected via 5-fold cross-validation using the one-standard-error rule. 
	\end{threeparttable}
\end{table}

\end{appendices}

\end{document}